\newif\iflncs
\newif\iftimes 
\lncstrue
\timesfalse

\newif\ifdraft
\draftfalse

\iflncs
\def\oxford{\inst{1}}
\def\edinburgh{\inst{2}}
\else
\def\oxford{$^{\dagger}$}
\def\edinburgh{$^{\ddagger}$}
\fi

\def\TITLE{Notions of Bidirectional Computation and~Entangled~State Monads}

\def\SHORTAUTHOR{F.~Abou-Saleh\oxford \and 
  J.~Cheney\edinburgh \and 
  J.~Gibbons\oxford \and 
  J.~McKinna\edinburgh \and 
  P.~Stevens\edinburgh
}
\def\AUTHOR{Faris~Abou-Saleh\oxford \and 
  James~Cheney\edinburgh \and 
  Jeremy~Gibbons\oxford \and 
  James~McKinna\edinburgh \and 
  Perdita~Stevens\edinburgh
}

\def\INST{Department of Computer Science, University of Oxford \\
\email{firstname.lastname@cs.ox.ac.uk}\and
School of Informatics, University of Edinburgh
\\ \email{firstname.lastname@ed.ac.uk}}

\documentclass[12pt,envcountsame,runningheads]{llncs}
\usepackage{fullpage}
\iftimes \usepackage{times,mathptmx} \fi
\usepackage{url}

\usepackage{xspace}
\usepackage{amsmath}
\usepackage{amssymb}

\usepackage{cite}
\usepackage{graphicx}
\usepackage[all]{xy}

\usepackage{definitions}

\makeatletter
\@ifundefined{lhs2tex.lhs2tex.sty.read}%
  {\@namedef{lhs2tex.lhs2tex.sty.read}{}%
   \newcommand\SkipToFmtEnd{}%
   \newcommand\EndFmtInput{}%
   \long\def\SkipToFmtEnd#1\EndFmtInput{}%
  }\SkipToFmtEnd

\newcommand\ReadOnlyOnce[1]{\@ifundefined{#1}{\@namedef{#1}{}}\SkipToFmtEnd}
\usepackage{amstext}
\usepackage{amssymb}
\usepackage{stmaryrd}
\DeclareFontFamily{OT1}{cmtex}{}
\DeclareFontShape{OT1}{cmtex}{m}{n}
  {<5><6><7><8>cmtex8
   <9>cmtex9
   <10><10.95><12><14.4><17.28><20.74><24.88>cmtex10}{}
\DeclareFontShape{OT1}{cmtex}{m}{it}
  {<-> ssub * cmtt/m/it}{}

\DeclareFontShape{OT1}{cmtt}{bx}{n}
  {<5><6><7><8>cmtt8
   <9>cmbtt9
   <10><10.95><12><14.4><17.28><20.74><24.88>cmbtt10}{}
\DeclareFontShape{OT1}{cmtex}{bx}{n}
  {<-> ssub * cmtt/bx/n}{}

\newcommand{\Conid}[1]{\mathit{#1}}
\newcommand{\Varid}[1]{\mathit{#1}}
\newcommand{\anonymous}{\kern0.06em \vbox{\hrule\@width.5em}}
\newcommand{\plus}{\mathbin{+\!\!\!+}}
\newcommand{\bind}{\mathbin{>\!\!\!>\mkern-6.7mu=}}
\newcommand{\sequ}{\mathbin{>\!\!\!>}}

\usepackage{polytable}

\@ifundefined{mathindent}%
  {\newdimen\mathindent\mathindent\leftmargini}%
  {}%

\def\resethooks{%
  \global\let\SaveRestoreHook\empty
  \global\let\ColumnHook\empty}
\newcommand*{\savecolumns}[1][default]%
  {\g@addto@macro\SaveRestoreHook{\savecolumns[#1]}}
\newcommand*{\restorecolumns}[1][default]%
  {\g@addto@macro\SaveRestoreHook{\restorecolumns[#1]}}
\newcommand*{\aligncolumn}[2]%
  {\g@addto@macro\ColumnHook{\column{#1}{#2}}}

\resethooks

\newcommand{\onelinecommentchars}{\quad-{}- }
\newcommand{\commentbeginchars}{\enskip\{-}
\newcommand{\commentendchars}{-\}\enskip}

\newcommand{\visiblecomments}{%
  \let\onelinecomment=\onelinecommentchars
  \let\commentbegin=\commentbeginchars
  \let\commentend=\commentendchars}

\newcommand{\invisiblecomments}{%
  \let\onelinecomment=\empty
  \let\commentbegin=\empty
  \let\commentend=\empty}

\visiblecomments

\newlength{\blanklineskip}
\newcommand{\hsindent}[1]{\quad}
\let\hspre\empty
\let\hspost\empty

\EndFmtInput
\makeatother
\ReadOnlyOnce{polycode.fmt}%
\makeatletter

\newcommand{\hsnewpar}[1]%
  {{\parskip=0pt\parindent=0pt\par\vskip #1\noindent}}

\newcommand{\hscodestyle}{}

\newcommand{\sethscode}[1]%
  {\expandafter\let\expandafter\hscode\csname #1\endcsname
   \expandafter\let\expandafter\endhscode\csname end#1\endcsname}

  {\par\noindent
   \advance\leftskip\mathindent
   \hscodestyle
   \let\\=\@normalcr
   \let\hspre\(\let\hspost\)%
   \pboxed}%
  {\endpboxed\)%
   \par\noindent
   \ignorespacesafterend}

  {\hsnewpar\abovedisplayskip
   \advance\leftskip\mathindent
   \hscodestyle
   \let\hspre\(\let\hspost\)%
   \pboxed}%
  {\endpboxed%
   \hsnewpar\belowdisplayskip
   \ignorespacesafterend}

  {\hsnewpar\abovedisplayskip
   \advance\leftskip\mathindent
   \hscodestyle
   \let\\=\@normalcr
   \(\pboxed}%
  {\endpboxed\)%
   \hsnewpar\belowdisplayskip
   \ignorespacesafterend}

\newcommand{\plainhs}{\sethscode{plainhscode}}

\plainhs

  {\hsnewpar\abovedisplayskip
   \advance\leftskip\mathindent
   \hscodestyle
   \let\\=\@normalcr
   \(\parray}%
  {\endparray\)%
   \hsnewpar\belowdisplayskip
   \ignorespacesafterend}

  {\parray}{\endparray}

  {\(\parray}{\endparray\)}

\def\codeframewidth{\arrayrulewidth}
\RequirePackage{calc}

  {\parskip=\abovedisplayskip\par\noindent
   \hscodestyle
   \arrayrulewidth=\codeframewidth
   \tabular{@{}|p{\linewidth-2\arraycolsep-2\arrayrulewidth-2pt}|@{}}%
   \hline\framedhslinecorrect\\{-1.5ex}%
   \let\endoflinesave=\\
   \let\\=\@normalcr
   \(\pboxed}%
  {\endpboxed\)%
   \framedhslinecorrect\endoflinesave{.5ex}\hline
   \endtabular
   \parskip=\belowdisplayskip\par\noindent
   \ignorespacesafterend}

\newcommand{\framedhslinecorrect}[2]%
  {#1[#2]}

  {\(\def\column##1##2{}%
   \let\>\undefined\let\<\undefined\let\\\undefined
   \newcommand\>[1][]{}\newcommand\<[1][]{}\newcommand\\[1][]{}%
   \def\fromto##1##2##3{##3}%
   }{\) }%

  {\let\orighscode=\hscode
   \let\origendhscode=\endhscode
   \def\endhscode{\def\hscode{\endgroup\def\@currenvir{hscode}\\}\begingroup}
   \orighscode\def\hscode{\endgroup\def\@currenvir{hscode}}}%
  {\origendhscode
   \global\let\hscode=\orighscode
   \global\let\endhscode=\origendhscode}%

\makeatother
\EndFmtInput
\def\get#1{\Varid{get}_{\mskip-2mu#1}}
\def\set#1{\Varid{set}_{\mskip-2mu#1}}

\def\commentbegin{\quad$[\![\enskip$}
\def\commentend{$\enskip]\!]$}

\ReadOnlyOnce{forall.fmt}%
\makeatletter

\let\HaskellResetHook\empty
\newcommand*{\AtHaskellReset}[1]{%
  \g@addto@macro\HaskellResetHook{#1}}
\newcommand*{\HaskellReset}{\HaskellResetHook}

\newcommand\hsforall{\global\let\hsdot=\hsperiodonce}
\newcommand*\hsperiodonce[2]{#2\global\let\hsdot=\hscompose}
\newcommand*\hscompose[2]{#1}

\AtHaskellReset{\global\let\hsdot=\hscompose}

\HaskellReset

\makeatother
\EndFmtInput
\ReadOnlyOnce{lambda.fmt}%
\makeatletter

\newcommand\hslambda{\global\let\hsarrow=\hsarrowperiodonce}
\newcommand*\hsarrowperiodonce[2]{#2\global\let\hsarrow=\hscompose}

\AtHaskellReset{\global\let\hsarrow=\hscompose}

\HaskellReset

\makeatother
\EndFmtInput
\title{\TITLE}

\author{\AUTHOR}
\authorrunning{\SHORTAUTHOR}
\institute{\INST}

\begin{document}
\maketitle

\begin{abstract}
Bidirectional transformations (bx) support principled consistency maintenance between data sources. Each data source corresponds to one perspective on a composite system, manifested by operations to `get' and `set' a view of the whole from that particular perspective. Bx are important in a wide range of settings, including databases, interactive applications, and model-driven development.
We show that bx are naturally modelled in terms of mutable state; in particular, the `set' operations are stateful functions. This leads naturally to considering bx that exploit other computational effects too, such as I/O, nondeterminism, and failure, all largely ignored in the bx literature to date.
We present a semantic foundation for symmetric bidirectional transformations with effects. We build on the mature theory of monadic encapsulation of effects in functional programming, develop the equational theory and important combinators for effectful bx, and provide a prototype implementation in Haskell along with several illustrative examples.

\end{abstract}


\section{Introduction} \label{sec:introduction}

Bidirectional transformations (\bx) arise when synchronising data in 
different data sources: updates to one source entail
corresponding updates to the others, in order to maintain
consistency. When a data source represents the complete information, this is a
straightforward task; an update can be matched by
discarding and regenerating the other sources. It becomes more interesting
when one data representation lacks some information that is recorded
by another; then the corresponding update has to merge new
information on one side with old information on the other side.
Such bidirectional
transformations have been the focus of a flurry of recent
activity---in databases, in programming languages, and in software
engineering, among other fields---giving rise to a flourishing series
of BX Workshops (see \url{http://bx-community.wikidot.com/}) and BX
Seminars (in Japan, Germany, and Canada so far: see \cite{GraceReport}
for an early report on the state of the art).

The different branches of the \bx{} community have come up with a variety
of different formalisations of \bx{} with
conflicting definitions and incompatible extensions, such as
\emph{lenses}~\cite{lens-toplas}, \emph{relational
  bx}~\cite{stevens09:sosym}, \emph{symmetric lenses}~\cite{symlens},
 \emph{putback-based
lenses}~\cite{pacheco14pepm}, and
\emph{profunctors}~\cite{lenseslibrary}.
We have been seeking a unification of the varying
approaches. It turns out that quite a satisfying unifying formalism can be
obtained from the perspective of the state monad. More specifically,
we are thinking about two data sources, and stateful computations on acting on 
pairs representing these two sources. However, the two components of the
pair are not independent, as two distinct memory cells would be, but are
\emph{entangled}---a change to one component generally entails a
consequent change to the other.

\begin{figure}[tb] \centering\scriptsize
\def\upstrut{\vrule width0pt height3ex}
\begin{tabular}[t]{@{}c@{}}
  \ensuremath{\Conid{A}} \\
  \fbox{\begin{tabular}{@{}l@{~}l@{}}
  $\{$ & (Schumann, \\
&Germany, \\
&1810--1856), \\
       & (Schubert, \\
&Austria, \\
&1797--1828), \ldots $\}$
  \end{tabular}}
\end{tabular}
\quad
\begin{tabular}[t]{@{}c@{}}
\quad \\
\begin{picture}(140,41)
\put(10,10){\vector(1,0){30}} 
\put(40,30){\vector(-1,0){30}} 
\put(10,20){\makebox(30,10){\ensuremath{\get{L}}}}
\put(10,0){\makebox(30,10){\ensuremath{\set{L}}}}
\put(90,30){\makebox(10,10){\ensuremath{\Conid{M}}}}
\put(100,30){\vector(1,0){30}}
\put(100,20){\makebox(30,10){\ensuremath{\get{R}}}}
\put(130,10){\vector(-1,0){30}}
\put(100,0){\makebox(30,10){\ensuremath{\set{R}}}}
\put(40,0){\framebox(60,40){\ensuremath{\Conid{S}}}}
\end{picture}
\end{tabular}
\quad
\begin{tabular}[t]{@{}c@{}}
  \ensuremath{\Conid{B}} \\
  \fbox{\begin{tabular}{@{}c|c@{}}
  \textit{Name} & \textit{Nationality} \\ \hline
  \upstrut 
  Schubert & Austria \\
  Schumann & Germany \\
  $\cdots$ & $\cdots$
  \end{tabular}}
\end{tabular}
\caption{Effectful bx between sources \ensuremath{\Conid{A}} and \ensuremath{\Conid{B}} and with effects in
  monad \ensuremath{\Conid{M}}, with hidden state \ensuremath{\Conid{S}}, illustrated using the
  Composers example}
\label{fig:composers}
\end{figure}
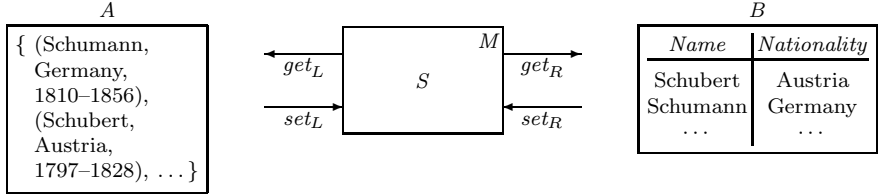

This stateful perspective 
suggests using monads for \bx{}, much as Moggi showed that monads
unify many computational effects~\cite{moggi}. But not only that; it
suggests a way to \emph{generalise} \bx{} to encompass 
other features that monads can handle. 
In fact, several approaches to lenses do in practice allow for monadic
operations~\cite{pacheco14pepm,lenseslibrary}.  But there are natural
concerns about such an extension: Are ``monadic lenses'' legitimate
bidirectional transformations?  Do they satisfy laws analogous to the
roundtripping (`GetPut' and `PutGet')
laws of traditional lenses?  Can we compose such
transformations?  We show that bidirectional computations can be
encapsulated in monads, and be combined with other standard monads to
accommodate effects, while still satisfying appropriate equational
laws and supporting composition.

To illustrate our approach informally, consider
Figure~\ref{fig:composers}, which is based on the
\emph{Composers} example \cite{composers}.  
This example relates two data sources \ensuremath{\Conid{A}} and \ensuremath{\Conid{B}}: on the
left, an \ensuremath{\Varid{a}\mathbin{::}\Conid{A}} consists of a set of triples \ensuremath{(\Varid{name},\Varid{nationality},\Varid{dates})},
and on the right, a \ensuremath{\Varid{b}\mathbin{::}\Conid{B}} consists of an ordered list of pairs
\ensuremath{(\Varid{name},\Varid{nationality})}.  The two sources \ensuremath{\Varid{a}} and \ensuremath{\Varid{b}} are consistent when
they contain the name--nationality pairs, ignoring dates and ordering.  The
centre of the figure illustrates the interface and typical operations
of our (effectful) \bx{}, including a monad \ensuremath{\Conid{M}}, operations \ensuremath{\get{L}\mathbin{::}\Conid{M}\;\Conid{A}} and \ensuremath{\get{R}\mathbin{::}\Conid{M}\;\Conid{B}} that return the current values of the left and
right sides, and operations \ensuremath{\set{L}\mathbin{::}\Conid{A}\hsarrow{\rightarrow }{\mathpunct{.}}\Conid{M}\;()} and \ensuremath{\set{R}\mathbin{::}\Conid{B}\hsarrow{\rightarrow }{\mathpunct{.}}\Conid{M}\;()} that accept a new value for the left-hand or right-hand side,
possibly performing side-effects in \ensuremath{\Conid{M}}.

In this example, neither \ensuremath{\Conid{A}} nor \ensuremath{\Conid{B}} is obtainable from the other; \ensuremath{\Conid{A}}
omits the ordering information from \ensuremath{\Conid{B}}, while \ensuremath{\Conid{B}} omits the date
information from \ensuremath{\Conid{A}}.  This means that there may be multiple ways to change one
side to match a change to the other.
For example, the monadic computation 
\ensuremath{\mathbf{do}\;\{\mskip1.5mu \Varid{b}\leftarrow \get{R};\set{R}\;(\Varid{b}\plus [\mskip1.5mu (\text{\tt \char34 Bach\char34},\text{\tt \char34 Germany\char34})\mskip1.5mu])\mskip1.5mu\}} 
looks at the current value of the
\ensuremath{\Conid{B}}-side, and modifies it to include a new pair at the end.  Of
course, this addition is ambiguous: do we mean J.\,S.\,Bach (1685--1750),
J.\,C.\,Bach (1735--1782), or another Bach? There is no way to give the
dates through the \ensuremath{\Conid{B}} interface; typically, pure \bx{} would
initialise the unspecified part to some default value such as
\ensuremath{\text{\tt \char34 ????-????\char34}}.  Conversely, had we inserted the triple \ensuremath{(\text{\tt \char34 Bach\char34},\text{\tt \char34 Germany\char34},\text{\tt \char34 1685-1750\char34})} on the left, then there may be several
consistent ways to change the right-hand side to match: for example,
inserting at the beginning, the end, or somewhere in the middle of the
list.  Again, conventional pure \bx{} must fix some strategy in
advance.

A conventional (pure) \bx{} corresponds (roughly) to taking \ensuremath{\Conid{M}\mathrel{=}\Conid{State}\;\Conid{S}}, where \ensuremath{\Conid{S}} is some type of states from which we can obtain both \ensuremath{\Conid{A}}
and \ensuremath{\Conid{B}}; for example, \ensuremath{\Conid{S}} could consist of \emph{lists} of triples
\ensuremath{(\Varid{name},\Varid{nationality},\Varid{dates})} from which we can easily extract both \ensuremath{\Conid{A}}
(by forgetting the order) and \ensuremath{\Conid{B}} (by forgetting the dates).  However,
our approach to effectful \bx{} allows many other choices for \ensuremath{\Conid{M}} that
allow us to use side-effects when restoring consistency.  Consider the
following two scenarios, 
taken from the model-driven development domain,
where the entities being synchronised are
`model states' \ensuremath{\Varid{a},\Varid{b}} such as UML models or RDBMS schemas,
drawn from suitable `model spaces' \ensuremath{\Conid{A},\Conid{B}}.  
We will revisit them among the concrete examples 
in \Section~\ref{sec:examples}.

\begin{scenario}[nondeterminism]\label{ex:nondeterminism}
  As mentioned above, most formal notions of \bx{} require that the
  transformation (or programmer) decide on a consistency-restoration
  strategy in advance.  By contrast, in the
  Janus Transformation Language (JTL)~\cite{CicchettiREP10},
  programmers need only specify a consistency relation, allowing 
  the \bx{} engine to resolve the underspecification
  nondeterministically.  Given a change to one source, JTL uses an
  external constraint solver to find a consistent choice for the other
  source; there might be multiple choices.

  Our effectful \bx{} can handle this by combining state with a
  nondeterministic choice monad, 
  taking \ensuremath{\Conid{M}\mathrel{=}\Conid{StateT}\;\Conid{S}\;[\mskip1.5mu \mskip1.5mu]}.  For
  example, an attempt to add Bach to the right-hand side could result
  in several possibilities for the left-hand side (perhaps using an
  external source such as Wikipedia to find the dates for candidate
  matches).  Conversely, adding Bach to the left-hand side could
  result in a nondeterministic choice of all possible positions 
  to add
  the matching record in the right-hand list.
No previous formalism permits such \emph{nondeterministic
    \bx{}} to be composed with conventional deterministic
  transformations, or characterises the laws that such transformations
  ought to satisfy.
\end{scenario}

\begin{scenario}[interaction]\label{ex:dynamic} 
  Alternatively, instead of automatically trying to find a single (or
  all possible) dates for Bach, why not ask the user for help?  In
  unidirectional model transformation settings, Varr\'o~\cite{Varro06}
  proposed ``model transformation by example'', where
  the transformation system `learns' gradually, by prompting its users over
  time, the desired way to restore consistency in various situations. 
  One can see this as an (interactive) instance of \emph{memoisation}. 
  
  Our effectful \bx{} can also handle this, using the \ensuremath{\Conid{IO}} monad in
  concert with some additional state to remember 
  past questions and
  their answers.  For example, when Bach is added to the right-hand
  side, the \bx{} can check to see whether it already knows how to
  restore consistency.  If so, it does so without further ado.  If
  not, it queries the user to determine how to fill in the missing
  values needed for the right-hand side, perhaps having first appealed
  to some external knowledge base 
  to generate helpful suggestions.  It
  then records the new version of the right-hand side and updates its
  state so that the same question is not asked again.  No previous
  formalism allows for I/O during consistency restoration.
\end{scenario}

The paper is structured as follows.
\Section~\ref{sec:background} reviews monads as a foundation for effectful
programming, fixing (idealised) Haskell notation used throughout the paper,
and recaps definitions of lenses.
Our contributions start in \Section~\ref{sec:monad},
with a presentation of our monadic approach to \bx{}.
\Section~\ref{sec:composition} considers a definition of composition for
effectful \bx{}.
In \Section~\ref{sec:examples} we  discuss initialisation, and
formalise the motivating examples above, along with other combinators
and examples of effectful \bx{}.  These examples are the first formal
treatments of effects such as nondeterminism or interaction for
symmetric bidirectional transformations, and they illustrate the
generality of our approach. 
Finally we discuss related work and conclude. 
This technical report version includes an appendix with all proofs;
executable code can be found 
on our project web page 
\url{http://groups.inf.ed.ac.uk/bx/}. 
\jgnote{check online supplement}
\section{Background}\label{sec:background}

Our approach to \bx{}
is semantics-driven, so we here provide some preliminaries on semantics of
effectful computation --
focusing on monads, Haskell's use of type classes for them,
and some key instances that we exploit heavily in what follows.
We also briefly recap the definitions of asymmetric and symmetric lenses.

\subsection{Effectful computation}

Moggi's seminal work on the semantics of effectful
computation~\cite{moggi}, and much continued investigation, shows
how computational effects can be described using monads.
Building on this, we assume that
computations are represented as Kleisli arrows for a strong monad \ensuremath{\Conid{T}}
defined on a cartesian closed category \ensuremath{\ensuremath{\mathbb{C}}} of `value'
types and `pure' functions.  
The reader uncomfortable with such generality can safely consider our
definitions in terms of the category of sets and total functions,
with \ensuremath{\Conid{T}} encapsulating the `ambient' programming language
effects: none in a total functional programming language like Agda, 
partiality in Haskell, global state in Pascal, network access in Java, etc.

\subsection{Notational conventions}

We write in Haskell notation, except for the following few
idealisations.  We assume a cartesian closed category \ensuremath{\ensuremath{\mathbb{C}}},
avoiding niceties about lifted types and undefined values in Haskell; 
we further restrict attention to terminating programs.
We use lowercase (Greek) letters for polymorphic type variables
in code, and uppercase (Roman) letters for monomorphic instantiations
of those variables in accompanying prose.
We elide constructors and destructors for a
\textbf{newtype}, the explicit witnesses to the isomorphism between
the defined type and its structure, and use instead a \textbf{type}
synonym that equates the defined type and its structure;
e.g., in \Section~\ref{sec:stateT} we omit 
the function \ensuremath{\Varid{runStateT}} from \ensuremath{\Conid{StateT}\;\Conid{S}\;\Conid{T}\;\Conid{A}} to \ensuremath{\Conid{S}\hsarrow{\rightarrow }{\mathpunct{.}}\Conid{T}\;(\Conid{A},\Conid{S})}.  
Except where expressly noted, we assume a kind of Barendregt
convention, that bound variables are chosen not to clash with free
variables; for example, in the definition below of a commutative
monad, we elide the explicit proviso ``for \ensuremath{\Varid{x},\Varid{y}} distinct variables
not free in \ensuremath{\Varid{m},\Varid{n}}'' (one might take the view that \ensuremath{\Varid{m},\Varid{n}} are themselves
variables, rather than possibly open terms that might capture \ensuremath{\Varid{x},\Varid{y}}).
We use a tightest-binding lowered dot for field access in records; e.g., in
\Definition~\ref{def:proper} we write \ensuremath{bx\mathord{.}\get{L}} rather than
\ensuremath{\get{L}\;bx}; we therefore write function composition using a centred dot, \ensuremath{\Varid{f}\hsdot{\cdot }{.}\Varid{g}}.
The code online expands these conventions into real Haskell.
We also
make extensive use of equational reasoning over monads in \ensuremath{\mathbf{do}} notation
\cite{Gibbons&Hinze11:Just}.
Different branches of the \bx{} community have
conflicting naming conventions for various operations, so we have
renamed some of them, favouring internal over external consistency. 


\subsection{Monads} \label{sec:monads}

\begin{definition*}[monad type class] 
Type constructors representing notions of effectful computation are
represented as instances of the Haskell type class \ensuremath{\Conid{Monad}}:
\begin{hscode}\SaveRestoreHook
\column{B}{@{}>{\hspre}l<{\hspost}@{}}%
\column{3}{@{}>{\hspre}l<{\hspost}@{}}%
\column{11}{@{}>{\hspre}l<{\hspost}@{}}%
\column{47}{@{}>{\hspre}l<{\hspost}@{}}%
\column{E}{@{}>{\hspre}l<{\hspost}@{}}%
\>[B]{}\mathbf{class}\;\Conid{Monad}\;\tau\;\mathbf{where}{}\<[E]%
\\
\>[B]{}\hsindent{3}{}\<[3]%
\>[3]{}\Varid{return}{}\<[11]%
\>[11]{}\mathbin{::}\alpha\hsarrow{\rightarrow }{\mathpunct{.}}\tau\;\alpha{}\<[E]%
\\
\>[B]{}\hsindent{3}{}\<[3]%
\>[3]{}(\bind ){}\<[11]%
\>[11]{}\mathbin{::}\tau\;\alpha\hsarrow{\rightarrow }{\mathpunct{.}}(\alpha\hsarrow{\rightarrow }{\mathpunct{.}}\tau\;\beta)\hsarrow{\rightarrow }{\mathpunct{.}}\tau\;\beta{}\<[47]%
\>[47]{}\mbox{\onelinecomment  pronounced `bind'}{}\<[E]%
\ColumnHook
\end{hscode}\resethooks
A \ensuremath{\Conid{Monad}} instance should satisfy the following laws:
\begin{hscode}\SaveRestoreHook
\column{B}{@{}>{\hspre}l<{\hspost}@{}}%
\column{3}{@{}>{\hspre}l<{\hspost}@{}}%
\column{20}{@{}>{\hspre}l<{\hspost}@{}}%
\column{E}{@{}>{\hspre}l<{\hspost}@{}}%
\>[3]{}\Varid{return}\;\Varid{x}\bind \Varid{f}{}\<[20]%
\>[20]{}\mathrel{=}\Varid{f}\;\Varid{x}{}\<[E]%
\\
\>[3]{}\Varid{m}\bind \Varid{return}{}\<[20]%
\>[20]{}\mathrel{=}\Varid{m}{}\<[E]%
\\
\>[3]{}(\Varid{m}\bind \Varid{f})\bind \Varid{g}{}\<[20]%
\>[20]{}\mathrel{=}\Varid{m}\bind \lambda \hslambda \Varid{x}\hsarrow{\rightarrow }{\mathpunct{.}}(\Varid{f}\;\Varid{x}\bind \Varid{g}){}\<[E]%
\ColumnHook
\end{hscode}\resethooks
\endswithdisplay
\end{definition*}
Common examples in Haskell (with which we assume familiarity) include: 
\begin{hscode}\SaveRestoreHook
\column{B}{@{}>{\hspre}l<{\hspost}@{}}%
\column{23}{@{}>{\hspre}l<{\hspost}@{}}%
\column{52}{@{}>{\hspre}l<{\hspost}@{}}%
\column{E}{@{}>{\hspre}l<{\hspost}@{}}%
\>[B]{}\mathbf{type}\;\Conid{Id}\;\alpha{}\<[23]%
\>[23]{}\mathrel{=}\alpha{}\<[52]%
\>[52]{}\mbox{\onelinecomment  no effects}{}\<[E]%
\\
\>[B]{}\mathbf{data}\;\Conid{Maybe}\;\alpha{}\<[23]%
\>[23]{}\mathrel{=}\Conid{Just}\;\alpha\mid \Conid{Nothing}{}\<[52]%
\>[52]{}\mbox{\onelinecomment  failure/exceptions}{}\<[E]%
\\
\>[B]{}\mathbf{data}\;[\mskip1.5mu \alpha\mskip1.5mu]{}\<[23]%
\>[23]{}\mathrel{=}[\mskip1.5mu \mskip1.5mu]\mid \alpha\mathbin{:}[\mskip1.5mu \alpha\mskip1.5mu]{}\<[52]%
\>[52]{}\mbox{\onelinecomment  choice}{}\<[E]%
\\
\>[B]{}\mathbf{type}\;\Conid{State}\;\sigma\;\alpha{}\<[23]%
\>[23]{}\mathrel{=}\sigma\hsarrow{\rightarrow }{\mathpunct{.}}(\alpha,\sigma){}\<[52]%
\>[52]{}\mbox{\onelinecomment  state}{}\<[E]%
\\
\>[B]{}\mathbf{type}\;\Conid{Reader}\;\sigma\;\alpha{}\<[23]%
\>[23]{}\mathrel{=}\sigma\hsarrow{\rightarrow }{\mathpunct{.}}\alpha{}\<[52]%
\>[52]{}\mbox{\onelinecomment  environment}{}\<[E]%
\\
\>[B]{}\mathbf{type}\;\Conid{Writer}\;\sigma\;\alpha{}\<[23]%
\>[23]{}\mathrel{=}(\alpha,\sigma){}\<[52]%
\>[52]{}\mbox{\onelinecomment  logging}{}\<[E]%
\ColumnHook
\end{hscode}\resethooks
as well as the (in)famous \ensuremath{\Conid{IO}} monad, which encapsulates interaction
with the outside world.  We need a \ensuremath{\Conid{Monoid}\;\sigma} instance for the \ensuremath{\Conid{Writer}\;\sigma} 
monad, in order to support empty and composite logs.

\begin{definition}\label{defn:do-notation}
In Haskell, monadic expressions may be written using \ensuremath{\mathbf{do}} 
notation, which is defined by translation into applications of bind:
\begin{hscode}\SaveRestoreHook
\column{B}{@{}>{\hspre}l<{\hspost}@{}}%
\column{3}{@{}>{\hspre}l<{\hspost}@{}}%
\column{26}{@{}>{\hspre}l<{\hspost}@{}}%
\column{E}{@{}>{\hspre}l<{\hspost}@{}}%
\>[3]{}\mathbf{do}\;\{\mskip1.5mu \mathbf{let}\;\Varid{decls};\Varid{ms}\mskip1.5mu\}{}\<[26]%
\>[26]{}\mathrel{=}\mathbf{let}\;\Varid{decls}\;\mathbf{in}\;\mathbf{do}\;\{\mskip1.5mu \Varid{ms}\mskip1.5mu\}{}\<[E]%
\\
\>[3]{}\mathbf{do}\;\{\mskip1.5mu \Varid{a}\leftarrow \Varid{m};\Varid{ms}\mskip1.5mu\}{}\<[26]%
\>[26]{}\mathrel{=}\Varid{m}\bind \lambda \hslambda \Varid{a}\hsarrow{\rightarrow }{\mathpunct{.}}\mathbf{do}\;\{\mskip1.5mu \Varid{ms}\mskip1.5mu\}{}\<[E]%
\\
\>[3]{}\mathbf{do}\;\{\mskip1.5mu \Varid{m}\mskip1.5mu\}{}\<[26]%
\>[26]{}\mathrel{=}\Varid{m}{}\<[E]%
\ColumnHook
\end{hscode}\resethooks
The \emph{body} \ensuremath{\Varid{ms}} of a \ensuremath{\mathbf{do}} expression 
consists of
zero or more `qualifiers', and
a final expression \ensuremath{\Varid{m}} of monadic
type; qualifiers are either
`declarations' \ensuremath{\mathbf{let}\;\Varid{decls}} (with \ensuremath{\Varid{decls}} a collection of bindings
\ensuremath{\Varid{a}\mathrel{=}\Varid{e}} of patterns \ensuremath{\Varid{a}} to expressions \ensuremath{\Varid{e}}) or 
`generators' \ensuremath{\Varid{a}\leftarrow \Varid{m}} (with pattern \ensuremath{\Varid{a}} and monadic expression
\ensuremath{\Varid{m}}). Variables bound in pattern \ensuremath{\Varid{a}} may appear free in the subsequent
body \ensuremath{\Varid{ms}}; in contrast to Haskell, we assume that the pattern cannot fail to match.
When the return value of \ensuremath{\Varid{m}} is not used -- e.g., when void
-- we write \ensuremath{\mathbf{do}\;\{\mskip1.5mu \Varid{m};\Varid{ms}\mskip1.5mu\}} as shorthand for \ensuremath{\mathbf{do}\;\{\mskip1.5mu \anonymous \leftarrow \Varid{m};\Varid{ms}\mskip1.5mu\}}
with its wildcard pattern.
\end{definition}

\begin{definition*}[commutative monad] 
  We say that \ensuremath{\Varid{m}\mathbin{::}\Conid{T}\;\Conid{A}} \emph{commutes in \ensuremath{\Conid{T}}} if the following
  holds for all \ensuremath{\Varid{n}\mathbin{::}\Conid{T}\;\Conid{B}}:
  \begin{hscode}\SaveRestoreHook
\column{B}{@{}>{\hspre}l<{\hspost}@{}}%
\column{5}{@{}>{\hspre}l<{\hspost}@{}}%
\column{40}{@{}>{\hspre}c<{\hspost}@{}}%
\column{40E}{@{}l@{}}%
\column{43}{@{}>{\hspre}l<{\hspost}@{}}%
\column{E}{@{}>{\hspre}l<{\hspost}@{}}%
\>[5]{}\mathbf{do}\;\{\mskip1.5mu \Varid{x}\leftarrow \Varid{m};\Varid{y}\leftarrow \Varid{n};\Varid{return}\;(\Varid{x},\Varid{y})\mskip1.5mu\}{}\<[40]%
\>[40]{}\mathrel{=}{}\<[40E]%
\>[43]{}\mathbf{do}\;\{\mskip1.5mu \Varid{y}\leftarrow \Varid{n};\Varid{x}\leftarrow \Varid{m};\Varid{return}\;(\Varid{x},\Varid{y})\mskip1.5mu\}{}\<[E]%
\ColumnHook
\end{hscode}\resethooks
  A monad \ensuremath{\Conid{T}} is \emph{commutative} if all \ensuremath{\Varid{m}\mathbin{::}\Conid{T}\;\Conid{A}} commute, for all \ensuremath{\Conid{A}}.
\end{definition*}
\begin{definition*}
  An element \ensuremath{\Varid{z}} of a monad is called a \emph{zero element} if it satisfies:
  \begin{hscode}\SaveRestoreHook
\column{B}{@{}>{\hspre}l<{\hspost}@{}}%
\column{5}{@{}>{\hspre}l<{\hspost}@{}}%
\column{E}{@{}>{\hspre}l<{\hspost}@{}}%
\>[5]{}\mathbf{do}\;\{\mskip1.5mu \Varid{x}\leftarrow \Varid{z};\Varid{f}\;\Varid{x}\mskip1.5mu\}\mathrel{=}\Varid{z}\mathrel{=}\mathbf{do}\;\{\mskip1.5mu \Varid{x}\leftarrow \Varid{m};\Varid{z}\mskip1.5mu\}{}\<[E]%
\ColumnHook
\end{hscode}\resethooks
\endswithdisplay
\end{definition*}
Among monads discussed so far, \ensuremath{\Conid{Id}}, \ensuremath{\Conid{Reader}} and \ensuremath{\Conid{Maybe}}
are commutative; if \ensuremath{\sigma} is a commutative
monoid, \ensuremath{\Conid{Writer}\;\sigma} is commutative; 
but many interesting monads, such as \ensuremath{\Conid{IO}} and \ensuremath{\Conid{State}}, are not.
The \ensuremath{\Conid{Maybe}} monad has zero element \ensuremath{\Conid{Nothing}},
and \ensuremath{\Conid{List}} has zero \ensuremath{\Conid{Nil}};
the zero element is unique if it exists.

\begin{definition*}[monad morphism]
Given monads \ensuremath{\Conid{T}} and \ensuremath{\Conid{T'}}, a \emph{monad morphism} is 
a polymorphic function \ensuremath{\varphi \mathbin{::}\forall \alpha\hsforall \hsdot{\cdot }{.}\Conid{T}\;\alpha\hsarrow{\rightarrow }{\mathpunct{.}}\Conid{T'}\;\alpha}
satisfying
\begin{hscode}\SaveRestoreHook
\column{B}{@{}>{\hspre}l<{\hspost}@{}}%
\column{29}{@{}>{\hspre}l<{\hspost}@{}}%
\column{E}{@{}>{\hspre}l<{\hspost}@{}}%
\>[B]{}\varphi \;(\ensuremath{\mathbf{do}_{\Conid{T}}}\;\{\mskip1.5mu \Varid{return}\;\Varid{a}\mskip1.5mu\}){}\<[29]%
\>[29]{}\mathrel{=}\ensuremath{\mathbf{do}_{\Conid{T'}}}\;\{\mskip1.5mu \Varid{return}\;\Varid{a}\mskip1.5mu\}{}\<[E]%
\\
\>[B]{}\varphi \;(\ensuremath{\mathbf{do}_{\Conid{T}}}\;\{\mskip1.5mu \Varid{a}\leftarrow \Varid{m};\Varid{k}\;\Varid{a}\mskip1.5mu\}){}\<[29]%
\>[29]{}\mathrel{=}\ensuremath{\mathbf{do}_{\Conid{T'}}}\;\{\mskip1.5mu \Varid{a}\leftarrow \varphi \;\Varid{m};\varphi \;(\Varid{k}\;\Varid{a})\mskip1.5mu\}{}\<[E]%
\ColumnHook
\end{hscode}\resethooks
(subscripting to make clear which monad is used where).
\end{definition*}

%
%


\subsection{Combining state and other effects}
\label{sec:stateT}

We recall the \emph{state monad transformer} (see e.g. Liang
\etal~\cite{liang95popl}).

\begin{definition}[state monad transformer] \label{def:state-monad-transformer}
State can be combined with effects arising from an arbitrary monad \ensuremath{\Conid{T}}
using the \ensuremath{\Conid{StateT}} monad transformer: 
\begin{hscode}\SaveRestoreHook
\column{B}{@{}>{\hspre}l<{\hspost}@{}}%
\column{3}{@{}>{\hspre}l<{\hspost}@{}}%
\column{13}{@{}>{\hspre}l<{\hspost}@{}}%
\column{36}{@{}>{\hspre}l<{\hspost}@{}}%
\column{E}{@{}>{\hspre}l<{\hspost}@{}}%
\>[B]{}\mathbf{type}\;\Conid{StateT}\;\sigma\;\tau\;\alpha\mathrel{=}\sigma\hsarrow{\rightarrow }{\mathpunct{.}}\tau\;(\alpha,\sigma){}\<[E]%
\\[\blanklineskip]%
\>[B]{}\mathbf{instance}\;\Conid{Monad}\;\tau\Rightarrow \Conid{Monad}\;(\Conid{StateT}\;\sigma\;\tau)\;\mathbf{where}{}\<[E]%
\\
\>[B]{}\hsindent{3}{}\<[3]%
\>[3]{}\Varid{return}\;\Varid{a}{}\<[13]%
\>[13]{}\mathrel{=}\lambda \hslambda \Varid{s}\hsarrow{\rightarrow }{\mathpunct{.}}\Varid{return}\;(\Varid{a},\Varid{s}){}\<[E]%
\\
\>[B]{}\hsindent{3}{}\<[3]%
\>[3]{}\Varid{m}\bind \Varid{k}{}\<[13]%
\>[13]{}\mathrel{=}\lambda \hslambda \Varid{s}\hsarrow{\rightarrow }{\mathpunct{.}}\mathbf{do}\;\{\mskip1.5mu {}\<[36]%
\>[36]{}(\Varid{a},\Varid{s'})\leftarrow \Varid{m}\;\Varid{s};\Varid{k}\;\Varid{a}\;\Varid{s'}\mskip1.5mu\}{}\<[E]%
\ColumnHook
\end{hscode}\resethooks
This provides \ensuremath{\Varid{get}} and \ensuremath{\Varid{set}} operations for the state type: 
\begin{hscode}\SaveRestoreHook
\column{B}{@{}>{\hspre}l<{\hspost}@{}}%
\column{E}{@{}>{\hspre}l<{\hspost}@{}}%
\>[B]{}\Varid{get}\mathbin{::}\Conid{Monad}\;\tau\Rightarrow \Conid{StateT}\;\sigma\;\tau\;\sigma{}\<[E]%
\\
\>[B]{}\Varid{get}\mathrel{=}\lambda \hslambda \Varid{s}\hsarrow{\rightarrow }{\mathpunct{.}}\Varid{return}\;(\Varid{s},\Varid{s}){}\<[E]%
\\[\blanklineskip]%
\>[B]{}\Varid{set}\mathbin{::}\Conid{Monad}\;\tau\Rightarrow \sigma\hsarrow{\rightarrow }{\mathpunct{.}}\Conid{StateT}\;\sigma\;\tau\;(){}\<[E]%
\\
\>[B]{}\Varid{set}\;\Varid{s'}\mathrel{=}\lambda \hslambda \Varid{s}\hsarrow{\rightarrow }{\mathpunct{.}}\Varid{return}\;((),\Varid{s'}){}\<[E]%
\ColumnHook
\end{hscode}\resethooks
which satisfy the following four laws \cite{Plotkin&Power2002:Notions}:
\begin{hscode}\SaveRestoreHook
\column{B}{@{}>{\hspre}l<{\hspost}@{}}%
\column{10}{@{}>{\hspre}c<{\hspost}@{}}%
\column{10E}{@{}l@{}}%
\column{14}{@{}>{\hspre}l<{\hspost}@{}}%
\column{59}{@{}>{\hspre}l<{\hspost}@{}}%
\column{E}{@{}>{\hspre}l<{\hspost}@{}}%
\>[B]{}\mathrm{(GG)}{}\<[10]%
\>[10]{}\quad{}\<[10E]%
\>[14]{}\mathbf{do}\;\{\mskip1.5mu \Varid{s}\leftarrow \Varid{get};\Varid{s'}\leftarrow \Varid{get};\Varid{return}\;(\Varid{s},\Varid{s'})\mskip1.5mu\}{}\<[59]%
\>[59]{}\mathrel{=}\mathbf{do}\;\{\mskip1.5mu \Varid{s}\leftarrow \Varid{get};\Varid{return}\;(\Varid{s},\Varid{s})\mskip1.5mu\}{}\<[E]%
\\
\>[B]{}\mathrm{(SG)}{}\<[10]%
\>[10]{}\quad{}\<[10E]%
\>[14]{}\mathbf{do}\;\{\mskip1.5mu \Varid{set}\;\Varid{s};\Varid{get}\mskip1.5mu\}{}\<[59]%
\>[59]{}\mathrel{=}\mathbf{do}\;\{\mskip1.5mu \Varid{set}\;\Varid{s};\Varid{return}\;\Varid{s}\mskip1.5mu\}{}\<[E]%
\\
\>[B]{}\mathrm{(GS)}{}\<[10]%
\>[10]{}\quad{}\<[10E]%
\>[14]{}\mathbf{do}\;\{\mskip1.5mu \Varid{s}\leftarrow \Varid{get};\Varid{set}\;\Varid{s}\mskip1.5mu\}{}\<[59]%
\>[59]{}\mathrel{=}\mathbf{do}\;\{\mskip1.5mu \Varid{return}\;()\mskip1.5mu\}{}\<[E]%
\\
\>[B]{}\mathrm{(SS)}{}\<[10]%
\>[10]{}\quad{}\<[10E]%
\>[14]{}\mathbf{do}\;\{\mskip1.5mu \Varid{set}\;\Varid{s};\Varid{set}\;\Varid{s'}\mskip1.5mu\}{}\<[59]%
\>[59]{}\mathrel{=}\mathbf{do}\;\{\mskip1.5mu \Varid{set}\;\Varid{s'}\mskip1.5mu\}{}\<[E]%
\ColumnHook
\end{hscode}\resethooks
Computations in \ensuremath{\Conid{T}} embed into 
\ensuremath{\Conid{StateT}\;\Conid{S}\;\Conid{T}} via the monad morphism \ensuremath{\Varid{lift}}:
\begin{hscode}\SaveRestoreHook
\column{B}{@{}>{\hspre}l<{\hspost}@{}}%
\column{12}{@{}>{\hspre}l<{\hspost}@{}}%
\column{E}{@{}>{\hspre}l<{\hspost}@{}}%
\>[B]{}\Varid{lift}{}\<[12]%
\>[12]{}\mathbin{::}\Conid{Monad}\;\tau\Rightarrow \tau\;\alpha\hsarrow{\rightarrow }{\mathpunct{.}}\Conid{StateT}\;\sigma\;\tau\;\alpha{}\<[E]%
\\
\>[B]{}\Varid{lift}\;\Varid{m}\mathrel{=}\lambda \hslambda \Varid{s}\hsarrow{\rightarrow }{\mathpunct{.}}\mathbf{do}\;\{\mskip1.5mu \Varid{a}\leftarrow \Varid{m};\Varid{return}\;(\Varid{a},\Varid{s})\mskip1.5mu\}{}\<[E]%
\ColumnHook
\end{hscode}\resethooks
\endswithdisplay
\end{definition}

\begin{lemma} \label{lem:discardable}
Unused \ensuremath{\Varid{get}}s are discardable:
\begin{hscode}\SaveRestoreHook
\column{B}{@{}>{\hspre}l<{\hspost}@{}}%
\column{46}{@{}>{\hspre}l<{\hspost}@{}}%
\column{E}{@{}>{\hspre}l<{\hspost}@{}}%
\>[B]{}\mathbf{do}\;\{\mskip1.5mu \anonymous \leftarrow \Varid{get};\Varid{m}\mskip1.5mu\}{}\<[46]%
\>[46]{}\mathrel{=}\mathbf{do}\;\{\mskip1.5mu \Varid{m}\mskip1.5mu\}{}\<[E]%
\ColumnHook
\end{hscode}\resethooks
\endswithdisplay
\end{lemma}


\begin{lemma}[liftings commute with \ensuremath{\Varid{get}} and \ensuremath{\Varid{set}}] \label{lem:liftings-commute}
We have:
\begin{hscode}\SaveRestoreHook
\column{B}{@{}>{\hspre}l<{\hspost}@{}}%
\column{40}{@{}>{\hspre}l<{\hspost}@{}}%
\column{E}{@{}>{\hspre}l<{\hspost}@{}}%
\>[B]{}\mathbf{do}\;\{\mskip1.5mu \Varid{a}\leftarrow \Varid{get};\Varid{b}\leftarrow \Varid{lift}\;\Varid{m};\Varid{return}\;(\Varid{a},\Varid{b})\mskip1.5mu\}{}\<[E]%
\\
\>[B]{}\hsindent{40}{}\<[40]%
\>[40]{}\mathrel{=}\mathbf{do}\;\{\mskip1.5mu \Varid{b}\leftarrow \Varid{lift}\;\Varid{m};\Varid{a}\leftarrow \Varid{get};\Varid{return}\;(\Varid{a},\Varid{b})\mskip1.5mu\}{}\<[E]%
\\[\blanklineskip]%
\>[B]{}\mathbf{do}\;\{\mskip1.5mu \Varid{set}\;\Varid{a};\Varid{b}\leftarrow \Varid{lift}\;\Varid{m};\Varid{return}\;\Varid{b}\mskip1.5mu\}{}\<[40]%
\>[40]{}\mathrel{=}\mathbf{do}\;\{\mskip1.5mu \Varid{b}\leftarrow \Varid{lift}\;\Varid{m};\Varid{set}\;\Varid{a};\Varid{return}\;\Varid{b}\mskip1.5mu\}{}\<[E]%
\ColumnHook
\end{hscode}\resethooks
\endswithdisplay
\end{lemma}

\begin{definition*}
Some convenient shorthands:
\begin{hscode}\SaveRestoreHook
\column{B}{@{}>{\hspre}l<{\hspost}@{}}%
\column{13}{@{}>{\hspre}l<{\hspost}@{}}%
\column{17}{@{}>{\hspre}l<{\hspost}@{}}%
\column{E}{@{}>{\hspre}l<{\hspost}@{}}%
\>[B]{}\Varid{gets}\mathbin{::}\Conid{Monad}\;\tau\Rightarrow (\sigma\hsarrow{\rightarrow }{\mathpunct{.}}\alpha)\hsarrow{\rightarrow }{\mathpunct{.}}\Conid{StateT}\;\sigma\;\tau\;\alpha{}\<[E]%
\\
\>[B]{}\Varid{gets}\;\Varid{f}\mathrel{=}\mathbf{do}\;\{\mskip1.5mu \Varid{s}\leftarrow \Varid{get};\Varid{return}\;(\Varid{f}\;\Varid{s})\mskip1.5mu\}{}\<[E]%
\\[\blanklineskip]%
\>[B]{}\Varid{eval}{}\<[13]%
\>[13]{}\mathbin{::}\Conid{Monad}\;\tau\Rightarrow \Conid{StateT}\;\sigma\;\tau\;\alpha\hsarrow{\rightarrow }{\mathpunct{.}}\sigma\hsarrow{\rightarrow }{\mathpunct{.}}\tau\;\alpha{}\<[E]%
\\
\>[B]{}\Varid{eval}\;\Varid{m}\;\Varid{s}{}\<[17]%
\>[17]{}\mathrel{=}\mathbf{do}\;\{\mskip1.5mu (\Varid{a},\Varid{s'})\leftarrow \Varid{m}\;\Varid{s};\Varid{return}\;\Varid{a}\mskip1.5mu\}{}\<[E]%
\\[\blanklineskip]%
\>[B]{}\Varid{exec}{}\<[13]%
\>[13]{}\mathbin{::}\Conid{Monad}\;\tau\Rightarrow \Conid{StateT}\;\sigma\;\tau\;\alpha\hsarrow{\rightarrow }{\mathpunct{.}}\sigma\hsarrow{\rightarrow }{\mathpunct{.}}\tau\;\sigma{}\<[E]%
\\
\>[B]{}\Varid{exec}\;\Varid{m}\;\Varid{s}{}\<[17]%
\>[17]{}\mathrel{=}\mathbf{do}\;\{\mskip1.5mu (\Varid{a},\Varid{s'})\leftarrow \Varid{m}\;\Varid{s};\Varid{return}\;\Varid{s'}\mskip1.5mu\}{}\<[E]%
\ColumnHook
\end{hscode}\resethooks
\endswithdisplay
\end{definition*}

\begin{definition*}
  We say that a computation \ensuremath{\Varid{m}\mathbin{::}\Conid{StateT}\;\Conid{S}\;\Conid{T}\;\Conid{A}} is a \emph{\ensuremath{\Conid{T}}-pure query}
  if it cannot change the state, and is pure with respect to the base
  monad \ensuremath{\Conid{T}}; that is, \ensuremath{\Varid{m}\mathrel{=}\Varid{gets}\;\Varid{h}} for some \ensuremath{\Varid{h}\mathbin{::}\Conid{S}\hsarrow{\rightarrow }{\mathpunct{.}}\Conid{A}}. 
Note that a \ensuremath{\Conid{T}}-pure query need not be pure
with respect to \ensuremath{\Conid{StateT}\;\Conid{S}\;\Conid{T}}; in particular, it will typically read the state.
\end{definition*}

\begin{definition}[data refinement] \label{def:data-refinement}
Given monads \ensuremath{\Conid{M}} of `abstract computations' and \ensuremath{\Conid{M'}} of `concrete computations',
various `abstract operations' \ensuremath{\Varid{op}\mathbin{::}\Conid{A}\hsarrow{\rightarrow }{\mathpunct{.}}\Conid{M}\;\Conid{B}} with
corresponding `concrete operations' \ensuremath{\Varid{op}^\prime\mathbin{::}\Conid{A}\hsarrow{\rightarrow }{\mathpunct{.}}\Conid{M'}\;\Conid{B}}, 
an `abstraction function' \ensuremath{\Varid{abs}\mathbin{::}\Conid{M'}\;\alpha\hsarrow{\rightarrow }{\mathpunct{.}}\Conid{M}\;\alpha} and a
`reification function' \ensuremath{\Varid{conc}\mathbin{::}\Conid{M}\;\alpha\hsarrow{\rightarrow }{\mathpunct{.}}\Conid{M'}\;\alpha}, we say that 
\ensuremath{\Varid{conc}} is a \emph{data refinement} from \ensuremath{(\Conid{M},\Varid{op})} to \ensuremath{(\Conid{M'},\Varid{op}^\prime)} if:
\begin{itemize}
\item \ensuremath{\Varid{conc}} distributes over \ensuremath{(\bind )}
\item \ensuremath{\Varid{abs}\hsdot{\cdot }{.}\Varid{conc}\mathrel{=}\Varid{id}}, and
\item \ensuremath{\Varid{op}^\prime\mathrel{=}\Varid{conc}\hsdot{\cdot }{.}\Varid{op}} for each of the operations.
\end{itemize}
\endswithdisplay
\end{definition}

\begin{remark*}
Given such a data refinement, a composite abstract
computation can be faithfully simulated by a concrete one:
\begin{hscode}\SaveRestoreHook
\column{B}{@{}>{\hspre}c<{\hspost}@{}}%
\column{BE}{@{}l@{}}%
\column{3}{@{}>{\hspre}l<{\hspost}@{}}%
\column{5}{@{}>{\hspre}l<{\hspost}@{}}%
\column{E}{@{}>{\hspre}l<{\hspost}@{}}%
\>[3]{}\mathbf{do}\;\{\mskip1.5mu \Varid{a}\leftarrow \Varid{op}_1\;();\Varid{b}\leftarrow \Varid{op}_2\;(\Varid{a});\Varid{op}_3\;(\Varid{a},\Varid{b})\mskip1.5mu\}{}\<[E]%
\\
\>[B]{}\mathrel{=}{}\<[BE]%
\>[5]{}\mbox{\commentbegin  \ensuremath{\Varid{abs}\hsdot{\cdot }{.}\Varid{conc}\mathrel{=}\Varid{id}}  \commentend}{}\<[E]%
\\
\>[B]{}\hsindent{3}{}\<[3]%
\>[3]{}\Varid{abs}\;(\Varid{conc}\;(\mathbf{do}\;\{\mskip1.5mu \Varid{a}\leftarrow \Varid{op}_1\;();\Varid{b}\leftarrow \Varid{op}_2\;(\Varid{a});\Varid{op}_3\;(\Varid{a},\Varid{b})\mskip1.5mu\})){}\<[E]%
\\
\>[B]{}\mathrel{=}{}\<[BE]%
\>[5]{}\mbox{\commentbegin  \ensuremath{\Varid{conc}} distributes over \ensuremath{(\bind )}  \commentend}{}\<[E]%
\\
\>[B]{}\hsindent{3}{}\<[3]%
\>[3]{}\Varid{abs}\;(\mathbf{do}\;\{\mskip1.5mu \Varid{a}\leftarrow \Varid{conc}\;(\Varid{op}_1\;());\Varid{b}\leftarrow \Varid{conc}\;(\Varid{op}_2\;(\Varid{a}));\Varid{conc}\;(\Varid{op}_3\;(\Varid{a},\Varid{b}))\mskip1.5mu\}){}\<[E]%
\\
\>[B]{}\mathrel{=}{}\<[BE]%
\>[5]{}\mbox{\commentbegin  concrete operations  \commentend}{}\<[E]%
\\
\>[B]{}\hsindent{3}{}\<[3]%
\>[3]{}\Varid{abs}\;(\mathbf{do}\;\{\mskip1.5mu \Varid{a}\leftarrow \Varid{op}^\prime_{1}\;();\Varid{b}\leftarrow \Varid{op}^\prime_{2}\;(\Varid{a});\Varid{op}^\prime_{3}\;(\Varid{a},\Varid{b})\mskip1.5mu\}){}\<[E]%
\ColumnHook
\end{hscode}\resethooks
If \ensuremath{\Varid{conc}} also preserves \ensuremath{\Varid{return}} (so \ensuremath{\Varid{conc}} is a monad
morphism), then we would have a similar result for `empty' abstract
computations too; but we don't need that stronger property in this paper,
and it does not hold for our main example (Remark~\ref{rem:wlog}).
\end{remark*}

\begin{lemma} \label{lem:StateT-data-refinement}
Given an arbitrary monad \ensuremath{\Conid{T}}, not assumed to be an instance of \ensuremath{\Conid{StateT}},
with operations \ensuremath{\Varid{get_T}\mathbin{::}\Conid{T}\;\Conid{S}} and \ensuremath{\Varid{set_T}\mathbin{::}\Conid{S}\hsarrow{\rightarrow }{\mathpunct{.}}\Conid{T}\;()} for a type \ensuremath{\Conid{S}}, such that \ensuremath{\Varid{get_T}} and \ensuremath{\Varid{set_T}} satisfy the
laws \ensuremath{\mathrm{(GG)}}, \ensuremath{\mathrm{(GS)}}, and \ensuremath{\mathrm{(SG)}} 
of Definition~\ref{def:state-monad-transformer}, then
there is a data refinement from \ensuremath{\Conid{T}} to \ensuremath{\Conid{StateT}\;\Conid{S}\;\Conid{T}}.
\jgnote{might be better introducing \ensuremath{\Conid{MonadState}} typeclass}
\end{lemma}

\begin{proof}[sketch]
The abstraction function \ensuremath{\Varid{abs}} from \ensuremath{\Conid{StateT}\;\Conid{S}\;\Conid{T}} to \ensuremath{\Conid{T}} and the
reification function \ensuremath{\Varid{conc}} in the opposite direction are given by
\begin{hscode}\SaveRestoreHook
\column{B}{@{}>{\hspre}l<{\hspost}@{}}%
\column{9}{@{}>{\hspre}l<{\hspost}@{}}%
\column{E}{@{}>{\hspre}l<{\hspost}@{}}%
\>[B]{}\Varid{abs}\;\Varid{m}{}\<[9]%
\>[9]{}\mathrel{=}\mathbf{do}\;\{\mskip1.5mu \Varid{s}\leftarrow \Varid{get_T};(\Varid{a},\Varid{s'})\leftarrow \Varid{m}\;\Varid{s};\Varid{set_T}\;\Varid{s'};\Varid{return}\;\Varid{a}\mskip1.5mu\}{}\<[E]%
\\[\blanklineskip]%
\>[B]{}\Varid{conc}\;\Varid{m}{}\<[9]%
\>[9]{}\mathrel{=}\lambda \hslambda \Varid{s}\hsarrow{\rightarrow }{\mathpunct{.}}\mathbf{do}\;\{\mskip1.5mu \Varid{a}\leftarrow \Varid{m};\Varid{s'}\leftarrow \Varid{get_T};\Varid{return}\;(\Varid{a},\Varid{s'})\mskip1.5mu\}{}\<[E]%
\\
\>[9]{}\mathrel{=}\mathbf{do}\;\{\mskip1.5mu \Varid{a}\leftarrow \Varid{lift}\;\Varid{m};\Varid{s'}\leftarrow \Varid{lift}\;\Varid{get_T};\Varid{set}\;\Varid{s'};\Varid{return}\;\Varid{a}\mskip1.5mu\}{}\<[E]%
\ColumnHook
\end{hscode}\resethooks
\endswithdisplay
\end{proof}

\begin{remark*}
Informally, if \ensuremath{\Conid{T}} provides suitable get and set operations, we can
without loss of generality assume it to be an instance of
\ensuremath{\Conid{StateT}}. The essence of the data refinement is for concrete
computations to maintain a shadow copy of the implicit state; \ensuremath{\Varid{conc}\;\Varid{m}}
synchronises the outer copy of the state with the inner copy after
executing \ensuremath{\Varid{m}}, and \ensuremath{\Varid{abs}\;\Varid{m}} runs the \ensuremath{\Conid{StateT}} computation \ensuremath{\Varid{m}} on an
initial state extracted from \ensuremath{\Conid{T}}, and stores the final state back
there.
\end{remark*}

\subsection{Lenses}

The notion of an (asymmetric) `lens' between a source and a view was
introduced by Foster \etal~\cite{lens-toplas}. We adapt their
notation, as follows.

\begin{definition} \label{def:lens}
A lens \ensuremath{\Varid{l}\mathbin{::}\Conid{Lens}\;\Conid{S}\;\Conid{V}} from source type \ensuremath{\Conid{S}} to view type \ensuremath{\Conid{V}} 
consists of a pair of functions
which get a \ensuremath{\Varid{view}} of the source, and \ensuremath{\Varid{update}} an old source with
a modified view:
\begin{hscode}\SaveRestoreHook
\column{B}{@{}>{\hspre}l<{\hspost}@{}}%
\column{20}{@{}>{\hspre}l<{\hspost}@{}}%
\column{35}{@{}>{\hspre}l<{\hspost}@{}}%
\column{56}{@{}>{\hspre}l<{\hspost}@{}}%
\column{E}{@{}>{\hspre}l<{\hspost}@{}}%
\>[B]{}\mathbf{data}\;\Conid{Lens}\;\alpha\;\beta\mathrel{=}{}\<[20]%
\>[20]{}\Conid{Lens}\;\{\mskip1.5mu \Varid{view}{}\<[35]%
\>[35]{}\mathbin{::}\alpha\hsarrow{\rightarrow }{\mathpunct{.}}\beta,\Varid{update}{}\<[56]%
\>[56]{}\mathbin{::}\alpha\hsarrow{\rightarrow }{\mathpunct{.}}\beta\hsarrow{\rightarrow }{\mathpunct{.}}\alpha\mskip1.5mu\}{}\<[E]%
\ColumnHook
\end{hscode}\resethooks
We say that a lens \ensuremath{\Varid{l}\mathbin{::}\Conid{Lens}\;\Conid{S}\;\Conid{V}} is \emph{well behaved} if it satisfies
the two round-tripping laws
\begin{hscode}\SaveRestoreHook
\column{B}{@{}>{\hspre}l<{\hspost}@{}}%
\column{10}{@{}>{\hspre}c<{\hspost}@{}}%
\column{10E}{@{}l@{}}%
\column{14}{@{}>{\hspre}l<{\hspost}@{}}%
\column{39}{@{}>{\hspre}l<{\hspost}@{}}%
\column{E}{@{}>{\hspre}l<{\hspost}@{}}%
\>[B]{}\mathrm{(UV)}{}\<[10]%
\>[10]{}\quad{}\<[10E]%
\>[14]{}\Varid{l}\mathord{.}\Varid{view}\;(\Varid{l}\mathord{.}\Varid{update}\;\Varid{s}\;\Varid{v}){}\<[39]%
\>[39]{}\mathrel{=}\Varid{v}{}\<[E]%
\\
\>[B]{}\mathrm{(VU)}{}\<[10]%
\>[10]{}\quad{}\<[10E]%
\>[14]{}\Varid{l}\mathord{.}\Varid{update}\;\Varid{s}\;(\Varid{l}\mathord{.}\Varid{view}\;\Varid{s}){}\<[39]%
\>[39]{}\mathrel{=}\Varid{s}{}\<[E]%
\ColumnHook
\end{hscode}\resethooks
and \emph{very well-behaved}  or
\emph{overwritable} if in addition
\begin{hscode}\SaveRestoreHook
\column{B}{@{}>{\hspre}l<{\hspost}@{}}%
\column{10}{@{}>{\hspre}c<{\hspost}@{}}%
\column{10E}{@{}l@{}}%
\column{14}{@{}>{\hspre}l<{\hspost}@{}}%
\column{E}{@{}>{\hspre}l<{\hspost}@{}}%
\>[B]{}\mathrm{(UU)}{}\<[10]%
\>[10]{}\quad{}\<[10E]%
\>[14]{}\Varid{l}\mathord{.}\Varid{update}\;(\Varid{l}\mathord{.}\Varid{update}\;\Varid{s}\;\Varid{v})\;\Varid{v'}\mathrel{=}\Varid{l}\mathord{.}\Varid{update}\;\Varid{s}\;\Varid{v'}{}\<[E]%
\ColumnHook
\end{hscode}\resethooks
\endswithdisplay
\end{definition}

\begin{remark*}
Very well-behavedness captures the idea that, after two successive
updates, the second update completely \emph{overwrites} the first.
It turns out to be a rather strong condition, and many natural lenses
do not satisfy it. Those that do generally have the special property
that source \ensuremath{\Conid{S}} factorises cleanly into \ensuremath{\Conid{V}\times\Conid{C}} for some type \ensuremath{\Conid{C}} of
`complements' independent of \ensuremath{\Conid{V}}, and so the view is a projection.
For example:
\begin{hscode}\SaveRestoreHook
\column{B}{@{}>{\hspre}l<{\hspost}@{}}%
\column{E}{@{}>{\hspre}l<{\hspost}@{}}%
\>[B]{}\Varid{fstLens}\mathbin{::}\Conid{Lens}\;(\Varid{a},\Varid{b})\;\Varid{a}{}\<[E]%
\\
\>[B]{}\Varid{fstLens}\mathrel{=}\Conid{Lens}\;\Varid{fst}\;\Varid{u}\;\mathbf{where}\;\Varid{u}\;(\Varid{a},\Varid{b})\;\Varid{a'}\mathrel{=}(\Varid{a'},\Varid{b}){}\<[E]%
\\[\blanklineskip]%
\>[B]{}\Varid{sndLens}\mathbin{::}\Conid{Lens}\;(\Varid{a},\Varid{b})\;\Varid{b}{}\<[E]%
\\
\>[B]{}\Varid{sndLens}\mathrel{=}\Conid{Lens}\;\Varid{snd}\;\Varid{u}\;\mathbf{where}\;\Varid{u}\;(\Varid{a},\Varid{b})\;\Varid{b'}\mathrel{=}(\Varid{a},\Varid{b'}){}\<[E]%
\ColumnHook
\end{hscode}\resethooks
But in general, the \ensuremath{\Conid{V}} may be
computed from and therefore depend on all of the \ensuremath{\Conid{S}} value,
and there is no clean factorisation of \ensuremath{\Conid{S}} into \ensuremath{\Conid{V}\times\Conid{C}}.
\end{remark*}

Asymmetric lenses are constrained, in the sense that they relate
two types \ensuremath{\Conid{S}} and \ensuremath{\Conid{V}} in which the view \ensuremath{\Conid{V}} is completely determined
by the source \ensuremath{\Conid{S}}. \HPWlong\ \cite{symlens} relaxed this constraint,
introducing \emph{symmetric lenses} between two types \ensuremath{\Conid{A}} and \ensuremath{\Conid{B}},
neither of which need determine the other:

\begin{definition} \label{def:symlens}
A \emph{symmetric lens} from \ensuremath{\Conid{A}} to \ensuremath{\Conid{B}} with complement
type \ensuremath{\Conid{C}} consists of two functions converting to and from \ensuremath{\Conid{A}} and \ensuremath{\Conid{B}},
each also operating on \ensuremath{\Conid{C}}. 
\begin{hscode}\SaveRestoreHook
\column{B}{@{}>{\hspre}l<{\hspost}@{}}%
\column{32}{@{}>{\hspre}l<{\hspost}@{}}%
\column{41}{@{}>{\hspre}l<{\hspost}@{}}%
\column{73}{@{}>{\hspre}l<{\hspost}@{}}%
\column{E}{@{}>{\hspre}l<{\hspost}@{}}%
\>[B]{}\mathbf{data}\;\Conid{SLens}\;\gamma\;\alpha\;\beta\mathrel{=}\Conid{SLens}\;\{\mskip1.5mu {}\<[32]%
\>[32]{}\Varid{put}_{R}{}\<[41]%
\>[41]{}\mathbin{::}(\alpha,\gamma)\hsarrow{\rightarrow }{\mathpunct{.}}(\beta,\gamma),\Varid{put}_{L}{}\<[73]%
\>[73]{}\mathbin{::}(\beta,\gamma)\hsarrow{\rightarrow }{\mathpunct{.}}(\alpha,\gamma)\mskip1.5mu\}{}\<[E]%
\ColumnHook
\end{hscode}\resethooks
We say that symmetric lens \ensuremath{\Varid{l}} is \emph{well-behaved} if it satisfies the 
following two laws:
\begin{hscode}\SaveRestoreHook
\column{B}{@{}>{\hspre}l<{\hspost}@{}}%
\column{13}{@{}>{\hspre}c<{\hspost}@{}}%
\column{13E}{@{}l@{}}%
\column{17}{@{}>{\hspre}l<{\hspost}@{}}%
\column{37}{@{}>{\hspre}l<{\hspost}@{}}%
\column{47}{@{}>{\hspre}c<{\hspost}@{}}%
\column{47E}{@{}l@{}}%
\column{51}{@{}>{\hspre}c<{\hspost}@{}}%
\column{51E}{@{}l@{}}%
\column{55}{@{}>{\hspre}c<{\hspost}@{}}%
\column{55E}{@{}l@{}}%
\column{59}{@{}>{\hspre}l<{\hspost}@{}}%
\column{82}{@{}>{\hspre}l<{\hspost}@{}}%
\column{E}{@{}>{\hspre}l<{\hspost}@{}}%
\>[B]{}\mathrm{(PutRL)}{}\<[13]%
\>[13]{}\quad{}\<[13E]%
\>[17]{}\Varid{l}\mathord{.}\Varid{put}_{R}\;(\Varid{a},\Varid{c}){}\<[37]%
\>[37]{}\mathrel{=}(\Varid{b},\Varid{c'}){}\<[47]%
\>[47]{}\quad{}\<[47E]%
\>[51]{}\Rightarrow {}\<[51E]%
\>[55]{}\quad{}\<[55E]%
\>[59]{}\Varid{l}\mathord{.}\Varid{put}_{L}\;(\Varid{b},\Varid{c'}){}\<[82]%
\>[82]{}\mathrel{=}(\Varid{a},\Varid{c'}){}\<[E]%
\\
\>[B]{}\mathrm{(PutLR)}{}\<[13]%
\>[13]{}\quad{}\<[13E]%
\>[17]{}\Varid{l}\mathord{.}\Varid{put}_{L}\;(\Varid{b},\Varid{c}){}\<[37]%
\>[37]{}\mathrel{=}(\Varid{a},\Varid{c'}){}\<[47]%
\>[47]{}\quad{}\<[47E]%
\>[51]{}\Rightarrow {}\<[51E]%
\>[55]{}\quad{}\<[55E]%
\>[59]{}\Varid{l}\mathord{.}\Varid{put}_{R}\;(\Varid{a},\Varid{c'}){}\<[82]%
\>[82]{}\mathrel{=}(\Varid{b},\Varid{c'}){}\<[E]%
\ColumnHook
\end{hscode}\resethooks
(There is also a stronger notion of very well-behavedness,
but we do not need it for this paper.)
\end{definition}

\begin{remark*}
The idea is that \ensuremath{\Conid{A}} and \ensuremath{\Conid{B}} represent two overlapping but distinct
views of some common underlying data, and the so-called complement \ensuremath{\Conid{C}}
represents their amalgamation (\emph{not} necessarily containing all
the information from both: rather, one view plus the complement together
contain enough information to reconstruct the other view). Each
function takes a new view and the old complement, and returns a new
opposite view and a new complement.  The two well-behavedness
properties each say that after one update
operation, the complement \ensuremath{\Varid{c'}} is fully consistent with the current
views, and so a subsequent opposite update with the same view has no
further effect on the complement.
\end{remark*}
\section{Monadic bidirectional transformations}\label{sec:monad}

We have seen that the state monad provides a pair \ensuremath{\Varid{get},\Varid{set}} of
operations on the state. A symmetric \bx{} should provide two such
pairs, one for each data source; these four operations should be
effectful, not least because they should operate on some shared
consistent state.  We therefore introduce the following general notion
of \emph{monadic \bx} (which we sometimes call `mbx', for short).

\begin{definition*}
We say that a data structure \ensuremath{\Varid{t}\mathbin{::}\Conid{BX}\;\Conid{T}\;\Conid{A}\;\Conid{B}} is a
\emph{bx between \ensuremath{\Conid{A}} and \ensuremath{\Conid{B}} in monad \ensuremath{\Conid{T}}} when it provides appropriately
typed functions:
\begin{hscode}\SaveRestoreHook
\column{B}{@{}>{\hspre}l<{\hspost}@{}}%
\column{32}{@{}>{\hspre}l<{\hspost}@{}}%
\column{38}{@{}>{\hspre}l<{\hspost}@{}}%
\column{50}{@{}>{\hspre}l<{\hspost}@{}}%
\column{59}{@{}>{\hspre}l<{\hspost}@{}}%
\column{E}{@{}>{\hspre}l<{\hspost}@{}}%
\>[B]{}\mathbf{data}\;\Conid{BX}\;\tau\;\alpha\;\beta\mathrel{=}\Conid{BX}\;\{\mskip1.5mu {}\<[32]%
\>[32]{}\get{L}{}\<[38]%
\>[38]{}\mathbin{::}\tau\;\alpha,{}\<[50]%
\>[50]{}\quad\set{L}{}\<[59]%
\>[59]{}\mathbin{::}\alpha\hsarrow{\rightarrow }{\mathpunct{.}}\tau\;(),{}\<[E]%
\\
\>[32]{}\get{R}{}\<[38]%
\>[38]{}\mathbin{::}\tau\;\beta,{}\<[50]%
\>[50]{}\quad\set{R}{}\<[59]%
\>[59]{}\mathbin{::}\beta\hsarrow{\rightarrow }{\mathpunct{.}}\tau\;()\mskip1.5mu\}{}\<[E]%
\ColumnHook
\end{hscode}\resethooks
\endswithdisplay
\end{definition*}

\subsection{Entangled state}

The \ensuremath{\Varid{get}} and \ensuremath{\Varid{set}} operations of the state monad 
satisfy the four laws \ensuremath{\mathrm{(GG)}}, \ensuremath{\mathrm{(SG)}}, \ensuremath{\mathrm{(GS)}}, \ensuremath{\mathrm{(SS)}}
of \Definition~\ref{def:state-monad-transformer}.
More generally, one can give an
equational theory of state with multiple memory locations;
in particular, with just two locations `left' (\ensuremath{\Conid{L}}) and `right' (\ensuremath{\Conid{R}}),
the equational theory has four operations
\ensuremath{\get{L}}, \ensuremath{\set{L}}, \ensuremath{\get{R}}, \ensuremath{\set{R}} that match the \ensuremath{\Conid{BX}} interface. 
This theory has four 
laws for \ensuremath{\Conid{L}} analogous to those of 
\Definition~\ref{def:state-monad-transformer}, 
another four such laws for \ensuremath{\Conid{R}}, and a
final four laws 
stating that the \ensuremath{\Conid{L}}-operations commute with the \ensuremath{\Conid{R}}-operations. 
But this equational theory of two memory locations is too strong for
interesting \bx, because of the commutativity requirement: the whole point
of the exercise is that invoking \ensuremath{\set{L}} should indeed affect the behaviour
of a subsequent \ensuremath{\get{R}}, and symmetrically.
We therefore impose only a subset of those twelve laws on the \ensuremath{\Conid{BX}} interface.

\begin{definition*}
A \emph{well-behaved} \ensuremath{\Conid{BX}} is one satisfying 
the following seven laws:
\savecolumns
\begin{hscode}\SaveRestoreHook
\column{B}{@{}>{\hspre}l<{\hspost}@{}}%
\column{12}{@{}>{\hspre}c<{\hspost}@{}}%
\column{12E}{@{}l@{}}%
\column{16}{@{}>{\hspre}l<{\hspost}@{}}%
\column{63}{@{}>{\hspre}l<{\hspost}@{}}%
\column{E}{@{}>{\hspre}l<{\hspost}@{}}%
\>[B]{}\mathrm{(G_LG_L)}{}\<[12]%
\>[12]{}\quad{}\<[12E]%
\>[16]{}\mathbf{do}\;\{\mskip1.5mu \Varid{a}\leftarrow \get{L};\Varid{a'}\leftarrow \get{L};\Varid{return}\;(\Varid{a},\Varid{a'})\mskip1.5mu\}{}\<[E]%
\\
\>[16]{}\hsindent{47}{}\<[63]%
\>[63]{}\mathrel{=}\mathbf{do}\;\{\mskip1.5mu \Varid{a}\leftarrow \get{L};\Varid{return}\;(\Varid{a},\Varid{a})\mskip1.5mu\}{}\<[E]%
\\
\>[B]{}\mathrm{(S_LG_L)}{}\<[12]%
\>[12]{}\quad{}\<[12E]%
\>[16]{}\mathbf{do}\;\{\mskip1.5mu \set{L}\;\Varid{a};\get{L}\mskip1.5mu\}{}\<[63]%
\>[63]{}\mathrel{=}\mathbf{do}\;\{\mskip1.5mu \set{L}\;\Varid{a};\Varid{return}\;\Varid{a}\mskip1.5mu\}{}\<[E]%
\\
\>[B]{}\mathrm{(G_LS_L)}{}\<[12]%
\>[12]{}\quad{}\<[12E]%
\>[16]{}\mathbf{do}\;\{\mskip1.5mu \Varid{a}\leftarrow \get{L};\set{L}\;\Varid{a}\mskip1.5mu\}{}\<[63]%
\>[63]{}\mathrel{=}\mathbf{do}\;\{\mskip1.5mu \Varid{return}\;()\mskip1.5mu\}{}\<[E]%
\\[\blanklineskip]%
\>[B]{}\mathrm{(G_RG_R)}{}\<[12]%
\>[12]{}\quad{}\<[12E]%
\>[16]{}\mathbf{do}\;\{\mskip1.5mu \Varid{a}\leftarrow \get{R};\Varid{a'}\leftarrow \get{R};\Varid{return}\;(\Varid{a},\Varid{a'})\mskip1.5mu\}{}\<[E]%
\\
\>[16]{}\hsindent{47}{}\<[63]%
\>[63]{}\mathrel{=}\mathbf{do}\;\{\mskip1.5mu \Varid{a}\leftarrow \get{R};\Varid{return}\;(\Varid{a},\Varid{a})\mskip1.5mu\}{}\<[E]%
\\
\>[B]{}\mathrm{(S_RG_R)}{}\<[12]%
\>[12]{}\quad{}\<[12E]%
\>[16]{}\mathbf{do}\;\{\mskip1.5mu \set{R}\;\Varid{a};\get{R}\mskip1.5mu\}{}\<[63]%
\>[63]{}\mathrel{=}\mathbf{do}\;\{\mskip1.5mu \set{R}\;\Varid{a};\Varid{return}\;\Varid{a}\mskip1.5mu\}{}\<[E]%
\\
\>[B]{}\mathrm{(G_RS_R)}{}\<[12]%
\>[12]{}\quad{}\<[12E]%
\>[16]{}\mathbf{do}\;\{\mskip1.5mu \Varid{a}\leftarrow \get{R};\set{R}\;\Varid{a}\mskip1.5mu\}{}\<[63]%
\>[63]{}\mathrel{=}\mathbf{do}\;\{\mskip1.5mu \Varid{return}\;()\mskip1.5mu\}{}\<[E]%
\\[\blanklineskip]%
\>[B]{}\mathrm{(G_LG_R)}{}\<[12]%
\>[12]{}\quad{}\<[12E]%
\>[16]{}\mathbf{do}\;\{\mskip1.5mu \Varid{a}\leftarrow \get{L};\Varid{b}\leftarrow \get{R};\Varid{return}\;(\Varid{a},\Varid{b})\mskip1.5mu\}{}\<[E]%
\\
\>[16]{}\hsindent{47}{}\<[63]%
\>[63]{}\mathrel{=}\mathbf{do}\;\{\mskip1.5mu \Varid{b}\leftarrow \get{R};\Varid{a}\leftarrow \get{L};\Varid{return}\;(\Varid{a},\Varid{b})\mskip1.5mu\}{}\<[E]%
\ColumnHook
\end{hscode}\resethooks
We further say that a  \ensuremath{\Conid{BX}} is \emph{overwritable} if it satisfies
\restorecolumns
\begin{hscode}\SaveRestoreHook
\column{B}{@{}>{\hspre}l<{\hspost}@{}}%
\column{12}{@{}>{\hspre}c<{\hspost}@{}}%
\column{12E}{@{}l@{}}%
\column{16}{@{}>{\hspre}l<{\hspost}@{}}%
\column{63}{@{}>{\hspre}l<{\hspost}@{}}%
\column{E}{@{}>{\hspre}l<{\hspost}@{}}%
\>[B]{}\mathrm{(S_LS_L)}{}\<[12]%
\>[12]{}\quad{}\<[12E]%
\>[16]{}\mathbf{do}\;\{\mskip1.5mu \set{L}\;\Varid{a};\set{L}\;\Varid{a'}\mskip1.5mu\}{}\<[63]%
\>[63]{}\mathrel{=}\mathbf{do}\;\{\mskip1.5mu \set{L}\;\Varid{a'}\mskip1.5mu\}{}\<[E]%
\\
\>[B]{}\mathrm{(S_RS_R)}{}\<[12]%
\>[12]{}\quad{}\<[12E]%
\>[16]{}\mathbf{do}\;\{\mskip1.5mu \set{R}\;\Varid{a};\set{R}\;\Varid{a'}\mskip1.5mu\}{}\<[63]%
\>[63]{}\mathrel{=}\mathbf{do}\;\{\mskip1.5mu \set{R}\;\Varid{a'}\mskip1.5mu\}{}\<[E]%
\ColumnHook
\end{hscode}\resethooks
\endswithdisplay
\end{definition*}
We might think of the \ensuremath{\Conid{A}} and \ensuremath{\Conid{B}} views as being \emph{entangled}; in
particular, we call the monad arising as the initial model of the
theory with
the four operations
\ensuremath{\get{L},\set{L},\get{R},\set{R}}
and the seven laws \ensuremath{\mathrm{(G_LG_L)}}\dots \ensuremath{\mathrm{(G_LG_R)}} 
the \emph{entangled state monad}.

\begin{remark*}
Overwritability is a strong condition, corresponding
to very well-behavedness of lenses \cite{lens-toplas}, 
history-ignorance of relational bx \cite{stevens09:sosym}
etc.; many interesting bx fail to satisfy it.
Indeed, 
in an effectful setting, 
a law such as \ensuremath{\mathrm{(S_LS_L)}} demands
that \ensuremath{\set{L}\;\Varid{a'}} be able to undo (or overwrite) any effects arising from
\ensuremath{\set{L}\;\Varid{a}}; such behaviour is plausible in the pure state-based setting,
but not in general. Consequently, we do not demand
overwritability in what follows. 
\end{remark*}

\begin{definition} \label{def:equivalence}
  A \emph{\bx{} morphism} from \ensuremath{bx_{1}\mathbin{::}\Conid{BX}\;\Conid{T}_{1}\;\Conid{A}\;\Conid{B}} to \ensuremath{bx_{2}\mathbin{::}\Conid{BX}\;\Conid{T}_{2}\;\Conid{A}\;\Conid{B}} is a monad morphism \ensuremath{\varphi \mathbin{:}\forall \alpha \hsforall \hsdot{\cdot }{.}\Conid{T}_{1}\;\alpha \hsarrow{\rightarrow }{\mathpunct{.}}\Conid{T}_{2}\;\alpha } that preserves the
  \bx{} operations, in the sense that \ensuremath{\varphi \;(bx_{1}\mathord{.}\get{L})\mathrel{=}bx_{2}\mathord{.}\get{L}} and so
  on.
  A \emph{\bx{} isomorphism} is an invertible \bx{} morphism, i.e. a
  pair of monad morphisms \ensuremath{\iota \mathbin{::}\forall \alpha\hsforall \hsdot{\cdot }{.}\Conid{T}_{1}\;\alpha\hsarrow{\rightarrow }{\mathpunct{.}}\Conid{T}_{2}\;\alpha} and \ensuremath{\iota ^{-1}\mathbin{:}\forall \alpha\hsforall \hsdot{\cdot }{.}\Conid{T}_{2}\;\alpha\hsarrow{\rightarrow }{\mathpunct{.}}\Conid{T}_{1}\;\alpha} which are mutually inverse, and which also
  preserve the operations.
  We say that
  \ensuremath{bx_{1}} and \ensuremath{bx_{2}} are \emph{equivalent} (and write \ensuremath{bx_1\equiv bx_2}) if
  there is a \bx{} isomorphism between them.
\end{definition}

\subsection{Stateful BX}

The get and set operations of a \ensuremath{\Conid{BX}}, and the relationship via
entanglement with the equational theory of the state monad, 
strongly suggest that there is something inherently stateful about
\bx{}; that will be a crucial observation in what follows. In
particular, the \ensuremath{\get{L}} and \ensuremath{\get{R}} operations of a \ensuremath{\Conid{BX}\;\Conid{T}\;\Conid{A}\;\Conid{B}} reveal
that it is in some sense storing an \ensuremath{\Conid{A}\times\Conid{B}} pair; conversely, the
\ensuremath{\set{L}} and \ensuremath{\set{R}} operations allow that pair to be updated.
We therefore focus on monads of the form \ensuremath{\Conid{StateT}\;\Conid{S}\;\Conid{T}}, where \ensuremath{\Conid{S}} is the `state' of the bx itself, 
capable of recording an \ensuremath{\Conid{A}} and a \ensuremath{\Conid{B}}, and \ensuremath{\Conid{T}} is a monad
encapsulating any other ambient effects that can be performed by the
`get' and `set' operations.

\begin{definition*}
We introduce the following instance of the \ensuremath{\Conid{BX}} signature
(note the inverted argument order):
\begin{hscode}\SaveRestoreHook
\column{B}{@{}>{\hspre}l<{\hspost}@{}}%
\column{28}{@{}>{\hspre}l<{\hspost}@{}}%
\column{E}{@{}>{\hspre}l<{\hspost}@{}}%
\>[B]{}\mathbf{type}\;\Conid{StateTBX}\;\tau\;\sigma\;\alpha\;\beta{}\<[28]%
\>[28]{}\mathrel{=}\Conid{BX}\;(\Conid{StateT}\;\sigma\;\tau)\;\alpha\;\beta{}\<[E]%
\ColumnHook
\end{hscode}\resethooks
\endswithdisplay
\end{definition*}

\begin{remark} \label{rem:consistent}
In fact, we can say more about the pair inside a \ensuremath{bx\mathbin{::}\Conid{BX}\;\Conid{T}\;\Conid{A}\;\Conid{B}}: it will generally be the case that only certain such pairs are
observable. Specifically, we can define the subset \ensuremath{\Conid{R}\subseteq\Conid{A}\times\Conid{B}} of \emph{consistent pairs} according to \ensuremath{bx}, namely those pairs \ensuremath{(\Varid{a},\Varid{b})}
that may be returned by
\begin{hscode}\SaveRestoreHook
\column{B}{@{}>{\hspre}l<{\hspost}@{}}%
\column{E}{@{}>{\hspre}l<{\hspost}@{}}%
\>[B]{}\mathbf{do}\;\{\mskip1.5mu \Varid{a}\leftarrow \get{L};\Varid{b}\leftarrow \get{R};\Varid{return}\;(\Varid{a},\Varid{b})\mskip1.5mu\}{}\<[E]%
\ColumnHook
\end{hscode}\resethooks
We can see this subset \ensuremath{\Conid{R}} as the \emph{consistency
  relation} between \ensuremath{\Conid{A}} and \ensuremath{\Conid{B}} maintained by \ensuremath{bx}.  We sometimes
write \ensuremath{\Conid{A}\mathbin{\!\Join\!}\Conid{B}} for this relation, when the \bx{} in question is clear from
context.
\end{remark}

\begin{remark} \label{rem:wlog} 
Note that restricting attention to instances of \ensuremath{\Conid{StateT}} 
is not as great a loss of generality as might at first appear.  
Consider a well-behaved \bx{} of type
\ensuremath{\Conid{BX}\;\Conid{T}\;\Conid{A}\;\Conid{B}}, over some monad \ensuremath{\Conid{T}} not assumed to be an instance of \ensuremath{\Conid{StateT}}.
We say that a consistent pair \ensuremath{(\Varid{a},\Varid{b})\mathbin{::}\Conid{A}\mathbin{\!\Join\!}\Conid{B}} is \emph{stable} if, when setting the components in either
order, the later one does not disturb the earlier:
\begin{hscode}\SaveRestoreHook
\column{B}{@{}>{\hspre}l<{\hspost}@{}}%
\column{32}{@{}>{\hspre}l<{\hspost}@{}}%
\column{E}{@{}>{\hspre}l<{\hspost}@{}}%
\>[B]{}\mathbf{do}\;\{\mskip1.5mu \set{L}\;\Varid{a};\set{R}\;\Varid{b};\get{L}\mskip1.5mu\}{}\<[32]%
\>[32]{}\mathrel{=}\mathbf{do}\;\{\mskip1.5mu \set{L}\;\Varid{a};\set{R}\;\Varid{b};\Varid{return}\;\Varid{a}\mskip1.5mu\}{}\<[E]%
\\
\>[B]{}\mathbf{do}\;\{\mskip1.5mu \set{R}\;\Varid{b};\set{L}\;\Varid{a};\get{R}\mskip1.5mu\}{}\<[32]%
\>[32]{}\mathrel{=}\mathbf{do}\;\{\mskip1.5mu \set{R}\;\Varid{b};\set{L}\;\Varid{a};\Varid{return}\;\Varid{b}\mskip1.5mu\}{}\<[E]%
\ColumnHook
\end{hscode}\resethooks
We say that the \bx{} itself is stable if all its consistent pairs are
stable. Stability does not follow from the laws, 
but many \bx{} do satisfy this stronger condition. And given a
stable \bx{}, we can construct get and set operations for \ensuremath{\Conid{A}\mathbin{\!\Join\!}\Conid{B}}
pairs, satisfying the three laws \ensuremath{\mathrm{(GG)}}, \ensuremath{\mathrm{(GS)}}, \ensuremath{\mathrm{(SG)}}
of Definition~\ref{def:state-monad-transformer}.
\jgnote{Might be better phrased in terms of three \ensuremath{\Conid{MonadState}}
  instances?}  By Lemma~\ref{lem:StateT-data-refinement}, this gives a
data refinement from \ensuremath{\Conid{T}} to \ensuremath{\Conid{StateT}\;\Conid{S}\;\Conid{T}}, and so we lose nothing by
using \ensuremath{\Conid{StateT}\;\Conid{S}\;\Conid{T}} instead of~\ensuremath{\Conid{T}}.
Despite this, we do not impose stability as a requirement in the following,
because some interesting \bx{} are not stable.
\end{remark}

We have not found convincing examples of \ensuremath{\Conid{StateTBX}}
in which the two get functions have effects from \ensuremath{\Conid{T}}, rather than being
\ensuremath{\Conid{T}}-pure queries.
When the get functions are \ensuremath{\Conid{T}}-pure queries,
we obtain the get/get commutation laws 
   \ensuremath{\mathrm{(G_LG_L)}}, \ensuremath{\mathrm{(G_RG_R)}}, \ensuremath{\mathrm{(G_LG_R)}} for free~\cite{Gibbons&Hinze11:Just}, 
   motivating the following: 

\begin{definition}\label{def:proper}
We say that a well-behaved
\ensuremath{bx\mathbin{::}\Conid{StateTBX}\;\Conid{T}\;\Conid{S}\;\Conid{A}\;\Conid{B}} in the monad \ensuremath{\Conid{StateT}\;\Conid{S}\;\Conid{T}} 
is \emph{transparent} if \ensuremath{\get{L}}, \ensuremath{\get{R}} are \ensuremath{\Conid{T}}-pure queries, 
\ie\ there exist \ensuremath{\Varid{read}_{L}\mathbin{::}\Conid{S}\hsarrow{\rightarrow }{\mathpunct{.}}\Conid{A}} and \ensuremath{\Varid{read}_{R}\mathbin{::}\Conid{S}\hsarrow{\rightarrow }{\mathpunct{.}}\Conid{B}} such that
\ensuremath{bx\mathord{.}\get{L}\mathrel{=}\Varid{gets}\;\Varid{read}_{L}} and \ensuremath{bx\mathord{.}\get{R}\mathrel{=}\Varid{gets}\;\Varid{read}_{R}}.
\end{definition}
\begin{remark}\label{rem:proper-unique}
Under the mild condition 
(Moggi's \emph{monomorphism condition}~\cite{moggi}) 
on \ensuremath{\Conid{T}} that \ensuremath{\Varid{return}} be injective, 
\ensuremath{\Varid{read}_{L}} and \ensuremath{\Varid{read}_{R}} are \emph{uniquely} determined for a transparent \bx; so
informally, we refer to \ensuremath{bx\mathord{.}\Varid{read}_{L}} and \ensuremath{bx\mathord{.}\Varid{read}_{R}}, regarding them as
part of the signature of \ensuremath{bx}.
The monomorphism condition holds for the various monads we consider here
(provided we have non-empty types \ensuremath{\sigma} for \ensuremath{\Conid{State}}, \ensuremath{\Conid{Reader}}, \ensuremath{\Conid{Writer}}).
\end{remark}

Now, transparent \ensuremath{\Conid{StateTBX}}
compose (\Section~\ref{sec:composition}), while general \bx{} with effectful gets do not. 
So, in what follows, we confine our attention to transparent \bx{}. 

\subsection{Subsuming lenses}

Asymmetric lenses, as in \Definition~\ref{def:lens}, are subsumed by
\ensuremath{\Conid{StateTBX}}. To simulate \ensuremath{\Varid{l}\mathbin{::}\Conid{Lens}\;\Conid{A}\;\Conid{B}}, one uses a \ensuremath{\Conid{StateTBX}} on
state \ensuremath{\Conid{A}} and underlying monad \ensuremath{\Conid{Id}}:
\begin{hscode}\SaveRestoreHook
\column{B}{@{}>{\hspre}l<{\hspost}@{}}%
\column{3}{@{}>{\hspre}l<{\hspost}@{}}%
\column{12}{@{}>{\hspre}l<{\hspost}@{}}%
\column{E}{@{}>{\hspre}l<{\hspost}@{}}%
\>[B]{}\Conid{BX}\;\Varid{get}\;\Varid{set}\;\get{R}\;\set{R}\;\mathbf{where}{}\<[E]%
\\
\>[B]{}\hsindent{3}{}\<[3]%
\>[3]{}\get{R}{}\<[12]%
\>[12]{}\mathrel{=}\mathbf{do}\;\{\mskip1.5mu \Varid{a}\leftarrow \Varid{get};\Varid{return}\;(\Varid{l}\mathord{.}\Varid{view}\;\Varid{a})\mskip1.5mu\}{}\<[E]%
\\
\>[B]{}\hsindent{3}{}\<[3]%
\>[3]{}\set{R}\;\Varid{b'}{}\<[12]%
\>[12]{}\mathrel{=}\mathbf{do}\;\{\mskip1.5mu \Varid{a}\leftarrow \Varid{get};\Varid{set}\;(\Varid{l}\mathord{.}\Varid{update}\;\Varid{a}\;\Varid{b'})\mskip1.5mu\}{}\<[E]%
\ColumnHook
\end{hscode}\resethooks
Symmetric lenses, as in \Definition~\ref{def:symlens}, are subsumed by our
effectful \bx{} too. In a nutshell, to simulate \ensuremath{\Varid{sl}\mathbin{::}\Conid{SLens}\;\Conid{C}\;\Conid{A}\;\Conid{B}}
one uses \ensuremath{\Conid{StateTBX}\;\Conid{Id}\;\Conid{S}} where \ensuremath{\Conid{S}\subseteq\Conid{A}\times\Conid{B}\times\Conid{C}} is the set of
`consistent triples' \ensuremath{(\Varid{a},\Varid{b},\Varid{c})}, in the sense that \ensuremath{\Varid{sl}\mathord{.}\Varid{put}_{R}\;(\Varid{a},\Varid{c})\mathrel{=}(\Varid{b},\Varid{c})} and \ensuremath{\Varid{sl}\mathord{.}\Varid{put}_{L}\;(\Varid{b},\Varid{c})\mathrel{=}(\Varid{a},\Varid{c})}:
\begin{hscode}\SaveRestoreHook
\column{B}{@{}>{\hspre}l<{\hspost}@{}}%
\column{3}{@{}>{\hspre}l<{\hspost}@{}}%
\column{12}{@{}>{\hspre}l<{\hspost}@{}}%
\column{E}{@{}>{\hspre}l<{\hspost}@{}}%
\>[B]{}\Conid{BX}\;\get{L}\;\set{L}\;\get{L}\;\set{R}\;\mathbf{where}{}\<[E]%
\\
\>[B]{}\hsindent{3}{}\<[3]%
\>[3]{}\get{L}{}\<[12]%
\>[12]{}\mathrel{=}\mathbf{do}\;\{\mskip1.5mu (\Varid{a},\Varid{b},\Varid{c})\leftarrow \Varid{get};\Varid{return}\;\Varid{a}\mskip1.5mu\}{}\<[E]%
\\
\>[B]{}\hsindent{3}{}\<[3]%
\>[3]{}\get{R}{}\<[12]%
\>[12]{}\mathrel{=}\mathbf{do}\;\{\mskip1.5mu (\Varid{a},\Varid{b},\Varid{c})\leftarrow \Varid{get};\Varid{return}\;\Varid{b}\mskip1.5mu\}{}\<[E]%
\\
\>[B]{}\hsindent{3}{}\<[3]%
\>[3]{}\set{L}\;\Varid{a'}{}\<[12]%
\>[12]{}\mathrel{=}\mathbf{do}\;\{\mskip1.5mu (\Varid{a},\Varid{b},\Varid{c})\leftarrow \Varid{get};\mathbf{let}\;(\Varid{b'},\Varid{c'})\mathrel{=}\Varid{sl}\mathord{.}\Varid{put}_{R}\;(\Varid{a},\Varid{c});\Varid{set}\;(\Varid{a'},\Varid{b'},\Varid{c'})\mskip1.5mu\}{}\<[E]%
\\
\>[B]{}\hsindent{3}{}\<[3]%
\>[3]{}\set{R}\;\Varid{b'}{}\<[12]%
\>[12]{}\mathrel{=}\mathbf{do}\;\{\mskip1.5mu (\Varid{a},\Varid{b},\Varid{c})\leftarrow \Varid{get};\mathbf{let}\;(\Varid{a'},\Varid{c'})\mathrel{=}\Varid{sl}\mathord{.}\Varid{put}_{L}\;(\Varid{b},\Varid{c});\Varid{set}\;(\Varid{a'},\Varid{b'},\Varid{c'})\mskip1.5mu\}{}\<[E]%
\ColumnHook
\end{hscode}\resethooks

Asymmetric lenses generalise straightforwardly to accommodate 
effects in an underlying monad too. One can define
\begin{hscode}\SaveRestoreHook
\column{B}{@{}>{\hspre}l<{\hspost}@{}}%
\column{32}{@{}>{\hspre}l<{\hspost}@{}}%
\column{41}{@{}>{\hspre}l<{\hspost}@{}}%
\column{E}{@{}>{\hspre}l<{\hspost}@{}}%
\>[B]{}\mathbf{data}\;\Conid{MLens}\;\tau\;\alpha\;\beta\mathrel{=}\Conid{MLens}\;\{\mskip1.5mu {}\<[32]%
\>[32]{}\Varid{mview}{}\<[41]%
\>[41]{}\mathbin{::}\alpha\hsarrow{\rightarrow }{\mathpunct{.}}\beta,{}\<[E]%
\\
\>[32]{}\Varid{mupdate}{}\<[41]%
\>[41]{}\mathbin{::}\alpha\hsarrow{\rightarrow }{\mathpunct{.}}\beta\hsarrow{\rightarrow }{\mathpunct{.}}\tau\;\alpha\mskip1.5mu\}{}\<[E]%
\ColumnHook
\end{hscode}\resethooks
with corresponding notions of well-behaved and very-well-behaved
monadic lens.  (Divi\'anszky~\cite{reddit} and Pacheco et
al.~\cite{pacheco14pepm}, among others, have proposed similar
notions.)  However, it turns out not to be straightforward to
establish a corresponding notion of `monadic symmetric lens'
incorporating other effects.  In this paper, we take a different
approach to combining symmetry and effects; we defer further
discussion of the different approaches to \ensuremath{\Conid{MLens}}es and the
complications involved in extending symmetric lenses with effects to a
future paper.

\section{Composition}\label{sec:composition}

An obviously crucial question is whether well-behaved monadic \bx{}
compose. They do, but the issue is more delicate than might at first
be expected. Of course, we cannot expect arbitrary \ensuremath{\Conid{BX}} to compose,
because arbitrary monads do not.  Here, we present one successful
approach for \ensuremath{\Conid{StateTBX}}, based on lifting the component operations on
different state types (but the same underlying monad of effects) into a
common compound state.

\begin{definition}[\ensuremath{\Conid{StateT}} embeddings from lenses]\label{def:theta}
  Given a lens from \ensuremath{\Conid{A}} to \ensuremath{\Conid{B}}, we can embed a \ensuremath{\Conid{StateT}}
  computation on the narrower type \ensuremath{\Conid{B}} into another computation
  on the wider type \ensuremath{\Conid{A}}, wrt the same underlying
  monad~\ensuremath{\Conid{T}}:
\begin{hscode}\SaveRestoreHook
\column{B}{@{}>{\hspre}l<{\hspost}@{}}%
\column{20}{@{}>{\hspre}l<{\hspost}@{}}%
\column{E}{@{}>{\hspre}l<{\hspost}@{}}%
\>[B]{}\vartheta \mathbin{::}\Conid{Monad}\;\tau\Rightarrow \Conid{Lens}\;\alpha\;\beta\hsarrow{\rightarrow }{\mathpunct{.}}\Conid{StateT}\;\beta\;\tau\;\gamma\hsarrow{\rightarrow }{\mathpunct{.}}\Conid{StateT}\;\alpha\;\tau\;\gamma{}\<[E]%
\\
\>[B]{}\vartheta \;\Varid{l}\;\Varid{m}\mathrel{=}\mathbf{do}\;{}\<[20]%
\>[20]{}\Varid{a}\leftarrow \Varid{get};\mathbf{let}\;\Varid{b}\mathrel{=}\Varid{l}\mathord{.}\Varid{view}\;\Varid{a};{}\<[E]%
\\
\>[20]{}(\Varid{c},\Varid{b'})\leftarrow \Varid{lift}\;(\Varid{m}\;\Varid{b});{}\<[E]%
\\
\>[20]{}\mathbf{let}\;\Varid{a'}\mathrel{=}\Varid{l}\mathord{.}\Varid{update}\;\Varid{a}\;\Varid{b'};{}\<[E]%
\\
\>[20]{}\Varid{set}\;\Varid{a'};\Varid{return}\;\Varid{c}{}\<[E]%
\ColumnHook
\end{hscode}\resethooks
\endswithdisplay
\end{definition}
Essentially, \ensuremath{\vartheta \;\Varid{l}\;\Varid{m}} uses \ensuremath{\Varid{l}} to get a view \ensuremath{\Varid{b}} of the source
\ensuremath{\Varid{a}}, runs \ensuremath{\Varid{m}} to get a return value \ensuremath{\Varid{c}} and updated view \ensuremath{\Varid{b'}}, uses \ensuremath{\Varid{l}}
to update the view yielding an updated source \ensuremath{\Varid{a'}}, and returns \ensuremath{\Varid{c}}.

\begin{lemma}\label{lem:vwb-monad-morphism}
  If \ensuremath{\Varid{l}\mathbin{::}\Conid{Lens}\;\Conid{A}\;\Conid{B}} is very well-behaved,
  then \ensuremath{\vartheta \;\Varid{l}} is a monad
  morphism.  
\end{lemma}

\begin{definition}\label{def:left-right-monad-morphism}
By Lemma~\ref{lem:vwb-monad-morphism}, 
and since \ensuremath{\Varid{fstLens}} and \ensuremath{\Varid{sndLens}} are very well-behaved,
we have
the following monad morphisms lifting stateful computations to a product state space:
\begin{hscode}\SaveRestoreHook
\column{B}{@{}>{\hspre}l<{\hspost}@{}}%
\column{8}{@{}>{\hspre}c<{\hspost}@{}}%
\column{8E}{@{}l@{}}%
\column{12}{@{}>{\hspre}l<{\hspost}@{}}%
\column{E}{@{}>{\hspre}l<{\hspost}@{}}%
\>[B]{}\Varid{left}{}\<[8]%
\>[8]{}\mathbin{::}{}\<[8E]%
\>[12]{}\Conid{Monad}\;\tau\Rightarrow \Conid{StateT}\;\sigma_1\;\tau\;\alpha\hsarrow{\rightarrow }{\mathpunct{.}}\Conid{StateT}\;(\sigma_1,\sigma_2)\;\tau\;\alpha{}\<[E]%
\\
\>[B]{}\Varid{left}{}\<[8]%
\>[8]{}\mathrel{=}{}\<[8E]%
\>[12]{}\vartheta \;\Varid{fstLens}{}\<[E]%
\\[\blanklineskip]%
\>[B]{}\Varid{right}{}\<[8]%
\>[8]{}\mathbin{::}{}\<[8E]%
\>[12]{}\Conid{Monad}\;\tau\Rightarrow \Conid{StateT}\;\sigma_2\;\tau\;\alpha\hsarrow{\rightarrow }{\mathpunct{.}}\Conid{StateT}\;(\sigma_1,\sigma_2)\;\tau\;\alpha{}\<[E]%
\\
\>[B]{}\Varid{right}{}\<[8]%
\>[8]{}\mathrel{=}{}\<[8E]%
\>[12]{}\vartheta \;\Varid{sndLens}{}\<[E]%
\ColumnHook
\end{hscode}\resethooks
\endswithdisplay
\end{definition}

\begin{definition*}
For
\ensuremath{bx_{1}\mathbin{::}\Conid{StateTBX}\;\Conid{T}\;\Conid{S}_{1}\;\Conid{A}\;\Conid{B}}, 
\ensuremath{bx_{2}\mathbin{::}\Conid{StateTBX}\;\Conid{T}\;\Conid{S}_{2}\;\Conid{B}\;\Conid{C}}, 
define the join \ensuremath{\Conid{S}_{1} \mathrel{{}_{\!bx_{1}}{\!\!\Join}_{bx_{2}\!}} \Conid{S}_{2}} informally as the
subset of \ensuremath{\Conid{S}_{1}\times\Conid{S}_{2}} consisting of the
pairs \ensuremath{(\Varid{s}_{1},\Varid{s}_{2})} in which observing the middle
component of type~\ensuremath{\Conid{B}} in state \ensuremath{\Varid{s}_{1}} yields the same 
result as in state 
\ensuremath{\Varid{s}_{2}}. We might express this set-theoretically as follows:
\begin{hscode}\SaveRestoreHook
\column{B}{@{}>{\hspre}l<{\hspost}@{}}%
\column{25}{@{}>{\hspre}c<{\hspost}@{}}%
\column{25E}{@{}l@{}}%
\column{28}{@{}>{\hspre}l<{\hspost}@{}}%
\column{37}{@{}>{\hspre}c<{\hspost}@{}}%
\column{37E}{@{}l@{}}%
\column{40}{@{}>{\hspre}l<{\hspost}@{}}%
\column{70}{@{}>{\hspre}l<{\hspost}@{}}%
\column{E}{@{}>{\hspre}l<{\hspost}@{}}%
\>[B]{}\Conid{S}_{1} \mathrel{{}_{\!bx_1}{\!\!\Join}_{bx_2\!}} \Conid{S}_{2}\mathrel{=}{}\<[25]%
\>[25]{}\{\mskip1.5mu {}\<[25E]%
\>[28]{}(\Varid{s}_{1},\Varid{s}_{2}){}\<[37]%
\>[37]{}\mid{}\<[37E]%
\>[40]{}\Varid{eval}\;(bx_{1}\mathord{.}\get{R})\;\Varid{s}_{1}\mathrel{=}{}\<[70]%
\>[70]{}\Varid{eval}\;(bx_{2}\mathord{.}\get{L})\;\Varid{s}_{2}\mskip1.5mu\}{}\<[E]%
\ColumnHook
\end{hscode}\resethooks
More generally, one could state a categorical definition in terms of pullbacks.
In Haskell, we can only work with the coarser type of raw pairs \ensuremath{(\Conid{S}_{1},\Conid{S}_{2})}.
Note that the equation in the set comprehension compares two computations of type \ensuremath{\Conid{T}\;\Conid{B}};
but if the \bx{} are transparent, and \ensuremath{\Varid{return}} injective as per \Remark~\ref{rem:proper-unique}, 
then the definition simplifies to: 
\begin{hscode}\SaveRestoreHook
\column{B}{@{}>{\hspre}l<{\hspost}@{}}%
\column{25}{@{}>{\hspre}c<{\hspost}@{}}%
\column{25E}{@{}l@{}}%
\column{28}{@{}>{\hspre}l<{\hspost}@{}}%
\column{37}{@{}>{\hspre}c<{\hspost}@{}}%
\column{37E}{@{}l@{}}%
\column{40}{@{}>{\hspre}l<{\hspost}@{}}%
\column{56}{@{}>{\hspre}c<{\hspost}@{}}%
\column{56E}{@{}l@{}}%
\column{60}{@{}>{\hspre}l<{\hspost}@{}}%
\column{E}{@{}>{\hspre}l<{\hspost}@{}}%
\>[B]{}\Conid{S}_{1} \mathrel{{}_{\!bx_1}{\!\!\Join}_{bx_2\!}} \Conid{S}_{2}\mathrel{=}{}\<[25]%
\>[25]{}\{\mskip1.5mu {}\<[25E]%
\>[28]{}(\Varid{s}_{1},\Varid{s}_{2}){}\<[37]%
\>[37]{}\mid{}\<[37E]%
\>[40]{}bx_{1}\mathord{.}\Varid{read}_{R}\;\Varid{s}_{1}{}\<[56]%
\>[56]{}\mathrel{=}{}\<[56E]%
\>[60]{}bx_{2}\mathord{.}\Varid{read}_{L}\;\Varid{s}_{2}\mskip1.5mu\}{}\<[E]%
\ColumnHook
\end{hscode}\resethooks
The notation \ensuremath{\Conid{S}_{1} \mathrel{{}_{\!bx_{1}}{\!\!\Join}_{bx_{2}\!}} \Conid{S}_{2}} explicitly mentions \ensuremath{bx_{1}} and \ensuremath{bx_{2}}, 
but we usually just write \ensuremath{\Conid{S}_{1}\mathbin{\!\Join\!}\Conid{S}_{2}}.
No confusion should arise from using the same symbol to denote the consistent pairs of a single \bx{}, as we did in Remark~\ref{rem:consistent}.
\end{definition*}

\begin{definition}
Using \ensuremath{\Varid{left}} and \ensuremath{\Varid{right}}, we can define composition by:
\begin{hscode}\SaveRestoreHook
\column{B}{@{}>{\hspre}l<{\hspost}@{}}%
\column{8}{@{}>{\hspre}l<{\hspost}@{}}%
\column{12}{@{}>{\hspre}l<{\hspost}@{}}%
\column{16}{@{}>{\hspre}l<{\hspost}@{}}%
\column{24}{@{}>{\hspre}l<{\hspost}@{}}%
\column{E}{@{}>{\hspre}l<{\hspost}@{}}%
\>[B]{}( \mathbin{\fatsemi} )\mathbin{::}{}\<[12]%
\>[12]{}\Conid{Monad}\;\tau\Rightarrow {}\<[E]%
\\
\>[12]{}\Conid{StateTBX}\;\sigma_1\;\tau\;\alpha\;\beta\hsarrow{\rightarrow }{\mathpunct{.}}\Conid{StateTBX}\;\sigma_2\;\tau\;\beta\;\gamma\hsarrow{\rightarrow }{\mathpunct{.}}\Conid{StateTBX}\;(\sigma_1\mathbin{\!\Join\!}\sigma_2)\;\tau\;\alpha\;\gamma{}\<[E]%
\\
\>[B]{}bx_{1} \mathbin{\fatsemi} bx_{2}\mathrel{=}\Conid{BX}\;\get{L}\;\set{L}\;\get{R}\;\set{R}\;\mathbf{where}{}\<[E]%
\\
\>[B]{}\hsindent{8}{}\<[8]%
\>[8]{}\get{L}{}\<[16]%
\>[16]{}\mathrel{=}\mathbf{do}\;\{\mskip1.5mu {}\<[24]%
\>[24]{}\Varid{left}\;(bx_{1}\mathord{.}\get{L})\mskip1.5mu\}{}\<[E]%
\\
\>[B]{}\hsindent{8}{}\<[8]%
\>[8]{}\get{R}{}\<[16]%
\>[16]{}\mathrel{=}\mathbf{do}\;\{\mskip1.5mu {}\<[24]%
\>[24]{}\Varid{right}\;(bx_{2}\mathord{.}\get{R})\mskip1.5mu\}{}\<[E]%
\\
\>[B]{}\hsindent{8}{}\<[8]%
\>[8]{}\set{L}\;\Varid{a}{}\<[16]%
\>[16]{}\mathrel{=}\mathbf{do}\;\{\mskip1.5mu {}\<[24]%
\>[24]{}\Varid{left}\;(bx_{1}\mathord{.}\set{L}\;\Varid{a});\Varid{b}\leftarrow \Varid{left}\;(bx_{1}\mathord{.}\get{R});\Varid{right}\;(bx_{2}\mathord{.}\set{L}\;\Varid{b})\mskip1.5mu\}{}\<[E]%
\\
\>[B]{}\hsindent{8}{}\<[8]%
\>[8]{}\set{R}\;\Varid{c}{}\<[16]%
\>[16]{}\mathrel{=}\mathbf{do}\;\{\mskip1.5mu {}\<[24]%
\>[24]{}\Varid{right}\;(bx_{2}\mathord{.}\set{R}\;\Varid{c});\Varid{b}\leftarrow \Varid{right}\;(bx_{2}\mathord{.}\get{L});\Varid{left}\;(bx_{1}\mathord{.}\set{R}\;\Varid{b})\mskip1.5mu\}{}\<[E]%
\ColumnHook
\end{hscode}\resethooks
Essentially,  to set the left-hand side of
the composed \bx, we first set the left-hand side of the left component
\ensuremath{bx_{1}}, then get \ensuremath{bx_{1}}'s \ensuremath{\Conid{R}}-value \ensuremath{\Varid{b}}, and set the left-hand side of \ensuremath{bx_{2}}
to this value; and similarly on the right.
Note that the composition maintains the invariant that the compound state
is in the subset \ensuremath{\sigma_1\mathbin{\!\Join\!}\sigma_2} of \ensuremath{\sigma_1\times\sigma_2}.
\end{definition}

\begin{theorem}[transparent composition]\label{thm:composition-wb}
    Given transparent (and hence well-behaved)
    \ensuremath{bx_{1}\mathbin{::}\Conid{StateTBX}\;\Conid{S}_{1}\;\Conid{T}\;\Conid{A}\;\Conid{B}} and \ensuremath{bx_{2}\mathbin{::}\Conid{StateTBX}\;\Conid{S}_{2}\;\Conid{T}\;\Conid{B}\;\Conid{C}}, 
    their composition \ensuremath{bx_{1} \mathbin{\fatsemi} bx_{2}\mathbin{::}\Conid{StateTBX}\;(\Conid{S}_{1}\mathbin{\!\Join\!}\Conid{S}_{2})\;\Conid{T}\;\Conid{A}\;\Conid{C}} 
    is also transparent.
\end{theorem}

\begin{remark*}
Unpacking and simplifying the definitions, we have:
\begin{hscode}\SaveRestoreHook
\column{B}{@{}>{\hspre}l<{\hspost}@{}}%
\column{4}{@{}>{\hspre}l<{\hspost}@{}}%
\column{13}{@{}>{\hspre}l<{\hspost}@{}}%
\column{21}{@{}>{\hspre}l<{\hspost}@{}}%
\column{E}{@{}>{\hspre}l<{\hspost}@{}}%
\>[B]{}bx_{1} \mathbin{\fatsemi} bx_{2}\mathrel{=}\Conid{BX}\;\get{L}\;\set{L}\;\get{R}\;\set{R}\;\mathbf{where}{}\<[E]%
\\
\>[B]{}\hsindent{4}{}\<[4]%
\>[4]{}\get{L}{}\<[13]%
\>[13]{}\mathrel{=}\mathbf{do}\;\{\mskip1.5mu {}\<[21]%
\>[21]{}(\Varid{s}_{1},\anonymous )\leftarrow \Varid{get};\Varid{return}\;(bx_{1}\mathord{.}\Varid{read}_{L}\;\Varid{s}_{1})\mskip1.5mu\}{}\<[E]%
\\
\>[B]{}\hsindent{4}{}\<[4]%
\>[4]{}\get{R}{}\<[13]%
\>[13]{}\mathrel{=}\mathbf{do}\;\{\mskip1.5mu {}\<[21]%
\>[21]{}(\anonymous ,\Varid{s}_{2})\leftarrow \Varid{get};\Varid{return}\;(bx_{2}\mathord{.}\Varid{read}_{R}\;\Varid{s}_{2})\mskip1.5mu\}{}\<[E]%
\\
\>[B]{}\hsindent{4}{}\<[4]%
\>[4]{}\set{L}\;\Varid{a'}{}\<[13]%
\>[13]{}\mathrel{=}\mathbf{do}\;\{\mskip1.5mu {}\<[21]%
\>[21]{}(\Varid{s}_{1},\Varid{s}_{2})\leftarrow \Varid{get};{}\<[E]%
\\
\>[21]{}((),\Varid{s}_{1}')\leftarrow \Varid{lift}\;(bx_{1}\mathord{.}\set{L}\;\Varid{a'}\;\Varid{s}_{1});{}\<[E]%
\\
\>[21]{}\mathbf{let}\;\Varid{b}\mathrel{=}bx_{1}\mathord{.}\Varid{read}_{R}\;\Varid{s}_{1}';{}\<[E]%
\\
\>[21]{}((),\Varid{s}_{2}')\leftarrow \Varid{lift}\;(bx_{2}\mathord{.}\set{L}\;\Varid{b}\;\Varid{s}_{2});{}\<[E]%
\\
\>[21]{}\Varid{set}\;(\Varid{s}_{1}',\Varid{s}_{2}')\mskip1.5mu\}{}\<[E]%
\\
\>[B]{}\hsindent{4}{}\<[4]%
\>[4]{}\set{R}\;\Varid{c'}{}\<[13]%
\>[13]{}\mathrel{=}\mathbf{do}\;\{\mskip1.5mu {}\<[21]%
\>[21]{}(\Varid{s}_{1},\Varid{s}_{2})\leftarrow \Varid{get};{}\<[E]%
\\
\>[21]{}((),\Varid{s}_{2}')\leftarrow \Varid{lift}\;(bx_{2}\mathord{.}\set{R}\;\Varid{c'}\;\Varid{s}_{2});{}\<[E]%
\\
\>[21]{}\mathbf{let}\;\Varid{b}\mathrel{=}bx_{2}\mathord{.}\Varid{read}_{L}\;\Varid{s}_{2}';{}\<[E]%
\\
\>[21]{}((),\Varid{s}_{1}')\leftarrow \Varid{lift}\;(bx_{1}\mathord{.}\set{R}\;\Varid{b}\;\Varid{s}_{1});{}\<[E]%
\\
\>[21]{}\Varid{set}\;(\Varid{s}_{1}',\Varid{s}_{2}')\mskip1.5mu\}{}\<[E]%
\ColumnHook
\end{hscode}\resethooks
\endswithdisplay
\end{remark*}

\begin{remark}\label{rem:effectful-gets}
  Allowing effectful \ensuremath{\Varid{get}}s turns out to impose appreciable extra technical
  difficulty. In particular, while it still appears possible to prove that
  composition preserves well-behavedness, the identity laws 
  of composition do not appear
  to hold in general. At the same time, we currently lack compelling
  examples that motivate effectful \ensuremath{\Varid{get}}s; the only example we have
  considered that requires this capability is \Example~\ref{ex:switch} in
  \Section~\ref{sec:examples}. 
  This is why we mostly limit attention to transparent bx.
\end{remark}

Composition is usually expected to be associative and to satisfy
identity laws.  We can define a family of identity \bx{} as follows:
\begin{definition}[identity] \label{def:identity}
For any underlying monad instance, 
  we can form the \emph{identity} bx as follows:
\begin{hscode}\SaveRestoreHook
\column{B}{@{}>{\hspre}l<{\hspost}@{}}%
\column{E}{@{}>{\hspre}l<{\hspost}@{}}%
\>[B]{}\Varid{identity}\mathbin{::}\Conid{Monad}\;\tau\Rightarrow \Conid{StateTBX}\;\tau\;\alpha\;\alpha\;\alpha{}\<[E]%
\\
\>[B]{}\Varid{identity}\mathrel{=}\Conid{BX}\;\Varid{get}\;\Varid{set}\;\Varid{get}\;\Varid{set}{}\<[E]%
\ColumnHook
\end{hscode}\resethooks
Clearly, this \bx{} 
is well-behaved, indeed transparent, and overwritable.
\end{definition}
However, if we ask whether \ensuremath{bx\mathrel{=}\Varid{identity} \mathbin{\fatsemi} bx}, we are immediately
faced with a problem: the two \bx{} do not even have the same state types.
Similarly when we ask about associativity of composition.
Apparently, therefore, 
as for symmetric lenses~\cite{symlens}, 
we must satisfy ourselves with 
equality only up to some notion of
equivalence, such as the one introduced in Definition~\ref{def:equivalence}.

\begin{theorem}\label{thm:category}
Composition of transparent \bx{} satisfies the identity and associativity laws,
modulo \ensuremath{\equiv }.
\begin{hscode}\SaveRestoreHook
\column{B}{@{}>{\hspre}l<{\hspost}@{}}%
\column{3}{@{}>{\hspre}l<{\hspost}@{}}%
\column{15}{@{}>{\hspre}c<{\hspost}@{}}%
\column{15E}{@{}l@{}}%
\column{20}{@{}>{\hspre}l<{\hspost}@{}}%
\column{46}{@{}>{\hspre}l<{\hspost}@{}}%
\column{E}{@{}>{\hspre}l<{\hspost}@{}}%
\>[3]{}\mathrm{(Identity)}{}\<[15]%
\>[15]{}\quad{}\<[15E]%
\>[20]{}\Varid{identity} \mathbin{\fatsemi} bx\equiv bx\equiv bx \mathbin{\fatsemi} \Varid{identity}{}\<[E]%
\\
\>[3]{}\mathrm{(Assoc)}{}\<[15]%
\>[15]{}\quad{}\<[15E]%
\>[20]{}bx_{1} \mathbin{\fatsemi} (bx_{2} \mathbin{\fatsemi} bx_{3}){}\<[46]%
\>[46]{}\equiv (bx_{1} \mathbin{\fatsemi} bx_{2}) \mathbin{\fatsemi} bx_{3}{}\<[E]%
\ColumnHook
\end{hscode}\resethooks
\endswithdisplay
\end{theorem}
\section{Examples}\label{sec:examples}

We now show how to use and combine \bx{}, and discuss how to extend
our approach to support initialisation. We adapt some standard
constructions on symmetric lenses, involving pairs, sums and lists. Finally
we investigate some effectful \bx{} primitives and combinators, culminating
with the two examples from \Section~\ref{sec:introduction}.

\subsection{Initialisation}
\label{sec:initialization}

Readers familiar with bx will have noticed that so far we have not
mentioned mechanisms for initialisation, e.g. `create' for asymmetric
lenses~\cite{lens-toplas}, `missing' in symmetric
lenses~\cite{symlens}, or $\Omega$ in relational bx
terminology~\cite{stevens09:sosym}. Moreover, as we shall see in 
\Section~\ref{sec:combinators}, initialisation is
also needed for certain combinators.

\begin{definition*}
An \emph{initialisable \ensuremath{\Conid{StateTBX}}} is a \ensuremath{\Conid{StateTBX}} with 
two additional operations for initialisation:
\begin{hscode}\SaveRestoreHook
\column{B}{@{}>{\hspre}l<{\hspost}@{}}%
\column{3}{@{}>{\hspre}l<{\hspost}@{}}%
\column{11}{@{}>{\hspre}l<{\hspost}@{}}%
\column{34}{@{}>{\hspre}l<{\hspost}@{}}%
\column{44}{@{}>{\hspre}l<{\hspost}@{}}%
\column{71}{@{}>{\hspre}l<{\hspost}@{}}%
\column{82}{@{}>{\hspre}l<{\hspost}@{}}%
\column{98}{@{}>{\hspre}c<{\hspost}@{}}%
\column{98E}{@{}l@{}}%
\column{E}{@{}>{\hspre}l<{\hspost}@{}}%
\>[B]{}\mathbf{data}\;\Conid{InitStateTBX}\;\tau\;\sigma\;\alpha\;\beta\mathrel{=}\Conid{InitStateTBX}\;\{\mskip1.5mu {}\<[E]%
\\
\>[B]{}\hsindent{3}{}\<[3]%
\>[3]{}\get{L}{}\<[11]%
\>[11]{}\mathbin{::}\Conid{StateT}\;\sigma\;\tau\;\alpha,{}\<[34]%
\>[34]{}\quad\set{L}{}\<[44]%
\>[44]{}\mathbin{::}\alpha\hsarrow{\rightarrow }{\mathpunct{.}}\Conid{StateT}\;\sigma\;\tau\;(),{}\<[71]%
\>[71]{}\quad\Varid{init}_{L}{}\<[82]%
\>[82]{}\mathbin{::}\alpha\hsarrow{\rightarrow }{\mathpunct{.}}\tau\;\sigma,{}\<[E]%
\\
\>[B]{}\hsindent{3}{}\<[3]%
\>[3]{}\get{R}{}\<[11]%
\>[11]{}\mathbin{::}\Conid{StateT}\;\sigma\;\tau\;\beta,{}\<[34]%
\>[34]{}\quad\set{R}{}\<[44]%
\>[44]{}\mathbin{::}\beta\hsarrow{\rightarrow }{\mathpunct{.}}\Conid{StateT}\;\sigma\;\tau\;(),{}\<[71]%
\>[71]{}\quad\Varid{init}_{R}{}\<[82]%
\>[82]{}\mathbin{::}\beta\hsarrow{\rightarrow }{\mathpunct{.}}\tau\;\sigma{}\<[98]%
\>[98]{}\mskip1.5mu\}{}\<[98E]%
\ColumnHook
\end{hscode}\resethooks
The \ensuremath{\Varid{init}_{L}} and \ensuremath{\Varid{init}_{R}} operations build an initial state from 
one view or the other, possibly incurring effects in
the underlying monad. Well-behavedness of the bx requires in addition:
\begin{hscode}\SaveRestoreHook
\column{B}{@{}>{\hspre}l<{\hspost}@{}}%
\column{12}{@{}>{\hspre}c<{\hspost}@{}}%
\column{12E}{@{}l@{}}%
\column{16}{@{}>{\hspre}l<{\hspost}@{}}%
\column{18}{@{}>{\hspre}l<{\hspost}@{}}%
\column{E}{@{}>{\hspre}l<{\hspost}@{}}%
\>[B]{}\mathrm{(I_LG_L)}{}\<[12]%
\>[12]{}\quad{}\<[12E]%
\>[16]{}\mathbf{do}\;\{\mskip1.5mu \Varid{s}\leftarrow bx\mathord{.}\Varid{init}_{L}\;\Varid{a};bx\mathord{.}\get{L}\;\Varid{s}\mskip1.5mu\}{}\<[E]%
\\
\>[16]{}\hsindent{2}{}\<[18]%
\>[18]{}\mathrel{=}\mathbf{do}\;\{\mskip1.5mu \Varid{s}\leftarrow bx\mathord{.}\Varid{init}_{L}\;\Varid{a};\Varid{return}\;(\Varid{a},\Varid{s})\mskip1.5mu\}{}\<[E]%
\\
\>[B]{}\mathrm{(I_RG_R)}{}\<[12]%
\>[12]{}\quad{}\<[12E]%
\>[16]{}\mathbf{do}\;\{\mskip1.5mu \Varid{s}\leftarrow bx\mathord{.}\Varid{init}_{R}\;\Varid{b};bx\mathord{.}\get{R}\;\Varid{s}\mskip1.5mu\}{}\<[E]%
\\
\>[16]{}\hsindent{2}{}\<[18]%
\>[18]{}\mathrel{=}\mathbf{do}\;\{\mskip1.5mu \Varid{s}\leftarrow bx\mathord{.}\Varid{init}_{R}\;\Varid{b};\Varid{return}\;(\Varid{b},\Varid{s})\mskip1.5mu\}{}\<[E]%
\ColumnHook
\end{hscode}\resethooks
stating informally that initialising then getting yields the initialised value. 
There are no laws concerning initialising then setting.
\end{definition*}
We can extend composition to handle initialisation as follows:
\begin{hscode}\SaveRestoreHook
\column{B}{@{}>{\hspre}l<{\hspost}@{}}%
\column{26}{@{}>{\hspre}c<{\hspost}@{}}%
\column{26E}{@{}l@{}}%
\column{29}{@{}>{\hspre}l<{\hspost}@{}}%
\column{35}{@{}>{\hspre}l<{\hspost}@{}}%
\column{E}{@{}>{\hspre}l<{\hspost}@{}}%
\>[B]{}(bx_{1} \mathbin{\fatsemi} bx_{2})\mathord{.}\Varid{init}_{L}\;\Varid{a}{}\<[26]%
\>[26]{}\mathrel{=}{}\<[26E]%
\>[29]{}\mathbf{do}\;\{\mskip1.5mu {}\<[35]%
\>[35]{}\Varid{s}_{1}\leftarrow bx_{1}\mathord{.}\Varid{init}_{L}\;\Varid{a};\Varid{b}\leftarrow bx_{1}\mathord{.}\get{R}\;\Varid{s}_{1};{}\<[E]%
\\
\>[35]{}\Varid{s}_{2}\leftarrow bx_{2}\mathord{.}\Varid{init}_{L}\;\Varid{b};\Varid{return}\;(\Varid{s}_{1},\Varid{s}_{2})\mskip1.5mu\}{}\<[E]%
\ColumnHook
\end{hscode}\resethooks
and symmetrically for \ensuremath{\Varid{init}_{R}}.
We refine the notions of bx isomorphism and equivalence to
\ensuremath{\Conid{InitStateTBX}} as follows.  
A monad isomorphism \ensuremath{\iota \mathbin{::}\Conid{StateT}\;\Conid{S}_{1}\;\Conid{T}\hsarrow{\rightarrow }{\mathpunct{.}}\Conid{StateT}\;\Conid{S}_{2}\;\Conid{T}}
amounts to a bijection \ensuremath{\Varid{h}\mathbin{::}\Conid{S}_{1}\hsarrow{\rightarrow }{\mathpunct{.}}\Conid{S}_{2}} on the state spaces.
An isomorphism of \ensuremath{\Conid{InitStateTBX}}s consists of such an \ensuremath{\iota } and \ensuremath{\Varid{h}} 
satisfying the following equations (and their duals):
\begin{hscode}\SaveRestoreHook
\column{B}{@{}>{\hspre}l<{\hspost}@{}}%
\column{3}{@{}>{\hspre}l<{\hspost}@{}}%
\column{41}{@{}>{\hspre}l<{\hspost}@{}}%
\column{E}{@{}>{\hspre}l<{\hspost}@{}}%
\>[3]{}\iota \;(bx_{1}\mathord{.}\get{L}){}\<[41]%
\>[41]{}\mathrel{=}bx_{2}\mathord{.}\get{L}{}\<[E]%
\\
\>[3]{}\iota \;(bx_{1}\mathord{.}\set{L}\;\Varid{a}){}\<[41]%
\>[41]{}\mathrel{=}bx_{2}\mathord{.}\set{L}\;\Varid{a}{}\<[E]%
\\
\>[3]{}\mathbf{do}\;\{\mskip1.5mu \Varid{s}\leftarrow bx_{1}\mathord{.}\Varid{init}_{L}\;\Varid{a};\Varid{return}\;(\Varid{h}\;\Varid{s})\mskip1.5mu\}{}\<[41]%
\>[41]{}\mathrel{=}bx_{2}\mathord{.}\Varid{init}_{L}\;\Varid{a}{}\<[E]%
\ColumnHook
\end{hscode}\resethooks
Note that the first two equations (and their duals) imply that \ensuremath{\iota }
is a conventional isomorphism between the underlying \bx{} structures
of \ensuremath{bx_{1}} and \ensuremath{bx_{2}} if we ignore the initialisation operations.  The
third equation simply says that \ensuremath{\Varid{h}} maps the state obtained by
initialising \ensuremath{bx_{1}} with \ensuremath{\Varid{a}} to the state obtained by initialising
\ensuremath{bx_{2}} with \ensuremath{\Varid{a}}.  Equivalence of \ensuremath{\Conid{InitStateTBX}}s amounts to the
existence of such an isomorphism.

\begin{remark*}
Of course, there may be situations where these operations are not what
is desired. We might prefer to provide both view values and ask
the \bx{} system to find a suitable hidden state consistent with both at once.
This can be accommodated, by providing a third initialisation function:
\begin{hscode}\SaveRestoreHook
\column{B}{@{}>{\hspre}l<{\hspost}@{}}%
\column{E}{@{}>{\hspre}l<{\hspost}@{}}%
\>[B]{}\Varid{initBoth}\mathbin{::}\alpha\hsarrow{\rightarrow }{\mathpunct{.}}\beta\hsarrow{\rightarrow }{\mathpunct{.}}\tau\;(\Conid{Maybe}\;\sigma){}\<[E]%
\ColumnHook
\end{hscode}\resethooks
However, \ensuremath{\Varid{initBoth}} and \ensuremath{\Varid{init}_{L},\Varid{init}_{R}} are not interdefinable: 
\ensuremath{\Varid{initBoth}} requires both initial values, so is no help in defining a
function that has access only to one; and conversely, given both
initial values, there are in general two different ways to initialise
from one of them (and two more to initialise from one and then set
with the other).  Furthermore, it is not clear how to define 
\ensuremath{\Varid{initBoth}} for the composition of two \bx{} equipped with \ensuremath{\Varid{initBoth}}.
\end{remark*}

\subsection{Basic constructions and combinators}
\label{sec:combinators}

It is obviously desirable -- and essential in the design of any future
bx programming language -- to be able to build up \bx{} from
components using combinators that preserve interesting properties, and therefore avoid
having to prove well-behavedness
from scratch
for each bx.  Symmetric lenses \cite{symlens} admit several
standard constructions, involving constants, duality, pairing, sum
types, and lists.  We  show that these constructions can be
generalised to  \ensuremath{\Conid{StateTBX}}, and establish that they
preserve well-behavedness.  
For most combinators, the initialisation operations are
straightforward; in the interests of brevity, they and obvious duals are 
omitted in what follows.

\begin{definition}[duality] \label{def:dual}
  Trivially, we can dualise any \bx{}:
\begin{hscode}\SaveRestoreHook
\column{B}{@{}>{\hspre}l<{\hspost}@{}}%
\column{E}{@{}>{\hspre}l<{\hspost}@{}}%
\>[B]{}\Varid{dual}\mathbin{::}\Conid{StateTBX}\;\tau\;\sigma\;\alpha\;\beta\hsarrow{\rightarrow }{\mathpunct{.}}\Conid{StateTBX}\;\tau\;\sigma\;\beta\;\alpha{}\<[E]%
\\
\>[B]{}\Varid{dual}\;bx\mathrel{=}\Conid{BX}\;bx\mathord{.}\get{R}\;bx\mathord{.}\set{R}\;bx\mathord{.}\get{L}\;bx\mathord{.}\set{L}{}\<[E]%
\ColumnHook
\end{hscode}\resethooks
This simply exchanges the left and right operations; it preserves
  transparency and overwritability of the underlying \bx{}.
\end{definition}

\begin{definition*}[constant and pair combinators]
  \ensuremath{\Conid{StateTBX}} also admits constant, pairing and projection
  operations: 
\begin{hscode}\SaveRestoreHook
\column{B}{@{}>{\hspre}l<{\hspost}@{}}%
\column{10}{@{}>{\hspre}c<{\hspost}@{}}%
\column{10E}{@{}l@{}}%
\column{16}{@{}>{\hspre}l<{\hspost}@{}}%
\column{E}{@{}>{\hspre}l<{\hspost}@{}}%
\>[B]{}\Varid{constBX}{}\<[10]%
\>[10]{}\mathbin{::}{}\<[10E]%
\>[16]{}\Conid{Monad}\;\tau\Rightarrow \alpha\hsarrow{\rightarrow }{\mathpunct{.}}\Conid{StateTBX}\;\tau\;\alpha\;()\;\alpha{}\<[E]%
\\
\>[B]{}\Varid{fstBX}{}\<[10]%
\>[10]{}\mathbin{::}{}\<[10E]%
\>[16]{}\Conid{Monad}\;\tau\Rightarrow \Conid{StateTBX}\;\tau\;(\alpha,\beta)\;(\alpha,\beta)\;\alpha{}\<[E]%
\\
\>[B]{}\Varid{sndBX}{}\<[10]%
\>[10]{}\mathbin{::}{}\<[10E]%
\>[16]{}\Conid{Monad}\;\tau\Rightarrow \Conid{StateTBX}\;\tau\;(\alpha,\beta)\;(\alpha,\beta)\;\beta{}\<[E]%
\ColumnHook
\end{hscode}\resethooks
These straightforwardly generalise to \bx{} the
corresponding operations for symmetric lenses. 
If they are to be initialisable, \ensuremath{\Varid{fstBX}} and \ensuremath{\Varid{sndBX}} also have to take a parameter for the initial
value of the opposite side:
\begin{hscode}\SaveRestoreHook
\column{B}{@{}>{\hspre}l<{\hspost}@{}}%
\column{11}{@{}>{\hspre}c<{\hspost}@{}}%
\column{11E}{@{}l@{}}%
\column{17}{@{}>{\hspre}l<{\hspost}@{}}%
\column{E}{@{}>{\hspre}l<{\hspost}@{}}%
\>[B]{}\Varid{fstIBX}{}\<[11]%
\>[11]{}\mathbin{::}{}\<[11E]%
\>[17]{}\Conid{Monad}\;\tau\Rightarrow \beta\hsarrow{\rightarrow }{\mathpunct{.}}\Conid{InitStateTBX}\;\tau\;(\alpha,\beta)\;(\alpha,\beta)\;\alpha{}\<[E]%
\\
\>[B]{}\Varid{sndIBX}{}\<[11]%
\>[11]{}\mathbin{::}{}\<[11E]%
\>[17]{}\Conid{Monad}\;\tau\Rightarrow \alpha\hsarrow{\rightarrow }{\mathpunct{.}}\Conid{InitStateTBX}\;\tau\;(\alpha,\beta)\;(\alpha,\beta)\;\beta{}\<[E]%
\ColumnHook
\end{hscode}\resethooks
Pairing is defined as follows:
\begin{hscode}\SaveRestoreHook
\column{B}{@{}>{\hspre}l<{\hspost}@{}}%
\column{7}{@{}>{\hspre}l<{\hspost}@{}}%
\column{10}{@{}>{\hspre}c<{\hspost}@{}}%
\column{10E}{@{}l@{}}%
\column{13}{@{}>{\hspre}l<{\hspost}@{}}%
\column{16}{@{}>{\hspre}l<{\hspost}@{}}%
\column{19}{@{}>{\hspre}l<{\hspost}@{}}%
\column{22}{@{}>{\hspre}l<{\hspost}@{}}%
\column{28}{@{}>{\hspre}l<{\hspost}@{}}%
\column{31}{@{}>{\hspre}l<{\hspost}@{}}%
\column{32}{@{}>{\hspre}l<{\hspost}@{}}%
\column{38}{@{}>{\hspre}l<{\hspost}@{}}%
\column{53}{@{}>{\hspre}l<{\hspost}@{}}%
\column{68}{@{}>{\hspre}l<{\hspost}@{}}%
\column{E}{@{}>{\hspre}l<{\hspost}@{}}%
\>[B]{}\Varid{pairBX}{}\<[10]%
\>[10]{}\mathbin{::}{}\<[10E]%
\>[16]{}\Conid{Monad}\;\tau\Rightarrow \Conid{StateTBX}\;\tau\;\sigma_1\;\alpha_1\;\beta_1\hsarrow{\rightarrow }{\mathpunct{.}}\Conid{StateTBX}\;\tau\;{}\<[68]%
\>[68]{}\sigma_2\;\alpha_2\;\beta_2\hsarrow{\rightarrow }{\mathpunct{.}}{}\<[E]%
\\
\>[16]{}\Conid{StateTBX}\;\tau\;(\sigma_1,\sigma_2)\;(\alpha_1,\alpha_2)\;(\beta_1,\beta_2){}\<[E]%
\\
\>[B]{}\Varid{pairBX}\;bx_{1}\;bx_{2}\mathrel{=}\Conid{BX}\;\Varid{gl}\;\Varid{sl}\;\Varid{gr}\;\Varid{sr}\;\mathbf{where}{}\<[E]%
\\
\>[B]{}\hsindent{7}{}\<[7]%
\>[7]{}\Varid{gl}\mathrel{=}{}\<[13]%
\>[13]{}\mathbf{do}\;\{\mskip1.5mu {}\<[19]%
\>[19]{}\Varid{a}_{1}\leftarrow \Varid{left}\;(bx_{1}\mathord{.}\get{L});\Varid{a}_{2}\leftarrow \Varid{right}\;(bx_{2}\mathord{.}\get{L});\Varid{return}\;(\Varid{a}_{1},\Varid{a}_{2})\mskip1.5mu\}{}\<[E]%
\\
\>[B]{}\hsindent{7}{}\<[7]%
\>[7]{}\Varid{sl}\;(\Varid{a}_{1},\Varid{a}_{2}){}\<[19]%
\>[19]{}\mathrel{=}{}\<[22]%
\>[22]{}\mathbf{do}\;\{\mskip1.5mu {}\<[28]%
\>[28]{}\Varid{left}\;(bx_{1}\mathord{.}\set{L}\;\Varid{a}_{1});\Varid{right}\;(bx_{2}\mathord{.}\set{L}\;\Varid{a}_{2})\mskip1.5mu\}{}\<[E]%
\\
\>[B]{}\hsindent{7}{}\<[7]%
\>[7]{}\Varid{gr}{}\<[19]%
\>[19]{}\mathrel{=}{}\<[22]%
\>[22]{}... \mbox{\onelinecomment dual}{}\<[32]%
\>[32]{}{}\<[38]%
\>[38]{}{}\<[E]%
\\
\>[B]{}\hsindent{7}{}\<[7]%
\>[7]{}\Varid{sr}{}\<[19]%
\>[19]{}\mathrel{=}{}\<[22]%
\>[22]{}... \mbox{\onelinecomment dual}{}\<[31]%
\>[31]{}{}\<[53]%
\>[53]{}{}\<[E]%
\ColumnHook
\end{hscode}\resethooks
\endswithdisplay
\end{definition*}

Other operations based on isomorphisms, such as associativity of
pairs, can be lifted to \ensuremath{\Conid{StateTBX}}s without
problems. Well-behavedness is immediate for \ensuremath{\Varid{constBX}}, \ensuremath{\Varid{fstBX}},
\ensuremath{\Varid{sndBX}} and for any other \bx{} that can be obtained from an asymmetric or
symmetric lens.  For the \ensuremath{\Varid{pairBX}} combinator we need to verify
preservation of transparency:
\begin{proposition}\label{prop:pair-wb}
  If \ensuremath{bx_1} and \ensuremath{bx_2} are
  transparent (and hence well-behaved), then so is \ensuremath{\Varid{pairBX}\;bx_1\;bx_2}.
\end{proposition}
\begin{remark*}
The pair combinator does not necessarily
  preserve overwritability.  
  For this to be the
  case, we need to be able to commute the \ensuremath{\Varid{set}} operations of the
  component bx, including any effects in \ensuremath{\Conid{T}}.  Moreover,
  the pairing combinator is not in general uniquely determined for
  non-commutative \ensuremath{\Conid{T}}, because the effects of \ensuremath{bx_1} and \ensuremath{bx_2} can be
  applied in different orders.
\end{remark*}

\begin{definition*}[sum combinators]
  Similarly, we can define combinators analogous to the `retentive
  sum' symmetric lenses and injection operations~\cite{symlens}.  The
  injection operations relate an \ensuremath{\alpha} and either the same \ensuremath{\alpha} or some
  unrelated \ensuremath{\beta}; the old \ensuremath{\alpha} value of the left side is retained when
  the right side is a \ensuremath{\beta}.  
\begin{hscode}\SaveRestoreHook
\column{B}{@{}>{\hspre}l<{\hspost}@{}}%
\column{8}{@{}>{\hspre}c<{\hspost}@{}}%
\column{8E}{@{}l@{}}%
\column{12}{@{}>{\hspre}l<{\hspost}@{}}%
\column{E}{@{}>{\hspre}l<{\hspost}@{}}%
\>[B]{}\Varid{inlBX}{}\<[8]%
\>[8]{}\mathbin{::}{}\<[8E]%
\>[12]{}\Conid{Monad}\;\tau\Rightarrow \alpha\hsarrow{\rightarrow }{\mathpunct{.}}\Conid{StateTBX}\;\tau\;(\alpha,\Conid{Maybe}\;\beta)\;\alpha\;(\Conid{Either}\;\alpha\;\beta){}\<[E]%
\\
\>[B]{}\Varid{inrBX}{}\<[8]%
\>[8]{}\mathbin{::}{}\<[8E]%
\>[12]{}\Conid{Monad}\;\tau\Rightarrow \beta\hsarrow{\rightarrow }{\mathpunct{.}}\Conid{StateTBX}\;\tau\;(\beta,\Conid{Maybe}\;\alpha)\;\beta\;(\Conid{Either}\;\alpha\;\beta){}\<[E]%
\ColumnHook
\end{hscode}\resethooks
  The \ensuremath{\Varid{sumBX}} combinator combines two
  underlying bx and allows switching between them; the state of both
  (including that of the bx that is not currently in focus) is retained.
\begin{hscode}\SaveRestoreHook
\column{B}{@{}>{\hspre}l<{\hspost}@{}}%
\column{8}{@{}>{\hspre}c<{\hspost}@{}}%
\column{8E}{@{}l@{}}%
\column{10}{@{}>{\hspre}l<{\hspost}@{}}%
\column{12}{@{}>{\hspre}l<{\hspost}@{}}%
\column{14}{@{}>{\hspre}l<{\hspost}@{}}%
\column{21}{@{}>{\hspre}l<{\hspost}@{}}%
\column{25}{@{}>{\hspre}l<{\hspost}@{}}%
\column{27}{@{}>{\hspre}l<{\hspost}@{}}%
\column{33}{@{}>{\hspre}l<{\hspost}@{}}%
\column{39}{@{}>{\hspre}l<{\hspost}@{}}%
\column{64}{@{}>{\hspre}l<{\hspost}@{}}%
\column{E}{@{}>{\hspre}l<{\hspost}@{}}%
\>[B]{}\Varid{sumBX}{}\<[8]%
\>[8]{}\mathbin{::}{}\<[8E]%
\>[12]{}\Conid{Monad}\;\tau\Rightarrow \Conid{StateTBX}\;\tau\;\sigma_1\;\alpha_1\;\beta_1\hsarrow{\rightarrow }{\mathpunct{.}}\Conid{StateTBX}\;\tau\;{}\<[64]%
\>[64]{}\sigma_2\;\alpha_2\;\beta_2\hsarrow{\rightarrow }{\mathpunct{.}}{}\<[E]%
\\
\>[12]{}\Conid{StateTBX}\;\tau\;(\Conid{Bool},\sigma_1,\sigma_2)\;(\Conid{Either}\;\alpha_1\;\alpha_2)\;(\Conid{Either}\;\beta_1\;\beta_2){}\<[E]%
\\
\>[B]{}\Varid{sumBX}\;bx_{1}\;bx_{2}\mathrel{=}\Conid{BX}\;\Varid{gl}\;\Varid{sl}\;\Varid{gr}\;\Varid{sr}\;\mathbf{where}{}\<[E]%
\\
\>[B]{}\hsindent{10}{}\<[10]%
\>[10]{}\Varid{gl}\mathrel{=}\mathbf{do}\;\{\mskip1.5mu {}\<[21]%
\>[21]{}(\Varid{b},\Varid{s}_{1},\Varid{s}_{2})\leftarrow \Varid{get};{}\<[E]%
\\
\>[21]{}\mathbf{if}\;\Varid{b}\;{}\<[27]%
\>[27]{}\mathbf{then}\;{}\<[33]%
\>[33]{}\mathbf{do}\;\{\mskip1.5mu {}\<[39]%
\>[39]{}(\Varid{a}_{1},\anonymous )\leftarrow \Varid{lift}\;(bx_{1}\mathord{.}\get{L}\;\Varid{s}_{1});\Varid{return}\;(\Conid{Left}\;\Varid{a}_{1})\mskip1.5mu\}{}\<[E]%
\\
\>[27]{}\mathbf{else}\;{}\<[33]%
\>[33]{}\mathbf{do}\;\{\mskip1.5mu {}\<[39]%
\>[39]{}(\Varid{a}_{2},\anonymous )\leftarrow \Varid{lift}\;(bx_{2}\mathord{.}\get{L}\;\Varid{s}_{2});\Varid{return}\;(\Conid{Right}\;\Varid{a}_{2})\mskip1.5mu\}\mskip1.5mu\}{}\<[E]%
\\
\>[B]{}\hsindent{10}{}\<[10]%
\>[10]{}\Varid{sl}\;(\Conid{Left}\;\Varid{a}_{1}){}\<[25]%
\>[25]{}\mathrel{=}\mathbf{do}\;\{\mskip1.5mu {}\<[33]%
\>[33]{}(\Varid{b},\Varid{s}_{1},\Varid{s}_{2})\leftarrow \Varid{get};{}\<[E]%
\\
\>[33]{}((),\Varid{s}_{1}')\leftarrow \Varid{lift}\;((bx_{1}\mathord{.}\set{L}\;\Varid{a}_{1})\;\Varid{s}_{1});{}\<[E]%
\\
\>[33]{}\Varid{set}\;(\Conid{True},\Varid{s}_{1}',\Varid{s}_{2})\mskip1.5mu\}{}\<[E]%
\\
\>[B]{}\hsindent{10}{}\<[10]%
\>[10]{}\Varid{sl}\;(\Conid{Right}\;\Varid{a}_{2}){}\<[25]%
\>[25]{}\mathrel{=}\mathbf{do}\;\{\mskip1.5mu {}\<[33]%
\>[33]{}(\Varid{b},\Varid{s}_{1},\Varid{s}_{2})\leftarrow \Varid{get};{}\<[E]%
\\
\>[33]{}((),\Varid{s}_{2}')\leftarrow \Varid{lift}\;((bx_{2}\mathord{.}\set{L}\;\Varid{a}_{2})\;\Varid{s}_{2});{}\<[E]%
\\
\>[33]{}\Varid{set}\;(\Conid{False},\Varid{s}_{1},\Varid{s}_{2}')\mskip1.5mu\}{}\<[E]%
\\
\>[B]{}\hsindent{10}{}\<[10]%
\>[10]{}\Varid{gr}{}\<[14]%
\>[14]{}\mathrel{=}... \mbox{\onelinecomment dual}{}\<[E]%
\\
\>[B]{}\hsindent{10}{}\<[10]%
\>[10]{}\Varid{sr}{}\<[14]%
\>[14]{}\mathrel{=}... \mbox{\onelinecomment dual}{}\<[E]%
\ColumnHook
\end{hscode}\resethooks
\endswithdisplay
\end{definition*}
\begin{proposition}\label{prop:sum-wb}
If \ensuremath{bx_{1}} and \ensuremath{bx_{2}} are transparent, then so is \ensuremath{\Varid{sumBX}\;bx_{1}\;bx_{2}}.
\end{proposition}

Finally, we turn to building a bx 
that operates on lists from one that operates 
on elements. The symmetric lens list combinators~\cite{symlens}
implicitly regard the length of the list as data that is shared between the
two views. The forgetful list combinator forgets all data beyond the
current length. The retentive version maintains list elements beyond the
current length, so that they can be restored if the list is lengthened
again. We demonstrate the (more interesting) retentive version, making the
shared list length explicit. 
Several other variants are possible.

\begin{definition*}[retentive list combinator] 
  This combinator relies on the initialisation functions to deal with the case where the
  new values are inserted into the list, because in this case we need
  the capability to create new values on the other side (and new
  states linking them).
\begin{hscode}\SaveRestoreHook
\column{B}{@{}>{\hspre}l<{\hspost}@{}}%
\column{9}{@{}>{\hspre}l<{\hspost}@{}}%
\column{13}{@{}>{\hspre}l<{\hspost}@{}}%
\column{19}{@{}>{\hspre}c<{\hspost}@{}}%
\column{19E}{@{}l@{}}%
\column{22}{@{}>{\hspre}l<{\hspost}@{}}%
\column{28}{@{}>{\hspre}l<{\hspost}@{}}%
\column{36}{@{}>{\hspre}l<{\hspost}@{}}%
\column{40}{@{}>{\hspre}l<{\hspost}@{}}%
\column{E}{@{}>{\hspre}l<{\hspost}@{}}%
\>[B]{}\Varid{listIBX}\mathbin{::}{}\<[13]%
\>[13]{}\Conid{Monad}\;\tau\Rightarrow {}\<[E]%
\\
\>[13]{}\Conid{InitStateTBX}\;\tau\;\sigma\;\alpha\;\beta\hsarrow{\rightarrow }{\mathpunct{.}}\Conid{InitStateTBX}\;\tau\;(\Conid{Int},[\mskip1.5mu \sigma\mskip1.5mu])\;[\mskip1.5mu \alpha\mskip1.5mu]\;[\mskip1.5mu \beta\mskip1.5mu]{}\<[E]%
\\
\>[B]{}\Varid{listIBX}\;bx\mathrel{=}\Conid{InitStateTBX}\;\Varid{gl}\;\Varid{sl}\;\Varid{il}\;\Varid{gr}\;\Varid{sr}\;\Varid{ir}\;\mathbf{where}{}\<[E]%
\\
\>[B]{}\hsindent{9}{}\<[9]%
\>[9]{}\Varid{gl}{}\<[19]%
\>[19]{}\mathrel{=}{}\<[19E]%
\>[22]{}\mathbf{do}\;\{\mskip1.5mu {}\<[28]%
\>[28]{}(\Varid{n},\Varid{cs})\leftarrow \Varid{get};\Varid{mapM}\;(\Varid{lift}\hsdot{\cdot }{.}\Varid{eval}\;bx\mathord{.}\get{L})\;(\Varid{take}\;\Varid{n}\;\Varid{cs})\mskip1.5mu\}{}\<[E]%
\\
\>[B]{}\hsindent{9}{}\<[9]%
\>[9]{}\Varid{sl}\;\Varid{as}{}\<[19]%
\>[19]{}\mathrel{=}{}\<[19E]%
\>[22]{}\mathbf{do}\;\{\mskip1.5mu {}\<[28]%
\>[28]{}(\anonymous ,\Varid{cs})\leftarrow \Varid{get};{}\<[E]%
\\
\>[28]{}\Varid{cs'}\leftarrow \Varid{lift}\;(\Varid{sets}\;(\Varid{exec}\hsdot{\cdot }{.}bx\mathord{.}\set{L})\;bx\mathord{.}\Varid{init}_{L}\;\Varid{as}\;\Varid{cs});{}\<[E]%
\\
\>[28]{}\Varid{set}\;(\Varid{length}\;\Varid{as},\Varid{cs'})\mskip1.5mu\}{}\<[E]%
\\
\>[B]{}\hsindent{9}{}\<[9]%
\>[9]{}\Varid{il}\;\Varid{as}{}\<[19]%
\>[19]{}\mathrel{=}{}\<[19E]%
\>[22]{}\mathbf{do}\;\{\mskip1.5mu {}\<[28]%
\>[28]{}\Varid{cs}\leftarrow \Varid{mapM}\;(bx\mathord{.}\Varid{init}_{L})\;\Varid{as};\Varid{return}\;(\Varid{length}\;\Varid{as},\Varid{cs})\mskip1.5mu\}{}\<[E]%
\\
\>[B]{}\hsindent{9}{}\<[9]%
\>[9]{}\Varid{gr}{}\<[19]%
\>[19]{}\mathrel{=}{}\<[19E]%
\>[22]{}... \mbox{\onelinecomment dual}{}\<[40]%
\>[40]{}{}\<[E]%
\\
\>[B]{}\hsindent{9}{}\<[9]%
\>[9]{}\Varid{sr}\;\Varid{bs}{}\<[19]%
\>[19]{}\mathrel{=}{}\<[19E]%
\>[22]{}... \mbox{\onelinecomment dual}{}\<[40]%
\>[40]{}{}\<[E]%
\\
\>[B]{}\hsindent{9}{}\<[9]%
\>[9]{}\Varid{ir}\;\Varid{bs}{}\<[19]%
\>[19]{}\mathrel{=}{}\<[19E]%
\>[22]{}... \mbox{\onelinecomment dual}{}\<[36]%
\>[36]{}{}\<[E]%
\ColumnHook
\end{hscode}\resethooks
Here, the standard Haskell function \ensuremath{\Varid{mapM}} sequences a list of
computations, and \ensuremath{\Varid{sets}} sequentially updates a list of states
from a list of views, retaining any leftover states if the view list
is shorter:
\begin{hscode}\SaveRestoreHook
\column{B}{@{}>{\hspre}l<{\hspost}@{}}%
\column{10}{@{}>{\hspre}l<{\hspost}@{}}%
\column{23}{@{}>{\hspre}l<{\hspost}@{}}%
\column{31}{@{}>{\hspre}c<{\hspost}@{}}%
\column{31E}{@{}l@{}}%
\column{34}{@{}>{\hspre}l<{\hspost}@{}}%
\column{40}{@{}>{\hspre}l<{\hspost}@{}}%
\column{E}{@{}>{\hspre}l<{\hspost}@{}}%
\>[B]{}\Varid{sets}\mathbin{::}{}\<[10]%
\>[10]{}\Conid{Monad}\;\tau\Rightarrow (\alpha\hsarrow{\rightarrow }{\mathpunct{.}}\gamma\hsarrow{\rightarrow }{\mathpunct{.}}\tau\;\gamma)\hsarrow{\rightarrow }{\mathpunct{.}}(\alpha\hsarrow{\rightarrow }{\mathpunct{.}}\tau\;\gamma)\hsarrow{\rightarrow }{\mathpunct{.}}[\mskip1.5mu \alpha\mskip1.5mu]\hsarrow{\rightarrow }{\mathpunct{.}}[\mskip1.5mu \gamma\mskip1.5mu]\hsarrow{\rightarrow }{\mathpunct{.}}\tau\;[\mskip1.5mu \gamma\mskip1.5mu]{}\<[E]%
\\
\>[B]{}\Varid{sets}\;\Varid{set}\;\Varid{init}\;[\mskip1.5mu \mskip1.5mu]\;{}\<[23]%
\>[23]{}\Varid{cs}{}\<[31]%
\>[31]{}\mathrel{=}{}\<[31E]%
\>[34]{}\Varid{return}\;\Varid{cs}{}\<[E]%
\\
\>[B]{}\Varid{sets}\;\Varid{set}\;\Varid{init}\;(\Varid{x}\mathbin{:}\Varid{xs})\;{}\<[23]%
\>[23]{}(\Varid{c}\mathbin{:}\Varid{cs}){}\<[31]%
\>[31]{}\mathrel{=}{}\<[31E]%
\>[34]{}\mathbf{do}\;\{\mskip1.5mu {}\<[40]%
\>[40]{}\Varid{c'}\leftarrow \Varid{set}\;\Varid{x}\;\Varid{c};\Varid{cs'}\leftarrow \Varid{sets}\;\Varid{set}\;\Varid{init}\;\Varid{xs}\;\Varid{cs};{}\<[E]%
\\
\>[40]{}\Varid{return}\;(\Varid{c'}\mathbin{:}\Varid{cs'})\mskip1.5mu\}{}\<[E]%
\\
\>[B]{}\Varid{sets}\;\Varid{set}\;\Varid{init}\;\Varid{xs}\;{}\<[23]%
\>[23]{}[\mskip1.5mu \mskip1.5mu]{}\<[31]%
\>[31]{}\mathrel{=}{}\<[31E]%
\>[34]{}\Varid{mapM}\;\Varid{init}\;\Varid{xs}{}\<[E]%
\ColumnHook
\end{hscode}\resethooks
\endswithdisplay
\end{definition*}
\begin{proposition}\label{prop:list-wb}
If \ensuremath{bx} is transparent, then so is \ensuremath{\Varid{listIBX}\;bx}.
\end{proposition}

\subsection{Effectful \bx{}}

We now consider examples of \bx{} that make nontrivial
use of monadic effects.  
The careful consideration we paid earlier to the requirements for
composability give rise to some interesting and non-obvious
constraints on the definitions, which we highlight as we go.

For accessibility, we use specific monads in the examples in
order to state and prove properties; for generality, the accompanying code 
abstracts from specific monads using
Haskell type class constraints instead.
In the interests of brevity, we omit dual cases and 
initialisation functions, but these are defined in the 
Appendix. 

\begin{example}[environment]\label{ex:switch}
  The \ensuremath{\Conid{Reader}} or environment monad is useful for modelling
  global parameters. Some classes of bidirectional transformations are
  naturally parametrised; for example, Voigtl\"ander \etal's approach
  \cite{voigtlaender13jfp} uses a \ensuremath{\Varid{bias}} parameter to determine how to merge
  changes back into lists.

Suppose we have a family of
\bx{} indexed by some parameter \ensuremath{\gamma}, over a monad \ensuremath{\Conid{Reader}\;\gamma}.  Then we can
define
\begin{hscode}\SaveRestoreHook
\column{B}{@{}>{\hspre}l<{\hspost}@{}}%
\column{3}{@{}>{\hspre}l<{\hspost}@{}}%
\column{12}{@{}>{\hspre}l<{\hspost}@{}}%
\column{15}{@{}>{\hspre}c<{\hspost}@{}}%
\column{15E}{@{}l@{}}%
\column{18}{@{}>{\hspre}l<{\hspost}@{}}%
\column{E}{@{}>{\hspre}l<{\hspost}@{}}%
\>[B]{}\Varid{switch}\mathbin{::}{}\<[12]%
\>[12]{}(\gamma\hsarrow{\rightarrow }{\mathpunct{.}}\Conid{StateTBX}\;(\Conid{Reader}\;\gamma)\;\sigma\;\alpha\;\beta)\hsarrow{\rightarrow }{\mathpunct{.}}\Conid{StateTBX}\;(\Conid{Reader}\;\gamma)\;\sigma\;\alpha\;\beta{}\<[E]%
\\
\>[B]{}\Varid{switch}\;\Varid{f}\mathrel{=}\Conid{BX}\;\Varid{gl}\;\Varid{sl}\;\Varid{gr}\;\Varid{sr}\;\mathbf{where}{}\<[E]%
\\
\>[B]{}\hsindent{3}{}\<[3]%
\>[3]{}\Varid{gl}{}\<[15]%
\>[15]{}\mathrel{=}{}\<[15E]%
\>[18]{}\mathbf{do}\;\{\mskip1.5mu \Varid{c}\leftarrow \Varid{lift}\;\Varid{ask};(\Varid{f}\;\Varid{c})\mathord{.}\get{L}\mskip1.5mu\}{}\<[E]%
\\
\>[B]{}\hsindent{3}{}\<[3]%
\>[3]{}\Varid{sl}\;\Varid{a}{}\<[15]%
\>[15]{}\mathrel{=}{}\<[15E]%
\>[18]{}\mathbf{do}\;\{\mskip1.5mu \Varid{c}\leftarrow \Varid{lift}\;\Varid{ask};(\Varid{f}\;\Varid{c})\mathord{.}\set{L}\;\Varid{a}\mskip1.5mu\}{}\<[E]%
\\
\>[B]{}\hsindent{3}{}\<[3]%
\>[3]{}\Varid{gr}{}\<[15]%
\>[15]{}\mathrel{=}{}\<[15E]%
\>[18]{}... \mbox{\onelinecomment dual}{}\<[E]%
\\
\>[B]{}\hsindent{3}{}\<[3]%
\>[3]{}\Varid{sr}{}\<[15]%
\>[15]{}\mathrel{=}{}\<[15E]%
\>[18]{}... \mbox{\onelinecomment dual}{}\<[E]%
\ColumnHook
\end{hscode}\resethooks
where the standard \ensuremath{\Varid{ask}\mathbin{::}\Conid{Reader}\;\gamma} operation reads the \ensuremath{\gamma} value.
\end{example}

\begin{proposition}\label{prop:reader-wb}
If \ensuremath{\Varid{f}\;\Varid{c}\mathbin{::}\Conid{StateTBX}\;(\Conid{Reader}\;\Conid{C})\;\Conid{S}\;\Conid{A}\;\Conid{B}} is transparent for any \ensuremath{\Varid{c}\mathbin{::}\Conid{C}}, then \ensuremath{\Varid{switch}\;\Varid{f}} is a well-behaved, but not necessarily transparent, \ensuremath{\Conid{StateTBX}\;(\Conid{Reader}\;\Conid{C})\;\Conid{S}\;\Conid{A}\;\Conid{B}}.
\end{proposition}

\begin{remark*}
Note that \ensuremath{\Varid{switch}\;\Varid{f}} is well-behaved but not necessarily transparent.  
This is because the
\ensuremath{\Varid{get}} operations read not only from the designated state of the
\ensuremath{\Conid{StateTBX}} but also from the \ensuremath{\Conid{Reader}} environment, and so they are not
\ensuremath{(\Conid{Reader}\;\Conid{C})}-pure. This turns out not to be a big problem in this
case, because \ensuremath{\Conid{Reader}\;\Conid{C}} is a commutative monad. But suppose that the
underlying monad were not \ensuremath{\Conid{Reader}\;\Conid{C}} but
a non-commutative monad such as \ensuremath{\Conid{State}\;\Conid{C}}, maintaining some
flag that may be changed by the \ensuremath{\Varid{set}} operations; in this scenario, it is not
difficult to construct a counterexample to the identity laws for
composition.
Such counterexamples are why 
we have largely restricted attention in this paper to transparent \bx.
(Besides, one what argue that it is never necessary for the get
operations to depend on the context; any such dependencies could be
handled entirely by the set operations.)
\end{remark*}

\begin{example*}[exceptions]\label{ex:exceptions}
  We turn next to the possibility of failure. Conventionally, the functions
  defining a \bx{} are required to be total, but often it is not
  possible to constrain the source and view types enough to make this
  literally true; for example, consider a \bx{} relating two \ensuremath{\Conid{Rational}}
  views whose consistency relation is $\{(x, \fraction 1 x) \mid x \neq 0\}$.
A principled approach to failure
is to use the \ensuremath{\Conid{Maybe}} (exception) monad, so that an attempt
to divide by zero yields \ensuremath{\Conid{Nothing}}.
\begin{hscode}\SaveRestoreHook
\column{B}{@{}>{\hspre}l<{\hspost}@{}}%
\column{3}{@{}>{\hspre}l<{\hspost}@{}}%
\column{15}{@{}>{\hspre}c<{\hspost}@{}}%
\column{15E}{@{}l@{}}%
\column{18}{@{}>{\hspre}l<{\hspost}@{}}%
\column{E}{@{}>{\hspre}l<{\hspost}@{}}%
\>[B]{}\Varid{invBX}\mathbin{::}\Conid{StateTBX}\;\Conid{Maybe}\;\Conid{Rational}\;\Conid{Rational}\;\Conid{Rational}{}\<[E]%
\\
\>[B]{}\Varid{invBX}\mathrel{=}\Conid{BX}\;\Varid{get}\;\set{L}\;(\Varid{gets}\;(\lambda \hslambda \Varid{a}\hsarrow{\rightarrow }{\mathpunct{.}}\fraction{\mathrm{1}}{\Varid{a}}))\;\set{R}\;\mathbf{where}{}\<[E]%
\\
\>[B]{}\hsindent{3}{}\<[3]%
\>[3]{}\set{L}\;\Varid{a}{}\<[15]%
\>[15]{}\mathrel{=}{}\<[15E]%
\>[18]{}\mathbf{do}\;\{\mskip1.5mu \Varid{lift}\;(\Varid{guard}\;(\Varid{a}\notequals\mathrm{0}));\Varid{set}\;\Varid{a}\mskip1.5mu\}{}\<[E]%
\\
\>[B]{}\hsindent{3}{}\<[3]%
\>[3]{}\set{R}\;\Varid{b}{}\<[15]%
\>[15]{}\mathrel{=}{}\<[15E]%
\>[18]{}\mathbf{do}\;\{\mskip1.5mu \Varid{lift}\;(\Varid{guard}\;(\Varid{b}\notequals\mathrm{0}));\Varid{set}\;(\fraction{\mathrm{1}}{\Varid{b}})\mskip1.5mu\}{}\<[E]%
\ColumnHook
\end{hscode}\resethooks
where \ensuremath{\Varid{guard}\;\Varid{b}\mathrel{=}\mathbf{do}\;\{\mskip1.5mu \mathbf{if}\;\Varid{b}\;\mathbf{then}\;\Conid{Just}\;()\;\mathbf{else}\;\Conid{Nothing}\mskip1.5mu\}} is a standard
operation in the \ensuremath{\Conid{Maybe}} monad.  As another example, suppose we know that \ensuremath{\Conid{A}} is in
the \ensuremath{\Conid{Read}} and \ensuremath{\Conid{Show}} type classes, so each \ensuremath{\Conid{A}} value can be printed
to and possibly read from a string.  We can define:
\begin{hscode}\SaveRestoreHook
\column{B}{@{}>{\hspre}l<{\hspost}@{}}%
\column{3}{@{}>{\hspre}l<{\hspost}@{}}%
\column{12}{@{}>{\hspre}l<{\hspost}@{}}%
\column{16}{@{}>{\hspre}l<{\hspost}@{}}%
\column{20}{@{}>{\hspre}l<{\hspost}@{}}%
\column{24}{@{}>{\hspre}l<{\hspost}@{}}%
\column{38}{@{}>{\hspre}l<{\hspost}@{}}%
\column{E}{@{}>{\hspre}l<{\hspost}@{}}%
\>[B]{}\Varid{readSomeBX}\mathbin{::}{}\<[16]%
\>[16]{}(\Conid{Read}\;\alpha,\Conid{Show}\;\alpha)\Rightarrow \Conid{StateTBX}\;\Conid{Maybe}\;(\alpha,\Conid{String})\;\alpha\;\Conid{String}{}\<[E]%
\\
\>[B]{}\Varid{readSomeBX}\mathrel{=}\Conid{BX}\;(\Varid{gets}\;\Varid{fst})\;\set{L}\;(\Varid{gets}\;\Varid{snd})\;\set{R}\;\mathbf{where}{}\<[E]%
\\
\>[B]{}\hsindent{3}{}\<[3]%
\>[3]{}\set{L}\;\Varid{a'}{}\<[12]%
\>[12]{}\mathrel{=}\Varid{set}\;(\Varid{a'},\Varid{show}\;\Varid{a'}){}\<[E]%
\\
\>[B]{}\hsindent{3}{}\<[3]%
\>[3]{}\set{R}\;\Varid{b'}{}\<[12]%
\>[12]{}\mathrel{=}\mathbf{do}\;\{\mskip1.5mu {}\<[20]%
\>[20]{}(\anonymous ,\Varid{b})\leftarrow \Varid{get};{}\<[E]%
\\
\>[20]{}\mathbf{if}\;\Varid{b}\equals\Varid{b'}\;\mathbf{then}\;\Varid{return}\;()\;\mathbf{else}\;\mathbf{case}\;\Varid{reads}\;\Varid{b'}\;\mathbf{of}{}\<[E]%
\\
\>[20]{}\hsindent{4}{}\<[24]%
\>[24]{}((\Varid{a'},\text{\tt \char34 \char34})\mathbin{:\char95 }){}\<[38]%
\>[38]{}\hsarrow{\rightarrow }{\mathpunct{.}}\Varid{set}\;(\Varid{a'},\Varid{b'}){}\<[E]%
\\
\>[20]{}\hsindent{4}{}\<[24]%
\>[24]{}\anonymous {}\<[38]%
\>[38]{}\hsarrow{\rightarrow }{\mathpunct{.}}\Varid{lift}\;\Conid{Nothing}\mskip1.5mu\}{}\<[E]%
\ColumnHook
\end{hscode}\resethooks
(The function \ensuremath{\Varid{reads}} returns a list of possible parses with remaining text.)
Note that the get operations are \ensuremath{\Conid{Maybe}}-pure: if there is a \ensuremath{\Conid{Read}} error, 
it is raised instead by the set operations.

The same approach can be generalised to any monad \ensuremath{\Conid{T}} having a polymorphic error value
\ensuremath{\Varid{err}\mathbin{::}\forall \alpha\hsforall \hsdot{\cdot }{.}\Conid{T}\;\alpha} and any pair of partial
inverse functions \ensuremath{\Varid{f}\mathbin{::}\Conid{A}\hsarrow{\rightarrow }{\mathpunct{.}}\Conid{Maybe}\;\Conid{B}} and \ensuremath{\Varid{g}\mathbin{::}\Conid{B}\hsarrow{\rightarrow }{\mathpunct{.}}\Conid{Maybe}\;\Conid{A}}
(\ie, \ensuremath{\Varid{f}\;\Varid{a}\mathrel{=}\Conid{Just}\;\Varid{b}} if and only if \ensuremath{\Varid{g}\;\Varid{b}\mathrel{=}\Conid{Just}\;\Varid{a}}, for all \ensuremath{\Varid{a},\Varid{b}}): 
\begin{hscode}\SaveRestoreHook
\column{B}{@{}>{\hspre}l<{\hspost}@{}}%
\column{3}{@{}>{\hspre}l<{\hspost}@{}}%
\column{12}{@{}>{\hspre}c<{\hspost}@{}}%
\column{12E}{@{}l@{}}%
\column{15}{@{}>{\hspre}l<{\hspost}@{}}%
\column{29}{@{}>{\hspre}l<{\hspost}@{}}%
\column{38}{@{}>{\hspre}l<{\hspost}@{}}%
\column{E}{@{}>{\hspre}l<{\hspost}@{}}%
\>[B]{}\Varid{partialBX}\mathbin{::}{}\<[15]%
\>[15]{}\Conid{Monad}\;\tau\Rightarrow (\forall \alpha\hsforall \hsdot{\cdot }{.}\tau\;\alpha)\hsarrow{\rightarrow }{\mathpunct{.}}(\alpha\hsarrow{\rightarrow }{\mathpunct{.}}\Conid{Maybe}\;\beta)\hsarrow{\rightarrow }{\mathpunct{.}}(\beta\hsarrow{\rightarrow }{\mathpunct{.}}\Conid{Maybe}\;\alpha)\hsarrow{\rightarrow }{\mathpunct{.}}{}\<[E]%
\\
\>[15]{}\Conid{StateTBX}\;\tau\;(\alpha,\beta)\;\alpha\;\beta{}\<[E]%
\\
\>[B]{}\Varid{partialBX}\;\Varid{err}\;\Varid{f}\;\Varid{g}\mathrel{=}\Conid{BX}\;(\Varid{gets}\;\Varid{fst})\;\set{L}\;(\Varid{gets}\;\Varid{snd})\;\set{R}\;\mathbf{where}{}\<[E]%
\\
\>[B]{}\hsindent{3}{}\<[3]%
\>[3]{}\set{L}\;\Varid{a'}{}\<[12]%
\>[12]{}\mathrel{=}{}\<[12E]%
\>[15]{}\mathbf{case}\;\Varid{f}\;\Varid{a'}\;\mathbf{of}\;{}\<[29]%
\>[29]{}\Conid{Just}\;\Varid{b'}{}\<[38]%
\>[38]{}\hsarrow{\rightarrow }{\mathpunct{.}}\Varid{set}\;(\Varid{a'},\Varid{b'}){}\<[E]%
\\
\>[29]{}\Conid{Nothing}{}\<[38]%
\>[38]{}\hsarrow{\rightarrow }{\mathpunct{.}}\Varid{lift}\;\Varid{err}{}\<[E]%
\\
\>[B]{}\hsindent{3}{}\<[3]%
\>[3]{}\set{R}\;\Varid{b'}{}\<[12]%
\>[12]{}\mathrel{=}{}\<[12E]%
\>[15]{}\mathbf{case}\;\Varid{g}\;\Varid{b'}\;\mathbf{of}\;{}\<[29]%
\>[29]{}\Conid{Just}\;\Varid{a'}{}\<[38]%
\>[38]{}\hsarrow{\rightarrow }{\mathpunct{.}}\Varid{set}\;(\Varid{a'},\Varid{b'}){}\<[E]%
\\
\>[29]{}\Conid{Nothing}{}\<[38]%
\>[38]{}\hsarrow{\rightarrow }{\mathpunct{.}}\Varid{lift}\;\Varid{err}{}\<[E]%
\ColumnHook
\end{hscode}\resethooks
Then we 
could define \ensuremath{\Varid{invBX}} and a stricter variation of \ensuremath{\Varid{readSomeBX}} 
\hyphenation{white-space}
(one that will \ensuremath{\Varid{read}} only a string that it \ensuremath{\Varid{show}}s---rejecting alternative renderings, whitespace, and so on)
as instances of \ensuremath{\Varid{partialBX}}.
\end{example*}

\begin{proposition}\label{prop:partial-wb}
  Let \ensuremath{\Varid{f}\mathbin{::}\Conid{A}\hsarrow{\rightarrow }{\mathpunct{.}}\Conid{Maybe}\;\Conid{B}} and \ensuremath{\Varid{g}\mathbin{::}\Conid{B}\hsarrow{\rightarrow }{\mathpunct{.}}\Conid{Maybe}\;\Conid{A}} be partial inverses
  and let \ensuremath{\Varid{err}} be a zero element for  \ensuremath{\Conid{T}}.  Then
  \ensuremath{\Varid{partialBX}\;\Varid{err}\;\Varid{f}\;\Varid{g}\mathbin{::}\Conid{StateTBX}\;\Conid{T}\;\Conid{S}\;\Conid{A}\;\Conid{B}} is well-behaved, where \ensuremath{\Conid{S}\mathrel{=}\{\mskip1.5mu (\Varid{a},\Varid{b})\mid \Varid{f}\;\Varid{a}\mathrel{=}\Conid{Just}\;\Varid{b}\mskip1.5mu\}}.
\end{proposition}

\begin{example*}[nondeterminism---Scenario~\ref{ex:nondeterminism} revisited]
For simplicity's sake, we model nondeterminism via the list monad:
a `nondeterministic function' from \ensuremath{\Conid{A}} to \ensuremath{\Conid{B}} is represented as a pure
function of type \ensuremath{\Conid{A}\hsarrow{\rightarrow }{\mathpunct{.}}[\mskip1.5mu \Conid{B}\mskip1.5mu]}. The following bx is parametrised on a
predicate \ensuremath{\Varid{ok}} that checks consistency of two states, a
fix-up function \ensuremath{\Varid{bs}} that returns the \ensuremath{\Conid{B}} values consistent with a given
\ensuremath{\Conid{A}}, and symmetrically a fix-up function \ensuremath{\Varid{as}}.
\begin{hscode}\SaveRestoreHook
\column{B}{@{}>{\hspre}l<{\hspost}@{}}%
\column{3}{@{}>{\hspre}l<{\hspost}@{}}%
\column{12}{@{}>{\hspre}l<{\hspost}@{}}%
\column{14}{@{}>{\hspre}l<{\hspost}@{}}%
\column{20}{@{}>{\hspre}l<{\hspost}@{}}%
\column{22}{@{}>{\hspre}l<{\hspost}@{}}%
\column{E}{@{}>{\hspre}l<{\hspost}@{}}%
\>[B]{}\Varid{nondetBX}\mathbin{::}{}\<[14]%
\>[14]{}(\alpha\hsarrow{\rightarrow }{\mathpunct{.}}\beta\hsarrow{\rightarrow }{\mathpunct{.}}\Conid{Bool})\hsarrow{\rightarrow }{\mathpunct{.}}(\alpha\hsarrow{\rightarrow }{\mathpunct{.}}[\mskip1.5mu \beta\mskip1.5mu])\hsarrow{\rightarrow }{\mathpunct{.}}(\beta\hsarrow{\rightarrow }{\mathpunct{.}}[\mskip1.5mu \alpha\mskip1.5mu])\hsarrow{\rightarrow }{\mathpunct{.}}{}\<[E]%
\\
\>[14]{}\Conid{StateTBX}\;[\mskip1.5mu \mskip1.5mu]\;(\alpha,\beta)\;\alpha\;\beta{}\<[E]%
\\
\>[B]{}\Varid{nondetBX}\;\Varid{ok}\;\Varid{bs}\;\Varid{as}\mathrel{=}\Conid{BX}\;(\Varid{gets}\;\Varid{fst})\;\set{L}\;(\Varid{gets}\;\Varid{snd})\;\set{R}\;\mathbf{where}{}\<[E]%
\\
\>[B]{}\hsindent{3}{}\<[3]%
\>[3]{}\set{L}\;\Varid{a'}{}\<[12]%
\>[12]{}\mathrel{=}\mathbf{do}\;\{\mskip1.5mu {}\<[20]%
\>[20]{}(\Varid{a},\Varid{b})\leftarrow \Varid{get};{}\<[E]%
\\
\>[20]{}\mathbf{if}\;\Varid{ok}\;\Varid{a'}\;\Varid{b}\;\mathbf{then}\;\Varid{set}\;(\Varid{a'},\Varid{b})\;\mathbf{else}{}\<[E]%
\\
\>[20]{}\hsindent{2}{}\<[22]%
\>[22]{}\mathbf{do}\;\{\mskip1.5mu \Varid{b'}\leftarrow \Varid{lift}\;(\Varid{bs}\;\Varid{a'});\Varid{set}\;(\Varid{a'},\Varid{b'})\mskip1.5mu\}\mskip1.5mu\}{}\<[E]%
\\
\>[B]{}\hsindent{3}{}\<[3]%
\>[3]{}\set{R}\;\Varid{b'}{}\<[12]%
\>[12]{}\mathrel{=}\mathbf{do}\;\{\mskip1.5mu {}\<[20]%
\>[20]{}(\Varid{a},\Varid{b})\leftarrow \Varid{get};{}\<[E]%
\\
\>[20]{}\mathbf{if}\;\Varid{ok}\;\Varid{a}\;\Varid{b'}\;\mathbf{then}\;\Varid{set}\;(\Varid{a},\Varid{b'})\;\mathbf{else}{}\<[E]%
\\
\>[20]{}\hsindent{2}{}\<[22]%
\>[22]{}\mathbf{do}\;\{\mskip1.5mu \Varid{a'}\leftarrow \Varid{lift}\;(\Varid{as}\;\Varid{b'});\Varid{set}\;(\Varid{a'},\Varid{b'})\mskip1.5mu\}\mskip1.5mu\}{}\<[E]%
\ColumnHook
\end{hscode}\resethooks
\endswithdisplay
\end{example*}

\begin{proposition}\label{prop:nondet-wb}
Given \ensuremath{\Varid{ok}}, \ensuremath{\Conid{S}\mathrel{=}\{\mskip1.5mu (\Varid{a},\Varid{b})\mid \Varid{ok}\;\Varid{a}\;\Varid{b}\mskip1.5mu\}}, and \ensuremath{\Varid{as}} and \ensuremath{\Varid{bs}} satisfying
\begin{hscode}\SaveRestoreHook
\column{B}{@{}>{\hspre}l<{\hspost}@{}}%
\column{16}{@{}>{\hspre}c<{\hspost}@{}}%
\column{16E}{@{}l@{}}%
\column{19}{@{}>{\hspre}c<{\hspost}@{}}%
\column{19E}{@{}l@{}}%
\column{23}{@{}>{\hspre}c<{\hspost}@{}}%
\column{23E}{@{}l@{}}%
\column{26}{@{}>{\hspre}l<{\hspost}@{}}%
\column{E}{@{}>{\hspre}l<{\hspost}@{}}%
\>[B]{}\Varid{a}\in\Varid{as}\;\Varid{b}{}\<[16]%
\>[16]{}\:{}\<[16E]%
\>[19]{}\Rightarrow {}\<[19E]%
\>[23]{}\:{}\<[23E]%
\>[26]{}\Varid{ok}\;\Varid{a}\;\Varid{b}{}\<[E]%
\\
\>[B]{}\Varid{b}\in\Varid{bs}\;\Varid{a}{}\<[16]%
\>[16]{}\:{}\<[16E]%
\>[19]{}\Rightarrow {}\<[19E]%
\>[23]{}\:{}\<[23E]%
\>[26]{}\Varid{ok}\;\Varid{a}\;\Varid{b}{}\<[E]%
\ColumnHook
\end{hscode}\resethooks
then \ensuremath{\Varid{nondetBX}\;\Varid{ok}\;\Varid{bs}\;\Varid{as}\mathbin{::}\Conid{StateTBX}\;[\mskip1.5mu \mskip1.5mu]\;\Conid{S}\;\Conid{A}\;\Conid{B}} is well-behaved
(indeed, it is clearly transparent).
It is not necessary for the two conditions to be equivalences.
\end{proposition}

\begin{remark*}
Note that, in addition to choice, the list monad also allows for
failure: the fix-up functions can return
the empty list. From a semantic point of view, nondeterminism is
usually modelled using the monad of finite nonempty sets. 
If we had used the nonempty set monad instead of lists, then failure
would not be possible.  
\end{remark*}

\begin{example*}[signalling] 
  We can define a \bx{} that sends a signal every time either side
  changes:
  \begin{hscode}\SaveRestoreHook
\column{B}{@{}>{\hspre}l<{\hspost}@{}}%
\column{3}{@{}>{\hspre}l<{\hspost}@{}}%
\column{11}{@{}>{\hspre}l<{\hspost}@{}}%
\column{14}{@{}>{\hspre}l<{\hspost}@{}}%
\column{19}{@{}>{\hspre}l<{\hspost}@{}}%
\column{E}{@{}>{\hspre}l<{\hspost}@{}}%
\>[B]{}\Varid{signalBX}\mathbin{::}{}\<[14]%
\>[14]{}(\Conid{Eq}\;\alpha,\Conid{Eq}\;\beta,\Conid{Monad}\;\tau)\Rightarrow (\alpha\hsarrow{\rightarrow }{\mathpunct{.}}\tau\;())\hsarrow{\rightarrow }{\mathpunct{.}}(\beta\hsarrow{\rightarrow }{\mathpunct{.}}\tau\;())\hsarrow{\rightarrow }{\mathpunct{.}}{}\<[E]%
\\
\>[14]{}\Conid{StateTBX}\;\tau\;\sigma\;\alpha\;\beta\hsarrow{\rightarrow }{\mathpunct{.}}\Conid{StateTBX}\;\tau\;\sigma\;\alpha\;\beta{}\<[E]%
\\
\>[B]{}\Varid{signalBX}\;\Varid{sigA}\;\Varid{sigB}\;bx\mathrel{=}\Conid{BX}\;(bx\mathord{.}\get{L})\;\Varid{sl}\;(bx\mathord{.}\get{R})\;\Varid{sr}\;\mathbf{where}{}\<[E]%
\\
\>[B]{}\hsindent{3}{}\<[3]%
\>[3]{}\Varid{sl}\;\Varid{a'}{}\<[11]%
\>[11]{}\mathrel{=}\mathbf{do}\;\{\mskip1.5mu {}\<[19]%
\>[19]{}\Varid{a}\leftarrow bx\mathord{.}\get{L};bx\mathord{.}\set{L}\;\Varid{a'};{}\<[E]%
\\
\>[19]{}\Varid{lift}\;(\mathbf{if}\;\Varid{a}\notequals\Varid{a'}\;\mathbf{then}\;\Varid{sigA}\;\Varid{a'}\;\mathbf{else}\;\Varid{return}\;())\mskip1.5mu\}{}\<[E]%
\\
\>[B]{}\hsindent{3}{}\<[3]%
\>[3]{}\Varid{sr}\;\Varid{b'}{}\<[11]%
\>[11]{}\mathrel{=}\mathbf{do}\;\{\mskip1.5mu {}\<[19]%
\>[19]{}\Varid{b}\leftarrow bx\mathord{.}\get{R};bx\mathord{.}\set{R}\;\Varid{b'};{}\<[E]%
\\
\>[19]{}\Varid{lift}\;(\mathbf{if}\;\Varid{b}\notequals\Varid{b'}\;\mathbf{then}\;\Varid{sigB}\;\Varid{b'}\;\mathbf{else}\;\Varid{return}\;())\mskip1.5mu\}{}\<[E]%
\ColumnHook
\end{hscode}\resethooks
  Note that \ensuremath{\Varid{sl}} checks to see whether the new value \ensuremath{\Varid{a'}}
  equals the old value \ensuremath{\Varid{a}}, and does nothing if so; 
  only if they are different does it perform \ensuremath{\Varid{sigA}\;\Varid{a'}}.
  If the \bx{} is to be well-behaved, then
  no action can be performed 
in the case that
  \ensuremath{\Varid{a}\equals\Varid{a'}}.

For example, instantiating the underlying monad to \ensuremath{\Conid{IO}} we have:
\begin{hscode}\SaveRestoreHook
\column{B}{@{}>{\hspre}l<{\hspost}@{}}%
\column{13}{@{}>{\hspre}l<{\hspost}@{}}%
\column{21}{@{}>{\hspre}l<{\hspost}@{}}%
\column{E}{@{}>{\hspre}l<{\hspost}@{}}%
\>[B]{}\Varid{alertBX}\mathbin{::}{}\<[13]%
\>[13]{}(\Conid{Eq}\;\alpha,\Conid{Eq}\;\beta)\Rightarrow \Conid{StateTBX}\;\Conid{IO}\;\sigma\;\alpha\;\beta\hsarrow{\rightarrow }{\mathpunct{.}}\Conid{StateTBX}\;\Conid{IO}\;\sigma\;\alpha\;\beta{}\<[E]%
\\
\>[B]{}\Varid{alertBX}\mathrel{=}\Varid{signalBX}\;{}\<[21]%
\>[21]{}(\lambda \hslambda \anonymous \hsarrow{\rightarrow }{\mathpunct{.}}\Varid{putStrLn}\;\text{\tt \char34 Left\char34})\;(\lambda \hslambda \anonymous \hsarrow{\rightarrow }{\mathpunct{.}}\Varid{putStrLn}\;\text{\tt \char34 Right\char34}){}\<[E]%
\ColumnHook
\end{hscode}\resethooks
which prints a message whenever one side changes.
This is well-behaved; the \ensuremath{\Varid{set}} operations are
side-effecting, but the side-effects only occur when the state is
changed. It is not overwritable, because multiple changes may lead to
different signals from a single change.   

  As another example, 
we can define a
  logging \bx{} as follows:
\begin{hscode}\SaveRestoreHook
\column{B}{@{}>{\hspre}l<{\hspost}@{}}%
\column{10}{@{}>{\hspre}c<{\hspost}@{}}%
\column{10E}{@{}l@{}}%
\column{14}{@{}>{\hspre}l<{\hspost}@{}}%
\column{33}{@{}>{\hspre}l<{\hspost}@{}}%
\column{E}{@{}>{\hspre}l<{\hspost}@{}}%
\>[B]{}\Varid{logBX}{}\<[10]%
\>[10]{}\mathbin{::}{}\<[10E]%
\>[14]{}(\Conid{Eq}\;\alpha,\Conid{Eq}\;\beta)\Rightarrow {}\<[33]%
\>[33]{}\Conid{StateTBX}\;(\Conid{Writer}\;[\mskip1.5mu \Conid{Either}\;\alpha\;\beta\mskip1.5mu])\;\sigma\;\alpha\;\beta\hsarrow{\rightarrow }{\mathpunct{.}}{}\<[E]%
\\
\>[33]{}\Conid{StateTBX}\;(\Conid{Writer}\;[\mskip1.5mu \Conid{Either}\;\alpha\;\beta\mskip1.5mu])\;\sigma\;\alpha\;\beta{}\<[E]%
\\
\>[B]{}\Varid{logBX}\mathrel{=}\Varid{signalBX}\;(\lambda \hslambda \Varid{a}\hsarrow{\rightarrow }{\mathpunct{.}}\Varid{tell}\;[\mskip1.5mu \Conid{Left}\;\Varid{a}\mskip1.5mu])\;(\lambda \hslambda \Varid{b}\hsarrow{\rightarrow }{\mathpunct{.}}\Varid{tell}\;[\mskip1.5mu \Conid{Right}\;\Varid{b}\mskip1.5mu]){}\<[E]%
\ColumnHook
\end{hscode}\resethooks
where \ensuremath{\Varid{tell}\mathbin{::}\sigma\hsarrow{\rightarrow }{\mathpunct{.}}\Conid{Writer}\;\sigma\;()} is a standard operation in the
\ensuremath{\Conid{Writer}} monad that writes a value to the output.
This \bx{} logs a list of all of the views as they are
changed.
Wrapping a component of a chain of composed \bx{} with \ensuremath{\Varid{log}} can 
provide insight into how changes at the ends of the chain propagate
through that component.  If memory use is a concern, then we could limit the 
length of the list to record only the most recent
updates -- lists of bounded length also form a monoid.
\end{example*}
\begin{proposition}\label{prop:signal-wb}
  If \ensuremath{\Conid{A}} and \ensuremath{\Conid{B}} are types equipped with a well-behaved notion of equality
  (in the sense that \ensuremath{(\Varid{a}\equals\Varid{b})\mathrel{=}\Conid{True}} if and only if \ensuremath{\Varid{a}\mathrel{=}\Varid{b}}), 
  and \ensuremath{bx\mathbin{::}\Conid{StateTBX}\;\Conid{T}\;\Conid{S}\;\Conid{A}\;\Conid{B}} is well-behaved, 
  then \ensuremath{\Varid{signalBX}\;\Varid{sigA}\;\Varid{sigB}\;bx\mathbin{::}\Conid{StateTBX}\;\Conid{T}\;\Conid{S}\;\Conid{A}\;\Conid{B}} is well-behaved. 
  Moreover, \ensuremath{\Varid{signalBX}} preserves transparency. 
\end{proposition}

\begin{example*}[interaction---Scenario~\ref{ex:dynamic} revisited]
For this example, we need to record both the current state (an \ensuremath{\Conid{A}}
and a \ensuremath{\Conid{B}}) and the learned 
collection of consistency restorations. 
The latter is represented as two lists; the first list contains a tuple
\ensuremath{((\Varid{a'},\Varid{b}),\Varid{b'})} for each invocation of \ensuremath{\set{L}\;\Varid{a'}} on a state \ensuremath{(\anonymous ,\Varid{b})}
resulting in an updated state \ensuremath{(\Varid{a'},\Varid{b'})}; the second is symmetric,
for \ensuremath{\set{R}\;\Varid{b'}} invocations.
The types \ensuremath{\Conid{A}} and \ensuremath{\Conid{B}} must each
support equality, so that we can check for previously asked questions.
We abstract from the base monad; we parametrise the \bx{} on two
monadic functions, each somehow determining a consistent match
for one state.
\begin{hscode}\SaveRestoreHook
\column{B}{@{}>{\hspre}l<{\hspost}@{}}%
\column{10}{@{}>{\hspre}l<{\hspost}@{}}%
\column{15}{@{}>{\hspre}l<{\hspost}@{}}%
\column{26}{@{}>{\hspre}l<{\hspost}@{}}%
\column{28}{@{}>{\hspre}l<{\hspost}@{}}%
\column{30}{@{}>{\hspre}l<{\hspost}@{}}%
\column{39}{@{}>{\hspre}l<{\hspost}@{}}%
\column{44}{@{}>{\hspre}l<{\hspost}@{}}%
\column{48}{@{}>{\hspre}l<{\hspost}@{}}%
\column{122}{@{}>{\hspre}l<{\hspost}@{}}%
\column{159}{@{}>{\hspre}l<{\hspost}@{}}%
\column{E}{@{}>{\hspre}l<{\hspost}@{}}%
\>[B]{}\Varid{dynamicBX}\mathbin{::}{}\<[15]%
\>[15]{}(\Conid{Eq}\;\alpha,\Conid{Eq}\;\beta,\Conid{Monad}\;\tau)\Rightarrow {}\<[E]%
\\
\>[15]{}(\alpha\hsarrow{\rightarrow }{\mathpunct{.}}\beta\hsarrow{\rightarrow }{\mathpunct{.}}\tau\;\beta)\hsarrow{\rightarrow }{\mathpunct{.}}(\alpha\hsarrow{\rightarrow }{\mathpunct{.}}\beta\hsarrow{\rightarrow }{\mathpunct{.}}\tau\;\alpha)\hsarrow{\rightarrow }{\mathpunct{.}}{}\<[E]%
\\
\>[15]{}\Conid{StateTBX}\;\tau\;((\alpha,\beta),[\mskip1.5mu ((\alpha,\beta),\beta)\mskip1.5mu],[\mskip1.5mu ((\alpha,\beta),\alpha)\mskip1.5mu])\;\alpha\;\beta{}\<[E]%
\\
\>[B]{}\Varid{dynamicBX}\;\Varid{f}\;\Varid{g}\mathrel{=}\Conid{BX}\;(\Varid{gets}\;(\Varid{fst}\hsdot{\cdot }{.}\Varid{fst3}))\;\set{L}\;(\Varid{gets}\;(\Varid{snd}\hsdot{\cdot }{.}\Varid{fst3}))\;\set{R}\;\mathbf{where}{}\<[E]%
\\
\>[B]{}\hsindent{10}{}\<[10]%
\>[10]{}\set{L}\;\Varid{a'}\mathrel{=}\mathbf{do}\;\{\mskip1.5mu {}\<[26]%
\>[26]{}((\Varid{a},\Varid{b}),\Varid{fs},\Varid{bs})\leftarrow \Varid{get};{}\<[E]%
\\
\>[26]{}\mathbf{if}\;\Varid{a}\equals\Varid{a'}\;\mathbf{then}\;\Varid{return}\;()\;\mathbf{else}{}\<[E]%
\\
\>[26]{}\hsindent{2}{}\<[28]%
\>[28]{}\mathbf{case}\;\Varid{lookup}\;(\Varid{a'},\Varid{b})\;\Varid{fs}\;\mathbf{of}{}\<[E]%
\\
\>[28]{}\hsindent{2}{}\<[30]%
\>[30]{}\Conid{Just}\;\Varid{b'}{}\<[39]%
\>[39]{}\hsarrow{\rightarrow }{\mathpunct{.}}\Varid{set}\;((\Varid{a'},\Varid{b'}),\Varid{fs},\Varid{bs}){}\<[E]%
\\
\>[28]{}\hsindent{2}{}\<[30]%
\>[30]{}\Conid{Nothing}{}\<[39]%
\>[39]{}\hsarrow{\rightarrow }{\mathpunct{.}}\mathbf{do}\;\{\mskip1.5mu {}\<[48]%
\>[48]{}\Varid{b'}\leftarrow \Varid{lift}\;(\Varid{f}\;\Varid{a'}\;\Varid{b});{}\<[E]%
\\
\>[48]{}\Varid{set}\;((\Varid{a'},\Varid{b'}),((\Varid{a'},\Varid{b}),\Varid{b'})\mathbin{:}\Varid{fs},\Varid{bs})\mskip1.5mu\}\mskip1.5mu\}{}\<[E]%
\\
\>[B]{}\hsindent{10}{}\<[10]%
\>[10]{}\set{R}\;\Varid{b'}\mathrel{=}... \mbox{\onelinecomment dual}{}\<[44]%
\>[44]{}{}\<[122]%
\>[122]{}{}\<[159]%
\>[159]{}{}\<[E]%
\ColumnHook
\end{hscode}\resethooks
where \ensuremath{\Varid{fst3}\;(\Varid{a},\Varid{b},\Varid{c})\mathrel{=}\Varid{a}}.
For example, the \bx{} below finds matching states by asking the user,
writing to and reading from the terminal.
\begin{hscode}\SaveRestoreHook
\column{B}{@{}>{\hspre}l<{\hspost}@{}}%
\column{17}{@{}>{\hspre}l<{\hspost}@{}}%
\column{21}{@{}>{\hspre}l<{\hspost}@{}}%
\column{E}{@{}>{\hspre}l<{\hspost}@{}}%
\>[B]{}\Varid{dynamicIOBX}\mathbin{::}{}\<[17]%
\>[17]{}(\Conid{Eq}\;\alpha,\Conid{Eq}\;\beta,\Conid{Show}\;\alpha,\Conid{Show}\;\beta,\Conid{Read}\;\alpha,\Conid{Read}\;\beta)\Rightarrow {}\<[E]%
\\
\>[17]{}\Conid{StateTBX}\;\Conid{IO}\;((\alpha,\beta),[\mskip1.5mu ((\alpha,\beta),\beta)\mskip1.5mu],[\mskip1.5mu ((\alpha,\beta),\alpha)\mskip1.5mu])\;\alpha\;\beta{}\<[E]%
\\
\>[B]{}\Varid{dynamicIOBX}\mathrel{=}\Varid{dynamicBX}\;\Varid{matchIO}\;(\Varid{flip}\;\Varid{matchIO}){}\<[E]%
\\[\blanklineskip]%
\>[B]{}\Varid{matchIO}\mathbin{::}(\Conid{Show}\;\alpha,\Conid{Show}\;\beta,\Conid{Read}\;\beta)\Rightarrow \alpha\hsarrow{\rightarrow }{\mathpunct{.}}\beta\hsarrow{\rightarrow }{\mathpunct{.}}\Conid{IO}\;\beta{}\<[E]%
\\
\>[B]{}\Varid{matchIO}\;\Varid{a}\;\Varid{b}\mathrel{=}\mathbf{do}\;\{\mskip1.5mu {}\<[21]%
\>[21]{}\Varid{putStrLn}\;(\text{\tt \char34 Setting~\char34}\plus \Varid{show}\;\Varid{a});{}\<[E]%
\\
\>[21]{}\Varid{putStr}\;(\text{\tt \char34 Replacement~for~\char34}\plus \Varid{show}\;\Varid{b}\plus \text{\tt \char34 ?\char34});{}\<[E]%
\\
\>[21]{}\Varid{s}\leftarrow \Varid{getLine};\Varid{return}\;(\Varid{read}\;\Varid{s})\mskip1.5mu\}{}\<[E]%
\ColumnHook
\end{hscode}\resethooks
An alternative way to find matching states, for a finite state space, would be
to search an enumeration \ensuremath{[\mskip1.5mu \Varid{minBound}..\Varid{maxBound}\mskip1.5mu]} of the possible values,
checking against a fixed oracle \ensuremath{\Varid{p}}:
\begin{hscode}\SaveRestoreHook
\column{B}{@{}>{\hspre}l<{\hspost}@{}}%
\column{3}{@{}>{\hspre}l<{\hspost}@{}}%
\column{E}{@{}>{\hspre}l<{\hspost}@{}}%
\>[B]{}\Varid{dynamicSearchBX}\mathbin{::}{}\<[E]%
\\
\>[B]{}\hsindent{3}{}\<[3]%
\>[3]{}(\Conid{Eq}\;\alpha,\Conid{Eq}\;\beta,\Conid{Enum}\;\alpha,\Conid{Bounded}\;\alpha,\Conid{Enum}\;\beta,\Conid{Bounded}\;\beta)\Rightarrow {}\<[E]%
\\
\>[B]{}\hsindent{3}{}\<[3]%
\>[3]{}(\alpha\hsarrow{\rightarrow }{\mathpunct{.}}\beta\hsarrow{\rightarrow }{\mathpunct{.}}\Conid{Bool})\hsarrow{\rightarrow }{\mathpunct{.}}{}\<[E]%
\\
\>[B]{}\hsindent{3}{}\<[3]%
\>[3]{}\Conid{StateTBX}\;\Conid{Maybe}\;((\alpha,\beta),[\mskip1.5mu ((\alpha,\beta),\beta)\mskip1.5mu],[\mskip1.5mu ((\alpha,\beta),\alpha)\mskip1.5mu])\;\alpha\;\beta{}\<[E]%
\\
\>[B]{}\Varid{dynamicSearchBX}\;\Varid{p}\mathrel{=}\Varid{dynamicBX}\;(\Varid{search}\;\Varid{p})\;(\Varid{flip}\;(\Varid{search}\;(\Varid{flip}\;\Varid{p}))){}\<[E]%
\\[\blanklineskip]%
\>[B]{}\Varid{search}\mathbin{::}(\Conid{Enum}\;\beta,\Conid{Bounded}\;\beta)\Rightarrow (\alpha\hsarrow{\rightarrow }{\mathpunct{.}}\beta\hsarrow{\rightarrow }{\mathpunct{.}}\Conid{Bool})\hsarrow{\rightarrow }{\mathpunct{.}}\alpha\hsarrow{\rightarrow }{\mathpunct{.}}\beta\hsarrow{\rightarrow }{\mathpunct{.}}\Conid{Maybe}\;\beta{}\<[E]%
\\
\>[B]{}\Varid{search}\;\Varid{p}\;\Varid{a}\;\anonymous \mathrel{=}\Varid{find}\;(\Varid{p}\;\Varid{a})\;[\mskip1.5mu \Varid{minBound}..\Varid{maxBound}\mskip1.5mu]{}\<[E]%
\ColumnHook
\end{hscode}\resethooks
\endswithdisplay
\end{example*}

\begin{proposition}\label{prop:dynamic-wb}
For any \ensuremath{\Varid{f}}, \ensuremath{\Varid{g}}, the 
\bx{} \ensuremath{\Varid{dynamicBX}\;\Varid{f}\;\Varid{g}} is well-behaved
(it is clearly transparent).
\end{proposition}

\section{Related work}\label{sec:related}
 
\paragraph{Bidirectional programming} This has a large literature; 
work on view update flourished in the early 1980s, and the term
`lens' was coined in 2005 \cite{Foster:2005}. 
The GRACE report~\cite{GraceReport} surveys work since.
We mention here only the closest related work.

Pacheco \etal~\cite{pacheco14pepm} present `putback-style' asymmetric
lenses; \ie\ their laws and combinators focus only on the `put'
functions, of type \ensuremath{\Conid{Maybe}\;\Varid{s}\hsarrow{\rightarrow }{\mathpunct{.}}\Varid{v}\hsarrow{\rightarrow }{\mathpunct{.}}\Varid{m}\;\Varid{s}}, for some monad \ensuremath{\Varid{m}}.
This allows for effects, and they include a combinator \ensuremath{\Varid{effect}} that
applies a monad morphism to a lens. Their laws assume that the
monad \ensuremath{\Varid{m}} admits a membership operation \ensuremath{(\in)\mathbin{::}\Varid{a}\hsarrow{\rightarrow }{\mathpunct{.}}\Varid{m}\;\Varid{a}\hsarrow{\rightarrow }{\mathpunct{.}}\Conid{Bool}}. For
monads such as \ensuremath{\Conid{List}} or
\ensuremath{\Conid{Maybe}} that support such an operation, 
their laws are similar to ours, but their approach does not appear to
work for other important monads such as \ensuremath{\Conid{IO}} or \ensuremath{\Conid{State}}.
In Divi\'ansky's monadic lens proposal~\cite{reddit}, the \ensuremath{\Varid{get}}
function is monadic, so in principle it too can have
side-effects; as we have seen, this possibility significantly
complicates composition.
 
Johnson and Rosebrugh~\cite{johnson14bx} analyse symmetric lenses in a
general setting of categories with finite products, 
showing that they correspond to
pairs of (asymmetric) lenses with a common source. Our composition for
\ensuremath{\Conid{StateTBX}}s uses a similar idea; however, their construction does not
apply directly to monadic lenses, because the Kleisli category of a
monad does not necessarily have finite products. They also identify a
different notion of equivalence of symmetric lenses.

Elsewhere, we have considered a coalgebraic approach to \bx{} \cite{coalg-bx:cmcs2014}.  
Relating such an approach to the one presented here, 
and investigating their associated equivalences, 
is an interesting future direction of research.

Macedo \etal~\cite{macedo14bx} observe that most bx research 
deals with 
just two models, but many tools and specifications, such as
QVT-R~\cite{qvtr},
allow relating 
multiple models. Our notion of bx generalises
straightforwardly to such multidirectional transformations, provided we
only update one source value at a time.

\paragraph{Monads and algebraic effects} The vast literature on
combining and reasoning about
monads~\cite{Jones&Duponcheel93:Composing,luth02icfp,liang95popl,Mossakowski*2010:Generic}
stems from Moggi's work~\cite{moggi}; we have shown
that bidirectionality can be viewed as another kind of computational
effect, so results about monads can be applied to bidirectional
computation.

A promising area to investigate is the \emph{algebraic} treatment of
effects~\cite{Plotkin&Power2002:Notions}, particularly recent work on
combining effects using operations such as sum and
tensor~\cite{hyland06tcs} and \emph{handlers} of algebraic
effects~\cite{plotkin13lmcs,kammar13icfp,brady13icfp}. It appears
straightforward to view entangled state as generated by operations and equations
analogous to the bx laws. 
What is less clear is whether operations such as composition
can be defined in terms of effect handlers: so far, the theory underlying
handlers~\cite{plotkin13lmcs} does not support `tensor-like' combinations
of computations.
We therefore leave this investigation for future work.

The relationship between lenses and state monad morphisms is
intriguing, and hints of it appear in previous work on
\emph{compositional references} by Kagawa~\cite{kagawa97icfp}.  
The fact that lenses determine state monad morphisms 
(Definition~\ref{def:theta})
appears to be
folklore; Shkaravska~\cite{Shkaravska2005:Side}
stated this result in a talk, 
and it is implicit in the design of the Haskell \texttt{Data.Lens} library \cite{lenseslibrary},
but we are not aware of any previous published proof.

\section{Conclusions and further work\label{sec:concs}}

We have presented a semantic framework for effectful bidirectional
transformations (bx).  Our framework encompasses
symmetric lenses, which (as is well-known) in turn encompass other
approaches to bx
such as asymmetric lenses \cite{lens-toplas}
and relational bx \cite{stevens09:sosym}; we have also given examples of
other monadic effects. This is an advance on the
state of the art of bidirectional transformations: 
ours is the first formalism to reconcile the stateful behavior of bx
with other effects such as nondeterminism, I/O or exceptions
with due attention paid to the corresponding laws. We have
defined composition for effectful bx and shown that composition is
associative and satisfies identity laws, up to a suitable notion of
equivalence based on monad isomorphisms. We have also demonstrated
some combinators suitable for grounding the design of future bx languages
based on our approach.

In future we plan to investigate equivalence, and the relationship
with the work of Johnson and Rosebrugh~\cite{johnson14bx},
further. The equivalence we present here is finer than theirs, and also finer
than the equivalence for symmetric lenses presented by Hofmann
\etal~\cite{symlens}. Early investigations, guided by an alternative
coalgebraic presentation \cite{coalg-bx:cmcs2014} of our
framework, suggest that the situation
for bx may be similar to that for processes given as labelled
transition systems: it is possible to give many different equivalences
which are `right' according to different criteria. We think the one we
have given here is the finest reasonable, equating just enough bx to
make composition work.  Another interesting area for exploration is
formalisation of our (on-paper) proofs.

Our framework provides a foundation for future languages, libraries,
or tools for effectful bx, and there are several natural next steps in
this direction.  In this paper we explored only the case
where the get and set operations read or write complete states,
but our framework allows for generalisation beyond the category \ensuremath{\ensuremath{\mathsf{Set}}}
and hence, perhaps, into delta-based bx~\cite{diskin11jot}, edit
lenses \cite{hofmann12popl} and ordered updates
\cite{Hegner2004:Order}, in which the operations record state changes rather than complete states. 
Another natural next step is to explore
different \emph{witness structures} encapsulating the dependencies
between views, in order to formulate candidate
principles of Least Change (informally, that ``a bx should not
change more than it has to in order to restore consistency'') that are
more practical and flexible than those that can be stated in terms of
views alone.

\section*{Acknowledgements}

Preliminary work on this topic was presented orally at the BIRS workshop
13w5115 
in December 2013; a four-page abstract
\cite{entangled-bx2014} of some of the ideas in this paper appeared
at the Athens BX Workshop in March 2014;
and a short presentation on an alternative coalgebraic approach 
\cite{coalg-bx:cmcs2014} was made at CMCS 2014.
We thank the organisers of and participants at those meetings 
and the anonymous reviewers
for their helpful
comments.
The work was supported by the UK EPSRC-funded project \textit{A
  Theory of Least Change for Bidirectional Transformations}
\cite{tlcbx} (EP/K020218/1, EP/K020919/1).

\bibliographystyle{plain}
\bibliography{strings,paper}

\newpage \appendix
\section*{Appendices}

We present 
proofs of the lemmas and theorems omitted from the body of the paper,
and expand the code that was abbreviated in the paper.

\section{Proofs from Section~\ref{sec:background}}

\restatableLemma{lem:discardable}
\begin{lem:discardable}
If the laws \ensuremath{\mathrm{(GG)}} and \ensuremath{\mathrm{(GS)}} are satisfied, then unused \ensuremath{\Varid{get}}s are discardable:
\begin{hscode}\SaveRestoreHook
\column{B}{@{}>{\hspre}l<{\hspost}@{}}%
\column{46}{@{}>{\hspre}l<{\hspost}@{}}%
\column{E}{@{}>{\hspre}l<{\hspost}@{}}%
\>[B]{}\mathbf{do}\;\{\mskip1.5mu \anonymous \leftarrow \Varid{get};\Varid{m}\mskip1.5mu\}{}\<[46]%
\>[46]{}\mathrel{=}\mathbf{do}\;\{\mskip1.5mu \Varid{m}\mskip1.5mu\}{}\<[E]%
\ColumnHook
\end{hscode}\resethooks
\endswithdisplay
\end{lem:discardable}

\begin{proof}
\mbox{}
\begin{hscode}\SaveRestoreHook
\column{B}{@{}>{\hspre}c<{\hspost}@{}}%
\column{BE}{@{}l@{}}%
\column{3}{@{}>{\hspre}l<{\hspost}@{}}%
\column{5}{@{}>{\hspre}l<{\hspost}@{}}%
\column{E}{@{}>{\hspre}l<{\hspost}@{}}%
\>[3]{}\mathbf{do}\;\{\mskip1.5mu \Varid{s}\leftarrow \Varid{get};\Varid{m}\mskip1.5mu\}{}\<[E]%
\\
\>[B]{}\mathrel{=}{}\<[BE]%
\>[5]{}\mbox{\commentbegin  \ensuremath{\mathrm{(GS)}}  \commentend}{}\<[E]%
\\
\>[B]{}\hsindent{3}{}\<[3]%
\>[3]{}\mathbf{do}\;\{\mskip1.5mu \Varid{s}\leftarrow \Varid{get};\Varid{s'}\leftarrow \Varid{get};\Varid{set}\;\Varid{s'};\Varid{m}\mskip1.5mu\}{}\<[E]%
\\
\>[B]{}\mathrel{=}{}\<[BE]%
\>[5]{}\mbox{\commentbegin  \ensuremath{\mathrm{(GG)}}  \commentend}{}\<[E]%
\\
\>[B]{}\hsindent{3}{}\<[3]%
\>[3]{}\mathbf{do}\;\{\mskip1.5mu \Varid{s}\leftarrow \Varid{get};\mathbf{let}\;\Varid{s'}\mathrel{=}\Varid{s};\Varid{set}\;\Varid{s'};\Varid{m}\mskip1.5mu\}{}\<[E]%
\\
\>[B]{}\mathrel{=}{}\<[BE]%
\>[5]{}\mbox{\commentbegin  \ensuremath{\mathbf{let}}  \commentend}{}\<[E]%
\\
\>[B]{}\hsindent{3}{}\<[3]%
\>[3]{}\mathbf{do}\;\{\mskip1.5mu \Varid{s}\leftarrow \Varid{get};\Varid{set}\;\Varid{s};\Varid{m}\mskip1.5mu\}{}\<[E]%
\\
\>[B]{}\mathrel{=}{}\<[BE]%
\>[5]{}\mbox{\commentbegin  \ensuremath{\mathrm{(GS)}}  \commentend}{}\<[E]%
\\
\>[B]{}\hsindent{3}{}\<[3]%
\>[3]{}\mathbf{do}\;\{\mskip1.5mu \Varid{m}\mskip1.5mu\}{}\<[E]%
\ColumnHook
\end{hscode}\resethooks
\endswithdisplay
\end{proof}

\restatableLemma{lem:liftings-commute}
\begin{lem:liftings-commute}
Suppose \ensuremath{\Varid{a},\Varid{b}} are distinct variables not appearing in expression \ensuremath{\Varid{m}}. 
Then:
\begin{hscode}\SaveRestoreHook
\column{B}{@{}>{\hspre}l<{\hspost}@{}}%
\column{47}{@{}>{\hspre}l<{\hspost}@{}}%
\column{E}{@{}>{\hspre}l<{\hspost}@{}}%
\>[B]{}\mathbf{do}\;\{\mskip1.5mu \Varid{a}\leftarrow \Varid{get};\Varid{b}\leftarrow \Varid{lift}\;\Varid{m};\Varid{return}\;(\Varid{a},\Varid{b})\mskip1.5mu\}{}\<[47]%
\>[47]{}\mathrel{=}\mathbf{do}\;\{\mskip1.5mu \Varid{b}\leftarrow \Varid{lift}\;\Varid{m};\Varid{a}\leftarrow \Varid{get};\Varid{return}\;(\Varid{a},\Varid{b})\mskip1.5mu\}{}\<[E]%
\\
\>[B]{}\mathbf{do}\;\{\mskip1.5mu \Varid{set}\;\Varid{a};\Varid{b}\leftarrow \Varid{lift}\;\Varid{m};\Varid{return}\;\Varid{b}\mskip1.5mu\}{}\<[47]%
\>[47]{}\mathrel{=}\mathbf{do}\;\{\mskip1.5mu \Varid{b}\leftarrow \Varid{lift}\;\Varid{m};\Varid{set}\;\Varid{a};\Varid{return}\;\Varid{b}\mskip1.5mu\}{}\<[E]%
\ColumnHook
\end{hscode}\resethooks
\endswithdisplay
\end{lem:liftings-commute}

\begin{proof}
For the first part:
  \begin{hscode}\SaveRestoreHook
\column{B}{@{}>{\hspre}c<{\hspost}@{}}%
\column{BE}{@{}l@{}}%
\column{4}{@{}>{\hspre}l<{\hspost}@{}}%
\column{6}{@{}>{\hspre}l<{\hspost}@{}}%
\column{10}{@{}>{\hspre}l<{\hspost}@{}}%
\column{E}{@{}>{\hspre}l<{\hspost}@{}}%
\>[4]{}\mathbf{do}\;\{\mskip1.5mu \Varid{a}\leftarrow \Varid{get};\Varid{b}\leftarrow \Varid{lift}\;\Varid{m};\Varid{return}\;(\Varid{a},\Varid{b})\mskip1.5mu\}\;\Varid{s}{}\<[E]%
\\
\>[B]{}\mathrel{=}{}\<[BE]%
\>[6]{}\mbox{\commentbegin  Definitions of bind, \ensuremath{\Varid{get}}, \ensuremath{\Varid{return}}  \commentend}{}\<[E]%
\\
\>[B]{}\hsindent{4}{}\<[4]%
\>[4]{}\mathbf{do}\;\{\mskip1.5mu {}\<[10]%
\>[10]{}(\Varid{a},\Varid{s'})\leftarrow \Varid{return}\;(\Varid{s},\Varid{s});(\Varid{b},\Varid{s''})\leftarrow \mathbf{do}\;\{\mskip1.5mu \Varid{b'}\leftarrow \Varid{m};\Varid{return}\;(\Varid{b'},\Varid{s'})\mskip1.5mu\};\Varid{return}\;((\Varid{a},\Varid{b}),\Varid{s''})\mskip1.5mu\}{}\<[E]%
\\
\>[B]{}\mathrel{=}{}\<[BE]%
\>[6]{}\mbox{\commentbegin  monad unit  \commentend}{}\<[E]%
\\
\>[B]{}\hsindent{4}{}\<[4]%
\>[4]{}\mathbf{do}\;\{\mskip1.5mu {}\<[10]%
\>[10]{}(\Varid{b},\Varid{s''})\leftarrow \mathbf{do}\;\{\mskip1.5mu \Varid{b'}\leftarrow \Varid{m};\Varid{return}\;(\Varid{b'},\Varid{s})\mskip1.5mu\};\Varid{return}\;((\Varid{s},\Varid{b}),\Varid{s''})\mskip1.5mu\}{}\<[E]%
\\
\>[B]{}\mathrel{=}{}\<[BE]%
\>[6]{}\mbox{\commentbegin  monad associativity  \commentend}{}\<[E]%
\\
\>[B]{}\hsindent{4}{}\<[4]%
\>[4]{}\mathbf{do}\;\{\mskip1.5mu {}\<[10]%
\>[10]{}\Varid{b'}\leftarrow \Varid{m};(\Varid{b},\Varid{s''})\leftarrow \Varid{return}\;(\Varid{b'},\Varid{s});\Varid{return}\;((\Varid{s},\Varid{b}),\Varid{s''})\mskip1.5mu\}{}\<[E]%
\\
\>[B]{}\mathrel{=}{}\<[BE]%
\>[6]{}\mbox{\commentbegin  monad unit  \commentend}{}\<[E]%
\\
\>[B]{}\hsindent{4}{}\<[4]%
\>[4]{}\mathbf{do}\;\{\mskip1.5mu \Varid{b'}\leftarrow \Varid{m};\Varid{return}\;((\Varid{s},\Varid{b'}),\Varid{s})\mskip1.5mu\}{}\<[E]%
\\
\>[B]{}\mathrel{=}{}\<[BE]%
\>[6]{}\mbox{\commentbegin  reversing above steps  \commentend}{}\<[E]%
\\
\>[B]{}\hsindent{4}{}\<[4]%
\>[4]{}\mathbf{do}\;\{\mskip1.5mu {}\<[10]%
\>[10]{}(\Varid{b},\Varid{s'})\leftarrow \mathbf{do}\;\{\mskip1.5mu \Varid{b'}\leftarrow \Varid{m};\Varid{return}\;(\Varid{b'},\Varid{s})\mskip1.5mu\};(\Varid{a},\Varid{s''})\leftarrow \Varid{return}\;(\Varid{s'},\Varid{s'});\Varid{return}\;((\Varid{a},\Varid{b}),\Varid{s''})\mskip1.5mu\}{}\<[E]%
\\
\>[B]{}\mathrel{=}{}\<[BE]%
\>[6]{}\mbox{\commentbegin  definition  \commentend}{}\<[E]%
\\
\>[B]{}\hsindent{4}{}\<[4]%
\>[4]{}\mathbf{do}\;\{\mskip1.5mu \Varid{b}\leftarrow \Varid{lift}\;\Varid{m};\Varid{a}\leftarrow \Varid{get};\Varid{return}\;(\Varid{a},\Varid{b})\mskip1.5mu\}\;\Varid{s}{}\<[E]%
\ColumnHook
\end{hscode}\resethooks
so the desired result holds by eta equivalence.
For the second part:
  \begin{hscode}\SaveRestoreHook
\column{B}{@{}>{\hspre}c<{\hspost}@{}}%
\column{BE}{@{}l@{}}%
\column{4}{@{}>{\hspre}l<{\hspost}@{}}%
\column{6}{@{}>{\hspre}l<{\hspost}@{}}%
\column{10}{@{}>{\hspre}l<{\hspost}@{}}%
\column{E}{@{}>{\hspre}l<{\hspost}@{}}%
\>[4]{}\mathbf{do}\;\{\mskip1.5mu \Varid{set}\;\Varid{a};\Varid{b}\leftarrow \Varid{lift}\;\Varid{m};\Varid{return}\;\Varid{b}\mskip1.5mu\}\;\Varid{s}{}\<[E]%
\\
\>[B]{}\mathrel{=}{}\<[BE]%
\>[6]{}\mbox{\commentbegin  definitions of bind, \ensuremath{\Varid{lift}}, \ensuremath{\Varid{set}}, \ensuremath{\Varid{return}}  \commentend}{}\<[E]%
\\
\>[B]{}\hsindent{4}{}\<[4]%
\>[4]{}\mathbf{do}\;\{\mskip1.5mu {}\<[10]%
\>[10]{}(\anonymous ,\Varid{s'})\leftarrow \Varid{return}\;((),\Varid{a});(\Varid{b},\Varid{s''})\leftarrow \mathbf{do}\;\{\mskip1.5mu \Varid{b'}\leftarrow \Varid{m};\Varid{return}\;(\Varid{b'},\Varid{s'})\mskip1.5mu\};\Varid{return}\;(\Varid{b},\Varid{s''})\mskip1.5mu\}{}\<[E]%
\\
\>[B]{}\mathrel{=}{}\<[BE]%
\>[6]{}\mbox{\commentbegin  monad unit  \commentend}{}\<[E]%
\\
\>[B]{}\hsindent{4}{}\<[4]%
\>[4]{}\mathbf{do}\;\{\mskip1.5mu (\Varid{b},\Varid{s''})\leftarrow \mathbf{do}\;\{\mskip1.5mu \Varid{b'}\leftarrow \Varid{m};\Varid{return}\;(\Varid{b'},\Varid{a})\mskip1.5mu\};\Varid{return}\;(\Varid{b},\Varid{s''})\mskip1.5mu\}{}\<[E]%
\\
\>[B]{}\mathrel{=}{}\<[BE]%
\>[6]{}\mbox{\commentbegin  monad associativity  \commentend}{}\<[E]%
\\
\>[B]{}\hsindent{4}{}\<[4]%
\>[4]{}\mathbf{do}\;\{\mskip1.5mu {}\<[10]%
\>[10]{}\Varid{b'}\leftarrow \Varid{m};(\Varid{b},\Varid{s''})\leftarrow \Varid{return}\;(\Varid{b'},\Varid{a});\Varid{return}\;(\Varid{b},\Varid{s''})\mskip1.5mu\}{}\<[E]%
\\
\>[B]{}\mathrel{=}{}\<[BE]%
\>[6]{}\mbox{\commentbegin  monad unit  \commentend}{}\<[E]%
\\
\>[B]{}\hsindent{4}{}\<[4]%
\>[4]{}\mathbf{do}\;\{\mskip1.5mu {}\<[10]%
\>[10]{}\Varid{b'}\leftarrow \Varid{m};\Varid{return}\;(\Varid{b'},\Varid{a})\mskip1.5mu\}{}\<[E]%
\\
\>[B]{}\mathrel{=}{}\<[BE]%
\>[6]{}\mbox{\commentbegin  reversing above steps  \commentend}{}\<[E]%
\\
\>[B]{}\hsindent{4}{}\<[4]%
\>[4]{}\mathbf{do}\;\{\mskip1.5mu {}\<[10]%
\>[10]{}(\Varid{b},\Varid{s'})\leftarrow \mathbf{do}\;\{\mskip1.5mu \Varid{b'}\leftarrow \Varid{m};\Varid{return}\;(\Varid{b'},\Varid{s})\mskip1.5mu\};(\anonymous ,\Varid{s''})\leftarrow \Varid{return}\;((),\Varid{a});\Varid{return}\;(\Varid{b},\Varid{s''})\mskip1.5mu\}{}\<[E]%
\\
\>[B]{}\mathrel{=}{}\<[BE]%
\>[6]{}\mbox{\commentbegin  definition  \commentend}{}\<[E]%
\\
\>[B]{}\hsindent{4}{}\<[4]%
\>[4]{}\mathbf{do}\;\{\mskip1.5mu \Varid{b}\leftarrow \Varid{lift}\;\Varid{m};\Varid{set}\;\Varid{a};\Varid{return}\;\Varid{b}\mskip1.5mu\}\;\Varid{s}{}\<[E]%
\ColumnHook
\end{hscode}\resethooks
so again the result holds by eta-equivalence.
\end{proof}

\restatableLemma{lem:StateT-data-refinement}
\begin{lem:StateT-data-refinement}
Given an arbitrary monad \ensuremath{\Conid{T}}, not assumed to be an instance of \ensuremath{\Conid{StateT}},
with operations \ensuremath{\Varid{get_T}\mathbin{::}\Conid{T}\;\Conid{S}} and \ensuremath{\Varid{set_T}\mathbin{::}\Conid{S}\hsarrow{\rightarrow }{\mathpunct{.}}\Conid{T}\;()} for a type \ensuremath{\Conid{S}}, such that \ensuremath{\Varid{get_T}} and \ensuremath{\Varid{set_T}} satisfy the
laws \ensuremath{\mathrm{(GG)}}, \ensuremath{\mathrm{(GS)}}, and \ensuremath{\mathrm{(SG)}} 
of Definition~\ref{def:state-monad-transformer}, then
there is a data refinement from \ensuremath{\Conid{T}} to \ensuremath{\Conid{StateT}\;\Conid{S}\;\Conid{T}}.
\end{lem:StateT-data-refinement}

\begin{proof}
Define the abstraction function \ensuremath{\Varid{abs}} from \ensuremath{\Conid{StateT}\;\Conid{S}\;\Conid{T}} to \ensuremath{\Conid{T}} and the
reification function \ensuremath{\Varid{conc}} in the opposite direction by
\begin{hscode}\SaveRestoreHook
\column{B}{@{}>{\hspre}l<{\hspost}@{}}%
\column{9}{@{}>{\hspre}l<{\hspost}@{}}%
\column{E}{@{}>{\hspre}l<{\hspost}@{}}%
\>[B]{}\Varid{abs}\;\Varid{m}{}\<[9]%
\>[9]{}\mathrel{=}\mathbf{do}\;\{\mskip1.5mu \Varid{s}\leftarrow \Varid{get_T};(\Varid{a},\Varid{s'})\leftarrow \Varid{m}\;\Varid{s};\Varid{set_T}\;\Varid{s'};\Varid{return}\;\Varid{a}\mskip1.5mu\}{}\<[E]%
\\[\blanklineskip]%
\>[B]{}\Varid{conc}\;\Varid{m}{}\<[9]%
\>[9]{}\mathrel{=}\lambda \hslambda \Varid{s}\hsarrow{\rightarrow }{\mathpunct{.}}\mathbf{do}\;\{\mskip1.5mu \Varid{a}\leftarrow \Varid{m};\Varid{s'}\leftarrow \Varid{get_T};\Varid{return}\;(\Varid{a},\Varid{s'})\mskip1.5mu\}{}\<[E]%
\\
\>[9]{}\mathrel{=}\mathbf{do}\;\{\mskip1.5mu \Varid{a}\leftarrow \Varid{lift}\;\Varid{m};\Varid{s'}\leftarrow \Varid{lift}\;\Varid{get_T};\Varid{set}\;\Varid{s'};\Varid{return}\;\Varid{a}\mskip1.5mu\}{}\<[E]%
\ColumnHook
\end{hscode}\resethooks
On account of the \ensuremath{\Varid{get_T}} and \ensuremath{\Varid{set_T}} operations that it provides, monad
\ensuremath{\Conid{T}} is implicitly recording a state \ensuremath{\Conid{S}}; the idea of the data
refinement is to track this state explicitly in monad \ensuremath{\Conid{StateT}\;\Conid{S}\;\Conid{T}}. We
say that a computation \ensuremath{\Varid{m}} in \ensuremath{\Conid{StateT}\;\Conid{S}\;\Conid{T}} is \emph{synchronised} if
on completion the inner implicit and the outer explicit \ensuremath{\Conid{S}} values agree:
\begin{hscode}\SaveRestoreHook
\column{B}{@{}>{\hspre}l<{\hspost}@{}}%
\column{E}{@{}>{\hspre}l<{\hspost}@{}}%
\>[B]{}\mathbf{do}\;\{\mskip1.5mu \Varid{a}\leftarrow \Varid{m};\Varid{s'}\leftarrow \Varid{get};\Varid{return}\;(\Varid{a},\Varid{s'})\mskip1.5mu\}\mathrel{=}\mathbf{do}\;\{\mskip1.5mu \Varid{a}\leftarrow \Varid{m};\Varid{s''}\leftarrow \Varid{lift}\;\Varid{get_T};\Varid{return}\;(\Varid{a},\Varid{s''})\mskip1.5mu\}{}\<[E]%
\ColumnHook
\end{hscode}\resethooks
or equivalently, as computations in \ensuremath{\Conid{T}},
\begin{hscode}\SaveRestoreHook
\column{B}{@{}>{\hspre}l<{\hspost}@{}}%
\column{E}{@{}>{\hspre}l<{\hspost}@{}}%
\>[B]{}\mathbf{do}\;\{\mskip1.5mu (\Varid{a},\Varid{s'})\leftarrow \Varid{m}\;\Varid{s};\Varid{return}\;(\Varid{a},\Varid{s'})\mskip1.5mu\}\mathrel{=}\mathbf{do}\;\{\mskip1.5mu (\Varid{a},\Varid{s'})\leftarrow \Varid{m}\;\Varid{s};\Varid{s''}\leftarrow \Varid{get_T};\Varid{return}\;(\Varid{a},\Varid{s''})\mskip1.5mu\}{}\<[E]%
\ColumnHook
\end{hscode}\resethooks
It is straightforward to check that \ensuremath{\Varid{return}\;\Varid{a}} is synchronised, that
bind preserves synchronisation, and that \ensuremath{\Varid{conc}} yields only
synchronised computations.

We have to verify the three conditions of \Definition~\ref{def:data-refinement}.
For the first, we have to show that \ensuremath{\Varid{conc}} distributes over \ensuremath{(\bind )}; we have
\begin{hscode}\SaveRestoreHook
\column{B}{@{}>{\hspre}c<{\hspost}@{}}%
\column{BE}{@{}l@{}}%
\column{3}{@{}>{\hspre}l<{\hspost}@{}}%
\column{4}{@{}>{\hspre}l<{\hspost}@{}}%
\column{5}{@{}>{\hspre}l<{\hspost}@{}}%
\column{E}{@{}>{\hspre}l<{\hspost}@{}}%
\>[3]{}\Varid{conc}\;(\Varid{m}\bind \Varid{k}){}\<[E]%
\\
\>[B]{}\mathrel{=}{}\<[BE]%
\>[5]{}\mbox{\commentbegin  \ensuremath{\Varid{conc}}  \commentend}{}\<[E]%
\\
\>[B]{}\hsindent{3}{}\<[3]%
\>[3]{}\mathbf{do}\;\{\mskip1.5mu \Varid{b}\leftarrow \Varid{lift}\;(\Varid{m}\bind \Varid{k});\Varid{s''}\leftarrow \Varid{lift}\;\Varid{get_T};\Varid{set}\;\Varid{s''};\Varid{return}\;\Varid{b}\mskip1.5mu\}{}\<[E]%
\\
\>[B]{}\mathrel{=}{}\<[BE]%
\>[5]{}\mbox{\commentbegin  \ensuremath{\Varid{lift}} and bind  \commentend}{}\<[E]%
\\
\>[B]{}\hsindent{3}{}\<[3]%
\>[3]{}\mathbf{do}\;\{\mskip1.5mu \Varid{a}\leftarrow \Varid{lift}\;\Varid{m};\Varid{b}\leftarrow \Varid{lift}\;(\Varid{k}\;\Varid{a});\Varid{s''}\leftarrow \Varid{lift}\;\Varid{get_T};\Varid{set}\;\Varid{s''};\Varid{return}\;\Varid{b}\mskip1.5mu\}{}\<[E]%
\\
\>[B]{}\mathrel{=}{}\<[BE]%
\>[4]{}\mbox{\commentbegin  \ensuremath{\Varid{get_T}} is discardable  \commentend}{}\<[E]%
\\
\>[B]{}\hsindent{3}{}\<[3]%
\>[3]{}\mathbf{do}\;\{\mskip1.5mu \Varid{a}\leftarrow \Varid{lift}\;\Varid{m};\Varid{s'}\leftarrow \Varid{lift}\;\Varid{get_T};\Varid{b}\leftarrow \Varid{lift}\;(\Varid{k}\;\Varid{a});\Varid{s''}\leftarrow \Varid{lift}\;\Varid{get_T};\Varid{set}\;\Varid{s''};\Varid{return}\;\Varid{b}\mskip1.5mu\}{}\<[E]%
\\
\>[B]{}\mathrel{=}{}\<[BE]%
\>[4]{}\mbox{\commentbegin  \ensuremath{\mathrm{(SS)}} for \ensuremath{\Conid{StateT}}  \commentend}{}\<[E]%
\\
\>[B]{}\hsindent{3}{}\<[3]%
\>[3]{}\mathbf{do}\;\{\mskip1.5mu \Varid{a}\leftarrow \Varid{lift}\;\Varid{m};\Varid{s'}\leftarrow \Varid{lift}\;\Varid{get_T};\Varid{b}\leftarrow \Varid{lift}\;(\Varid{k}\;\Varid{a});\Varid{s''}\leftarrow \Varid{lift}\;\Varid{get_T};\Varid{set}\;\Varid{s'};\Varid{set}\;\Varid{s''};\Varid{return}\;\Varid{b}\mskip1.5mu\}{}\<[E]%
\\
\>[B]{}\mathrel{=}{}\<[BE]%
\>[4]{}\mbox{\commentbegin  Lemma~\ref{lem:liftings-commute}  \commentend}{}\<[E]%
\\
\>[B]{}\hsindent{3}{}\<[3]%
\>[3]{}\mathbf{do}\;\{\mskip1.5mu \Varid{a}\leftarrow \Varid{lift}\;\Varid{m};\Varid{s'}\leftarrow \Varid{lift}\;\Varid{get_T};\Varid{set}\;\Varid{s'};\Varid{b}\leftarrow \Varid{lift}\;(\Varid{k}\;\Varid{a});\Varid{s''}\leftarrow \Varid{lift}\;\Varid{get_T};\Varid{set}\;\Varid{s''};\Varid{return}\;\Varid{b}\mskip1.5mu\}{}\<[E]%
\\
\>[B]{}\mathrel{=}{}\<[BE]%
\>[4]{}\mbox{\commentbegin  bind  \commentend}{}\<[E]%
\\
\>[B]{}\hsindent{3}{}\<[3]%
\>[3]{}\mathbf{do}\;\{\mskip1.5mu \Varid{a}\leftarrow \Varid{lift}\;\Varid{m};\Varid{s'}\leftarrow \Varid{lift}\;\Varid{get_T};\Varid{set}\;\Varid{s'};\Varid{return}\;\Varid{a}\mskip1.5mu\}\bind \lambda \hslambda \Varid{a}\hsarrow{\rightarrow }{\mathpunct{.}}{}\<[E]%
\\
\>[B]{}\hsindent{3}{}\<[3]%
\>[3]{}\mathbf{do}\;\{\mskip1.5mu \Varid{b}\leftarrow \Varid{lift}\;(\Varid{k}\;\Varid{a});\Varid{s''}\leftarrow \Varid{lift}\;\Varid{get_T};\Varid{set}\;\Varid{s''};\Varid{return}\;\Varid{b}\mskip1.5mu\}{}\<[E]%
\\
\>[B]{}\mathrel{=}{}\<[BE]%
\>[4]{}\mbox{\commentbegin  \ensuremath{\Varid{conc}}  \commentend}{}\<[E]%
\\
\>[B]{}\hsindent{3}{}\<[3]%
\>[3]{}\Varid{conc}\;\Varid{m}\bind \lambda \hslambda \Varid{a}\hsarrow{\rightarrow }{\mathpunct{.}}\Varid{conc}\;(\Varid{k}\;\Varid{a}){}\<[E]%
\\
\>[B]{}\mathrel{=}{}\<[BE]%
\>[4]{}\mbox{\commentbegin  eta contraction  \commentend}{}\<[E]%
\\
\>[B]{}\hsindent{3}{}\<[3]%
\>[3]{}\Varid{conc}\;\Varid{m}\bind (\Varid{conc}\hsdot{\cdot }{.}\Varid{k}){}\<[E]%
\ColumnHook
\end{hscode}\resethooks
(Note that \ensuremath{\Varid{conc}} does not preserve \ensuremath{\Varid{return}}, and so \ensuremath{\Varid{conc}} is not a
monad morphism: \ensuremath{\Varid{conc}\;(\Varid{return}\;\Varid{a})} not only returns \ensuremath{\Varid{a}}, it also
synchronises the two copies of the state.)
For the second, we have to show that \ensuremath{\Varid{abs}\hsdot{\cdot }{.}\Varid{conc}} is the identity:
\begin{hscode}\SaveRestoreHook
\column{B}{@{}>{\hspre}c<{\hspost}@{}}%
\column{BE}{@{}l@{}}%
\column{3}{@{}>{\hspre}l<{\hspost}@{}}%
\column{5}{@{}>{\hspre}l<{\hspost}@{}}%
\column{E}{@{}>{\hspre}l<{\hspost}@{}}%
\>[3]{}\Varid{abs}\;(\Varid{conc}\;\Varid{m}){}\<[E]%
\\
\>[B]{}\mathrel{=}{}\<[BE]%
\>[5]{}\mbox{\commentbegin  \ensuremath{\Varid{abs}}  \commentend}{}\<[E]%
\\
\>[B]{}\hsindent{3}{}\<[3]%
\>[3]{}\mathbf{do}\;\{\mskip1.5mu \Varid{s}\leftarrow \Varid{get_T};(\Varid{a},\Varid{s'})\leftarrow \Varid{conc}\;\Varid{m}\;\Varid{s};\Varid{set_T}\;\Varid{s'};\Varid{return}\;\Varid{a}\mskip1.5mu\}{}\<[E]%
\\
\>[B]{}\mathrel{=}{}\<[BE]%
\>[5]{}\mbox{\commentbegin  \ensuremath{\Varid{conc}}  \commentend}{}\<[E]%
\\
\>[B]{}\hsindent{3}{}\<[3]%
\>[3]{}\mathbf{do}\;\{\mskip1.5mu \Varid{s}\leftarrow \Varid{get_T};\Varid{a}\leftarrow \Varid{m};\Varid{s'}\leftarrow \Varid{get_T};\Varid{set_T}\;\Varid{s'};\Varid{return}\;\Varid{a}\mskip1.5mu\}{}\<[E]%
\\
\>[B]{}\mathrel{=}{}\<[BE]%
\>[5]{}\mbox{\commentbegin  \ensuremath{\mathrm{(GS)}} for \ensuremath{\Conid{T}}  \commentend}{}\<[E]%
\\
\>[B]{}\hsindent{3}{}\<[3]%
\>[3]{}\mathbf{do}\;\{\mskip1.5mu \Varid{s}\leftarrow \Varid{get_T};\Varid{a}\leftarrow \Varid{m};\Varid{return}\;\Varid{a}\mskip1.5mu\}{}\<[E]%
\\
\>[B]{}\mathrel{=}{}\<[BE]%
\>[5]{}\mbox{\commentbegin  monad unit  \commentend}{}\<[E]%
\\
\>[B]{}\hsindent{3}{}\<[3]%
\>[3]{}\mathbf{do}\;\{\mskip1.5mu \Varid{s}\leftarrow \Varid{get_T};\Varid{m}\mskip1.5mu\}{}\<[E]%
\\
\>[B]{}\mathrel{=}{}\<[BE]%
\>[5]{}\mbox{\commentbegin  unused \ensuremath{\Varid{get_T}} is discardable  \commentend}{}\<[E]%
\\
\>[B]{}\hsindent{3}{}\<[3]%
\>[3]{}\mathbf{do}\;\{\mskip1.5mu \Varid{m}\mskip1.5mu\}{}\<[E]%
\ColumnHook
\end{hscode}\resethooks
For the third, we have to show that post-composition with \ensuremath{\Varid{abs}}
transforms the \ensuremath{\Conid{T}} operations into the corresponding \ensuremath{\Conid{StateT}\;\Conid{S}\;\Conid{T}}
operations. We do this by construction, defining: 
\begin{hscode}\SaveRestoreHook
\column{B}{@{}>{\hspre}l<{\hspost}@{}}%
\column{10}{@{}>{\hspre}l<{\hspost}@{}}%
\column{E}{@{}>{\hspre}l<{\hspost}@{}}%
\>[B]{}\Varid{sget}{}\<[10]%
\>[10]{}\mathrel{=}\Varid{conc}\;\Varid{get_T}{}\<[E]%
\\
\>[B]{}\Varid{sset}\;\Varid{s'}{}\<[10]%
\>[10]{}\mathrel{=}\Varid{conc}\;(\Varid{set_T}\;\Varid{s'}){}\<[E]%
\ColumnHook
\end{hscode}\resethooks
Expanding and simplifying, we see that
synchronised set writes to both copies of the state, and
synchronised get reads from the inner copy, but overwrites the outer copy to ensure that it agrees:
\begin{hscode}\SaveRestoreHook
\column{B}{@{}>{\hspre}l<{\hspost}@{}}%
\column{10}{@{}>{\hspre}l<{\hspost}@{}}%
\column{E}{@{}>{\hspre}l<{\hspost}@{}}%
\>[B]{}\Varid{sget}{}\<[10]%
\>[10]{}\mathrel{=}\mathbf{do}\;\{\mskip1.5mu \Varid{s'}\leftarrow \Varid{lift}\;\Varid{get_T};\Varid{set}\;\Varid{s'};\Varid{return}\;\Varid{s'}\mskip1.5mu\}{}\<[E]%
\\
\>[B]{}\Varid{sset}\;\Varid{s'}{}\<[10]%
\>[10]{}\mathrel{=}\mathbf{do}\;\{\mskip1.5mu \Varid{lift}\;(\Varid{set_T}\;\Varid{s'});\Varid{set}\;\Varid{s'}\mskip1.5mu\}{}\<[E]%
\ColumnHook
\end{hscode}\resethooks
\endswithdisplay
\end{proof}

\section{Proofs from Section~\ref{sec:composition}}
\label{app:composition}

\restatableLemma{lem:vwb-monad-morphism}
\begin{lem:vwb-monad-morphism} 
  If \ensuremath{\Varid{l}\mathbin{::}\Conid{Lens}\;\Conid{S}\;\Conid{V}} is very well-behaved, then \ensuremath{\vartheta \;\Varid{l}} is a monad morphism.  
\end{lem:vwb-monad-morphism}
\begin{proof}
We first show that \ensuremath{\vartheta \;\Varid{l}} preserves \ensuremath{\Varid{return}}s:
\begin{hscode}\SaveRestoreHook
\column{B}{@{}>{\hspre}c<{\hspost}@{}}%
\column{BE}{@{}l@{}}%
\column{3}{@{}>{\hspre}l<{\hspost}@{}}%
\column{5}{@{}>{\hspre}l<{\hspost}@{}}%
\column{9}{@{}>{\hspre}l<{\hspost}@{}}%
\column{E}{@{}>{\hspre}l<{\hspost}@{}}%
\>[3]{}\vartheta \;\Varid{l}\;(\Varid{return}\;\Varid{x}){}\<[E]%
\\
\>[B]{}\mathrel{=}{}\<[BE]%
\>[5]{}\mbox{\commentbegin  \ensuremath{\vartheta }  \commentend}{}\<[E]%
\\
\>[B]{}\hsindent{3}{}\<[3]%
\>[3]{}\mathbf{do}\;\{\mskip1.5mu {}\<[9]%
\>[9]{}\Varid{s}\leftarrow \Varid{get};\mathbf{let}\;\Varid{v}\mathrel{=}\Varid{l}\mathord{.}\Varid{view}\;\Varid{s};(\Varid{a},\Varid{v'})\leftarrow \Varid{lift}\;(\Varid{return}\;\Varid{x}\;\Varid{v});{}\<[E]%
\\
\>[9]{}\mathbf{let}\;\Varid{s'}\mathrel{=}\Varid{l}\mathord{.}\Varid{update}\;\Varid{s}\;\Varid{v'};\Varid{set}\;\Varid{s'};\Varid{return}\;\Varid{a}\mskip1.5mu\}{}\<[E]%
\\
\>[B]{}\mathrel{=}{}\<[BE]%
\>[5]{}\mbox{\commentbegin  \ensuremath{\Varid{return}} for \ensuremath{\Conid{StateT}}  \commentend}{}\<[E]%
\\
\>[B]{}\hsindent{3}{}\<[3]%
\>[3]{}\mathbf{do}\;\{\mskip1.5mu {}\<[9]%
\>[9]{}\Varid{s}\leftarrow \Varid{get};\mathbf{let}\;\Varid{v}\mathrel{=}\Varid{l}\mathord{.}\Varid{view}\;\Varid{s};\mathbf{let}\;(\Varid{a},\Varid{v'})\mathrel{=}(\Varid{x},\Varid{v});{}\<[E]%
\\
\>[9]{}\mathbf{let}\;\Varid{s'}\mathrel{=}\Varid{l}\mathord{.}\Varid{update}\;\Varid{s}\;\Varid{v'};\Varid{set}\;\Varid{s'};\Varid{return}\;\Varid{a}\mskip1.5mu\}{}\<[E]%
\\
\>[B]{}\mathrel{=}{}\<[BE]%
\>[5]{}\mbox{\commentbegin  \ensuremath{\mathrm{(UV)}}  \commentend}{}\<[E]%
\\
\>[B]{}\hsindent{3}{}\<[3]%
\>[3]{}\mathbf{do}\;\{\mskip1.5mu {}\<[9]%
\>[9]{}\Varid{s}\leftarrow \Varid{get};\mathbf{let}\;\Varid{v}\mathrel{=}\Varid{l}\mathord{.}\Varid{view}\;\Varid{s};\mathbf{let}\;(\Varid{a},\Varid{v'})\mathrel{=}(\Varid{x},\Varid{v});\mathbf{let}\;\Varid{s'}\mathrel{=}\Varid{s};\Varid{set}\;\Varid{s'};\Varid{return}\;\Varid{a}\mskip1.5mu\}{}\<[E]%
\\
\>[B]{}\mathrel{=}{}\<[BE]%
\>[5]{}\mbox{\commentbegin  \ensuremath{\mathrm{(GS)}} for \ensuremath{\Conid{StateT}}  \commentend}{}\<[E]%
\\
\>[B]{}\hsindent{3}{}\<[3]%
\>[3]{}\mathbf{do}\;\{\mskip1.5mu {}\<[9]%
\>[9]{}\Varid{s}\leftarrow \Varid{get};\mathbf{let}\;\Varid{v}\mathrel{=}\Varid{l}\mathord{.}\Varid{view}\;\Varid{s};\Varid{return}\;\Varid{x}\mskip1.5mu\}{}\<[E]%
\\
\>[B]{}\mathrel{=}{}\<[BE]%
\>[5]{}\mbox{\commentbegin  \ensuremath{\Varid{get}} is discardable  \commentend}{}\<[E]%
\\
\>[B]{}\hsindent{3}{}\<[3]%
\>[3]{}\Varid{return}\;\Varid{x}{}\<[E]%
\ColumnHook
\end{hscode}\resethooks
Now we show that \ensuremath{\vartheta \;\Varid{l}} respects sequential composition:
\begin{hscode}\SaveRestoreHook
\column{B}{@{}>{\hspre}c<{\hspost}@{}}%
\column{BE}{@{}l@{}}%
\column{3}{@{}>{\hspre}l<{\hspost}@{}}%
\column{5}{@{}>{\hspre}l<{\hspost}@{}}%
\column{9}{@{}>{\hspre}l<{\hspost}@{}}%
\column{E}{@{}>{\hspre}l<{\hspost}@{}}%
\>[3]{}\vartheta \;\Varid{l}\;(\mathbf{do}\;\{\mskip1.5mu \Varid{a}\leftarrow \Varid{m};\Varid{k}\;\Varid{a}\mskip1.5mu\}){}\<[E]%
\\
\>[B]{}\mathrel{=}{}\<[BE]%
\>[5]{}\mbox{\commentbegin  \ensuremath{\vartheta }  \commentend}{}\<[E]%
\\
\>[B]{}\hsindent{3}{}\<[3]%
\>[3]{}\mathbf{do}\;\{\mskip1.5mu {}\<[9]%
\>[9]{}\Varid{s}\leftarrow \Varid{get};\mathbf{let}\;\Varid{v}\mathrel{=}\Varid{l}\mathord{.}\Varid{view}\;\Varid{s};(\Varid{b},\Varid{v'''})\leftarrow \Varid{lift}\;(\mathbf{do}\;\{\mskip1.5mu \Varid{a}\leftarrow \Varid{m};\Varid{k}\;\Varid{a}\mskip1.5mu\}\;\Varid{v});{}\<[E]%
\\
\>[9]{}\mathbf{let}\;\Varid{s'''}\mathrel{=}\Varid{l}\mathord{.}\Varid{update}\;\Varid{s}\;\Varid{v'''};\Varid{set}\;\Varid{s'''};\Varid{return}\;\Varid{b}\mskip1.5mu\}{}\<[E]%
\\
\>[B]{}\mathrel{=}{}\<[BE]%
\>[5]{}\mbox{\commentbegin  \ensuremath{\Varid{lift}}; bind for \ensuremath{\Conid{StateT}}  \commentend}{}\<[E]%
\\
\>[B]{}\hsindent{3}{}\<[3]%
\>[3]{}\mathbf{do}\;\{\mskip1.5mu {}\<[9]%
\>[9]{}\Varid{s}\leftarrow \Varid{get};\mathbf{let}\;\Varid{v}\mathrel{=}\Varid{l}\mathord{.}\Varid{view}\;\Varid{s};(\Varid{a},\Varid{v'})\leftarrow \Varid{lift}\;(\Varid{m}\;\Varid{v});(\Varid{b},\Varid{v'''})\leftarrow \Varid{lift}\;(\Varid{k}\;\Varid{a}\;\Varid{v'});{}\<[E]%
\\
\>[9]{}\mathbf{let}\;\Varid{s'''}\mathrel{=}\Varid{l}\mathord{.}\Varid{update}\;\Varid{s}\;\Varid{v'''};\Varid{set}\;\Varid{s'''};\Varid{return}\;\Varid{b}\mskip1.5mu\}{}\<[E]%
\\
\>[B]{}\mathrel{=}{}\<[BE]%
\>[5]{}\mbox{\commentbegin  \ensuremath{\mathrm{(UU)}}  \commentend}{}\<[E]%
\\
\>[B]{}\hsindent{3}{}\<[3]%
\>[3]{}\mathbf{do}\;\{\mskip1.5mu {}\<[9]%
\>[9]{}\Varid{s}\leftarrow \Varid{get};\mathbf{let}\;\Varid{v}\mathrel{=}\Varid{l}\mathord{.}\Varid{view}\;\Varid{s};(\Varid{a},\Varid{v'})\leftarrow \Varid{lift}\;(\Varid{m}\;\Varid{v});(\Varid{b},\Varid{v'''})\leftarrow \Varid{lift}\;(\Varid{k}\;\Varid{a}\;\Varid{v'});{}\<[E]%
\\
\>[9]{}\mathbf{let}\;\Varid{s'}\mathrel{=}\Varid{l}\mathord{.}\Varid{update}\;\Varid{s}\;\Varid{v'};\mathbf{let}\;\Varid{s'''}\mathrel{=}\Varid{l}\mathord{.}\Varid{update}\;\Varid{s'}\;\Varid{v'''};\Varid{set}\;\Varid{s'''};\Varid{return}\;\Varid{b}\mskip1.5mu\}{}\<[E]%
\\
\>[B]{}\mathrel{=}{}\<[BE]%
\>[5]{}\mbox{\commentbegin  move \ensuremath{\mathbf{let}}s  \commentend}{}\<[E]%
\\
\>[B]{}\hsindent{3}{}\<[3]%
\>[3]{}\mathbf{do}\;\{\mskip1.5mu {}\<[9]%
\>[9]{}\Varid{s}\leftarrow \Varid{get};\mathbf{let}\;\Varid{v}\mathrel{=}\Varid{l}\mathord{.}\Varid{view}\;\Varid{s};(\Varid{a},\Varid{v'})\leftarrow \Varid{lift}\;(\Varid{m}\;\Varid{v});\mathbf{let}\;\Varid{s'}\mathrel{=}\Varid{l}\mathord{.}\Varid{update}\;\Varid{s}\;\Varid{v'};{}\<[E]%
\\
\>[9]{}(\Varid{b},\Varid{v'''})\leftarrow \Varid{lift}\;(\Varid{k}\;\Varid{a}\;\Varid{v'});\mathbf{let}\;\Varid{s'''}\mathrel{=}\Varid{l}\mathord{.}\Varid{update}\;\Varid{s'}\;\Varid{v'''};\Varid{set}\;\Varid{s'''};\Varid{return}\;\Varid{b}\mskip1.5mu\}{}\<[E]%
\\
\>[B]{}\mathrel{=}{}\<[BE]%
\>[5]{}\mbox{\commentbegin  \ensuremath{\mathrm{(SS)}} for \ensuremath{\Conid{StateT}}  \commentend}{}\<[E]%
\\
\>[B]{}\hsindent{3}{}\<[3]%
\>[3]{}\mathbf{do}\;\{\mskip1.5mu {}\<[9]%
\>[9]{}\Varid{s}\leftarrow \Varid{get};\mathbf{let}\;\Varid{v}\mathrel{=}\Varid{l}\mathord{.}\Varid{view}\;\Varid{s};(\Varid{a},\Varid{v'})\leftarrow \Varid{lift}\;(\Varid{m}\;\Varid{v});\mathbf{let}\;\Varid{s'}\mathrel{=}\Varid{l}\mathord{.}\Varid{update}\;\Varid{s}\;\Varid{v'};{}\<[E]%
\\
\>[9]{}(\Varid{b},\Varid{v'''})\leftarrow \Varid{lift}\;(\Varid{k}\;\Varid{a}\;\Varid{v'});\mathbf{let}\;\Varid{s'''}\mathrel{=}\Varid{l}\mathord{.}\Varid{update}\;\Varid{s'}\;\Varid{v'''};\Varid{set}\;\Varid{s'};\Varid{set}\;\Varid{s'''};\Varid{return}\;\Varid{b}\mskip1.5mu\}{}\<[E]%
\\
\>[B]{}\mathrel{=}{}\<[BE]%
\>[5]{}\mbox{\commentbegin  Lemma~\ref{lem:liftings-commute}  \commentend}{}\<[E]%
\\
\>[B]{}\hsindent{3}{}\<[3]%
\>[3]{}\mathbf{do}\;\{\mskip1.5mu {}\<[9]%
\>[9]{}\Varid{s}\leftarrow \Varid{get};\mathbf{let}\;\Varid{v}\mathrel{=}\Varid{l}\mathord{.}\Varid{view}\;\Varid{s};(\Varid{a},\Varid{v'})\leftarrow \Varid{lift}\;(\Varid{m}\;\Varid{v});\mathbf{let}\;\Varid{s'}\mathrel{=}\Varid{l}\mathord{.}\Varid{update}\;\Varid{s}\;\Varid{v'};\Varid{set}\;\Varid{s'};{}\<[E]%
\\
\>[9]{}(\Varid{b},\Varid{v'''})\leftarrow \Varid{lift}\;(\Varid{k}\;\Varid{a}\;\Varid{v'});\mathbf{let}\;\Varid{s'''}\mathrel{=}\Varid{l}\mathord{.}\Varid{update}\;\Varid{s'}\;\Varid{v'''};\Varid{set}\;\Varid{s'''};\Varid{return}\;\Varid{b}\mskip1.5mu\}{}\<[E]%
\\
\>[B]{}\mathrel{=}{}\<[BE]%
\>[5]{}\mbox{\commentbegin  \ensuremath{\mathrm{(UV)}}; \ensuremath{\mathrm{(SG)}} for \ensuremath{\Conid{StateT}}  \commentend}{}\<[E]%
\\
\>[B]{}\hsindent{3}{}\<[3]%
\>[3]{}\mathbf{do}\;\{\mskip1.5mu {}\<[9]%
\>[9]{}\Varid{s}\leftarrow \Varid{get};\mathbf{let}\;\Varid{v}\mathrel{=}\Varid{l}\mathord{.}\Varid{view}\;\Varid{s};(\Varid{a},\Varid{v'})\leftarrow \Varid{lift}\;(\Varid{m}\;\Varid{v});\mathbf{let}\;\Varid{s'}\mathrel{=}\Varid{l}\mathord{.}\Varid{update}\;\Varid{s}\;\Varid{v'};\Varid{set}\;\Varid{s'};{}\<[E]%
\\
\>[9]{}\Varid{s''}\leftarrow \Varid{get};\mathbf{let}\;\Varid{v''}\mathrel{=}\Varid{l}\mathord{.}\Varid{view}\;\Varid{s'};(\Varid{b},\Varid{v'''})\leftarrow \Varid{lift}\;(\Varid{k}\;\Varid{a}\;\Varid{v''});\mathbf{let}\;\Varid{s'''}\mathrel{=}\Varid{l}\mathord{.}\Varid{update}\;\Varid{s''}\;\Varid{v'''};\Varid{set}\;\Varid{s'''};{}\<[E]%
\\
\>[9]{}\Varid{return}\;\Varid{b}\mskip1.5mu\}{}\<[E]%
\\
\>[B]{}\mathrel{=}{}\<[BE]%
\>[5]{}\mbox{\commentbegin  \ensuremath{\vartheta }  \commentend}{}\<[E]%
\\
\>[B]{}\hsindent{3}{}\<[3]%
\>[3]{}\mathbf{do}\;\{\mskip1.5mu {}\<[9]%
\>[9]{}\Varid{a}\leftarrow \vartheta \;\Varid{l}\;\Varid{m};\vartheta \;\Varid{k}\;(\Varid{k}\;\Varid{a})\mskip1.5mu\}{}\<[E]%
\ColumnHook
\end{hscode}\resethooks
\endswithdisplay
\end{proof}

\begin{lemma} \label{lem:simplify-left-right}
Simplifying definitions, we have
\begin{hscode}\SaveRestoreHook
\column{B}{@{}>{\hspre}l<{\hspost}@{}}%
\column{18}{@{}>{\hspre}l<{\hspost}@{}}%
\column{26}{@{}>{\hspre}l<{\hspost}@{}}%
\column{E}{@{}>{\hspre}l<{\hspost}@{}}%
\>[B]{}\vartheta \;\Varid{fstLens}\;\Varid{m}{}\<[18]%
\>[18]{}\mathrel{=}\mathbf{do}\;\{\mskip1.5mu {}\<[26]%
\>[26]{}(\Varid{s}_{1},\Varid{s}_{2})\leftarrow \Varid{get};(\Varid{c},\Varid{s}_{1}')\leftarrow \Varid{lift}\;(\Varid{m}\;\Varid{s}_{1});\Varid{set}\;(\Varid{s}_{1}',\Varid{s}_{2});\Varid{return}\;\Varid{c}\mskip1.5mu\}{}\<[E]%
\\[\blanklineskip]%
\>[B]{}\vartheta \;\Varid{sndLens}\;\Varid{m}{}\<[18]%
\>[18]{}\mathrel{=}\mathbf{do}\;\{\mskip1.5mu {}\<[26]%
\>[26]{}(\Varid{s}_{1},\Varid{s}_{2})\leftarrow \Varid{get};(\Varid{a},\Varid{s}_{2}')\leftarrow \Varid{lift}\;(\Varid{m}\;\Varid{s}_{2});\Varid{set}\;(\Varid{s}_{1},\Varid{s}_{2}');\Varid{return}\;\Varid{a}\mskip1.5mu\}{}\<[E]%
\ColumnHook
\end{hscode}\resethooks
This will be convenient in what follows.
\end{lemma}

\begin{lemma} \label{lem:left-right-commute}
For arbitrary \ensuremath{\Varid{f}\mathbin{::}\sigma_1\hsarrow{\rightarrow }{\mathpunct{.}}\alpha} and \ensuremath{\Varid{m}\mathbin{::}\Conid{StateT}\;\sigma_2\;\tau\;\beta},
the liftings \ensuremath{\Varid{left}\;(\Varid{gets}\;\Varid{f})} and \ensuremath{\Varid{right}\;\Varid{m}} commute;
and symmetrically for \ensuremath{\Varid{right}\;(\Varid{gets}\;\Varid{f})} and \ensuremath{\Varid{left}\;\Varid{m}}.
(In fact, this holds for any \ensuremath{\Conid{T}}-pure computation, not just \ensuremath{\Varid{gets}\;\Varid{f}};
but we do not need the more general result.)
\end{lemma}
\begin{proof} We have:
\begin{hscode}\SaveRestoreHook
\column{B}{@{}>{\hspre}c<{\hspost}@{}}%
\column{BE}{@{}l@{}}%
\column{3}{@{}>{\hspre}l<{\hspost}@{}}%
\column{5}{@{}>{\hspre}l<{\hspost}@{}}%
\column{9}{@{}>{\hspre}l<{\hspost}@{}}%
\column{E}{@{}>{\hspre}l<{\hspost}@{}}%
\>[3]{}\mathbf{do}\;\{\mskip1.5mu {}\<[9]%
\>[9]{}\Varid{a}\leftarrow \Varid{left}\;(\Varid{gets}\;\Varid{f});\Varid{b}\leftarrow \Varid{right}\;\Varid{m};\Varid{return}\;(\Varid{a},\Varid{b})\mskip1.5mu\}{}\<[E]%
\\
\>[B]{}\mathrel{=}{}\<[BE]%
\>[5]{}\mbox{\commentbegin  Lemma~\ref{lem:simplify-left-right}, \ensuremath{\Varid{gets}}  \commentend}{}\<[E]%
\\
\>[B]{}\hsindent{3}{}\<[3]%
\>[3]{}\mathbf{do}\;\{\mskip1.5mu {}\<[9]%
\>[9]{}(\Varid{s}_{1},\Varid{s}_{2})\leftarrow \Varid{get};\mathbf{let}\;\Varid{a}\mathrel{=}\Varid{f}\;\Varid{s}_{1};(\Varid{b},\Varid{s}_{2}')\leftarrow \Varid{lift}\;(\Varid{m}\;\Varid{s}_{2});\Varid{set}\;(\Varid{s}_{1},\Varid{s}_{2}');\Varid{return}\;(\Varid{a},\Varid{b})\mskip1.5mu\}{}\<[E]%
\\
\>[B]{}\mathrel{=}{}\<[BE]%
\>[5]{}\mbox{\commentbegin  move \ensuremath{\mathbf{let}}  \commentend}{}\<[E]%
\\
\>[B]{}\hsindent{3}{}\<[3]%
\>[3]{}\mathbf{do}\;\{\mskip1.5mu {}\<[9]%
\>[9]{}(\Varid{s}_{1},\Varid{s}_{2})\leftarrow \Varid{get};(\Varid{b},\Varid{s}_{2}')\leftarrow \Varid{lift}\;(\Varid{m}\;\Varid{s}_{2});\Varid{set}\;(\Varid{s}_{1},\Varid{s}_{2}');\mathbf{let}\;\Varid{a}\mathrel{=}\Varid{f}\;\Varid{s}_{1};\Varid{return}\;(\Varid{a},\Varid{b})\mskip1.5mu\}{}\<[E]%
\\
\>[B]{}\mathrel{=}{}\<[BE]%
\>[5]{}\mbox{\commentbegin  \ensuremath{\mathrm{(SG)}} for \ensuremath{\Conid{StateT}}  \commentend}{}\<[E]%
\\
\>[B]{}\hsindent{3}{}\<[3]%
\>[3]{}\mathbf{do}\;\{\mskip1.5mu {}\<[9]%
\>[9]{}(\Varid{s}_{1},\Varid{s}_{2})\leftarrow \Varid{get};(\Varid{b},\Varid{s}_{2}')\leftarrow \Varid{lift}\;(\Varid{m}\;\Varid{s}_{2});\Varid{set}\;(\Varid{s}_{1},\Varid{s}_{2}');(\Varid{s}_{1}'',\Varid{s}_{2}'')\leftarrow \Varid{get};\mathbf{let}\;\Varid{a}\mathrel{=}\Varid{f}\;\Varid{s}_{1}'';\Varid{return}\;(\Varid{a},\Varid{b})\mskip1.5mu\}{}\<[E]%
\\
\>[B]{}\mathrel{=}{}\<[BE]%
\>[5]{}\mbox{\commentbegin  Lemma~\ref{lem:simplify-left-right}  \commentend}{}\<[E]%
\\
\>[B]{}\hsindent{3}{}\<[3]%
\>[3]{}\mathbf{do}\;\{\mskip1.5mu {}\<[9]%
\>[9]{}\Varid{b}\leftarrow \Varid{right}\;\Varid{m};\Varid{a}\leftarrow \Varid{left}\;(\Varid{gets}\;\Varid{f});\Varid{return}\;(\Varid{a},\Varid{b})\mskip1.5mu\}{}\<[E]%
\ColumnHook
\end{hscode}\resethooks
The symmetric property of course has a symmetric proof too.
\end{proof}

\restatableTheorem{thm:composition-wb}
\begin{thm:composition-wb}[transparent composition]
    Given transparent \bx{}
\begin{hscode}\SaveRestoreHook
\column{B}{@{}>{\hspre}l<{\hspost}@{}}%
\column{E}{@{}>{\hspre}l<{\hspost}@{}}%
\>[B]{}bx_{1}\mathbin{::}\Conid{StateTBX}\;\Conid{S}_{1}\;\Conid{T}\;\Conid{A}\;\Conid{B}{}\<[E]%
\\
\>[B]{}bx_{2}\mathbin{::}\Conid{StateTBX}\;\Conid{S}_{2}\;\Conid{T}\;\Conid{B}\;\Conid{C}{}\<[E]%
\ColumnHook
\end{hscode}\resethooks
    their composition \ensuremath{bx_{1} \mathbin{\fatsemi} bx_{2}\mathbin{::}\Conid{StateTBX}\;(\Conid{S}_{1}\mathbin{\!\Join\!}\Conid{S}_{2})\;\Conid{T}\;\Conid{A}\;\Conid{C}} 
    is also transparent.
\end{thm:composition-wb}
\begin{proof}
We first have to check that the composition does indeed operate only
on the state space \ensuremath{\Conid{S}_{1}\mathbin{\!\Join\!}\Conid{S}_{2}}, by verifying that \ensuremath{\set{L}} and \ensuremath{\set{R}}
maintain this invariant. For \ensuremath{\set{L}}, we have:
\begin{hscode}\SaveRestoreHook
\column{B}{@{}>{\hspre}c<{\hspost}@{}}%
\column{BE}{@{}l@{}}%
\column{3}{@{}>{\hspre}l<{\hspost}@{}}%
\column{5}{@{}>{\hspre}l<{\hspost}@{}}%
\column{9}{@{}>{\hspre}l<{\hspost}@{}}%
\column{E}{@{}>{\hspre}l<{\hspost}@{}}%
\>[3]{}\mathbf{do}\;\{\mskip1.5mu {}\<[9]%
\>[9]{}\set{L}\;\Varid{a'};\Varid{left}\;(bx_{1}\mathord{.}\get{R})\mskip1.5mu\}{}\<[E]%
\\
\>[B]{}\mathrel{=}{}\<[BE]%
\>[5]{}\mbox{\commentbegin  \ensuremath{\set{L}}; \ensuremath{\Varid{return}}  \commentend}{}\<[E]%
\\
\>[B]{}\hsindent{3}{}\<[3]%
\>[3]{}\mathbf{do}\;\{\mskip1.5mu {}\<[9]%
\>[9]{}\Varid{left}\;(bx_{1}\mathord{.}\set{L}\;\Varid{a'});\Varid{b'}\leftarrow \Varid{left}\;(bx_{1}\mathord{.}\get{R});\Varid{right}\;(bx_{2}\mathord{.}\set{L}\;\Varid{b'});\Varid{b''}\leftarrow \Varid{left}\;(bx_{1}\mathord{.}\get{R});\Varid{return}\;\Varid{b''}\mskip1.5mu\}{}\<[E]%
\\
\>[B]{}\mathrel{=}{}\<[BE]%
\>[5]{}\mbox{\commentbegin  \ensuremath{bx_{1}\mathord{.}\get{R}\mathrel{=}\Varid{gets}\;(bx_{1}\mathord{.}\Varid{read}_{R})}, and Lemma~\ref{lem:left-right-commute}  \commentend}{}\<[E]%
\\
\>[B]{}\hsindent{3}{}\<[3]%
\>[3]{}\mathbf{do}\;\{\mskip1.5mu {}\<[9]%
\>[9]{}\Varid{left}\;(bx_{1}\mathord{.}\set{L}\;\Varid{a'});\Varid{b'}\leftarrow \Varid{left}\;(bx_{1}\mathord{.}\get{R});\Varid{b''}\leftarrow \Varid{left}\;(bx_{1}\mathord{.}\get{R});\Varid{right}\;(bx_{2}\mathord{.}\set{L}\;\Varid{b'});\Varid{return}\;\Varid{b''}\mskip1.5mu\}{}\<[E]%
\\
\>[B]{}\mathrel{=}{}\<[BE]%
\>[5]{}\mbox{\commentbegin  \ensuremath{\Varid{left}} is a monad morphism, and \ensuremath{\mathrm{(G_RG_R)}} for \ensuremath{bx_{1}}  \commentend}{}\<[E]%
\\
\>[B]{}\hsindent{3}{}\<[3]%
\>[3]{}\mathbf{do}\;\{\mskip1.5mu {}\<[9]%
\>[9]{}\Varid{left}\;(bx_{1}\mathord{.}\set{L}\;\Varid{a'});\Varid{b'}\leftarrow \Varid{left}\;(bx_{1}\mathord{.}\get{R});\mathbf{let}\;\Varid{b''}\mathrel{=}\Varid{b'};\Varid{right}\;(bx_{2}\mathord{.}\set{L}\;\Varid{b'});\Varid{return}\;\Varid{b''}\mskip1.5mu\}{}\<[E]%
\\
\>[B]{}\mathrel{=}{}\<[BE]%
\>[5]{}\mbox{\commentbegin  move \ensuremath{\mathbf{let}}  \commentend}{}\<[E]%
\\
\>[B]{}\hsindent{3}{}\<[3]%
\>[3]{}\mathbf{do}\;\{\mskip1.5mu {}\<[9]%
\>[9]{}\Varid{left}\;(bx_{1}\mathord{.}\set{L}\;\Varid{a'});\Varid{b'}\leftarrow \Varid{left}\;(bx_{1}\mathord{.}\get{R});\Varid{right}\;(bx_{2}\mathord{.}\set{L}\;\Varid{b'});\mathbf{let}\;\Varid{b''}\mathrel{=}\Varid{b'};\Varid{return}\;\Varid{b''}\mskip1.5mu\}{}\<[E]%
\\
\>[B]{}\mathrel{=}{}\<[BE]%
\>[5]{}\mbox{\commentbegin  \ensuremath{\Varid{right}} is a monad morphism, and \ensuremath{\mathrm{(S_LG_L)}} for \ensuremath{bx_{2}}  \commentend}{}\<[E]%
\\
\>[B]{}\hsindent{3}{}\<[3]%
\>[3]{}\mathbf{do}\;\{\mskip1.5mu {}\<[9]%
\>[9]{}\Varid{left}\;(bx_{1}\mathord{.}\set{L}\;\Varid{a'});\Varid{b'}\leftarrow \Varid{left}\;(bx_{1}\mathord{.}\get{R});\Varid{right}\;(bx_{2}\mathord{.}\set{L}\;\Varid{b'});\Varid{b''}\leftarrow \Varid{right}\;(bx_{2}\mathord{.}\get{L});\Varid{return}\;\Varid{b''}\mskip1.5mu\}{}\<[E]%
\\
\>[B]{}\mathrel{=}{}\<[BE]%
\>[5]{}\mbox{\commentbegin  \ensuremath{\set{L}}  \commentend}{}\<[E]%
\\
\>[B]{}\hsindent{3}{}\<[3]%
\>[3]{}\mathbf{do}\;\{\mskip1.5mu {}\<[9]%
\>[9]{}\set{L}\;\Varid{a'};\Varid{right}\;(bx_{2}\mathord{.}\get{L})\mskip1.5mu\}{}\<[E]%
\ColumnHook
\end{hscode}\resethooks
Of course, \ensuremath{\set{R}} is symmetric. Note that \ensuremath{\get{L}} (and symmetrically,
\ensuremath{\get{R}}) are \ensuremath{\Conid{T}}-pure queries, so do not affect the state:
\begin{hscode}\SaveRestoreHook
\column{B}{@{}>{\hspre}c<{\hspost}@{}}%
\column{BE}{@{}l@{}}%
\column{3}{@{}>{\hspre}l<{\hspost}@{}}%
\column{5}{@{}>{\hspre}l<{\hspost}@{}}%
\column{9}{@{}>{\hspre}l<{\hspost}@{}}%
\column{E}{@{}>{\hspre}l<{\hspost}@{}}%
\>[3]{}\get{L}{}\<[E]%
\\
\>[B]{}\mathrel{=}{}\<[BE]%
\>[5]{}\mbox{\commentbegin  \ensuremath{\get{L}}  \commentend}{}\<[E]%
\\
\>[B]{}\hsindent{3}{}\<[3]%
\>[3]{}\Varid{left}\;(bx_{1}\mathord{.}\get{L}){}\<[E]%
\\
\>[B]{}\mathrel{=}{}\<[BE]%
\>[5]{}\mbox{\commentbegin  Lemma~\ref{lem:simplify-left-right}  \commentend}{}\<[E]%
\\
\>[B]{}\hsindent{3}{}\<[3]%
\>[3]{}\mathbf{do}\;\{\mskip1.5mu {}\<[9]%
\>[9]{}(\Varid{s}_{1},\Varid{s}_{2})\leftarrow \Varid{get};(\Varid{c},\Varid{s}_{1}')\leftarrow \Varid{lift}\;(bx_{1}\mathord{.}\get{L}\;\Varid{s}_{1});\Varid{set}\;(\Varid{s}_{1}',\Varid{s}_{2});\Varid{return}\;\Varid{c}\mskip1.5mu\}{}\<[E]%
\\
\>[B]{}\mathrel{=}{}\<[BE]%
\>[5]{}\mbox{\commentbegin  \ensuremath{bx_{1}} is transparent  \commentend}{}\<[E]%
\\
\>[B]{}\hsindent{3}{}\<[3]%
\>[3]{}\mathbf{do}\;\{\mskip1.5mu {}\<[9]%
\>[9]{}(\Varid{s}_{1},\Varid{s}_{2})\leftarrow \Varid{get};\mathbf{let}\;(\Varid{c},\Varid{s}_{1}')\mathrel{=}(bx_{1}\mathord{.}\Varid{read}_{L}\;\Varid{s}_{1},\Varid{s}_{1});\Varid{set}\;(\Varid{s}_{1}',\Varid{s}_{2});\Varid{return}\;\Varid{c}\mskip1.5mu\}{}\<[E]%
\\
\>[B]{}\mathrel{=}{}\<[BE]%
\>[5]{}\mbox{\commentbegin  \ensuremath{\mathrm{(GS)}} for \ensuremath{\Conid{StateT}}  \commentend}{}\<[E]%
\\
\>[B]{}\hsindent{3}{}\<[3]%
\>[3]{}\mathbf{do}\;\{\mskip1.5mu {}\<[9]%
\>[9]{}(\Varid{s}_{1},\Varid{s}_{2})\leftarrow \Varid{get};\mathbf{let}\;\Varid{c}\mathrel{=}bx_{1}\mathord{.}\Varid{read}_{L}\;\Varid{s}_{1};\Varid{return}\;\Varid{c}\mskip1.5mu\}{}\<[E]%
\\
\>[B]{}\mathrel{=}{}\<[BE]%
\>[5]{}\mbox{\commentbegin  \ensuremath{\Varid{gets}}  \commentend}{}\<[E]%
\\
\>[B]{}\hsindent{3}{}\<[3]%
\>[3]{}\Varid{gets}\;(\Varid{read}_{L}\hsdot{\cdot }{.}\Varid{fst}){}\<[E]%
\ColumnHook
\end{hscode}\resethooks
So the composition is transparent, and hence we get the laws
\ensuremath{\mathrm{(G_LG_L)}}, \ensuremath{\mathrm{(G_RG_R)}}, and \ensuremath{\mathrm{(G_LG_R)}} for free: we need only check
the laws involving sets. For \ensuremath{\mathrm{(S_LG_L)}}, we have:
\begin{hscode}\SaveRestoreHook
\column{B}{@{}>{\hspre}c<{\hspost}@{}}%
\column{BE}{@{}l@{}}%
\column{3}{@{}>{\hspre}l<{\hspost}@{}}%
\column{5}{@{}>{\hspre}l<{\hspost}@{}}%
\column{9}{@{}>{\hspre}l<{\hspost}@{}}%
\column{E}{@{}>{\hspre}l<{\hspost}@{}}%
\>[3]{}\mathbf{do}\;\{\mskip1.5mu {}\<[9]%
\>[9]{}\set{L}\;\Varid{a'};\get{L}\mskip1.5mu\}{}\<[E]%
\\
\>[B]{}\mathrel{=}{}\<[BE]%
\>[5]{}\mbox{\commentbegin  \ensuremath{\set{L}}, \ensuremath{\get{L}}  \commentend}{}\<[E]%
\\
\>[B]{}\hsindent{3}{}\<[3]%
\>[3]{}\mathbf{do}\;\{\mskip1.5mu {}\<[9]%
\>[9]{}\Varid{left}\;(bx_{1}\mathord{.}\set{L}\;\Varid{a'});\Varid{b'}\leftarrow \Varid{left}\;(bx_{1}\mathord{.}\get{R});\Varid{right}\;(bx_{2}\mathord{.}\set{L}\;\Varid{b'});\Varid{left}\;(bx_{1}\mathord{.}\get{L})\mskip1.5mu\}{}\<[E]%
\\
\>[B]{}\mathrel{=}{}\<[BE]%
\>[5]{}\mbox{\commentbegin  Lemma~\ref{lem:left-right-commute}  \commentend}{}\<[E]%
\\
\>[B]{}\hsindent{3}{}\<[3]%
\>[3]{}\mathbf{do}\;\{\mskip1.5mu {}\<[9]%
\>[9]{}\Varid{left}\;(bx_{1}\mathord{.}\set{L}\;\Varid{a'});\Varid{b'}\leftarrow \Varid{left}\;(bx_{1}\mathord{.}\get{R});\Varid{a''}\leftarrow \Varid{left}\;(bx_{1}\mathord{.}\get{L});\Varid{right}\;(bx_{2}\mathord{.}\set{L}\;\Varid{b'});\Varid{return}\;\Varid{a''}\mskip1.5mu\}{}\<[E]%
\\
\>[B]{}\mathrel{=}{}\<[BE]%
\>[5]{}\mbox{\commentbegin  \ensuremath{\Varid{left}} is a monad morphism; \ensuremath{\mathrm{(G_LG_R)}} for \ensuremath{bx_{1}}  \commentend}{}\<[E]%
\\
\>[B]{}\hsindent{3}{}\<[3]%
\>[3]{}\mathbf{do}\;\{\mskip1.5mu {}\<[9]%
\>[9]{}\Varid{left}\;(bx_{1}\mathord{.}\set{L}\;\Varid{a'});\Varid{a''}\leftarrow \Varid{left}\;(bx_{1}\mathord{.}\get{L});\Varid{b'}\leftarrow \Varid{left}\;(bx_{1}\mathord{.}\get{R});\Varid{right}\;(bx_{2}\mathord{.}\set{L}\;\Varid{b'});\Varid{return}\;\Varid{a''}\mskip1.5mu\}{}\<[E]%
\\
\>[B]{}\mathrel{=}{}\<[BE]%
\>[5]{}\mbox{\commentbegin  \ensuremath{\Varid{left}} is a monad morphism; \ensuremath{\mathrm{(S_LG_L)}} for \ensuremath{bx_{1}}  \commentend}{}\<[E]%
\\
\>[B]{}\hsindent{3}{}\<[3]%
\>[3]{}\mathbf{do}\;\{\mskip1.5mu {}\<[9]%
\>[9]{}\Varid{left}\;(bx_{1}\mathord{.}\set{L}\;\Varid{a'});\mathbf{let}\;\Varid{a''}\mathrel{=}\Varid{a'};\Varid{b'}\leftarrow \Varid{left}\;(bx_{1}\mathord{.}\get{R});\Varid{right}\;(bx_{2}\mathord{.}\set{L}\;\Varid{b'});\Varid{return}\;\Varid{a''}\mskip1.5mu\}{}\<[E]%
\\
\>[B]{}\mathrel{=}{}\<[BE]%
\>[5]{}\mbox{\commentbegin  move \ensuremath{\mathbf{let}}; \ensuremath{\set{L}}  \commentend}{}\<[E]%
\\
\>[B]{}\hsindent{3}{}\<[3]%
\>[3]{}\mathbf{do}\;\{\mskip1.5mu {}\<[9]%
\>[9]{}\set{L}\;\Varid{a'};\Varid{return}\;\Varid{a'}\mskip1.5mu\}{}\<[E]%
\ColumnHook
\end{hscode}\resethooks
And for \ensuremath{\mathrm{(G_LS_L)}}, using the fact that the initial state is in \ensuremath{\Conid{S}_{1}\mathbin{\!\Join\!}\Conid{S}_{2}}:
\begin{hscode}\SaveRestoreHook
\column{B}{@{}>{\hspre}c<{\hspost}@{}}%
\column{BE}{@{}l@{}}%
\column{3}{@{}>{\hspre}l<{\hspost}@{}}%
\column{5}{@{}>{\hspre}l<{\hspost}@{}}%
\column{9}{@{}>{\hspre}l<{\hspost}@{}}%
\column{E}{@{}>{\hspre}l<{\hspost}@{}}%
\>[3]{}\mathbf{do}\;\{\mskip1.5mu {}\<[9]%
\>[9]{}\Varid{a}\leftarrow \get{L};\set{L}\;\Varid{a}\mskip1.5mu\}{}\<[E]%
\\
\>[B]{}\mathrel{=}{}\<[BE]%
\>[5]{}\mbox{\commentbegin  \ensuremath{\get{L}}, \ensuremath{\set{L}}  \commentend}{}\<[E]%
\\
\>[B]{}\hsindent{3}{}\<[3]%
\>[3]{}\mathbf{do}\;\{\mskip1.5mu {}\<[9]%
\>[9]{}\Varid{a}\leftarrow \Varid{left}\;(bx_{1}\mathord{.}\get{L});\Varid{left}\;(bx_{1}\mathord{.}\set{L}\;\Varid{a});\Varid{b'}\leftarrow \Varid{left}\;(bx_{1}\mathord{.}\get{R});\Varid{right}\;(bx_{2}\mathord{.}\set{L}\;\Varid{b'})\mskip1.5mu\}{}\<[E]%
\\
\>[B]{}\mathrel{=}{}\<[BE]%
\>[5]{}\mbox{\commentbegin  \ensuremath{\Varid{left}} is a monad morphism; \ensuremath{\mathrm{(G_LS_L)}} for \ensuremath{bx_{1}}  \commentend}{}\<[E]%
\\
\>[B]{}\hsindent{3}{}\<[3]%
\>[3]{}\mathbf{do}\;\{\mskip1.5mu {}\<[9]%
\>[9]{}\Varid{b'}\leftarrow \Varid{left}\;(bx_{1}\mathord{.}\get{R});\Varid{right}\;(bx_{2}\mathord{.}\set{L}\;\Varid{b'})\mskip1.5mu\}{}\<[E]%
\\
\>[B]{}\mathrel{=}{}\<[BE]%
\>[5]{}\mbox{\commentbegin  initial state is consistent, so \ensuremath{\Varid{left}\;(bx_{1}\mathord{.}\get{R})\mathrel{=}\Varid{right}\;(bx_{2}\mathord{.}\get{L})}  \commentend}{}\<[E]%
\\
\>[B]{}\hsindent{3}{}\<[3]%
\>[3]{}\mathbf{do}\;\{\mskip1.5mu {}\<[9]%
\>[9]{}\Varid{b'}\leftarrow \Varid{right}\;(bx_{2}\mathord{.}\get{L});\Varid{right}\;(bx_{2}\mathord{.}\set{L}\;\Varid{b'})\mskip1.5mu\}{}\<[E]%
\\
\>[B]{}\mathrel{=}{}\<[BE]%
\>[5]{}\mbox{\commentbegin  \ensuremath{\Varid{right}} is a monad morphism; \ensuremath{\mathrm{(G_LS_L)}} for \ensuremath{bx_{2}}  \commentend}{}\<[E]%
\\
\>[B]{}\hsindent{3}{}\<[3]%
\>[3]{}\Varid{return}\;(){}\<[E]%
\ColumnHook
\end{hscode}\resethooks
And of course, \ensuremath{\mathrm{(S_RG_R)}} and \ensuremath{\mathrm{(G_RS_R)}} are symmetric.
\end{proof}

\begin{proposition}
The \ensuremath{\Varid{identity}} bx (Definition~\ref{def:identity}) is transparent and overwritable.
\end{proposition}
\begin{proof}
  The well-behavedness and overwritability laws are all
  immediate from the laws \ensuremath{\mathrm{(GG)}}, \ensuremath{\mathrm{(GS)}}, \ensuremath{\mathrm{(SG)}}, and
  \ensuremath{\mathrm{(SS)}} of the monad \ensuremath{\Conid{StateT}\;\Conid{S}\;\Conid{T}}.
  The transparency of \ensuremath{\Varid{identity}} is obvious from its definition (\ensuremath{\Varid{get}\mathrel{=}\Varid{gets}\;\Varid{id}}).  
\end{proof}

\begin{definition*}
In the case of \ensuremath{\Conid{StateTBX}}s, with \ensuremath{\Conid{T}_{1}\mathrel{=}\Conid{StateT}\;\Conid{S}_{1}\;\Conid{T}} and \ensuremath{\Conid{T}_{2}\mathrel{=}\Conid{StateT}\;\Conid{S}_{2}\;\Conid{T}} for some state types \ensuremath{\Conid{S}_{1},\Conid{S}_{2}}, we can construct a monad isomorphism
from \ensuremath{\Conid{T}_{1}} to \ensuremath{\Conid{T}_{2}} by lifting an isomorphism on the state spaces, using
the following construction:
\begin{hscode}\SaveRestoreHook
\column{B}{@{}>{\hspre}l<{\hspost}@{}}%
\column{18}{@{}>{\hspre}l<{\hspost}@{}}%
\column{22}{@{}>{\hspre}l<{\hspost}@{}}%
\column{E}{@{}>{\hspre}l<{\hspost}@{}}%
\>[B]{}\mathbf{data}\;\Conid{Iso}\;\alpha\;\beta\mathrel{=}\Conid{Iso}\;\{\mskip1.5mu \Varid{to}\mathbin{::}\alpha\hsarrow{\rightarrow }{\mathpunct{.}}\beta,\Varid{from}\mathbin{::}\beta\hsarrow{\rightarrow }{\mathpunct{.}}\alpha\mskip1.5mu\}{}\<[E]%
\\
\>[B]{}\Varid{inv}\;\Varid{h}\mathrel{=}\Conid{Iso}\;(\Varid{h}\mathord{.}\Varid{from})\;(\Varid{h}\mathord{.}\Varid{to}){}\<[E]%
\\[\blanklineskip]%
\>[B]{}\iota \mathbin{::}\Conid{Monad}\;\tau\Rightarrow {}\<[22]%
\>[22]{}\Conid{Iso}\;\sigma_1\;\sigma_2\hsarrow{\rightarrow }{\mathpunct{.}}\Conid{StateT}\;\sigma_1\;\tau\;\alpha\hsarrow{\rightarrow }{\mathpunct{.}}\Conid{StateT}\;\sigma_2\;\tau\;\alpha{}\<[E]%
\\
\>[B]{}\iota \;\Varid{h}\;\Varid{m}\mathrel{=}\mathbf{do}\;\{\mskip1.5mu {}\<[18]%
\>[18]{}\Varid{s}_{2}\leftarrow \Varid{get};(\Varid{a},\Varid{s}_{1})\leftarrow \Varid{lift}\;(\Varid{m}\;(\Varid{h}\mathord{.}\Varid{from}\;\Varid{s}_{2}));{}\<[E]%
\\
\>[18]{}\Varid{set}\;(\Varid{h}\mathord{.}\Varid{to}\;\Varid{s}_{1});\Varid{return}\;\Varid{a}\mskip1.5mu\}{}\<[E]%
\\
\>[B]{}\iota ^{-1}\;\Varid{h}\mathrel{=}\iota \;(\Varid{inv}\;\Varid{h}){}\<[E]%
\ColumnHook
\end{hscode}\resethooks
\endswithdisplay
\end{definition*}

So to show equivalence of \ensuremath{bx_{1}\mathbin{::}\Conid{StateTBX}\;\Conid{T}\;\Conid{S}_{1}\;\Conid{A}\;\Conid{B}} and
\ensuremath{bx_{2}\mathbin{::}\Conid{StateTBX}\;\Conid{T}\;\Conid{S}_{2}\;\Conid{A}\;\Conid{B}}, we just need to find an invertible function
\ensuremath{\Varid{h}\mathbin{::}\Conid{S}_{1}\hsarrow{\rightarrow }{\mathpunct{.}}\Conid{S}_{2}} such that the induced monad isomorphism \ensuremath{\iota \;\Varid{h}\mathbin{::}\Conid{StateT}\;\Conid{S}_{1}\;\Conid{T}\hsarrow{\rightarrow }{\mathpunct{.}}\Conid{StateT}\;\Conid{S}_{2}\;\Conid{T}} satisfies \ensuremath{\iota \;(bx_{1}\mathord{.}\get{L})\mathrel{=}bx_{2}\mathord{.}\get{L}} and \ensuremath{\iota \;(bx_{1}\mathord{.}\set{L}\;\Varid{a})\mathrel{=}bx_{2}\mathord{.}\set{L}\;\Varid{a}} and dually:
\begin{lemma}\label{lem:invertible-iso}
If \ensuremath{\Varid{h}\mathbin{::}\Conid{S}_{1}\hsarrow{\rightarrow }{\mathpunct{.}}\Conid{S}_{2}} is invertible, then \ensuremath{\iota \;\Varid{h}} is a monad isomorphism from 
\ensuremath{\Conid{StateT}\;\Conid{S}_{1}\;\Conid{T}} to \ensuremath{\Conid{StateT}\;\Conid{S}_{2}\;\Conid{T}}. 
\end{lemma}

\begin{proof}
  The fact that \ensuremath{\iota \;\Varid{h}} is a monad morphism follows from the fact
  that any isomorphism determines a very well-behaved \ensuremath{\Conid{T}}-lens whose updates commute in \ensuremath{\Conid{T}},  \fanote{is this better?} by
  Lemma~\ref{lem:vwb-monad-morphism}.  It is straightforward to verify
  that if \ensuremath{\Varid{h}} and \ensuremath{\Varid{inv}\;\Varid{h}} are inverses then so are \ensuremath{\iota \;\Varid{h}} and
  \ensuremath{\iota ^{-1}\;\Varid{h}}.
\end{proof}

\begin{lemma}\label{lem:identity}
  For any well-behaved \ensuremath{bx}, we have
\begin{hscode}\SaveRestoreHook
\column{B}{@{}>{\hspre}l<{\hspost}@{}}%
\column{31}{@{}>{\hspre}c<{\hspost}@{}}%
\column{31E}{@{}l@{}}%
\column{36}{@{}>{\hspre}l<{\hspost}@{}}%
\column{E}{@{}>{\hspre}l<{\hspost}@{}}%
\>[B]{}bx\equiv \Varid{identity} \mathbin{\fatsemi} bx{}\<[31]%
\>[31]{}\quad \mbox{and} \quad{}\<[31E]%
\>[36]{}bx\equiv bx \mathbin{\fatsemi} \Varid{identity}{}\<[E]%
\ColumnHook
\end{hscode}\resethooks
\endswithdisplay
\end{lemma}
\begin{proof}
For the LHS of
\ensuremath{\mathrm{(Identity)}}, consider 
\begin{hscode}\SaveRestoreHook
\column{B}{@{}>{\hspre}l<{\hspost}@{}}%
\column{11}{@{}>{\hspre}c<{\hspost}@{}}%
\column{11E}{@{}l@{}}%
\column{15}{@{}>{\hspre}l<{\hspost}@{}}%
\column{E}{@{}>{\hspre}l<{\hspost}@{}}%
\>[B]{}\Varid{identity}{}\<[11]%
\>[11]{}\mathbin{::}{}\<[11E]%
\>[15]{}\Conid{StateTBX}\;\Conid{T}\;\Conid{A}\;\Conid{A}\;\Conid{A}{}\<[E]%
\\
\>[B]{}bx{}\<[11]%
\>[11]{}\mathbin{::}{}\<[11E]%
\>[15]{}\Conid{StateTBX}\;\Conid{T}\;\Conid{S}\;\Conid{A}\;\Conid{B}{}\<[E]%
\ColumnHook
\end{hscode}\resethooks
so \ensuremath{\Varid{identity} \mathbin{\fatsemi} bx\mathbin{::}\Conid{StateTBX}\;\Conid{T}\;(\Conid{A}\mathbin{\!\Join\!}\Conid{S})\;\Conid{A}\;\Conid{B}}, 
where 
\begin{hscode}\SaveRestoreHook
\column{B}{@{}>{\hspre}l<{\hspost}@{}}%
\column{11}{@{}>{\hspre}l<{\hspost}@{}}%
\column{24}{@{}>{\hspre}l<{\hspost}@{}}%
\column{57}{@{}>{\hspre}c<{\hspost}@{}}%
\column{57E}{@{}l@{}}%
\column{60}{@{}>{\hspre}l<{\hspost}@{}}%
\column{E}{@{}>{\hspre}l<{\hspost}@{}}%
\>[B]{}\Conid{A}\mathbin{\!\Join\!}\Conid{S}{}\<[11]%
\>[11]{}\mathrel{=}\{\mskip1.5mu (\Varid{a},\Varid{s})\mid {}\<[24]%
\>[24]{}\Varid{eval}\;(\Varid{identity}\mathord{.}\get{R})\;\Varid{a}{}\<[57]%
\>[57]{}\mathrel{=}{}\<[57E]%
\>[60]{}\Varid{eval}\;(bx\mathord{.}\get{L})\;\Varid{s}\mskip1.5mu\}{}\<[E]%
\\
\>[11]{}\mathrel{=}\{\mskip1.5mu (bx\mathord{.}\Varid{read}_{L}\;\Varid{s},\Varid{s})\mid \Varid{s}\in\Conid{S}\mskip1.5mu\}{}\<[E]%
\ColumnHook
\end{hscode}\resethooks
To define an isomorphism between \ensuremath{\Conid{StateT}\;\Conid{S}\;\Conid{T}} and \ensuremath{\Conid{StateT}\;(\Conid{A}\mathbin{\!\Join\!}\Conid{S})\;\Conid{T}}, we
need to define an isomorphism \ensuremath{\Varid{f}\mathbin{:}\Conid{S}\hsarrow{\rightarrow }{\mathpunct{.}}(\Conid{A}\mathbin{\!\Join\!}\Conid{S})}. This is 
straightforward: the two directions are 
\begin{hscode}\SaveRestoreHook
\column{B}{@{}>{\hspre}l<{\hspost}@{}}%
\column{25}{@{}>{\hspre}c<{\hspost}@{}}%
\column{25E}{@{}l@{}}%
\column{29}{@{}>{\hspre}c<{\hspost}@{}}%
\column{29E}{@{}l@{}}%
\column{33}{@{}>{\hspre}l<{\hspost}@{}}%
\column{E}{@{}>{\hspre}l<{\hspost}@{}}%
\>[B]{}\Varid{f}\;\Varid{s}\mathrel{=}(bx\mathord{.}\Varid{read}_{L}\;\Varid{s},\Varid{s}){}\<[25]%
\>[25]{}\quad{}\<[25E]%
\>[29]{}\quad{}\<[29E]%
\>[33]{}\Varid{f}^{-1}\mathrel{=}\Varid{snd}{}\<[E]%
\ColumnHook
\end{hscode}\resethooks

We just need to verify compatibility of the isomorphism
\ensuremath{\Conid{StateT}\;\Conid{S}\;\Conid{T}\cong \Conid{StateT}\;(\Conid{A}\mathbin{\!\Join\!}\Conid{S})\;\Conid{T}} 
that results
with the operations of \ensuremath{\Varid{identity} \mathbin{\fatsemi} bx} and \ensuremath{bx}, that is:
\begin{hscode}\SaveRestoreHook
\column{B}{@{}>{\hspre}l<{\hspost}@{}}%
\column{3}{@{}>{\hspre}l<{\hspost}@{}}%
\column{36}{@{}>{\hspre}c<{\hspost}@{}}%
\column{36E}{@{}l@{}}%
\column{38}{@{}>{\hspre}c<{\hspost}@{}}%
\column{38E}{@{}l@{}}%
\column{39}{@{}>{\hspre}l<{\hspost}@{}}%
\column{41}{@{}>{\hspre}l<{\hspost}@{}}%
\column{E}{@{}>{\hspre}l<{\hspost}@{}}%
\>[3]{}(\iota \;\Varid{h})\;bx\mathord{.}\get{L}{}\<[36]%
\>[36]{}\mathrel{=}{}\<[36E]%
\>[39]{}(\Varid{id} \mathbin{\fatsemi} bx)\mathord{.}\get{L}{}\<[E]%
\\
\>[3]{}(\iota \;\Varid{h})\;bx\mathord{.}\get{R}{}\<[36]%
\>[36]{}\mathrel{=}{}\<[36E]%
\>[39]{}(\Varid{id} \mathbin{\fatsemi} bx)\mathord{.}\get{R}{}\<[E]%
\\
\>[3]{}(\iota \;\Varid{h})\;bx\mathord{.}\set{L}\;\Varid{a}{}\<[38]%
\>[38]{}\mathrel{=}{}\<[38E]%
\>[41]{}(\Varid{id} \mathbin{\fatsemi} bx)\mathord{.}\set{L}\;\Varid{a}{}\<[E]%
\\
\>[3]{}(\iota \;\Varid{h})\;bx\mathord{.}\set{R}\;\Varid{b}{}\<[38]%
\>[38]{}\mathrel{=}{}\<[38E]%
\>[41]{}(\Varid{id} \mathbin{\fatsemi} bx)\mathord{.}\set{R}\;\Varid{b}{}\<[E]%
\ColumnHook
\end{hscode}\resethooks

\jrcnote{TODO: Fill in a few cases}
We illustrate the \ensuremath{\get{L},\set{L}} cases, as they are more interesting. 
\begin{hscode}\SaveRestoreHook
\column{B}{@{}>{\hspre}l<{\hspost}@{}}%
\column{3}{@{}>{\hspre}l<{\hspost}@{}}%
\column{5}{@{}>{\hspre}l<{\hspost}@{}}%
\column{7}{@{}>{\hspre}c<{\hspost}@{}}%
\column{7E}{@{}l@{}}%
\column{10}{@{}>{\hspre}l<{\hspost}@{}}%
\column{E}{@{}>{\hspre}l<{\hspost}@{}}%
\>[3]{}(\iota \;\Varid{h})\;bx\mathord{.}\get{L}{}\<[E]%
\\
\>[3]{}\hsindent{2}{}\<[5]%
\>[5]{}\mathrel{=}\mbox{\commentbegin  definition of \ensuremath{\iota \;\Varid{h}} (using \ensuremath{\Varid{f}} and \ensuremath{\Varid{f}^{-1}})  \commentend}{}\<[E]%
\\
\>[3]{}\mathbf{do}\;{}\<[7]%
\>[7]{}\{\mskip1.5mu {}\<[7E]%
\>[10]{}(\Varid{a},\Varid{s})\leftarrow \Varid{get};{}\<[E]%
\\
\>[10]{}(\Varid{a'},\Varid{s'})\leftarrow \Varid{lift}\;(bx\mathord{.}\get{L}\;(\Varid{snd}\;(\Varid{a},\Varid{s})));{}\<[E]%
\\
\>[10]{}\Varid{set}\;(bx\mathord{.}\Varid{read}_{L}\;\Varid{s'},\Varid{s'});{}\<[E]%
\\
\>[10]{}\Varid{return}\;\Varid{a'}\mskip1.5mu\}{}\<[E]%
\\
\>[3]{}\hsindent{2}{}\<[5]%
\>[5]{}\mathrel{=}\mbox{\commentbegin  simplifying \ensuremath{\Varid{snd}}; \ensuremath{bx} is transparent   \commentend}{}\<[E]%
\\
\>[3]{}\mathbf{do}\;{}\<[7]%
\>[7]{}\{\mskip1.5mu {}\<[7E]%
\>[10]{}(\Varid{a},\Varid{s})\leftarrow \Varid{get};{}\<[E]%
\\
\>[10]{}(\Varid{a'},\Varid{s'})\leftarrow \Varid{lift}\;(\Varid{return}\;(bx\mathord{.}\Varid{read}_{L}\;\Varid{s},\Varid{s}));{}\<[E]%
\\
\>[10]{}\Varid{set}\;(bx\mathord{.}\Varid{read}_{L}\;\Varid{s'},\Varid{s'});{}\<[E]%
\\
\>[10]{}\Varid{return}\;\Varid{a'}\mskip1.5mu\}{}\<[E]%
\\
\>[3]{}\hsindent{2}{}\<[5]%
\>[5]{}\mathrel{=}\mbox{\commentbegin  \ensuremath{\Varid{lift}} monad morphism  \commentend}{}\<[E]%
\\
\>[3]{}\mathbf{do}\;{}\<[7]%
\>[7]{}\{\mskip1.5mu {}\<[7E]%
\>[10]{}(\Varid{a},\Varid{s})\leftarrow \Varid{get};{}\<[E]%
\\
\>[10]{}(\Varid{a'},\Varid{s'})\leftarrow \Varid{return}\;(bx\mathord{.}\Varid{read}_{L}\;\Varid{s},\Varid{s});{}\<[E]%
\\
\>[10]{}\Varid{set}\;(bx\mathord{.}\Varid{read}_{L}\;\Varid{s'},\Varid{s'});{}\<[E]%
\\
\>[10]{}\Varid{return}\;\Varid{a'}\mskip1.5mu\}{}\<[E]%
\\
\>[3]{}\hsindent{2}{}\<[5]%
\>[5]{}\mathrel{=}\mbox{\commentbegin  inlining \ensuremath{\Varid{a'}} and \ensuremath{\Varid{s'}}  \commentend}{}\<[E]%
\\
\>[3]{}\mathbf{do}\;{}\<[7]%
\>[7]{}\{\mskip1.5mu {}\<[7E]%
\>[10]{}(\Varid{a},\Varid{s})\leftarrow \Varid{get};{}\<[E]%
\\
\>[10]{}\Varid{set}\;(bx\mathord{.}\Varid{read}_{L}\;\Varid{s},\Varid{s});{}\<[E]%
\\
\>[10]{}\Varid{return}\;bx\mathord{.}\Varid{read}_{L}\;\Varid{s}\mskip1.5mu\}{}\<[E]%
\\
\>[3]{}\hsindent{2}{}\<[5]%
\>[5]{}\mathrel{=}\mbox{\commentbegin  \ensuremath{(\Varid{a},\Varid{s})\mathbin{::}(\Conid{A}\mathbin{\!\Join\!}\Conid{S})} implies \ensuremath{bx\mathord{.}\Varid{read}_{L}\;\Varid{s}\mathrel{=}\Varid{a}}  \commentend}{}\<[E]%
\\
\>[3]{}\mathbf{do}\;{}\<[7]%
\>[7]{}\{\mskip1.5mu {}\<[7E]%
\>[10]{}(\Varid{a},\Varid{s})\leftarrow \Varid{get};{}\<[E]%
\\
\>[10]{}\Varid{set}\;(\Varid{a},\Varid{s});{}\<[E]%
\\
\>[10]{}\Varid{return}\;\Varid{a}\mskip1.5mu\}{}\<[E]%
\\
\>[3]{}\hsindent{2}{}\<[5]%
\>[5]{}\mathrel{=}\mbox{\commentbegin  \ensuremath{\mathrm{(GG)}} and then \ensuremath{\mathrm{(GS)}} for \ensuremath{\Conid{StateT}\;(\Conid{A}\mathbin{\!\Join\!}\Conid{S})\;\Conid{T}}  \commentend}{}\<[E]%
\\
\>[3]{}\mathbf{do}\;{}\<[7]%
\>[7]{}\{\mskip1.5mu {}\<[7E]%
\>[10]{}(\Varid{a},\Varid{s})\leftarrow \Varid{get};{}\<[E]%
\\
\>[10]{}\Varid{return}\;\Varid{a}\mskip1.5mu\}{}\<[E]%
\\
\>[3]{}\hsindent{2}{}\<[5]%
\>[5]{}\mathrel{=}\mbox{\commentbegin  introduce trivial binding (where \ensuremath{\Varid{get}\mathbin{::}\Conid{StateT}\;\Conid{A}\;\Conid{A}})  \commentend}{}\<[E]%
\\
\>[3]{}\mathbf{do}\;{}\<[7]%
\>[7]{}\{\mskip1.5mu {}\<[7E]%
\>[10]{}(\Varid{a},\Varid{s})\leftarrow \Varid{get};(\Varid{a'},\anonymous )\leftarrow \Varid{lift}\;(\Varid{get}\;\Varid{a});{}\<[E]%
\\
\>[10]{}\Varid{return}\;\Varid{a'}\mskip1.5mu\}{}\<[E]%
\\
\>[3]{}\hsindent{2}{}\<[5]%
\>[5]{}\mathrel{=}\mbox{\commentbegin  definition of \ensuremath{\Varid{id}\mathord{.}\get{L}}  \commentend}{}\<[E]%
\\
\>[3]{}\mathbf{do}\;{}\<[7]%
\>[7]{}\{\mskip1.5mu {}\<[7E]%
\>[10]{}(\Varid{a},\Varid{s})\leftarrow \Varid{get};(\Varid{a'},\anonymous )\leftarrow \Varid{lift}\;(\Varid{id}\mathord{.}\get{L}\;\Varid{a});{}\<[E]%
\\
\>[10]{}\Varid{return}\;\Varid{a}\mskip1.5mu\}{}\<[E]%
\\
\>[3]{}\hsindent{2}{}\<[5]%
\>[5]{}\mathrel{=}\mbox{\commentbegin  Form of \ensuremath{\get{L}} preceding Remark \ref{rem:effectful-gets}  \commentend}{}\<[E]%
\\
\>[3]{}(\Varid{id} \mathbin{\fatsemi} bx)\mathord{.}\get{L}{}\<[E]%
\ColumnHook
\end{hscode}\resethooks
\fanote{Second-last step appealed to the forms of \ensuremath{\get{L}} at the end of the `alternative approach to composition' material... so it may be good to keep those forms at least}
Here is the proof for \ensuremath{\set{L}}:
\begin{hscode}\SaveRestoreHook
\column{B}{@{}>{\hspre}l<{\hspost}@{}}%
\column{3}{@{}>{\hspre}l<{\hspost}@{}}%
\column{4}{@{}>{\hspre}l<{\hspost}@{}}%
\column{5}{@{}>{\hspre}l<{\hspost}@{}}%
\column{7}{@{}>{\hspre}c<{\hspost}@{}}%
\column{7E}{@{}l@{}}%
\column{10}{@{}>{\hspre}l<{\hspost}@{}}%
\column{20}{@{}>{\hspre}l<{\hspost}@{}}%
\column{E}{@{}>{\hspre}l<{\hspost}@{}}%
\>[3]{}(\Varid{id} \mathbin{\fatsemi} bx)\mathord{.}\set{L}\;\Varid{a'}{}\<[E]%
\\
\>[3]{}\hsindent{2}{}\<[5]%
\>[5]{}\mathrel{=}\mbox{\commentbegin  Form of \ensuremath{\set{L}} preceding Remark \ref{rem:effectful-gets}  \commentend}{}\<[E]%
\\
\>[3]{}\mathbf{do}\;\{\mskip1.5mu {}\<[10]%
\>[10]{}(\Varid{a},\Varid{s})\leftarrow \Varid{get};((),\Varid{a'})\leftarrow \Varid{lift}\;(\Varid{id}\mathord{.}\set{L}\;\Varid{a'}\;\Varid{a});{}\<[E]%
\\
\>[10]{}\hsindent{10}{}\<[20]%
\>[20]{}(\Varid{b},\anonymous )\leftarrow \Varid{lift}\;(\Varid{id}\mathord{.}\get{R}\;\Varid{a'});{}\<[E]%
\\
\>[10]{}\hsindent{10}{}\<[20]%
\>[20]{}((),\Varid{s'})\leftarrow \Varid{lift}\;(bx\mathord{.}\set{L}\;\Varid{b}\;\Varid{s});\Varid{set}\;(\Varid{a'},\Varid{s'});\Varid{return}\;()\mskip1.5mu\}{}\<[E]%
\\
\>[3]{}\hsindent{2}{}\<[5]%
\>[5]{}\mathrel{=}\mbox{\commentbegin  definition of \ensuremath{\set{L},\get{R}} for \ensuremath{\Varid{id}}   \commentend}{}\<[E]%
\\
\>[3]{}\mathbf{do}\;\{\mskip1.5mu {}\<[10]%
\>[10]{}(\Varid{a},\Varid{s})\leftarrow \Varid{get};((),\Varid{a'})\leftarrow \Varid{lift}\;(\Varid{set}\;\Varid{a'}\;\Varid{a});{}\<[E]%
\\
\>[10]{}\hsindent{10}{}\<[20]%
\>[20]{}(\Varid{b},\anonymous )\leftarrow \Varid{lift}\;(\Varid{get}\;\Varid{a'});{}\<[E]%
\\
\>[10]{}\hsindent{10}{}\<[20]%
\>[20]{}((),\Varid{s'})\leftarrow \Varid{lift}\;(bx\mathord{.}\set{L}\;\Varid{b}\;\Varid{s});\Varid{set}\;(\Varid{a'},\Varid{s'});\Varid{return}\;()\mskip1.5mu\}{}\<[E]%
\\
\>[3]{}\hsindent{2}{}\<[5]%
\>[5]{}\mathrel{=}\mbox{\commentbegin  definitions of \ensuremath{\Varid{set}} and \ensuremath{\Varid{get}}  \commentend}{}\<[E]%
\\
\>[3]{}\mathbf{do}\;\{\mskip1.5mu {}\<[10]%
\>[10]{}(\Varid{a},\Varid{s})\leftarrow \Varid{get};((),\Varid{a'})\leftarrow \Varid{lift}\;(\Varid{return}\;((),\Varid{a'}));{}\<[E]%
\\
\>[10]{}\hsindent{10}{}\<[20]%
\>[20]{}(\Varid{b},\anonymous )\leftarrow \Varid{lift}\;(\Varid{return}\;(\Varid{a'},\Varid{a'}));{}\<[E]%
\\
\>[10]{}\hsindent{10}{}\<[20]%
\>[20]{}((),\Varid{s'})\leftarrow \Varid{lift}\;(bx\mathord{.}\set{L}\;\Varid{b}\;\Varid{s});\Varid{set}\;(\Varid{a'},\Varid{s'})\mskip1.5mu\}{}\<[E]%
\\
\>[3]{}\hsindent{2}{}\<[5]%
\>[5]{}\mathrel{=}\mbox{\commentbegin  \ensuremath{\Varid{lift}\;(\Varid{return}\;\Varid{x})\mathrel{=}\Varid{return}\;\Varid{x}}; inline resulting \ensuremath{\mathbf{let}}s  \commentend}{}\<[E]%
\\
\>[3]{}\mathbf{do}\;\{\mskip1.5mu {}\<[10]%
\>[10]{}(\Varid{a},\Varid{s})\leftarrow \Varid{get};{}\<[E]%
\\
\>[10]{}((),\Varid{s'})\leftarrow \Varid{lift}\;(bx\mathord{.}\set{L}\;\Varid{a'}\;\Varid{s});{}\<[E]%
\\
\>[10]{}\Varid{set}\;(\Varid{a'},\Varid{s'})\mskip1.5mu\}{}\<[E]%
\\
\>[3]{}\hsindent{1}{}\<[4]%
\>[4]{}\mathrel{=}\mbox{\commentbegin  introducing binding (\ensuremath{\Varid{return}\;\Varid{a'}\mathbin{::}\Conid{StateT}\;\Conid{S}\;\Conid{T}\;\Conid{A}})  \commentend}{}\<[E]%
\\
\>[3]{}\mathbf{do}\;\{\mskip1.5mu {}\<[10]%
\>[10]{}(\Varid{a},\Varid{s})\leftarrow \Varid{get};{}\<[E]%
\\
\>[10]{}((),\Varid{s'})\leftarrow \Varid{lift}\;(bx\mathord{.}\set{L}\;\Varid{a'}\;\Varid{s});{}\<[E]%
\\
\>[10]{}(\Varid{a''},\Varid{s''})\leftarrow \Varid{lift}\;(\Varid{return}\;\Varid{a'}\;\Varid{s'});{}\<[E]%
\\
\>[10]{}\Varid{set}\;(\Varid{a''},\Varid{s''})\mskip1.5mu\}{}\<[E]%
\\
\>[3]{}\hsindent{1}{}\<[4]%
\>[4]{}\mathrel{=}\mbox{\commentbegin  \ensuremath{\Varid{lift}} monad morphism \fanote{slight gloss}  \commentend}{}\<[E]%
\\
\>[3]{}\mathbf{do}\;\{\mskip1.5mu {}\<[10]%
\>[10]{}(\Varid{a},\Varid{s})\leftarrow \Varid{get};{}\<[E]%
\\
\>[10]{}(\Varid{a''},\Varid{s''})\leftarrow \Varid{lift}\;(\mathbf{do}\;\{\mskip1.5mu bx\mathord{.}\set{L}\;\Varid{a'};\Varid{return}\;\Varid{a'}\mskip1.5mu\}\;\Varid{s});{}\<[E]%
\\
\>[10]{}\Varid{set}\;(\Varid{a''},\Varid{s''})\mskip1.5mu\}{}\<[E]%
\\
\>[3]{}\hsindent{1}{}\<[4]%
\>[4]{}\mathrel{=}\mbox{\commentbegin  \ensuremath{\mathrm{(G_LS_L)}} for \ensuremath{bx}  \commentend}{}\<[E]%
\\
\>[3]{}\mathbf{do}\;\{\mskip1.5mu {}\<[10]%
\>[10]{}(\Varid{a},\Varid{s})\leftarrow \Varid{get};{}\<[E]%
\\
\>[10]{}(\Varid{a''},\Varid{s''})\leftarrow \Varid{lift}\;(\mathbf{do}\;\{\mskip1.5mu bx\mathord{.}\set{L}\;\Varid{a'};bx\mathord{.}\get{L}\mskip1.5mu\}\;\Varid{s});{}\<[E]%
\\
\>[10]{}\Varid{set}\;(\Varid{a''},\Varid{s''})\mskip1.5mu\}{}\<[E]%
\\
\>[3]{}\hsindent{1}{}\<[4]%
\>[4]{}\mathrel{=}\mbox{\commentbegin  \ensuremath{\Varid{lift}} monad morphism  \commentend}{}\<[E]%
\\
\>[3]{}\mathbf{do}\;\{\mskip1.5mu {}\<[10]%
\>[10]{}(\Varid{a},\Varid{s})\leftarrow \Varid{get};{}\<[E]%
\\
\>[10]{}((),\Varid{s'})\leftarrow \Varid{lift}\;(bx\mathord{.}\set{L}\;\Varid{a'}\;\Varid{s});{}\<[E]%
\\
\>[10]{}(\Varid{a''},\Varid{s''})\leftarrow \Varid{lift}\;(bx\mathord{.}\get{L}\;\Varid{s'});{}\<[E]%
\\
\>[10]{}\Varid{set}\;(\Varid{a''},\Varid{s''})\mskip1.5mu\}{}\<[E]%
\\
\>[3]{}\hsindent{1}{}\<[4]%
\>[4]{}\mathrel{=}\mbox{\commentbegin  \ensuremath{bx} is transparent  \commentend}{}\<[E]%
\\
\>[3]{}\mathbf{do}\;\{\mskip1.5mu {}\<[10]%
\>[10]{}(\Varid{a},\Varid{s})\leftarrow \Varid{get};{}\<[E]%
\\
\>[10]{}((),\Varid{s'})\leftarrow \Varid{lift}\;(bx\mathord{.}\set{L}\;\Varid{a'}\;\Varid{s});{}\<[E]%
\\
\>[10]{}(\Varid{a''},\Varid{s''})\leftarrow \Varid{lift}\;(\Varid{return}\;(bx\mathord{.}\Varid{read}_{L}\;\Varid{s'},\Varid{s'}));{}\<[E]%
\\
\>[10]{}\Varid{set}\;(\Varid{a''},\Varid{s''})\mskip1.5mu\}{}\<[E]%
\\
\>[3]{}\hsindent{1}{}\<[4]%
\>[4]{}\mathrel{=}\mbox{\commentbegin  \ensuremath{\Varid{lift}\hsdot{\cdot }{.}\Varid{return}\mathrel{=}\Varid{return}}; inlining \ensuremath{\Varid{a''}} and \ensuremath{\Varid{s''}}  \commentend}{}\<[E]%
\\
\>[3]{}\mathbf{do}\;\{\mskip1.5mu {}\<[10]%
\>[10]{}(\Varid{a},\Varid{s})\leftarrow \Varid{get};{}\<[E]%
\\
\>[10]{}((),\Varid{s'})\leftarrow \Varid{lift}\;(bx\mathord{.}\set{L}\;\Varid{a'}\;\Varid{s});{}\<[E]%
\\
\>[10]{}\Varid{set}\;(bx\mathord{.}\Varid{read}_{L}\;\Varid{s'},\Varid{s'})\mskip1.5mu\}{}\<[E]%
\\
\>[3]{}\hsindent{2}{}\<[5]%
\>[5]{}\mathrel{=}\mbox{\commentbegin  introducing trivial \ensuremath{\Varid{snd}} and \ensuremath{\Varid{return}}  \commentend}{}\<[E]%
\\
\>[3]{}\mathbf{do}\;{}\<[7]%
\>[7]{}\{\mskip1.5mu {}\<[7E]%
\>[10]{}(\Varid{a},\Varid{s})\leftarrow \Varid{get};{}\<[E]%
\\
\>[10]{}((),\Varid{s'})\leftarrow \Varid{lift}\;(bx\mathord{.}\set{L}\;\Varid{a'}\;(\Varid{snd}\;(\Varid{a},\Varid{s})));{}\<[E]%
\\
\>[10]{}\Varid{set}\;(bx\mathord{.}\Varid{read}_{L}\;\Varid{s'},\Varid{s'});{}\<[E]%
\\
\>[10]{}\Varid{return}\;()\mskip1.5mu\}{}\<[E]%
\\
\>[3]{}\hsindent{2}{}\<[5]%
\>[5]{}\mathrel{=}\mbox{\commentbegin  definition of \ensuremath{\iota \;\Varid{h}}  \commentend}{}\<[E]%
\\
\>[3]{}(\iota \;\Varid{h})\;bx\mathord{.}\set{L}{}\<[E]%
\ColumnHook
\end{hscode}\resethooks
Thus \ensuremath{\Varid{identity} \mathbin{\fatsemi} bx\equiv bx}.  The reasoning for the second
equation  is symmetric.
\end{proof}

\begin{lemma}\label{lem:associativity}
    For any well-behaved \ensuremath{bx}, we have
\begin{hscode}\SaveRestoreHook
\column{B}{@{}>{\hspre}l<{\hspost}@{}}%
\column{E}{@{}>{\hspre}l<{\hspost}@{}}%
\>[B]{}(bx_{1} \mathbin{\fatsemi} bx_{2}) \mathbin{\fatsemi} bx_{3}\equiv bx_{1} \mathbin{\fatsemi} (bx_{2} \mathbin{\fatsemi} bx_{3}){}\<[E]%
\ColumnHook
\end{hscode}\resethooks
\endswithdisplay
\end{lemma}
\begin{proof}
Consider composing
\begin{hscode}\SaveRestoreHook
\column{B}{@{}>{\hspre}l<{\hspost}@{}}%
\column{6}{@{}>{\hspre}l<{\hspost}@{}}%
\column{E}{@{}>{\hspre}l<{\hspost}@{}}%
\>[B]{}bx_{1}{}\<[6]%
\>[6]{}\mathbin{::}\Conid{StateTBX}\;\Conid{T}\;\Conid{S}_{1}\;\Conid{A}\;\Conid{B}{}\<[E]%
\\
\>[B]{}bx_{2}{}\<[6]%
\>[6]{}\mathbin{::}\Conid{StateTBX}\;\Conid{T}\;\Conid{S}_{2}\;\Conid{B}\;\Conid{C}{}\<[E]%
\\
\>[B]{}bx_{3}{}\<[6]%
\>[6]{}\mathbin{::}\Conid{StateTBX}\;\Conid{T}\;\Conid{S}_{3}\;\Conid{C}\;\Conid{D}{}\<[E]%
\ColumnHook
\end{hscode}\resethooks
in the following two ways:
\begin{hscode}\SaveRestoreHook
\column{B}{@{}>{\hspre}l<{\hspost}@{}}%
\column{3}{@{}>{\hspre}l<{\hspost}@{}}%
\column{29}{@{}>{\hspre}l<{\hspost}@{}}%
\column{E}{@{}>{\hspre}l<{\hspost}@{}}%
\>[3]{}bx_{1} \mathbin{\fatsemi} (bx_{2} \mathbin{\fatsemi} bx_{3}){}\<[29]%
\>[29]{}\mathbin{::}\Conid{StateTBX}\;\Conid{T}\;(\Conid{S}_{1}\mathbin{\!\Join\!}(\Conid{S}_{2}\mathbin{\!\Join\!}\Conid{S}_{3}))\;\Conid{A}\;\Conid{C}{}\<[E]%
\\
\>[3]{}(bx_{1} \mathbin{\fatsemi} bx_{2}) \mathbin{\fatsemi} bx_{3}{}\<[29]%
\>[29]{}\mathbin{::}\Conid{StateTBX}\;\Conid{T}\;((\Conid{S}_{1}\mathbin{\!\Join\!}\Conid{S}_{2})\mathbin{\!\Join\!}\Conid{S}_{3})\;\Conid{A}\;\Conid{C}{}\<[E]%
\ColumnHook
\end{hscode}\resethooks
Proving \ensuremath{\mathrm{(Assoc)}} 
amounts to showing that the obvious
isomorphism \ensuremath{\Varid{h}\mathbin{::}(\Varid{a},(\Varid{b},\Varid{c}))\hsarrow{\rightarrow }{\mathpunct{.}}((\Varid{a},\Varid{b}),\Varid{c})}
induces a monad isomorphism 
\ensuremath{\iota \;\Varid{h}\mathbin{::}\Conid{StateT}\;(\Conid{S}_{1}\mathbin{\!\Join\!}(\Conid{S}_{2}\mathbin{\!\Join\!}\Conid{S}_{3}))\;\Conid{T}\;\alpha\hsarrow{\rightarrow }{\mathpunct{.}}\Conid{StateT}\;((\Conid{S}_{1}\mathbin{\!\Join\!}\Conid{S}_{2})\mathbin{\!\Join\!}\Conid{S}_{3})\;\Conid{T}\;\alpha} 
and checking that 
\begin{hscode}\SaveRestoreHook
\column{B}{@{}>{\hspre}l<{\hspost}@{}}%
\column{3}{@{}>{\hspre}l<{\hspost}@{}}%
\column{62}{@{}>{\hspre}c<{\hspost}@{}}%
\column{62E}{@{}l@{}}%
\column{65}{@{}>{\hspre}l<{\hspost}@{}}%
\column{E}{@{}>{\hspre}l<{\hspost}@{}}%
\>[3]{}(\iota \;\Varid{h})\;(bx_{1} \mathbin{\fatsemi} (bx_{2} \mathbin{\fatsemi} bx_{3}))\mathord{.}\get{L}{}\<[62]%
\>[62]{}\mathrel{=}{}\<[62E]%
\>[65]{}((bx_{1} \mathbin{\fatsemi} bx_{2}) \mathbin{\fatsemi} bx_{3})\mathord{.}\get{L}{}\<[E]%
\\
\>[3]{}(\iota \;\Varid{h})\;(bx_{1} \mathbin{\fatsemi} (bx_{2} \mathbin{\fatsemi} bx_{3}))\mathord{.}\get{R}{}\<[62]%
\>[62]{}\mathrel{=}{}\<[62E]%
\>[65]{}((bx_{1} \mathbin{\fatsemi} bx_{2}) \mathbin{\fatsemi} bx_{3})\mathord{.}\get{R}{}\<[E]%
\\
\>[3]{}(\iota \;\Varid{h})\;(bx_{1} \mathbin{\fatsemi} (bx_{2} \mathbin{\fatsemi} bx_{3}))\mathord{.}\set{L}\;\Varid{a}{}\<[62]%
\>[62]{}\mathrel{=}{}\<[62E]%
\>[65]{}((bx_{1} \mathbin{\fatsemi} bx_{2}) \mathbin{\fatsemi} bx_{3})\mathord{.}\set{L}\;\Varid{a}{}\<[E]%
\\
\>[3]{}(\iota \;\Varid{h})\;(bx_{1} \mathbin{\fatsemi} (bx_{2} \mathbin{\fatsemi} bx_{3}))\mathord{.}\set{R}\;\Varid{b}{}\<[62]%
\>[62]{}\mathrel{=}{}\<[62E]%
\>[65]{}((bx_{1} \mathbin{\fatsemi} bx_{2}) \mathbin{\fatsemi} bx_{3})\mathord{.}\set{R}\;\Varid{b}{}\<[E]%
\ColumnHook
\end{hscode}\resethooks
We outline the proofs of the \ensuremath{\get{L}} and \ensuremath{\set{L}} cases. For \ensuremath{\get{L}}, we make use of the following property: 
\begin{hscode}\SaveRestoreHook
\column{B}{@{}>{\hspre}l<{\hspost}@{}}%
\column{12}{@{}>{\hspre}c<{\hspost}@{}}%
\column{12E}{@{}l@{}}%
\column{16}{@{}>{\hspre}l<{\hspost}@{}}%
\column{E}{@{}>{\hspre}l<{\hspost}@{}}%
\>[B]{}\mathrm{(\ast)}{}\<[12]%
\>[12]{}\quad{}\<[12E]%
\>[16]{}(bx_{1} \mathbin{\fatsemi} (bx_{2} \mathbin{\fatsemi} bx_{3}))\mathord{.}\get{L}\;(\Varid{s}_{1},(\Varid{s}_{2},\Varid{s}_{3}))\mathrel{=}\Varid{return}\;(bx_{1}\mathord{.}\Varid{read}_{L}\;\Varid{s}_{1},(\Varid{s}_{1},(\Varid{s}_{2},\Varid{s}_{3}))){}\<[E]%
\ColumnHook
\end{hscode}\resethooks
This allows us to consider the \ensuremath{\get{L}} condition above:
\begin{hscode}\SaveRestoreHook
\column{B}{@{}>{\hspre}l<{\hspost}@{}}%
\column{3}{@{}>{\hspre}l<{\hspost}@{}}%
\column{5}{@{}>{\hspre}l<{\hspost}@{}}%
\column{7}{@{}>{\hspre}c<{\hspost}@{}}%
\column{7E}{@{}l@{}}%
\column{10}{@{}>{\hspre}l<{\hspost}@{}}%
\column{E}{@{}>{\hspre}l<{\hspost}@{}}%
\>[3]{}(\iota \;\Varid{h})\;(bx_{1} \mathbin{\fatsemi} (bx_{2} \mathbin{\fatsemi} bx_{3}))\mathord{.}\get{L}{}\<[E]%
\\
\>[3]{}\hsindent{2}{}\<[5]%
\>[5]{}\mathrel{=}\mbox{\commentbegin  definition of \ensuremath{\iota \;\Varid{h}}  \commentend}{}\<[E]%
\\
\>[3]{}\mathbf{do}\;{}\<[7]%
\>[7]{}\{\mskip1.5mu {}\<[7E]%
\>[10]{}((\Varid{s}_{1},\Varid{s}_{2}),\Varid{s}_{3})\leftarrow \Varid{get};{}\<[E]%
\\
\>[10]{}(\Varid{a'},(\Varid{s}_{1}',(\Varid{s}_{2}',\Varid{s}_{3}')))\leftarrow \Varid{lift}\;((bx_{1} \mathbin{\fatsemi} (bx_{2} \mathbin{\fatsemi} bx_{3}))\mathord{.}\get{L}\;(\Varid{s}_{1},(\Varid{s}_{2},\Varid{s}_{3})));{}\<[E]%
\\
\>[10]{}\Varid{set}\;((\Varid{s}_{1}',\Varid{s}_{2}'),\Varid{s}_{3}');{}\<[E]%
\\
\>[10]{}\Varid{return}\;\Varid{a'}\mskip1.5mu\}{}\<[E]%
\\
\>[3]{}\hsindent{2}{}\<[5]%
\>[5]{}\mathrel{=}\mbox{\commentbegin  property \ensuremath{\mathrm{(\ast)}}  \commentend}{}\<[E]%
\\
\>[3]{}\mathbf{do}\;{}\<[7]%
\>[7]{}\{\mskip1.5mu {}\<[7E]%
\>[10]{}((\Varid{s}_{1},\Varid{s}_{2}),\Varid{s}_{3})\leftarrow \Varid{get};{}\<[E]%
\\
\>[10]{}(\Varid{a'},(\Varid{s}_{1}',(\Varid{s}_{2}',\Varid{s}_{3}')))\leftarrow \Varid{lift}\;(\Varid{return}\;(bx_{1}\mathord{.}\Varid{read}_{L}\;\Varid{s}_{1},(\Varid{s}_{1},(\Varid{s}_{2},\Varid{s}_{3}))));{}\<[E]%
\\
\>[10]{}\Varid{set}\;((\Varid{s}_{1}',\Varid{s}_{2}'),\Varid{s}_{3}');{}\<[E]%
\\
\>[10]{}\Varid{return}\;\Varid{a'}\mskip1.5mu\}{}\<[E]%
\\
\>[3]{}\hsindent{2}{}\<[5]%
\>[5]{}\mathrel{=}\mbox{\commentbegin  \ensuremath{\Varid{lift}\hsdot{\cdot }{.}\Varid{return}\mathrel{=}\Varid{return}}; inlining \ensuremath{\Varid{a'},\Varid{s}_{1}',\Varid{s}_{2}',\Varid{s}_{3}'}  \commentend}{}\<[E]%
\\
\>[3]{}\mathbf{do}\;{}\<[7]%
\>[7]{}\{\mskip1.5mu {}\<[7E]%
\>[10]{}((\Varid{s}_{1},\Varid{s}_{2}),\Varid{s}_{3})\leftarrow \Varid{get};{}\<[E]%
\\
\>[10]{}\Varid{set}\;((\Varid{s}_{1},\Varid{s}_{2}),\Varid{s}_{3});{}\<[E]%
\\
\>[10]{}\Varid{return}\;(bx_{1}\mathord{.}\Varid{read}_{L}\;\Varid{s}_{1},((\Varid{s}_{1},\Varid{s}_{2}),\Varid{s}_{3}))\mskip1.5mu\}{}\<[E]%
\\
\>[3]{}\hsindent{2}{}\<[5]%
\>[5]{}\mathrel{=}\mbox{\commentbegin  \ensuremath{\mathrm{(GG)}} and \ensuremath{\mathrm{(GS)}}  \commentend}{}\<[E]%
\\
\>[3]{}\mathbf{do}\;{}\<[7]%
\>[7]{}\{\mskip1.5mu {}\<[7E]%
\>[10]{}((\Varid{s}_{1},\Varid{s}_{2}),\Varid{s}_{3})\leftarrow \Varid{get};{}\<[E]%
\\
\>[10]{}\Varid{return}\;(bx_{1}\mathord{.}\Varid{read}_{L}\;\Varid{s}_{1},((\Varid{s}_{1},\Varid{s}_{2}),\Varid{s}_{3}))\mskip1.5mu\}{}\<[E]%
\ColumnHook
\end{hscode}\resethooks

Analogously to property \ensuremath{\mathrm{(\ast)}}, we may similarly show that
\begin{hscode}\SaveRestoreHook
\column{B}{@{}>{\hspre}l<{\hspost}@{}}%
\column{4}{@{}>{\hspre}l<{\hspost}@{}}%
\column{57}{@{}>{\hspre}l<{\hspost}@{}}%
\column{E}{@{}>{\hspre}l<{\hspost}@{}}%
\>[4]{}((bx_{1} \mathbin{\fatsemi} bx_{2}) \mathbin{\fatsemi} bx_{3})\mathord{.}\get{L}\;((\Varid{s}_{1},\Varid{s}_{2}),\Varid{s}_{3})\mathrel{=}{}\<[57]%
\>[57]{}\Varid{return}\;(bx_{1}\mathord{.}\Varid{read}_{L}\;\Varid{s}_{1},((\Varid{s}_{1},\Varid{s}_{2}),\Varid{s}_{3})){}\<[E]%
\ColumnHook
\end{hscode}\resethooks
and the previous proof can be adapted to show that
\ensuremath{((bx_{1} \mathbin{\fatsemi} bx_{2}) \mathbin{\fatsemi} bx_{3})\mathord{.}\get{L}}
also simplifies into
\begin{hscode}\SaveRestoreHook
\column{B}{@{}>{\hspre}l<{\hspost}@{}}%
\column{5}{@{}>{\hspre}c<{\hspost}@{}}%
\column{5E}{@{}l@{}}%
\column{8}{@{}>{\hspre}l<{\hspost}@{}}%
\column{E}{@{}>{\hspre}l<{\hspost}@{}}%
\>[B]{}\mathbf{do}\;{}\<[5]%
\>[5]{}\{\mskip1.5mu {}\<[5E]%
\>[8]{}((\Varid{s}_{1},\Varid{s}_{2}),\Varid{s}_{3})\leftarrow \Varid{get};\Varid{return}\;(bx_{1}\mathord{.}\Varid{read}_{L}\;\Varid{s}_{1},((\Varid{s}_{1},\Varid{s}_{2}),\Varid{s}_{3}))\mskip1.5mu\}{}\<[E]%
\ColumnHook
\end{hscode}\resethooks
which concludes the proof of the \ensuremath{\get{L}} condition.

As for \ensuremath{\set{L}}, we use the following property:
\begin{hscode}\SaveRestoreHook
\column{B}{@{}>{\hspre}l<{\hspost}@{}}%
\column{12}{@{}>{\hspre}c<{\hspost}@{}}%
\column{12E}{@{}l@{}}%
\column{16}{@{}>{\hspre}l<{\hspost}@{}}%
\column{18}{@{}>{\hspre}c<{\hspost}@{}}%
\column{18E}{@{}l@{}}%
\column{21}{@{}>{\hspre}l<{\hspost}@{}}%
\column{27}{@{}>{\hspre}l<{\hspost}@{}}%
\column{E}{@{}>{\hspre}l<{\hspost}@{}}%
\>[B]{}\mathrm{(\dag)}{}\<[12]%
\>[12]{}\quad{}\<[12E]%
\>[16]{}(bx_{1} \mathbin{\fatsemi} (bx_{2} \mathbin{\fatsemi} bx_{3}))\mathord{.}\set{L}\;\Varid{a'}\;(\Varid{s}_{1},(\Varid{s}_{2},\Varid{s}_{3})){}\<[E]%
\\
\>[16]{}\hsindent{2}{}\<[18]%
\>[18]{}\mathrel{=}{}\<[18E]%
\>[21]{}\mathbf{do}\;\{\mskip1.5mu {}\<[27]%
\>[27]{}((),\Varid{s}_{1}')\leftarrow bx_{1}\mathord{.}\set{L}\;\Varid{a'}\;\Varid{s}_{1};{}\<[E]%
\\
\>[27]{}(\Varid{a''},\Varid{s}_{2}')\leftarrow bx_{1}\mathord{.}\get{R}\;\Varid{s}_{1}';{}\<[E]%
\\
\>[27]{}((),\Varid{s}_{2}')\leftarrow bx_{2}\mathord{.}\set{L}\;\Varid{a''}\;\Varid{s}_{2};{}\<[E]%
\\
\>[27]{}(\Varid{b'},\Varid{s}_{2}')\leftarrow bx_{2}\mathord{.}\get{R}\;\Varid{s}_{2}';{}\<[E]%
\\
\>[27]{}((),\Varid{s}_{3}')\leftarrow bx_{3}\mathord{.}\set{L}\;\Varid{b'}\;\Varid{s}_{3};{}\<[E]%
\\
\>[27]{}\Varid{return}\;((),(\Varid{s}_{1}',(\Varid{s}_{2}',\Varid{s}_{3}')))\mskip1.5mu\}{}\<[E]%
\ColumnHook
\end{hscode}\resethooks
and one may prove an analogous form for 
\begin{hscode}\SaveRestoreHook
\column{B}{@{}>{\hspre}l<{\hspost}@{}}%
\column{E}{@{}>{\hspre}l<{\hspost}@{}}%
\>[B]{}((bx_{1} \mathbin{\fatsemi} bx_{2}) \mathbin{\fatsemi} bx_{3})\mathord{.}\set{L}\;\Varid{a'}\;((\Varid{s}_{1},\Varid{s}_{2}),\Varid{s}_{3}){}\<[E]%
\ColumnHook
\end{hscode}\resethooks
where the final line instead reads \ensuremath{\Varid{return}\;((),((\Varid{s}_{1}',\Varid{s}_{2}'),\Varid{s}_{3}'))}.
This allows us to show:
\begin{hscode}\SaveRestoreHook
\column{B}{@{}>{\hspre}l<{\hspost}@{}}%
\column{3}{@{}>{\hspre}l<{\hspost}@{}}%
\column{5}{@{}>{\hspre}l<{\hspost}@{}}%
\column{7}{@{}>{\hspre}c<{\hspost}@{}}%
\column{7E}{@{}l@{}}%
\column{10}{@{}>{\hspre}l<{\hspost}@{}}%
\column{E}{@{}>{\hspre}l<{\hspost}@{}}%
\>[3]{}(\iota \;\Varid{h})\;(bx_{1} \mathbin{\fatsemi} (bx_{2} \mathbin{\fatsemi} bx_{3}))\mathord{.}\set{L}\;\Varid{a'}{}\<[E]%
\\
\>[3]{}\hsindent{2}{}\<[5]%
\>[5]{}\mathrel{=}\mbox{\commentbegin  definition of \ensuremath{\iota \;\Varid{h}}  \commentend}{}\<[E]%
\\
\>[3]{}\mathbf{do}\;{}\<[7]%
\>[7]{}\{\mskip1.5mu {}\<[7E]%
\>[10]{}((\Varid{s}_{1},\Varid{s}_{2}),\Varid{s}_{3})\leftarrow \Varid{get};{}\<[E]%
\\
\>[10]{}((),(\Varid{s}_{1}',(\Varid{s}_{2}',\Varid{s}_{3}')))\leftarrow \Varid{lift}\;((bx_{1} \mathbin{\fatsemi} (bx_{2} \mathbin{\fatsemi} bx_{3}))\mathord{.}\set{L}\;\Varid{a'}\;(\Varid{s}_{1},(\Varid{s}_{2},\Varid{s}_{3})));{}\<[E]%
\\
\>[10]{}\Varid{set}\;((\Varid{s}_{1}',\Varid{s}_{2}'),\Varid{s}_{3}');{}\<[E]%
\\
\>[10]{}\Varid{return}\;()\mskip1.5mu\}{}\<[E]%
\\
\>[3]{}\hsindent{2}{}\<[5]%
\>[5]{}\mathrel{=}\mbox{\commentbegin  Property \ensuremath{\mathrm{(\dag)}}  \commentend}{}\<[E]%
\\
\>[3]{}\mathbf{do}\;{}\<[7]%
\>[7]{}\{\mskip1.5mu {}\<[7E]%
\>[10]{}((\Varid{s}_{1},\Varid{s}_{2}),\Varid{s}_{3})\leftarrow \Varid{get};{}\<[E]%
\\
\>[10]{}((),((\Varid{s}_{1}',\Varid{s}_{2}'),\Varid{s}_{3}'))\leftarrow \Varid{lift}\;(((bx_{1} \mathbin{\fatsemi} bx_{2}) \mathbin{\fatsemi} bx_{3})\mathord{.}\set{L}\;\Varid{a'}\;((\Varid{s}_{1},\Varid{s}_{2}),\Varid{s}_{3}));{}\<[E]%
\\
\>[10]{}\Varid{set}\;((\Varid{s}_{1}',\Varid{s}_{2}'),\Varid{s}_{3}');{}\<[E]%
\\
\>[10]{}\Varid{return}\;()\mskip1.5mu\}{}\<[E]%
\\
\>[3]{}\hsindent{2}{}\<[5]%
\>[5]{}\mathrel{=}\mbox{\commentbegin  analogous form of Property \ensuremath{\mathrm{(\dag)}}  \commentend}{}\<[E]%
\\
\>[3]{}((bx_{1} \mathbin{\fatsemi} bx_{2}) \mathbin{\fatsemi} bx_{3})\mathord{.}\set{L}\;\Varid{a'}{}\<[E]%
\ColumnHook
\end{hscode}\resethooks
\endswithdisplay
\end{proof}

\restatableTheorem{thm:category}
\begin{thm:category}
Composition of transparent \bx{} satisfies the identity and associativity laws,
modulo \ensuremath{\equiv }.
\begin{hscode}\SaveRestoreHook
\column{B}{@{}>{\hspre}l<{\hspost}@{}}%
\column{3}{@{}>{\hspre}l<{\hspost}@{}}%
\column{19}{@{}>{\hspre}l<{\hspost}@{}}%
\column{45}{@{}>{\hspre}l<{\hspost}@{}}%
\column{E}{@{}>{\hspre}l<{\hspost}@{}}%
\>[3]{}\mathrm{(Identity)}\:{}\<[19]%
\>[19]{}\Varid{identity} \mathbin{\fatsemi} bx\equiv bx\equiv bx \mathbin{\fatsemi} \Varid{identity}{}\<[E]%
\\
\>[3]{}\mathrm{(Assoc)}\:{}\<[19]%
\>[19]{}bx_{1} \mathbin{\fatsemi} (bx_{2} \mathbin{\fatsemi} bx_{3}){}\<[45]%
\>[45]{}\equiv (bx_{1} \mathbin{\fatsemi} bx_{2}) \mathbin{\fatsemi} bx_{3}{}\<[E]%
\ColumnHook
\end{hscode}\resethooks
\endswithdisplay
\end{thm:category}
\begin{proof}
  By Lemmas~\ref{lem:identity} and~\ref{lem:associativity}.
\end{proof}
\section{Proofs from Section~\ref{sec:examples}}

\begin{proposition}
The \ensuremath{\Varid{dual}} operator (Definition~\ref{def:dual}) preserves transparency and overwritability.
\end{proposition}
\begin{proof}
Immediate, since those properties are invariant under transposing left and right.
\end{proof}

\begin{lemma}
\ensuremath{\Varid{left}} and \ensuremath{\Varid{right}} are monad morphisms.
\end{lemma}
\begin{proof}
  Immediate from Lemma~\ref{lem:vwb-monad-morphism}, because \ensuremath{\Varid{left}} and
  \ensuremath{\Varid{right}} are of the form \ensuremath{\vartheta \;\Varid{l}} where \ensuremath{\Varid{l}} is a very well-behaved
  lens.  
\end{proof}

\restatableProposition{prop:pair-wb}
\begin{prop:pair-wb}
  If \ensuremath{bx_{1}} and \ensuremath{bx_{2}} are transparent, then \ensuremath{\Varid{pairBX}\;bx_{1}\;bx_{2}} is transparent. 
\end{prop:pair-wb}
\begin{proof}
Let \ensuremath{bx\mathrel{=}\Varid{pairBX}\;bx_{1}\;bx_{2}}.  Then to show 
\ensuremath{\mathrm{(G_LS_L)}}:
\begin{hscode}\SaveRestoreHook
\column{B}{@{}>{\hspre}c<{\hspost}@{}}%
\column{BE}{@{}l@{}}%
\column{4}{@{}>{\hspre}l<{\hspost}@{}}%
\column{6}{@{}>{\hspre}l<{\hspost}@{}}%
\column{10}{@{}>{\hspre}l<{\hspost}@{}}%
\column{E}{@{}>{\hspre}l<{\hspost}@{}}%
\>[4]{}\mathbf{do}\;\{\mskip1.5mu \Varid{a}\leftarrow bx\mathord{.}\get{L};bx\mathord{.}\set{L}\;\Varid{a}\mskip1.5mu\}{}\<[E]%
\\
\>[B]{}\mathrel{=}{}\<[BE]%
\>[6]{}\mbox{\commentbegin  eta-expansion  \commentend}{}\<[E]%
\\
\>[B]{}\hsindent{4}{}\<[4]%
\>[4]{}\mathbf{do}\;\{\mskip1.5mu (\Varid{a}_{1},\Varid{a}_{2})\leftarrow bx\mathord{.}\get{L};bx\mathord{.}\set{L}\;(\Varid{a}_{1},\Varid{a}_{2})\mskip1.5mu\}{}\<[E]%
\\
\>[B]{}\mathrel{=}{}\<[BE]%
\>[6]{}\mbox{\commentbegin  Definition  \commentend}{}\<[E]%
\\
\>[B]{}\hsindent{4}{}\<[4]%
\>[4]{}\mathbf{do}\;\{\mskip1.5mu {}\<[10]%
\>[10]{}\Varid{a}_{1}\leftarrow \Varid{left}\;(bx_{1}\mathord{.}\get{L});\Varid{a}_{2}\leftarrow \Varid{right}\;(bx_{2}\mathord{.}\get{L});{}\<[E]%
\\
\>[10]{}(\Varid{a}_{1}',\Varid{a}_{2}')\leftarrow \Varid{return}\;(\Varid{a}_{1},\Varid{a}_{2});{}\<[E]%
\\
\>[10]{}\Varid{left}\;(bx_{1}\mathord{.}\set{L}\;\Varid{a}_{1}');\Varid{right}\;(bx_{2}\mathord{.}\set{L}\;\Varid{a}_{2}')\mskip1.5mu\}{}\<[E]%
\\
\>[B]{}\mathrel{=}{}\<[BE]%
\>[6]{}\mbox{\commentbegin  Monad unit  \commentend}{}\<[E]%
\\
\>[B]{}\hsindent{4}{}\<[4]%
\>[4]{}\mathbf{do}\;\{\mskip1.5mu {}\<[10]%
\>[10]{}\Varid{a}_{1}\leftarrow \Varid{left}\;(bx_{1}\mathord{.}\get{L});\Varid{a}_{2}\leftarrow \Varid{right}\;(bx_{2}\mathord{.}\get{L});{}\<[E]%
\\
\>[10]{}\Varid{left}\;(bx_{1}\mathord{.}\set{L}\;\Varid{a}_{1});\Varid{right}\;(bx_{2}\mathord{.}\set{L}\;\Varid{a}_{2})\mskip1.5mu\}{}\<[E]%
\\
\>[B]{}\mathrel{=}{}\<[BE]%
\>[6]{}\mbox{\commentbegin  Lemma~\ref{lem:left-right-commute}, since \ensuremath{bx_{2}\mathord{.}\get{L}} is \ensuremath{\Conid{T}}-pure  \commentend}{}\<[E]%
\\
\>[B]{}\hsindent{4}{}\<[4]%
\>[4]{}\mathbf{do}\;\{\mskip1.5mu {}\<[10]%
\>[10]{}\Varid{a}_{1}\leftarrow \Varid{left}\;(bx_{1}\mathord{.}\get{L});\Varid{left}\;(bx_{1}\mathord{.}\set{L}\;\Varid{a}_{1});{}\<[E]%
\\
\>[10]{}\Varid{a}_{2}\leftarrow \Varid{right}\;(bx_{2}\mathord{.}\get{L});\Varid{right}\;(bx_{2}\mathord{.}\set{L}\;\Varid{a}_{2})\mskip1.5mu\}{}\<[E]%
\\
\>[B]{}\mathrel{=}{}\<[BE]%
\>[6]{}\mbox{\commentbegin  \ensuremath{\Varid{left}}, \ensuremath{\Varid{right}} monad morphisms  \commentend}{}\<[E]%
\\
\>[B]{}\hsindent{4}{}\<[4]%
\>[4]{}\mathbf{do}\;\{\mskip1.5mu \Varid{left}\;(\mathbf{do}\;\{\mskip1.5mu \Varid{a}_{1}\leftarrow bx_{1}\mathord{.}\get{L};bx_{1}\mathord{.}\set{L}\;\Varid{a}_{1}\mskip1.5mu\});{}\<[E]%
\\
\>[4]{}\hsindent{6}{}\<[10]%
\>[10]{}\Varid{right}\;(\mathbf{do}\;\{\mskip1.5mu \Varid{a}_{2}\leftarrow bx_{2}\mathord{.}\get{L};bx_{2}\mathord{.}\set{L}\;\Varid{a}_{2}\mskip1.5mu\})\mskip1.5mu\}{}\<[E]%
\\
\>[B]{}\mathrel{=}{}\<[BE]%
\>[6]{}\mbox{\commentbegin  \ensuremath{\mathrm{(G_LS_L)}} twice  \commentend}{}\<[E]%
\\
\>[B]{}\hsindent{4}{}\<[4]%
\>[4]{}\mathbf{do}\;\{\mskip1.5mu \Varid{left}\;(\Varid{return}\;());\Varid{right}\;(\Varid{return}\;())\mskip1.5mu\}{}\<[E]%
\\
\>[B]{}\mathrel{=}{}\<[BE]%
\>[6]{}\mbox{\commentbegin  monad morphism, unit  \commentend}{}\<[E]%
\\
\>[B]{}\hsindent{4}{}\<[4]%
\>[4]{}\Varid{return}\;(){}\<[E]%
\ColumnHook
\end{hscode}\resethooks

Likewise, to show \ensuremath{\mathrm{(S_LG_L)}}:
\begin{hscode}\SaveRestoreHook
\column{B}{@{}>{\hspre}c<{\hspost}@{}}%
\column{BE}{@{}l@{}}%
\column{4}{@{}>{\hspre}l<{\hspost}@{}}%
\column{6}{@{}>{\hspre}l<{\hspost}@{}}%
\column{10}{@{}>{\hspre}l<{\hspost}@{}}%
\column{68}{@{}>{\hspre}l<{\hspost}@{}}%
\column{E}{@{}>{\hspre}l<{\hspost}@{}}%
\>[4]{}\mathbf{do}\;\{\mskip1.5mu bx\mathord{.}\set{L}\;\Varid{a};bx\mathord{.}\get{L}\mskip1.5mu\}{}\<[E]%
\\
\>[B]{}\mathrel{=}{}\<[BE]%
\>[6]{}\mbox{\commentbegin  eta-expansion  \commentend}{}\<[E]%
\\
\>[B]{}\hsindent{4}{}\<[4]%
\>[4]{}\mathbf{do}\;\{\mskip1.5mu bx\mathord{.}\set{L}\;(\Varid{a}_{1},\Varid{a}_{2});bx\mathord{.}\get{L}\mskip1.5mu\}{}\<[E]%
\\
\>[B]{}\mathrel{=}{}\<[BE]%
\>[6]{}\mbox{\commentbegin  definition  \commentend}{}\<[E]%
\\
\>[B]{}\hsindent{4}{}\<[4]%
\>[4]{}\mathbf{do}\;\{\mskip1.5mu {}\<[10]%
\>[10]{}\Varid{left}\;(bx_{1}\mathord{.}\set{L}\;\Varid{a}_{1});\Varid{right}\;(bx_{2}\mathord{.}\set{L}\;\Varid{a}_{2});{}\<[E]%
\\
\>[10]{}\Varid{a}_{1}'\leftarrow \Varid{left}\;(bx_{1}\mathord{.}\get{L});\Varid{a}_{2}'\leftarrow \Varid{right}\;(bx_{2}\mathord{.}\get{L});{}\<[E]%
\\
\>[10]{}\Varid{return}\;(\Varid{a}_{1}',\Varid{a}_{2}')\mskip1.5mu\}{}\<[E]%
\\
\>[B]{}\mathrel{=}{}\<[BE]%
\>[6]{}\mbox{\commentbegin  Lemma~\ref{lem:left-right-commute}, since \ensuremath{bx_{1}\mathord{.}\get{L}} \ensuremath{\Conid{T}}-pure  \commentend}{}\<[E]%
\\
\>[B]{}\hsindent{4}{}\<[4]%
\>[4]{}\mathbf{do}\;\{\mskip1.5mu {}\<[10]%
\>[10]{}\Varid{left}\;(bx_{1}\mathord{.}\set{L}\;\Varid{a}_{1});\Varid{a}_{1}'\leftarrow \Varid{left}\;(bx_{1}\mathord{.}\get{L});{}\<[E]%
\\
\>[10]{}\Varid{right}\;(bx_{2}\mathord{.}\set{L}\;\Varid{a}_{2});\Varid{a}_{2}'\leftarrow \Varid{right}\;(bx_{2}\mathord{.}\get{L});{}\<[E]%
\\
\>[10]{}\Varid{return}\;(\Varid{a}_{1}',\Varid{a}_{2}')\mskip1.5mu\}{}\<[E]%
\\
\>[B]{}\mathrel{=}{}\<[BE]%
\>[6]{}\mbox{\commentbegin  \ensuremath{\mathrm{(S_LG_L)}} twice  \commentend}{}\<[E]%
\\
\>[B]{}\hsindent{4}{}\<[4]%
\>[4]{}\mathbf{do}\;\{\mskip1.5mu {}\<[10]%
\>[10]{}\Varid{left}\;(bx_{1}\mathord{.}\set{L}\;\Varid{a}_{1});\Varid{a}_{1}\leftarrow \Varid{return}\;\Varid{a}_{1};{}\<[E]%
\\
\>[10]{}\Varid{right}\;(bx_{2}\mathord{.}\set{L}\;\Varid{a}_{2});\Varid{a}_{2}\leftarrow \Varid{return}\;\Varid{a}_{2};\Varid{return}\;(\Varid{a}_{1},\Varid{a}_{2})\mskip1.5mu\}{}\<[E]%
\\
\>[B]{}\mathrel{=}{}\<[BE]%
\>[6]{}\mbox{\commentbegin  Monad unit  \commentend}{}\<[E]%
\\
\>[B]{}\hsindent{4}{}\<[4]%
\>[4]{}\mathbf{do}\;\{\mskip1.5mu {}\<[10]%
\>[10]{}\Varid{left}\;(bx_{1}\mathord{.}\set{L}\;\Varid{a}_{1});\Varid{right}\;(bx_{2}\mathord{.}\set{L}\;\Varid{a}_{2});{}\<[68]%
\>[68]{}\Varid{return}\;(\Varid{a}_{1},\Varid{a}_{2})\mskip1.5mu\}{}\<[E]%
\\
\>[B]{}\mathrel{=}{}\<[BE]%
\>[6]{}\mbox{\commentbegin  Definition  \commentend}{}\<[E]%
\\
\>[B]{}\hsindent{4}{}\<[4]%
\>[4]{}\mathbf{do}\;\{\mskip1.5mu bx\mathord{.}\set{L}\;(\Varid{a}_{1},\Varid{a}_{2});\Varid{return}\;(\Varid{a}_{1},\Varid{a}_{2})\mskip1.5mu\}{}\<[E]%
\\
\>[B]{}\mathrel{=}{}\<[BE]%
\>[6]{}\mbox{\commentbegin  eta-contraction for pairs  \commentend}{}\<[E]%
\\
\>[B]{}\hsindent{4}{}\<[4]%
\>[4]{}\mathbf{do}\;\{\mskip1.5mu bx\mathord{.}\set{L}\;\Varid{a};\Varid{return}\;\Varid{a}\mskip1.5mu\}{}\<[E]%
\ColumnHook
\end{hscode}\resethooks
\endswithdisplay
\end{proof}

\restatableProposition{prop:sum-wb}
\begin{prop:sum-wb}
  If \ensuremath{bx_{1}} and \ensuremath{bx_{2}} are transparent, then so is \ensuremath{\Varid{sumBX}\;bx_{1}\;bx_{2}}.
\end{prop:sum-wb}
\begin{proof}
Let \ensuremath{bx\mathrel{=}\Varid{sumBX}\;bx_{1}\;bx_{2}}. 
We first show that \ensuremath{bx} has \ensuremath{\Conid{T}}-pure queries.
Suppose that \ensuremath{bx_{1}} is transparent, with \ensuremath{\Varid{read}} functions \ensuremath{\Varid{rl}_{1}} and \ensuremath{\Varid{rr}_{1}}; and similarly for \ensuremath{bx_{2}}. Then
\begin{hscode}\SaveRestoreHook
\column{B}{@{}>{\hspre}c<{\hspost}@{}}%
\column{BE}{@{}l@{}}%
\column{3}{@{}>{\hspre}l<{\hspost}@{}}%
\column{5}{@{}>{\hspre}l<{\hspost}@{}}%
\column{9}{@{}>{\hspre}l<{\hspost}@{}}%
\column{15}{@{}>{\hspre}l<{\hspost}@{}}%
\column{18}{@{}>{\hspre}l<{\hspost}@{}}%
\column{21}{@{}>{\hspre}l<{\hspost}@{}}%
\column{27}{@{}>{\hspre}l<{\hspost}@{}}%
\column{37}{@{}>{\hspre}l<{\hspost}@{}}%
\column{43}{@{}>{\hspre}l<{\hspost}@{}}%
\column{E}{@{}>{\hspre}l<{\hspost}@{}}%
\>[3]{}bx\mathord{.}\get{L}{}\<[E]%
\\
\>[B]{}\mathrel{=}{}\<[BE]%
\>[5]{}\mbox{\commentbegin  definition of \ensuremath{\Varid{sumBX}}  \commentend}{}\<[E]%
\\
\>[B]{}\hsindent{3}{}\<[3]%
\>[3]{}\mathbf{do}\;\{\mskip1.5mu {}\<[9]%
\>[9]{}(\Varid{b},\Varid{s}_{1},\Varid{s}_{2})\leftarrow \Varid{get};{}\<[E]%
\\
\>[9]{}\mathbf{if}\;\Varid{b}\;{}\<[15]%
\>[15]{}\mathbf{then}\;{}\<[21]%
\>[21]{}\mathbf{do}\;\{\mskip1.5mu {}\<[27]%
\>[27]{}(\Varid{a}_{1},\anonymous )\leftarrow \Varid{lift}\;(bx_{1}\mathord{.}\get{L}\;\Varid{s}_{1});{}\<[E]%
\\
\>[27]{}\Varid{return}\;(\Conid{Left}\;\Varid{a}_{1})\mskip1.5mu\}{}\<[E]%
\\
\>[15]{}\mathbf{else}\;{}\<[21]%
\>[21]{}\mathbf{do}\;\{\mskip1.5mu {}\<[27]%
\>[27]{}(\Varid{a}_{2},\anonymous )\leftarrow \Varid{lift}\;(bx_{2}\mathord{.}\get{L}\;\Varid{s}_{2});{}\<[E]%
\\
\>[27]{}\Varid{return}\;(\Conid{Right}\;\Varid{a}_{2})\mskip1.5mu\}\mskip1.5mu\}{}\<[E]%
\\
\>[B]{}\mathrel{=}{}\<[BE]%
\>[5]{}\mbox{\commentbegin  \ensuremath{bx_{1}} and \ensuremath{bx_{2}} are transparent  \commentend}{}\<[E]%
\\
\>[B]{}\hsindent{3}{}\<[3]%
\>[3]{}\mathbf{do}\;\{\mskip1.5mu {}\<[9]%
\>[9]{}(\Varid{b},\Varid{s}_{1},\Varid{s}_{2})\leftarrow \Varid{get};{}\<[E]%
\\
\>[9]{}\mathbf{if}\;\Varid{b}\;{}\<[15]%
\>[15]{}\mathbf{then}\;{}\<[21]%
\>[21]{}\mathbf{do}\;\{\mskip1.5mu {}\<[27]%
\>[27]{}(\Varid{a}_{1},\anonymous )\leftarrow \Varid{lift}\;(\Varid{gets}\;\Varid{rl}_{1}\;\Varid{s}_{1});{}\<[E]%
\\
\>[27]{}\Varid{return}\;(\Conid{Left}\;\Varid{a}_{1})\mskip1.5mu\}{}\<[E]%
\\
\>[15]{}\mathbf{else}\;{}\<[21]%
\>[21]{}\mathbf{do}\;\{\mskip1.5mu {}\<[27]%
\>[27]{}(\Varid{a}_{2},\anonymous )\leftarrow \Varid{lift}\;(\Varid{gets}\;\Varid{rl}_{2}\;\Varid{s}_{2});{}\<[E]%
\\
\>[27]{}\Varid{return}\;(\Conid{Right}\;\Varid{a}_{2})\mskip1.5mu\}\mskip1.5mu\}{}\<[E]%
\\
\>[B]{}\mathrel{=}{}\<[BE]%
\>[5]{}\mbox{\commentbegin  definition of \ensuremath{\Varid{gets}}  \commentend}{}\<[E]%
\\
\>[B]{}\hsindent{3}{}\<[3]%
\>[3]{}\mathbf{do}\;\{\mskip1.5mu {}\<[9]%
\>[9]{}(\Varid{b},\Varid{s}_{1},\Varid{s}_{2})\leftarrow \Varid{get};{}\<[E]%
\\
\>[9]{}\mathbf{if}\;\Varid{b}\;{}\<[15]%
\>[15]{}\mathbf{then}\;{}\<[21]%
\>[21]{}\mathbf{do}\;\{\mskip1.5mu {}\<[27]%
\>[27]{}(\Varid{a}_{1},\anonymous )\leftarrow \Varid{lift}\;(\Varid{return}\;(\Varid{rl}_{1}\;\Varid{s}_{1},\Varid{s}_{1}));{}\<[E]%
\\
\>[27]{}\Varid{return}\;(\Conid{Left}\;\Varid{a}_{1})\mskip1.5mu\}{}\<[E]%
\\
\>[15]{}\mathbf{else}\;{}\<[21]%
\>[21]{}\mathbf{do}\;\{\mskip1.5mu {}\<[27]%
\>[27]{}(\Varid{a}_{2},\anonymous )\leftarrow \Varid{lift}\;(\Varid{return}\;(\Varid{rl}_{2}\;\Varid{s}_{2},\Varid{s}_{2}));{}\<[E]%
\\
\>[27]{}\Varid{return}\;(\Conid{Right}\;\Varid{a}_{2})\mskip1.5mu\}\mskip1.5mu\}{}\<[E]%
\\
\>[B]{}\mathrel{=}{}\<[BE]%
\>[5]{}\mbox{\commentbegin  \ensuremath{\Varid{lift}} is a monad morphism  \commentend}{}\<[E]%
\\
\>[B]{}\hsindent{3}{}\<[3]%
\>[3]{}\mathbf{do}\;\{\mskip1.5mu {}\<[9]%
\>[9]{}(\Varid{b},\Varid{s}_{1},\Varid{s}_{2})\leftarrow \Varid{get};{}\<[E]%
\\
\>[9]{}\mathbf{if}\;\Varid{b}\;{}\<[15]%
\>[15]{}\mathbf{then}\;{}\<[21]%
\>[21]{}\mathbf{do}\;\{\mskip1.5mu {}\<[27]%
\>[27]{}(\Varid{a}_{1},\anonymous )\leftarrow \Varid{return}\;(\Varid{rl}_{1}\;\Varid{s}_{1},\Varid{s}_{1});{}\<[E]%
\\
\>[27]{}\Varid{return}\;(\Conid{Left}\;\Varid{a}_{1})\mskip1.5mu\}{}\<[E]%
\\
\>[15]{}\mathbf{else}\;{}\<[21]%
\>[21]{}\mathbf{do}\;\{\mskip1.5mu {}\<[27]%
\>[27]{}(\Varid{a}_{2},\anonymous )\leftarrow \Varid{return}\;(\Varid{rl}_{2}\;\Varid{s}_{2},\Varid{s}_{2});{}\<[E]%
\\
\>[27]{}\Varid{return}\;(\Conid{Right}\;\Varid{a}_{2})\mskip1.5mu\}\mskip1.5mu\}{}\<[E]%
\\
\>[B]{}\mathrel{=}{}\<[BE]%
\>[5]{}\mbox{\commentbegin  monads  \commentend}{}\<[E]%
\\
\>[B]{}\hsindent{3}{}\<[3]%
\>[3]{}\mathbf{do}\;\{\mskip1.5mu {}\<[9]%
\>[9]{}(\Varid{b},\Varid{s}_{1},\Varid{s}_{2})\leftarrow \Varid{get};{}\<[E]%
\\
\>[9]{}\mathbf{if}\;\Varid{b}\;{}\<[15]%
\>[15]{}\mathbf{then}\;{}\<[21]%
\>[21]{}\mathbf{do}\;\{\mskip1.5mu {}\<[27]%
\>[27]{}\mathbf{let}\;\Varid{a}_{1}\mathrel{=}\Varid{rl}_{1}\;\Varid{s}_{1};\Varid{return}\;(\Conid{Left}\;\Varid{a}_{1})\mskip1.5mu\}{}\<[E]%
\\
\>[15]{}\mathbf{else}\;{}\<[21]%
\>[21]{}\mathbf{do}\;\{\mskip1.5mu {}\<[27]%
\>[27]{}\mathbf{let}\;\Varid{a}_{2}\mathrel{=}\Varid{rl}_{2}\;\Varid{s}_{2};\Varid{return}\;(\Conid{Right}\;\Varid{a}_{2})\mskip1.5mu\}\mskip1.5mu\}{}\<[E]%
\\
\>[B]{}\mathrel{=}{}\<[BE]%
\>[5]{}\mbox{\commentbegin  \ensuremath{\mathbf{do}} notation  \commentend}{}\<[E]%
\\
\>[B]{}\hsindent{3}{}\<[3]%
\>[3]{}\mathbf{do}\;\{\mskip1.5mu {}\<[9]%
\>[9]{}(\Varid{b},\Varid{s}_{1},\Varid{s}_{2})\leftarrow \Varid{get};{}\<[E]%
\\
\>[9]{}\mathbf{if}\;\Varid{b}\;{}\<[15]%
\>[15]{}\mathbf{then}\;{}\<[21]%
\>[21]{}\mathbf{do}\;\{\mskip1.5mu {}\<[27]%
\>[27]{}\Varid{return}\;(\Conid{Left}\;(\Varid{rl}_{1}\;\Varid{s}_{1}))\mskip1.5mu\}{}\<[E]%
\\
\>[15]{}\mathbf{else}\;{}\<[21]%
\>[21]{}\mathbf{do}\;\{\mskip1.5mu {}\<[27]%
\>[27]{}\Varid{return}\;(\Conid{Right}\;(\Varid{rl}_{2}\;\Varid{s}_{2}))\mskip1.5mu\}\mskip1.5mu\}{}\<[E]%
\\
\>[B]{}\mathrel{=}{}\<[BE]%
\>[5]{}\mbox{\commentbegin  definition of \ensuremath{\Varid{gets}}  \commentend}{}\<[E]%
\\
\>[B]{}\hsindent{3}{}\<[3]%
\>[3]{}\mathbf{do}\;\{\mskip1.5mu {}\<[9]%
\>[9]{}\Varid{gets}\;(\lambda \hslambda {}\<[18]%
\>[18]{}(\Varid{b},\Varid{s}_{1},\Varid{s}_{2})\hsarrow{\rightarrow }{\mathpunct{.}}\mathbf{if}\;\Varid{b}\;{}\<[37]%
\>[37]{}\mathbf{then}\;{}\<[43]%
\>[43]{}\Conid{Left}\;(\Varid{rl}_{1}\;\Varid{s}_{1}){}\<[E]%
\\
\>[37]{}\mathbf{else}\;{}\<[43]%
\>[43]{}\Conid{Right}\;(\Varid{rl}_{2}\;\Varid{s}_{2}))\mskip1.5mu\}{}\<[E]%
\ColumnHook
\end{hscode}\resethooks
Similarly for \ensuremath{bx\mathord{.}\get{R}}. 

Now suppose also that \ensuremath{bx_{1}} and \ensuremath{bx_{2}} are well-behaved; we show that \ensuremath{bx} is well-behaved too. 
Because \ensuremath{bx} has \ensuremath{\Conid{T}}-pure queries, it satisfies \ensuremath{\mathrm{(G_LG_L)}}, \ensuremath{\mathrm{(G_RG_R)}}, and \ensuremath{\mathrm{(G_LG_R)}}.
For \ensuremath{\mathrm{(G_LS_L)}} we have:
\begin{hscode}\SaveRestoreHook
\column{B}{@{}>{\hspre}c<{\hspost}@{}}%
\column{BE}{@{}l@{}}%
\column{3}{@{}>{\hspre}l<{\hspost}@{}}%
\column{5}{@{}>{\hspre}l<{\hspost}@{}}%
\column{9}{@{}>{\hspre}l<{\hspost}@{}}%
\column{15}{@{}>{\hspre}l<{\hspost}@{}}%
\column{21}{@{}>{\hspre}l<{\hspost}@{}}%
\column{27}{@{}>{\hspre}l<{\hspost}@{}}%
\column{E}{@{}>{\hspre}l<{\hspost}@{}}%
\>[3]{}\mathbf{do}\;\{\mskip1.5mu \Varid{a}\leftarrow bx\mathord{.}\get{L};bx\mathord{.}\set{L}\;\Varid{a}\mskip1.5mu\}{}\<[E]%
\\
\>[B]{}\mathrel{=}{}\<[BE]%
\>[5]{}\mbox{\commentbegin  definition of \ensuremath{bx\mathord{.}\get{L}}  \commentend}{}\<[E]%
\\
\>[B]{}\hsindent{3}{}\<[3]%
\>[3]{}\mathbf{do}\;\{\mskip1.5mu {}\<[9]%
\>[9]{}(\Varid{b},\Varid{s}_{1},\Varid{s}_{2})\leftarrow \Varid{get};{}\<[E]%
\\
\>[9]{}\mathbf{if}\;\Varid{b}\;{}\<[15]%
\>[15]{}\mathbf{then}\;{}\<[21]%
\>[21]{}\mathbf{do}\;\{\mskip1.5mu {}\<[27]%
\>[27]{}(\Varid{a}_{1},\anonymous )\leftarrow \Varid{lift}\;(bx_{1}\mathord{.}\get{L}\;\Varid{s}_{1});{}\<[E]%
\\
\>[27]{}\mathbf{let}\;\Varid{a}\mathrel{=}\Conid{Left}\;\Varid{a}_{1};bx\mathord{.}\set{L}\;\Varid{a}\mskip1.5mu\}{}\<[E]%
\\
\>[15]{}\mathbf{else}\;{}\<[21]%
\>[21]{}\mathbf{do}\;\{\mskip1.5mu {}\<[27]%
\>[27]{}(\Varid{a}_{2},\anonymous )\leftarrow \Varid{lift}\;(bx_{2}\mathord{.}\get{L}\;\Varid{s}_{2});{}\<[E]%
\\
\>[27]{}\mathbf{let}\;\Varid{a}\mathrel{=}\Conid{Right}\;\Varid{a}_{2};bx\mathord{.}\set{L}\;\Varid{a}\mskip1.5mu\}{}\<[E]%
\\
\>[B]{}\mathrel{=}{}\<[BE]%
\>[5]{}\mbox{\commentbegin  assume \ensuremath{\Varid{b}} is \ensuremath{\Conid{True}} (the \ensuremath{\Conid{False}} case is symmetric)  \commentend}{}\<[E]%
\\
\>[B]{}\hsindent{3}{}\<[3]%
\>[3]{}\mathbf{do}\;\{\mskip1.5mu {}\<[9]%
\>[9]{}(\Varid{b},\Varid{s}_{1},\Varid{s}_{2})\leftarrow \Varid{get};{}\<[E]%
\\
\>[9]{}(\Varid{a}_{1},\anonymous )\leftarrow \Varid{lift}\;(bx_{1}\mathord{.}\get{L}\;\Varid{s}_{1});{}\<[E]%
\\
\>[9]{}\mathbf{let}\;\Varid{a}\mathrel{=}\Conid{Left}\;\Varid{a}_{1};bx\mathord{.}\set{L}\;\Varid{a}\mskip1.5mu\}{}\<[E]%
\\
\>[B]{}\mathrel{=}{}\<[BE]%
\>[5]{}\mbox{\commentbegin  definition of \ensuremath{bx\mathord{.}\set{L}}  \commentend}{}\<[E]%
\\
\>[B]{}\hsindent{3}{}\<[3]%
\>[3]{}\mathbf{do}\;\{\mskip1.5mu {}\<[9]%
\>[9]{}(\Varid{b},\Varid{s}_{1},\Varid{s}_{2})\leftarrow \Varid{get};{}\<[E]%
\\
\>[9]{}(\Varid{a}_{1},\anonymous )\leftarrow \Varid{lift}\;(bx_{1}\mathord{.}\get{L}\;\Varid{s}_{1});{}\<[E]%
\\
\>[9]{}(\Varid{b},\Varid{s}_{1},\Varid{s}_{2})\leftarrow \Varid{get};{}\<[E]%
\\
\>[9]{}((),\Varid{s}_{1}')\leftarrow \Varid{lift}\;((bx_{1}\mathord{.}\set{L}\;\Varid{a}_{1})\;\Varid{s}_{1});{}\<[E]%
\\
\>[9]{}\Varid{set}\;(\Conid{True},\Varid{s}_{1}',\Varid{s}_{2})\mskip1.5mu\}{}\<[E]%
\\
\>[B]{}\mathrel{=}{}\<[BE]%
\>[5]{}\mbox{\commentbegin  \ensuremath{\Varid{get}} commutes with lifting  \commentend}{}\<[E]%
\\
\>[B]{}\hsindent{3}{}\<[3]%
\>[3]{}\mathbf{do}\;\{\mskip1.5mu {}\<[9]%
\>[9]{}(\Varid{b},\Varid{s}_{1},\Varid{s}_{2})\leftarrow \Varid{get};{}\<[E]%
\\
\>[9]{}(\Varid{b},\Varid{s}_{1},\Varid{s}_{2})\leftarrow \Varid{get};{}\<[E]%
\\
\>[9]{}(\Varid{a}_{1},\anonymous )\leftarrow \Varid{lift}\;(bx_{1}\mathord{.}\get{L}\;\Varid{s}_{1});{}\<[E]%
\\
\>[9]{}((),\Varid{s}_{1}')\leftarrow \Varid{lift}\;((bx_{1}\mathord{.}\set{L}\;\Varid{a}_{1})\;\Varid{s}_{1});{}\<[E]%
\\
\>[9]{}\Varid{set}\;(\Conid{True},\Varid{s}_{1}',\Varid{s}_{2})\mskip1.5mu\}{}\<[E]%
\\
\>[B]{}\mathrel{=}{}\<[BE]%
\>[5]{}\mbox{\commentbegin  \ensuremath{\mathrm{(GG)}}  \commentend}{}\<[E]%
\\
\>[B]{}\hsindent{3}{}\<[3]%
\>[3]{}\mathbf{do}\;\{\mskip1.5mu {}\<[9]%
\>[9]{}(\Varid{b},\Varid{s}_{1},\Varid{s}_{2})\leftarrow \Varid{get};{}\<[E]%
\\
\>[9]{}(\Varid{a}_{1},\anonymous )\leftarrow \Varid{lift}\;(bx_{1}\mathord{.}\get{L}\;\Varid{s}_{1});{}\<[E]%
\\
\>[9]{}((),\Varid{s}_{1}')\leftarrow \Varid{lift}\;((bx_{1}\mathord{.}\set{L}\;\Varid{a}_{1})\;\Varid{s}_{1});{}\<[E]%
\\
\>[9]{}\Varid{set}\;(\Conid{True},\Varid{s}_{1}',\Varid{s}_{2})\mskip1.5mu\}{}\<[E]%
\\
\>[B]{}\mathrel{=}{}\<[BE]%
\>[5]{}\mbox{\commentbegin  \ensuremath{\Varid{lift}} is a monad morphism  \commentend}{}\<[E]%
\\
\>[B]{}\hsindent{3}{}\<[3]%
\>[3]{}\mathbf{do}\;\{\mskip1.5mu {}\<[9]%
\>[9]{}(\Varid{b},\Varid{s}_{1},\Varid{s}_{2})\leftarrow \Varid{get};{}\<[E]%
\\
\>[9]{}((),\Varid{s}_{1}')\leftarrow \Varid{lift}\;(\mathbf{do}\;\{\mskip1.5mu \Varid{a}_{1}\leftarrow bx_{1}\mathord{.}\get{L};bx_{1}\mathord{.}\set{L}\;\Varid{a}_{1}\mskip1.5mu\}\;\Varid{s}_{1});{}\<[E]%
\\
\>[9]{}\Varid{set}\;(\Conid{True},\Varid{s}_{1}',\Varid{s}_{2})\mskip1.5mu\}{}\<[E]%
\\
\>[B]{}\mathrel{=}{}\<[BE]%
\>[5]{}\mbox{\commentbegin  \ensuremath{\mathrm{(G_LS_L)}} for \ensuremath{bx_{1}}  \commentend}{}\<[E]%
\\
\>[B]{}\hsindent{3}{}\<[3]%
\>[3]{}\mathbf{do}\;\{\mskip1.5mu {}\<[9]%
\>[9]{}(\Varid{b},\Varid{s}_{1},\Varid{s}_{2})\leftarrow \Varid{get};{}\<[E]%
\\
\>[9]{}((),\Varid{s}_{1}')\leftarrow \Varid{lift}\;(\Varid{return}\;()\;\Varid{s}_{1});{}\<[E]%
\\
\>[9]{}\Varid{set}\;(\Conid{True},\Varid{s}_{1}',\Varid{s}_{2})\mskip1.5mu\}{}\<[E]%
\\
\>[B]{}\mathrel{=}{}\<[BE]%
\>[5]{}\mbox{\commentbegin  definition of \ensuremath{\Varid{lift}}  \commentend}{}\<[E]%
\\
\>[B]{}\hsindent{3}{}\<[3]%
\>[3]{}\mathbf{do}\;\{\mskip1.5mu {}\<[9]%
\>[9]{}(\Varid{b},\Varid{s}_{1},\Varid{s}_{2})\leftarrow \Varid{get};{}\<[E]%
\\
\>[9]{}\mathbf{let}\;\Varid{s}_{1}'\mathrel{=}\Varid{s}_{1};{}\<[E]%
\\
\>[9]{}\Varid{set}\;(\Conid{True},\Varid{s}_{1}',\Varid{s}_{2})\mskip1.5mu\}{}\<[E]%
\\
\>[B]{}\mathrel{=}{}\<[BE]%
\>[5]{}\mbox{\commentbegin  substituting \ensuremath{\Varid{b}\mathrel{=}\Conid{True}} and \ensuremath{\Varid{s}_{1}'\mathrel{=}\Varid{s}_{1}}  \commentend}{}\<[E]%
\\
\>[B]{}\hsindent{3}{}\<[3]%
\>[3]{}\mathbf{do}\;\{\mskip1.5mu {}\<[9]%
\>[9]{}(\Varid{b},\Varid{s}_{1},\Varid{s}_{2})\leftarrow \Varid{get};{}\<[E]%
\\
\>[9]{}\Varid{set}\;(\Varid{b},\Varid{s}_{1},\Varid{s}_{2})\mskip1.5mu\}{}\<[E]%
\\
\>[B]{}\mathrel{=}{}\<[BE]%
\>[5]{}\mbox{\commentbegin  \ensuremath{\mathrm{(GS)}}  \commentend}{}\<[E]%
\\
\>[B]{}\hsindent{3}{}\<[3]%
\>[3]{}\mathbf{do}\;\{\mskip1.5mu {}\<[9]%
\>[9]{}\Varid{return}\;()\mskip1.5mu\}{}\<[E]%
\ColumnHook
\end{hscode}\resethooks
For \ensuremath{\mathrm{(S_LG_L)}}, and setting a \ensuremath{\Conid{Left}} value, we have:
\begin{hscode}\SaveRestoreHook
\column{B}{@{}>{\hspre}c<{\hspost}@{}}%
\column{BE}{@{}l@{}}%
\column{3}{@{}>{\hspre}l<{\hspost}@{}}%
\column{5}{@{}>{\hspre}l<{\hspost}@{}}%
\column{9}{@{}>{\hspre}l<{\hspost}@{}}%
\column{15}{@{}>{\hspre}l<{\hspost}@{}}%
\column{21}{@{}>{\hspre}l<{\hspost}@{}}%
\column{27}{@{}>{\hspre}l<{\hspost}@{}}%
\column{E}{@{}>{\hspre}l<{\hspost}@{}}%
\>[3]{}\mathbf{do}\;\{\mskip1.5mu bx\mathord{.}\set{L}\;(\Conid{Left}\;\Varid{a}_{1});bx\mathord{.}\get{L}\mskip1.5mu\}{}\<[E]%
\\
\>[B]{}\mathrel{=}{}\<[BE]%
\>[5]{}\mbox{\commentbegin  definition of \ensuremath{bx\mathord{.}\set{L}}  \commentend}{}\<[E]%
\\
\>[B]{}\hsindent{3}{}\<[3]%
\>[3]{}\mathbf{do}\;\{\mskip1.5mu {}\<[9]%
\>[9]{}(\Varid{b},\Varid{s}_{1},\Varid{s}_{2})\leftarrow \Varid{get};{}\<[E]%
\\
\>[9]{}((),\Varid{s}_{1}')\leftarrow \Varid{lift}\;((bx_{1}\mathord{.}\set{L}\;\Varid{a}_{1})\;\Varid{s}_{1});{}\<[E]%
\\
\>[9]{}\Varid{set}\;(\Conid{True},\Varid{s}_{1}',\Varid{s}_{2});{}\<[E]%
\\
\>[9]{}bx\mathord{.}\get{L}\mskip1.5mu\}{}\<[E]%
\\
\>[B]{}\mathrel{=}{}\<[BE]%
\>[5]{}\mbox{\commentbegin  definition of \ensuremath{bx\mathord{.}\get{L}}  \commentend}{}\<[E]%
\\
\>[B]{}\hsindent{3}{}\<[3]%
\>[3]{}\mathbf{do}\;\{\mskip1.5mu {}\<[9]%
\>[9]{}(\Varid{b},\Varid{s}_{1},\Varid{s}_{2})\leftarrow \Varid{get};{}\<[E]%
\\
\>[9]{}((),\Varid{s}_{1}')\leftarrow \Varid{lift}\;((bx_{1}\mathord{.}\set{L}\;\Varid{a}_{1})\;\Varid{s}_{1});{}\<[E]%
\\
\>[9]{}\Varid{set}\;(\Conid{True},\Varid{s}_{1}',\Varid{s}_{2});{}\<[E]%
\\
\>[9]{}(\Varid{b},\Varid{s}_{1}'',\Varid{s}_{2}')\leftarrow \Varid{get};{}\<[E]%
\\
\>[9]{}\mathbf{if}\;\Varid{b}\;{}\<[15]%
\>[15]{}\mathbf{then}\;{}\<[21]%
\>[21]{}\mathbf{do}\;\{\mskip1.5mu {}\<[27]%
\>[27]{}(\Varid{a}_{1}',\anonymous )\leftarrow \Varid{lift}\;(bx_{1}\mathord{.}\get{L}\;\Varid{s}_{1}'');{}\<[E]%
\\
\>[27]{}\Varid{return}\;(\Conid{Left}\;\Varid{a}_{1}')\mskip1.5mu\}{}\<[E]%
\\
\>[15]{}\mathbf{else}\;{}\<[21]%
\>[21]{}\mathbf{do}\;\{\mskip1.5mu {}\<[27]%
\>[27]{}(\Varid{a}_{2},\anonymous )\leftarrow \Varid{lift}\;(bx_{2}\mathord{.}\get{L}\;\Varid{s}_{2}');{}\<[E]%
\\
\>[27]{}\Varid{return}\;(\Conid{Right}\;\Varid{a}_{2})\mskip1.5mu\}\mskip1.5mu\}{}\<[E]%
\\
\>[B]{}\mathrel{=}{}\<[BE]%
\>[5]{}\mbox{\commentbegin  \ensuremath{\mathrm{(SG)}}  \commentend}{}\<[E]%
\\
\>[B]{}\hsindent{3}{}\<[3]%
\>[3]{}\mathbf{do}\;\{\mskip1.5mu {}\<[9]%
\>[9]{}(\Varid{b},\Varid{s}_{1},\Varid{s}_{2})\leftarrow \Varid{get};{}\<[E]%
\\
\>[9]{}((),\Varid{s}_{1}')\leftarrow \Varid{lift}\;((bx_{1}\mathord{.}\set{L}\;\Varid{a}_{1})\;\Varid{s}_{1});{}\<[E]%
\\
\>[9]{}\Varid{set}\;(\Conid{True},\Varid{s}_{1}',\Varid{s}_{2});{}\<[E]%
\\
\>[9]{}\mathbf{let}\;(\Varid{b},\Varid{s}_{1}'',\Varid{s}_{2}')\mathrel{=}(\Conid{True},\Varid{s}_{1}',\Varid{s}_{2});{}\<[E]%
\\
\>[9]{}\mathbf{if}\;\Varid{b}\;{}\<[15]%
\>[15]{}\mathbf{then}\;{}\<[21]%
\>[21]{}\mathbf{do}\;\{\mskip1.5mu {}\<[27]%
\>[27]{}(\Varid{a}_{1}',\anonymous )\leftarrow \Varid{lift}\;(bx_{1}\mathord{.}\get{L}\;\Varid{s}_{1}'');{}\<[E]%
\\
\>[27]{}\Varid{return}\;(\Conid{Left}\;\Varid{a}_{1}')\mskip1.5mu\}{}\<[E]%
\\
\>[15]{}\mathbf{else}\;{}\<[21]%
\>[21]{}\mathbf{do}\;\{\mskip1.5mu {}\<[27]%
\>[27]{}(\Varid{a}_{2},\anonymous )\leftarrow \Varid{lift}\;(bx_{2}\mathord{.}\get{L}\;\Varid{s}_{2}');{}\<[E]%
\\
\>[27]{}\Varid{return}\;(\Conid{Right}\;\Varid{a}_{2})\mskip1.5mu\}\mskip1.5mu\}{}\<[E]%
\\
\>[B]{}\mathrel{=}{}\<[BE]%
\>[5]{}\mbox{\commentbegin  substituting; conditional  \commentend}{}\<[E]%
\\
\>[B]{}\hsindent{3}{}\<[3]%
\>[3]{}\mathbf{do}\;\{\mskip1.5mu {}\<[9]%
\>[9]{}(\Varid{b},\Varid{s}_{1},\Varid{s}_{2})\leftarrow \Varid{get};{}\<[E]%
\\
\>[9]{}((),\Varid{s}_{1}')\leftarrow \Varid{lift}\;((bx_{1}\mathord{.}\set{L}\;\Varid{a}_{1})\;\Varid{s}_{1});{}\<[E]%
\\
\>[9]{}\Varid{set}\;(\Conid{True},\Varid{s}_{1}',\Varid{s}_{2});{}\<[E]%
\\
\>[9]{}(\Varid{a}_{1}',\anonymous )\leftarrow \Varid{lift}\;(bx_{1}\mathord{.}\get{L}\;\Varid{s}_{1}');{}\<[E]%
\\
\>[9]{}\Varid{return}\;(\Conid{Left}\;\Varid{a}_{1}')\mskip1.5mu\}{}\<[E]%
\\
\>[B]{}\mathrel{=}{}\<[BE]%
\>[5]{}\mbox{\commentbegin  \ensuremath{\Varid{set}} commutes with lifting; naming the wildcard;   \commentend}{}\<[E]%
\\
\>[B]{}\hsindent{3}{}\<[3]%
\>[3]{}\mathbf{do}\;\{\mskip1.5mu {}\<[9]%
\>[9]{}(\Varid{b},\Varid{s}_{1},\Varid{s}_{2})\leftarrow \Varid{get};{}\<[E]%
\\
\>[9]{}((),\Varid{s}_{1}')\leftarrow \Varid{lift}\;((bx_{1}\mathord{.}\set{L}\;\Varid{a}_{1})\;\Varid{s}_{1});{}\<[E]%
\\
\>[9]{}(\Varid{a}_{1}',\Varid{s}_{1}'')\leftarrow \Varid{lift}\;(bx_{1}\mathord{.}\get{L}\;\Varid{s}_{1}');{}\<[E]%
\\
\>[9]{}\Varid{set}\;(\Conid{True},\Varid{s}_{1}',\Varid{s}_{2});{}\<[E]%
\\
\>[9]{}\Varid{return}\;(\Conid{Left}\;\Varid{a}_{1}')\mskip1.5mu\}{}\<[E]%
\\
\>[B]{}\mathrel{=}{}\<[BE]%
\>[5]{}\mbox{\commentbegin  \ensuremath{bx_{1}\mathord{.}\get{L}} is \ensuremath{\Conid{T}}-pure, so \ensuremath{\Varid{s}_{1}''\mathrel{=}\Varid{s}_{1}'}  \commentend}{}\<[E]%
\\
\>[B]{}\hsindent{3}{}\<[3]%
\>[3]{}\mathbf{do}\;\{\mskip1.5mu {}\<[9]%
\>[9]{}(\Varid{b},\Varid{s}_{1},\Varid{s}_{2})\leftarrow \Varid{get};{}\<[E]%
\\
\>[9]{}((),\Varid{s}_{1}')\leftarrow \Varid{lift}\;((bx_{1}\mathord{.}\set{L}\;\Varid{a}_{1})\;\Varid{s}_{1});{}\<[E]%
\\
\>[9]{}(\Varid{a}_{1}',\Varid{s}_{1}'')\leftarrow \Varid{lift}\;(bx_{1}\mathord{.}\get{L}\;\Varid{s}_{1}');{}\<[E]%
\\
\>[9]{}\Varid{set}\;(\Conid{True},\Varid{s}_{1}'',\Varid{s}_{2});{}\<[E]%
\\
\>[9]{}\Varid{return}\;(\Conid{Left}\;\Varid{a}_{1}')\mskip1.5mu\}{}\<[E]%
\\
\>[B]{}\mathrel{=}{}\<[BE]%
\>[5]{}\mbox{\commentbegin  \ensuremath{\Varid{lift}} is a monad morphism  \commentend}{}\<[E]%
\\
\>[B]{}\hsindent{3}{}\<[3]%
\>[3]{}\mathbf{do}\;\{\mskip1.5mu {}\<[9]%
\>[9]{}(\Varid{b},\Varid{s}_{1},\Varid{s}_{2})\leftarrow \Varid{get};{}\<[E]%
\\
\>[9]{}(\Varid{a}_{1}',\Varid{s}_{1}'')\leftarrow \Varid{lift}\;(\mathbf{do}\;\{\mskip1.5mu bx_{1}\mathord{.}\set{L}\;\Varid{a}_{1};bx_{1}\mathord{.}\get{L}\mskip1.5mu\}\;\Varid{s}_{1});{}\<[E]%
\\
\>[9]{}\Varid{set}\;(\Conid{True},\Varid{s}_{1}'',\Varid{s}_{2});{}\<[E]%
\\
\>[9]{}\Varid{return}\;(\Conid{Left}\;\Varid{a}_{1}')\mskip1.5mu\}{}\<[E]%
\\
\>[B]{}\mathrel{=}{}\<[BE]%
\>[5]{}\mbox{\commentbegin  \ensuremath{\mathrm{(S_LG_L)}} for \ensuremath{bx_{1}}  \commentend}{}\<[E]%
\\
\>[B]{}\hsindent{3}{}\<[3]%
\>[3]{}\mathbf{do}\;\{\mskip1.5mu {}\<[9]%
\>[9]{}(\Varid{b},\Varid{s}_{1},\Varid{s}_{2})\leftarrow \Varid{get};{}\<[E]%
\\
\>[9]{}(\Varid{a}_{1}',\Varid{s}_{1}'')\leftarrow \Varid{lift}\;(\mathbf{do}\;\{\mskip1.5mu bx_{1}\mathord{.}\set{L}\;\Varid{a}_{1};\Varid{return}\;\Varid{a}_{1}\mskip1.5mu\}\;\Varid{s}_{1});{}\<[E]%
\\
\>[9]{}\Varid{set}\;(\Conid{True},\Varid{s}_{1}'',\Varid{s}_{2});{}\<[E]%
\\
\>[9]{}\Varid{return}\;(\Conid{Left}\;\Varid{a}_{1}')\mskip1.5mu\}{}\<[E]%
\\
\>[B]{}\mathrel{=}{}\<[BE]%
\>[5]{}\mbox{\commentbegin  \ensuremath{\Varid{lift}} is a monad morphism  \commentend}{}\<[E]%
\\
\>[B]{}\hsindent{3}{}\<[3]%
\>[3]{}\mathbf{do}\;\{\mskip1.5mu {}\<[9]%
\>[9]{}(\Varid{b},\Varid{s}_{1},\Varid{s}_{2})\leftarrow \Varid{get};{}\<[E]%
\\
\>[9]{}((),\Varid{s}_{1}')\leftarrow \Varid{lift}\;((bx_{1}\mathord{.}\set{L}\;\Varid{a}_{1})\;\Varid{s}_{1});{}\<[E]%
\\
\>[9]{}(\Varid{a}_{1}',\Varid{s}_{1}'')\leftarrow \Varid{lift}\;(\Varid{return}\;\Varid{a}_{1}\;\Varid{s}_{1}');{}\<[E]%
\\
\>[9]{}\Varid{set}\;(\Conid{True},\Varid{s}_{1}'',\Varid{s}_{2});{}\<[E]%
\\
\>[9]{}\Varid{return}\;(\Conid{Left}\;\Varid{a}_{1}')\mskip1.5mu\}{}\<[E]%
\\
\>[B]{}\mathrel{=}{}\<[BE]%
\>[5]{}\mbox{\commentbegin  \ensuremath{\Varid{return}} for \ensuremath{\Conid{StateT}}: \ensuremath{\Varid{s}_{1}''\mathrel{=}\Varid{s}_{1}'}  \commentend}{}\<[E]%
\\
\>[B]{}\hsindent{3}{}\<[3]%
\>[3]{}\mathbf{do}\;\{\mskip1.5mu {}\<[9]%
\>[9]{}(\Varid{b},\Varid{s}_{1},\Varid{s}_{2})\leftarrow \Varid{get};{}\<[E]%
\\
\>[9]{}((),\Varid{s}_{1}')\leftarrow \Varid{lift}\;((bx_{1}\mathord{.}\set{L}\;\Varid{a}_{1})\;\Varid{s}_{1});{}\<[E]%
\\
\>[9]{}(\Varid{a}_{1}',\anonymous )\leftarrow \Varid{lift}\;(\Varid{return}\;\Varid{a}_{1}\;\Varid{s}_{1}');{}\<[E]%
\\
\>[9]{}\Varid{set}\;(\Conid{True},\Varid{s}_{1}',\Varid{s}_{2});{}\<[E]%
\\
\>[9]{}\Varid{return}\;(\Conid{Left}\;\Varid{a}_{1}')\mskip1.5mu\}{}\<[E]%
\\
\>[B]{}\mathrel{=}{}\<[BE]%
\>[5]{}\mbox{\commentbegin  \ensuremath{\Varid{set}} commutes with lifting  \commentend}{}\<[E]%
\\
\>[B]{}\hsindent{3}{}\<[3]%
\>[3]{}\mathbf{do}\;\{\mskip1.5mu {}\<[9]%
\>[9]{}(\Varid{b},\Varid{s}_{1},\Varid{s}_{2})\leftarrow \Varid{get};{}\<[E]%
\\
\>[9]{}((),\Varid{s}_{1}')\leftarrow \Varid{lift}\;((bx_{1}\mathord{.}\set{L}\;\Varid{a}_{1})\;\Varid{s}_{1});{}\<[E]%
\\
\>[9]{}\Varid{set}\;(\Conid{True},\Varid{s}_{1}',\Varid{s}_{2});{}\<[E]%
\\
\>[9]{}(\Varid{a}_{1}',\Varid{s}_{1}'')\leftarrow \Varid{lift}\;(\Varid{return}\;\Varid{a}_{1}\;\Varid{s}_{1}');{}\<[E]%
\\
\>[9]{}\Varid{return}\;(\Conid{Left}\;\Varid{a}_{1}')\mskip1.5mu\}{}\<[E]%
\\
\>[B]{}\mathrel{=}{}\<[BE]%
\>[5]{}\mbox{\commentbegin  definitions of \ensuremath{\Varid{return}} and \ensuremath{\Varid{lift}}  \commentend}{}\<[E]%
\\
\>[B]{}\hsindent{3}{}\<[3]%
\>[3]{}\mathbf{do}\;\{\mskip1.5mu {}\<[9]%
\>[9]{}(\Varid{b},\Varid{s}_{1},\Varid{s}_{2})\leftarrow \Varid{get};{}\<[E]%
\\
\>[9]{}((),\Varid{s}_{1}')\leftarrow \Varid{lift}\;((bx_{1}\mathord{.}\set{L}\;\Varid{a}_{1})\;\Varid{s}_{1});{}\<[E]%
\\
\>[9]{}\Varid{set}\;(\Conid{True},\Varid{s}_{1}',\Varid{s}_{2});{}\<[E]%
\\
\>[9]{}\mathbf{let}\;(\Varid{a}_{1}',\Varid{s}_{1}'')\mathrel{=}(\Varid{a}_{1},\Varid{s}_{1}');{}\<[E]%
\\
\>[9]{}\Varid{return}\;(\Conid{Left}\;\Varid{a}_{1}')\mskip1.5mu\}{}\<[E]%
\\
\>[B]{}\mathrel{=}{}\<[BE]%
\>[5]{}\mbox{\commentbegin  definition of \ensuremath{bx\mathord{.}\set{L}}; substituting \ensuremath{\Varid{a}_{1}'\mathrel{=}\Varid{a}_{1}}  \commentend}{}\<[E]%
\\
\>[B]{}\hsindent{3}{}\<[3]%
\>[3]{}\mathbf{do}\;\{\mskip1.5mu {}\<[9]%
\>[9]{}bx_{1}\mathord{.}\set{L}\;(\Conid{Left}\;\Varid{a}_{1});\Varid{return}\;(\Conid{Left}\;\Varid{a}_{1})\mskip1.5mu\}{}\<[E]%
\ColumnHook
\end{hscode}\resethooks
Of course, setting a \ensuremath{\Conid{Right}} value, and \ensuremath{\mathrm{(G_RS_R)}} and \ensuremath{\mathrm{(S_RG_R)}}, are symmetric.
\end{proof}

\restatableProposition{prop:list-wb}
\begin{prop:list-wb}
If \ensuremath{bx} is transparent, then so is \ensuremath{\Varid{listIBX}\;bx}.
\end{prop:list-wb}
\begin{proof}
Suppose \ensuremath{bx} is transparent. We first show that \ensuremath{\Varid{listIBX}\;bx} has \ensuremath{\Conid{T}}-pure queries; we consider only \ensuremath{(\Varid{listIBX}\;bx)\mathord{.}\get{L}}, as \ensuremath{\get{R}} is symmetric. 
\begin{hscode}\SaveRestoreHook
\column{B}{@{}>{\hspre}c<{\hspost}@{}}%
\column{BE}{@{}l@{}}%
\column{3}{@{}>{\hspre}l<{\hspost}@{}}%
\column{5}{@{}>{\hspre}l<{\hspost}@{}}%
\column{9}{@{}>{\hspre}l<{\hspost}@{}}%
\column{E}{@{}>{\hspre}l<{\hspost}@{}}%
\>[3]{}(\Varid{listIBX}\;bx)\mathord{.}\get{L}{}\<[E]%
\\
\>[B]{}\mathrel{=}{}\<[BE]%
\>[5]{}\mbox{\commentbegin  definition of \ensuremath{\Varid{listIBX}}  \commentend}{}\<[E]%
\\
\>[B]{}\hsindent{3}{}\<[3]%
\>[3]{}\mathbf{do}\;\{\mskip1.5mu {}\<[9]%
\>[9]{}(\Varid{n},\Varid{cs})\leftarrow \Varid{get};\Varid{mapM}\;(\Varid{lift}\hsdot{\cdot }{.}\Varid{eval}\;bx\mathord{.}\get{L})\;(\Varid{take}\;\Varid{n}\;\Varid{cs})\mskip1.5mu\}{}\<[E]%
\\
\>[B]{}\mathrel{=}{}\<[BE]%
\>[5]{}\mbox{\commentbegin  \ensuremath{bx} is transparent  \commentend}{}\<[E]%
\\
\>[B]{}\hsindent{3}{}\<[3]%
\>[3]{}\mathbf{do}\;\{\mskip1.5mu {}\<[9]%
\>[9]{}(\Varid{n},\Varid{cs})\leftarrow \Varid{get};\Varid{mapM}\;(\Varid{lift}\hsdot{\cdot }{.}\Varid{eval}\;(\Varid{gets}\;(bx\mathord{.}\Varid{read}_{L})))\;(\Varid{take}\;\Varid{n}\;\Varid{cs})\mskip1.5mu\}{}\<[E]%
\\
\>[B]{}\mathrel{=}{}\<[BE]%
\>[5]{}\mbox{\commentbegin  definition of \ensuremath{\Varid{eval}}  \commentend}{}\<[E]%
\\
\>[B]{}\hsindent{3}{}\<[3]%
\>[3]{}\mathbf{do}\;\{\mskip1.5mu {}\<[9]%
\>[9]{}(\Varid{n},\Varid{cs})\leftarrow \Varid{get};\Varid{mapM}\;(\Varid{lift}\hsdot{\cdot }{.}\Varid{return}\hsdot{\cdot }{.}bx\mathord{.}\Varid{read}_{L})\;(\Varid{take}\;\Varid{n}\;\Varid{cs})\mskip1.5mu\}{}\<[E]%
\\
\>[B]{}\mathrel{=}{}\<[BE]%
\>[5]{}\mbox{\commentbegin  \ensuremath{\Varid{lift}} is a monad morphism  \commentend}{}\<[E]%
\\
\>[B]{}\hsindent{3}{}\<[3]%
\>[3]{}\mathbf{do}\;\{\mskip1.5mu {}\<[9]%
\>[9]{}(\Varid{n},\Varid{cs})\leftarrow \Varid{get};\Varid{mapM}\;(\Varid{return}\hsdot{\cdot }{.}bx\mathord{.}\Varid{read}_{L})\;(\Varid{take}\;\Varid{n}\;\Varid{cs})\mskip1.5mu\}{}\<[E]%
\\
\>[B]{}\mathrel{=}{}\<[BE]%
\>[5]{}\mbox{\commentbegin  \ensuremath{\Varid{mapM}\;(\Varid{return}\hsdot{\cdot }{.}\Varid{f})\mathrel{=}\Varid{return}\hsdot{\cdot }{.}\Varid{map}\;\Varid{f}}  \commentend}{}\<[E]%
\\
\>[B]{}\hsindent{3}{}\<[3]%
\>[3]{}\mathbf{do}\;\{\mskip1.5mu {}\<[9]%
\>[9]{}(\Varid{n},\Varid{cs})\leftarrow \Varid{get};\Varid{return}\;(\Varid{map}\;(bx\mathord{.}\Varid{read}_{L}\;(\Varid{take}\;\Varid{n}\;\Varid{cs})))\mskip1.5mu\}{}\<[E]%
\\
\>[B]{}\mathrel{=}{}\<[BE]%
\>[5]{}\mbox{\commentbegin  definition of \ensuremath{\Varid{gets}}  \commentend}{}\<[E]%
\\
\>[B]{}\hsindent{3}{}\<[3]%
\>[3]{}\Varid{gets}\;(\lambda \hslambda (\Varid{n},\Varid{cs})\hsarrow{\rightarrow }{\mathpunct{.}}\Varid{map}\;(bx\mathord{.}\Varid{read}_{L})\;(\Varid{take}\;\Varid{n}\;\Varid{cs})){}\<[E]%
\ColumnHook
\end{hscode}\resethooks
Now to show that \ensuremath{\Varid{listIBX}} preserves well-behavedness. 
Note that \ensuremath{\Varid{sets}} simplifies when its two list arguments have the same length, to:
\begin{hscode}\SaveRestoreHook
\column{B}{@{}>{\hspre}l<{\hspost}@{}}%
\column{E}{@{}>{\hspre}l<{\hspost}@{}}%
\>[B]{}\Varid{sets}\;\Varid{s}\;\Varid{i}\;\Varid{as}\;\Varid{cs}\mathrel{=}\Varid{mapM}\;(\Varid{uncurry}\;\Varid{s})\;(\Varid{zip}\;\Varid{as}\;\Varid{cs}){}\<[E]%
\ColumnHook
\end{hscode}\resethooks
Then for \ensuremath{\mathrm{(G_LS_L)}}, we have:
\begin{hscode}\SaveRestoreHook
\column{B}{@{}>{\hspre}c<{\hspost}@{}}%
\column{BE}{@{}l@{}}%
\column{3}{@{}>{\hspre}l<{\hspost}@{}}%
\column{5}{@{}>{\hspre}l<{\hspost}@{}}%
\column{9}{@{}>{\hspre}l<{\hspost}@{}}%
\column{28}{@{}>{\hspre}l<{\hspost}@{}}%
\column{E}{@{}>{\hspre}l<{\hspost}@{}}%
\>[3]{}\mathbf{do}\;\{\mskip1.5mu \Varid{as}\leftarrow (\Varid{listIBX}\;bx)\mathord{.}\get{L};(\Varid{listIBX}\;bx)\hsdot{\cdot }{.}\set{L}\;\Varid{as}\mskip1.5mu\}{}\<[E]%
\\
\>[B]{}\mathrel{=}{}\<[BE]%
\>[5]{}\mbox{\commentbegin  definition of \ensuremath{\Varid{listIBX}}  \commentend}{}\<[E]%
\\
\>[B]{}\hsindent{3}{}\<[3]%
\>[3]{}\mathbf{do}\;\{\mskip1.5mu {}\<[9]%
\>[9]{}(\Varid{n},\Varid{cs})\leftarrow \Varid{get};{}\<[E]%
\\
\>[9]{}\Varid{as}\leftarrow \Varid{mapM}\;(\Varid{lift}\hsdot{\cdot }{.}\Varid{eval}\;bx\mathord{.}\get{L})\;(\Varid{take}\;\Varid{n}\;\Varid{cs});{}\<[E]%
\\
\>[9]{}(\anonymous ,\Varid{cs})\leftarrow \Varid{get};{}\<[E]%
\\
\>[9]{}\Varid{cs'}\leftarrow \Varid{lift}\;(\Varid{sets}\;(\Varid{exec}\hsdot{\cdot }{.}bx\mathord{.}\set{L})\;bx\mathord{.}\Varid{init}_{L}\;\Varid{as}\;\Varid{cs});{}\<[E]%
\\
\>[9]{}\Varid{set}\;(\Varid{length}\;\Varid{as},\Varid{cs'})\mskip1.5mu\}{}\<[E]%
\\
\>[B]{}\mathrel{=}{}\<[BE]%
\>[5]{}\mbox{\commentbegin  \ensuremath{\Varid{get}} commutes with liftings; \ensuremath{\mathrm{(GG)}}  \commentend}{}\<[E]%
\\
\>[B]{}\hsindent{3}{}\<[3]%
\>[3]{}\mathbf{do}\;\{\mskip1.5mu {}\<[9]%
\>[9]{}(\Varid{n},\Varid{cs})\leftarrow \Varid{get};{}\<[E]%
\\
\>[9]{}\Varid{as}\leftarrow \Varid{mapM}\;(\Varid{lift}\hsdot{\cdot }{.}\Varid{eval}\;bx\mathord{.}\get{L})\;(\Varid{take}\;\Varid{n}\;\Varid{cs});{}\<[E]%
\\
\>[9]{}\Varid{cs'}\leftarrow \Varid{lift}\;(\Varid{sets}\;(\Varid{exec}\hsdot{\cdot }{.}bx\mathord{.}\set{L})\;bx\mathord{.}\Varid{init}_{L}\;\Varid{as}\;\Varid{cs});{}\<[E]%
\\
\>[9]{}\Varid{set}\;(\Varid{length}\;\Varid{as},\Varid{cs'})\mskip1.5mu\}{}\<[E]%
\\
\>[B]{}\mathrel{=}{}\<[BE]%
\>[5]{}\mbox{\commentbegin  \ensuremath{bx} is transparent, as above  \commentend}{}\<[E]%
\\
\>[B]{}\hsindent{3}{}\<[3]%
\>[3]{}\mathbf{do}\;\{\mskip1.5mu {}\<[9]%
\>[9]{}(\Varid{n},\Varid{cs})\leftarrow \Varid{get};{}\<[E]%
\\
\>[9]{}\mathbf{let}\;\Varid{as}\mathrel{=}\Varid{map}\;(bx\mathord{.}\Varid{read}_{R})\;(\Varid{take}\;\Varid{n}\;\Varid{cs});{}\<[E]%
\\
\>[9]{}\Varid{cs'}\leftarrow \Varid{lift}\;(\Varid{sets}\;(\Varid{exec}\hsdot{\cdot }{.}bx\mathord{.}\set{L})\;bx\mathord{.}\Varid{init}_{L}\;\Varid{as}\;\Varid{cs});{}\<[E]%
\\
\>[9]{}\Varid{set}\;(\Varid{length}\;\Varid{as},\Varid{cs'})\mskip1.5mu\}{}\<[E]%
\\
\>[B]{}\mathrel{=}{}\<[BE]%
\>[5]{}\mbox{\commentbegin  \ensuremath{\Varid{length}\;\Varid{as}\mathrel{=}\Varid{length}\;(\Varid{take}\;\Varid{n}\;\Varid{cs})\mathrel{=}\Varid{n}}, so \ensuremath{\Varid{sets}} simplifies  \commentend}{}\<[E]%
\\
\>[B]{}\hsindent{3}{}\<[3]%
\>[3]{}\mathbf{do}\;\{\mskip1.5mu {}\<[9]%
\>[9]{}(\Varid{n},\Varid{cs})\leftarrow \Varid{get};{}\<[E]%
\\
\>[9]{}\mathbf{let}\;\Varid{as}\mathrel{=}\Varid{map}\;(bx\mathord{.}\Varid{read}_{R})\;(\Varid{take}\;\Varid{n}\;\Varid{cs});{}\<[E]%
\\
\>[9]{}\Varid{cs'}\leftarrow \Varid{lift}\;(\Varid{mapM}\;(\Varid{uncurry}\;(\Varid{exec}\hsdot{\cdot }{.}bx\mathord{.}\set{L}))\;(\Varid{zip}\;\Varid{as}\;\Varid{cs}));{}\<[E]%
\\
\>[9]{}\Varid{set}\;(\Varid{n},\Varid{cs'})\mskip1.5mu\}{}\<[E]%
\\
\>[B]{}\mathrel{=}{}\<[BE]%
\>[5]{}\mbox{\commentbegin  \ensuremath{\Varid{zip}\;(\Varid{map}\;\Varid{f}\;\Varid{xs})\;\Varid{xs}\mathrel{=}\Varid{map}\;(\lambda \hslambda \Varid{x}\hsarrow{\rightarrow }{\mathpunct{.}}(\Varid{f}\;\Varid{x},\Varid{x}))\;\Varid{xs}}  \commentend}{}\<[E]%
\\
\>[B]{}\hsindent{3}{}\<[3]%
\>[3]{}\mathbf{do}\;\{\mskip1.5mu {}\<[9]%
\>[9]{}(\Varid{n},\Varid{cs})\leftarrow \Varid{get};{}\<[E]%
\\
\>[9]{}\Varid{cs'}\leftarrow \Varid{lift}\;(\Varid{mapM}\;{}\<[28]%
\>[28]{}(\Varid{uncurry}\;(\Varid{exec}\hsdot{\cdot }{.}bx\mathord{.}\set{L}))\;{}\<[E]%
\\
\>[28]{}(\Varid{map}\;(\lambda \hslambda \Varid{c}\hsarrow{\rightarrow }{\mathpunct{.}}(bx\mathord{.}\Varid{read}_{R}\;\Varid{c},\Varid{c}))\;\Varid{cs}));{}\<[E]%
\\
\>[9]{}\Varid{set}\;(\Varid{n},\Varid{cs'})\mskip1.5mu\}{}\<[E]%
\\
\>[B]{}\mathrel{=}{}\<[BE]%
\>[5]{}\mbox{\commentbegin   \ensuremath{\Varid{exec}\;(bx\mathord{.}\set{L}\;(bx\mathord{.}\Varid{read}_{R}\;\Varid{c}))\;\Varid{c}\mathrel{=}\Varid{return}\;\Varid{c}}, 
        by \ensuremath{\mathrm{(G_LS_L)}}   \commentend}{}\<[E]%
\\[\blanklineskip]%
\>[B]{}\hsindent{3}{}\<[3]%
\>[3]{}\mathbf{do}\;\{\mskip1.5mu {}\<[9]%
\>[9]{}(\Varid{n},\Varid{cs})\leftarrow \Varid{get};{}\<[E]%
\\
\>[9]{}\mathbf{let}\;\Varid{cs'}\mathrel{=}\Varid{cs};{}\<[E]%
\\
\>[9]{}\Varid{set}\;(\Varid{n},\Varid{cs'})\mskip1.5mu\}{}\<[E]%
\\
\>[B]{}\mathrel{=}{}\<[BE]%
\>[5]{}\mbox{\commentbegin  substituting \ensuremath{\Varid{cs'}\mathrel{=}\Varid{cs}}; \ensuremath{\mathrm{(GS)}}  \commentend}{}\<[E]%
\\
\>[B]{}\hsindent{3}{}\<[3]%
\>[3]{}\Varid{return}\;(){}\<[E]%
\ColumnHook
\end{hscode}\resethooks
And for \ensuremath{\mathrm{(S_LG_L)}}, we note first that
\begin{hscode}\SaveRestoreHook
\column{B}{@{}>{\hspre}c<{\hspost}@{}}%
\column{BE}{@{}l@{}}%
\column{3}{@{}>{\hspre}l<{\hspost}@{}}%
\column{5}{@{}>{\hspre}l<{\hspost}@{}}%
\column{9}{@{}>{\hspre}l<{\hspost}@{}}%
\column{E}{@{}>{\hspre}l<{\hspost}@{}}%
\>[3]{}\mathbf{do}\;\{\mskip1.5mu {}\<[9]%
\>[9]{}\Varid{c''}\leftarrow \Varid{lift}\;(\Varid{uncurry}\;(\Varid{exec}\hsdot{\cdot }{.}bx\mathord{.}\set{L})\;(\Varid{a},\Varid{c}));{}\<[E]%
\\
\>[9]{}\mathbf{let}\;\Varid{a'}\mathrel{=}bx\mathord{.}\Varid{read}_{L}\;\Varid{c''};\Varid{return}\;(\Varid{a'},\Varid{c''})\mskip1.5mu\}{}\<[E]%
\\
\>[B]{}\mathrel{=}{}\<[BE]%
\>[5]{}\mbox{\commentbegin  \ensuremath{\Varid{uncurry}}, \ensuremath{\Varid{exec}}; \ensuremath{bx} is transparent  \commentend}{}\<[E]%
\\
\>[B]{}\hsindent{3}{}\<[3]%
\>[3]{}\mathbf{do}\;\{\mskip1.5mu {}\<[9]%
\>[9]{}((),\Varid{c'})\leftarrow \Varid{lift}\;((bx\mathord{.}\set{L}\;\Varid{a})\;\Varid{c});{}\<[E]%
\\
\>[9]{}(\Varid{a'},\Varid{c''})\leftarrow \Varid{lift}\;(bx\mathord{.}\get{L}\;\Varid{c'});\Varid{return}\;(\Varid{a'},\Varid{c''})\mskip1.5mu\}{}\<[E]%
\\
\>[B]{}\mathrel{=}{}\<[BE]%
\>[5]{}\mbox{\commentbegin  \ensuremath{\Varid{lift}} is a monad morphism  \commentend}{}\<[E]%
\\
\>[B]{}\hsindent{3}{}\<[3]%
\>[3]{}\mathbf{do}\;\{\mskip1.5mu {}\<[9]%
\>[9]{}(\Varid{a'},\Varid{c''})\leftarrow \Varid{lift}\;(\mathbf{do}\;\{\mskip1.5mu bx\mathord{.}\set{L}\;\Varid{a};bx\mathord{.}\get{L}\mskip1.5mu\}\;\Varid{c});{}\<[E]%
\\
\>[9]{}\Varid{return}\;(\Varid{a'},\Varid{c''})\mskip1.5mu\}{}\<[E]%
\\
\>[B]{}\mathrel{=}{}\<[BE]%
\>[5]{}\mbox{\commentbegin  \ensuremath{\mathrm{(S_LG_L)}} for \ensuremath{bx}  \commentend}{}\<[E]%
\\
\>[B]{}\hsindent{3}{}\<[3]%
\>[3]{}\mathbf{do}\;\{\mskip1.5mu {}\<[9]%
\>[9]{}(\Varid{a'},\Varid{c''})\leftarrow \Varid{lift}\;(\mathbf{do}\;\{\mskip1.5mu bx\mathord{.}\set{L}\;\Varid{a};\Varid{return}\;\Varid{a}\mskip1.5mu\}\;\Varid{c});{}\<[E]%
\\
\>[9]{}\Varid{return}\;(\Varid{a'},\Varid{c''})\mskip1.5mu\}{}\<[E]%
\\
\>[B]{}\mathrel{=}{}\<[BE]%
\>[5]{}\mbox{\commentbegin  monads: \ensuremath{\Varid{a'}} will be bound to \ensuremath{\Varid{a}}  \commentend}{}\<[E]%
\\
\>[B]{}\hsindent{3}{}\<[3]%
\>[3]{}\mathbf{do}\;\{\mskip1.5mu {}\<[9]%
\>[9]{}(\anonymous ,\Varid{c''})\leftarrow \Varid{lift}\;(\mathbf{do}\;\{\mskip1.5mu bx\mathord{.}\set{L}\;\Varid{a};\Varid{return}\;\Varid{a}\mskip1.5mu\}\;\Varid{c});{}\<[E]%
\\
\>[9]{}\mathbf{let}\;\Varid{a'}\mathrel{=}\Varid{a};\Varid{return}\;(\Varid{a'},\Varid{c''})\mskip1.5mu\}{}\<[E]%
\\
\>[B]{}\mathrel{=}{}\<[BE]%
\>[5]{}\mbox{\commentbegin  reversing the above steps  \commentend}{}\<[E]%
\\
\>[B]{}\hsindent{3}{}\<[3]%
\>[3]{}\mathbf{do}\;\{\mskip1.5mu {}\<[9]%
\>[9]{}\Varid{c''}\leftarrow \Varid{lift}\;(\Varid{uncurry}\;(\Varid{exec}\hsdot{\cdot }{.}bx\mathord{.}\set{L})\;(\Varid{a},\Varid{c}));{}\<[E]%
\\
\>[9]{}\mathbf{let}\;\Varid{a'}\mathrel{=}\Varid{a};\Varid{return}\;(\Varid{a'},\Varid{c''})\mskip1.5mu\}{}\<[E]%
\ColumnHook
\end{hscode}\resethooks
and similarly
\begin{hscode}\SaveRestoreHook
\column{B}{@{}>{\hspre}c<{\hspost}@{}}%
\column{BE}{@{}l@{}}%
\column{3}{@{}>{\hspre}l<{\hspost}@{}}%
\column{5}{@{}>{\hspre}l<{\hspost}@{}}%
\column{9}{@{}>{\hspre}l<{\hspost}@{}}%
\column{E}{@{}>{\hspre}l<{\hspost}@{}}%
\>[3]{}\mathbf{do}\;\{\mskip1.5mu {}\<[9]%
\>[9]{}\Varid{c''}\leftarrow \Varid{lift}\;(bx\mathord{.}\Varid{init}_{L}\;\Varid{a});{}\<[E]%
\\
\>[9]{}\mathbf{let}\;\Varid{a'}\mathrel{=}bx\mathord{.}\Varid{read}_{L}\;\Varid{c''};\Varid{return}\;(\Varid{a'},\Varid{c''})\mskip1.5mu\}{}\<[E]%
\\
\>[B]{}\mathrel{=}{}\<[BE]%
\>[5]{}\mbox{\commentbegin  \ensuremath{bx} is transparent  \commentend}{}\<[E]%
\\
\>[B]{}\hsindent{3}{}\<[3]%
\>[3]{}\mathbf{do}\;\{\mskip1.5mu {}\<[9]%
\>[9]{}\Varid{c'}\leftarrow \Varid{lift}\;(bx\mathord{.}\Varid{init}_{L}\;\Varid{a});{}\<[E]%
\\
\>[9]{}(\Varid{a'},\Varid{c''})\leftarrow \Varid{lift}\;(bx\mathord{.}\get{L}\;\Varid{c'});\Varid{return}\;(\Varid{a'},\Varid{c''})\mskip1.5mu\}{}\<[E]%
\\
\>[B]{}\mathrel{=}{}\<[BE]%
\>[5]{}\mbox{\commentbegin  \ensuremath{\Varid{lift}} is a monad morphism  \commentend}{}\<[E]%
\\
\>[B]{}\hsindent{3}{}\<[3]%
\>[3]{}\mathbf{do}\;\{\mskip1.5mu {}\<[9]%
\>[9]{}(\Varid{a'},\Varid{c''})\leftarrow \Varid{lift}\;(\mathbf{do}\;\{\mskip1.5mu \Varid{c'}\leftarrow bx\mathord{.}\Varid{init}_{L}\;\Varid{a};bx\mathord{.}\get{L}\;\Varid{c'}\mskip1.5mu\});{}\<[E]%
\\
\>[9]{}\Varid{return}\;(\Varid{a'},\Varid{c''})\mskip1.5mu\}{}\<[E]%
\\
\>[B]{}\mathrel{=}{}\<[BE]%
\>[5]{}\mbox{\commentbegin  \ensuremath{\mathrm{(I_LG_L)}} for \ensuremath{bx}  \commentend}{}\<[E]%
\\
\>[B]{}\hsindent{3}{}\<[3]%
\>[3]{}\mathbf{do}\;\{\mskip1.5mu {}\<[9]%
\>[9]{}(\Varid{a'},\Varid{c''})\leftarrow \Varid{lift}\;(\mathbf{do}\;\{\mskip1.5mu \Varid{c'}\leftarrow bx\mathord{.}\Varid{init}_{L}\;\Varid{a};\Varid{return}\;(\Varid{a},\Varid{c'})\mskip1.5mu\});{}\<[E]%
\\
\>[9]{}\Varid{return}\;(\Varid{a'},\Varid{c''})\mskip1.5mu\}{}\<[E]%
\\
\>[B]{}\mathrel{=}{}\<[BE]%
\>[5]{}\mbox{\commentbegin  monads: \ensuremath{\Varid{a'}} will be bound to \ensuremath{\Varid{a}}  \commentend}{}\<[E]%
\\
\>[B]{}\hsindent{3}{}\<[3]%
\>[3]{}\mathbf{do}\;\{\mskip1.5mu {}\<[9]%
\>[9]{}(\anonymous ,\Varid{c''})\leftarrow \Varid{lift}\;(\mathbf{do}\;\{\mskip1.5mu \Varid{c'}\leftarrow bx\mathord{.}\Varid{init}_{L}\;\Varid{a};\Varid{return}\;(\Varid{a},\Varid{c'})\mskip1.5mu\});{}\<[E]%
\\
\>[9]{}\mathbf{let}\;\Varid{a'}\mathrel{=}\Varid{a};\Varid{return}\;(\Varid{a'},\Varid{c''})\mskip1.5mu\}{}\<[E]%
\\
\>[B]{}\mathrel{=}{}\<[BE]%
\>[5]{}\mbox{\commentbegin  reversing the above steps  \commentend}{}\<[E]%
\\
\>[B]{}\hsindent{3}{}\<[3]%
\>[3]{}\mathbf{do}\;\{\mskip1.5mu {}\<[9]%
\>[9]{}\Varid{c''}\leftarrow \Varid{lift}\;(bx\mathord{.}\Varid{init}_{L}\;\Varid{a});{}\<[E]%
\\
\>[9]{}\mathbf{let}\;\Varid{a'}\mathrel{=}\Varid{a};\Varid{return}\;(\Varid{a'},\Varid{c''})\mskip1.5mu\}{}\<[E]%
\ColumnHook
\end{hscode}\resethooks
and therefore (by induction on \ensuremath{\Varid{as}}):
\jgnote{a bit of a big step\ldots}
\begin{hscode}\SaveRestoreHook
\column{B}{@{}>{\hspre}c<{\hspost}@{}}%
\column{BE}{@{}l@{}}%
\column{3}{@{}>{\hspre}l<{\hspost}@{}}%
\column{5}{@{}>{\hspre}l<{\hspost}@{}}%
\column{9}{@{}>{\hspre}l<{\hspost}@{}}%
\column{E}{@{}>{\hspre}l<{\hspost}@{}}%
\>[3]{}\mathbf{do}\;\{\mskip1.5mu {}\<[9]%
\>[9]{}\Varid{cs'}\leftarrow \Varid{lift}\;(\Varid{sets}\;(\Varid{exec}\hsdot{\cdot }{.}bx\mathord{.}\set{L})\;bx\mathord{.}\Varid{init}_{L}\;\Varid{as}\;\Varid{cs});{}\<[E]%
\\
\>[9]{}\mathbf{let}\;\Varid{as'}\mathrel{=}\Varid{map}\;(bx\mathord{.}\Varid{read}_{L}\;(\Varid{take}\;(\Varid{length}\;\Varid{as})\;\Varid{cs'}));\Varid{k}\;\Varid{as'}\;\Varid{cs'}\mskip1.5mu\}{}\<[E]%
\\
\>[B]{}\mathrel{=}{}\<[BE]%
\>[5]{}\mbox{\commentbegin  either way, \ensuremath{\Varid{as'}} gets bound to \ensuremath{\Varid{as}}  \commentend}{}\<[E]%
\\
\>[B]{}\hsindent{3}{}\<[3]%
\>[3]{}\mathbf{do}\;\{\mskip1.5mu {}\<[9]%
\>[9]{}\Varid{cs'}\leftarrow \Varid{lift}\;(\Varid{sets}\;(\Varid{exec}\hsdot{\cdot }{.}bx\mathord{.}\set{L})\;bx\mathord{.}\Varid{init}_{L}\;\Varid{as}\;\Varid{cs});{}\<[E]%
\\
\>[9]{}\Varid{k}\;\Varid{as}\;\Varid{cs'}\mskip1.5mu\}{}\<[E]%
\ColumnHook
\end{hscode}\resethooks
Then we have:
\begin{hscode}\SaveRestoreHook
\column{B}{@{}>{\hspre}c<{\hspost}@{}}%
\column{BE}{@{}l@{}}%
\column{3}{@{}>{\hspre}l<{\hspost}@{}}%
\column{5}{@{}>{\hspre}l<{\hspost}@{}}%
\column{9}{@{}>{\hspre}l<{\hspost}@{}}%
\column{E}{@{}>{\hspre}l<{\hspost}@{}}%
\>[3]{}\mathbf{do}\;\{\mskip1.5mu (\Varid{listIBX}\;bx)\hsdot{\cdot }{.}\set{L}\;\Varid{as};(\Varid{listIBX}\;bx)\mathord{.}\get{L}\mskip1.5mu\}{}\<[E]%
\\
\>[B]{}\mathrel{=}{}\<[BE]%
\>[5]{}\mbox{\commentbegin  definition of \ensuremath{\Varid{listIBX}}  \commentend}{}\<[E]%
\\
\>[B]{}\hsindent{3}{}\<[3]%
\>[3]{}\mathbf{do}\;\{\mskip1.5mu {}\<[9]%
\>[9]{}(\anonymous ,\Varid{cs})\leftarrow \Varid{get};{}\<[E]%
\\
\>[9]{}\Varid{cs'}\leftarrow \Varid{lift}\;(\Varid{sets}\;(\Varid{exec}\hsdot{\cdot }{.}bx\mathord{.}\set{L})\;bx\mathord{.}\Varid{init}_{L}\;\Varid{as}\;\Varid{cs});{}\<[E]%
\\
\>[9]{}\Varid{set}\;(\Varid{length}\;\Varid{as},\Varid{cs'});{}\<[E]%
\\
\>[9]{}(\Varid{n},\Varid{cs})\leftarrow \Varid{get};{}\<[E]%
\\
\>[9]{}\Varid{mapM}\;(\Varid{lift}\hsdot{\cdot }{.}\Varid{eval}\;bx\mathord{.}\get{L})\;(\Varid{take}\;\Varid{n}\;\Varid{cs})\mskip1.5mu\}{}\<[E]%
\\
\>[B]{}\mathrel{=}{}\<[BE]%
\>[5]{}\mbox{\commentbegin  \ensuremath{\mathrm{(SG)}}  \commentend}{}\<[E]%
\\
\>[B]{}\hsindent{3}{}\<[3]%
\>[3]{}\mathbf{do}\;\{\mskip1.5mu {}\<[9]%
\>[9]{}(\anonymous ,\Varid{cs})\leftarrow \Varid{get};{}\<[E]%
\\
\>[9]{}\Varid{cs'}\leftarrow \Varid{lift}\;(\Varid{sets}\;(\Varid{exec}\hsdot{\cdot }{.}bx\mathord{.}\set{L})\;bx\mathord{.}\Varid{init}_{L}\;\Varid{as}\;\Varid{cs});{}\<[E]%
\\
\>[9]{}\Varid{set}\;(\Varid{length}\;\Varid{as},\Varid{cs'});{}\<[E]%
\\
\>[9]{}\Varid{mapM}\;(\Varid{lift}\hsdot{\cdot }{.}\Varid{eval}\;bx\mathord{.}\get{L})\;(\Varid{take}\;(\Varid{length}\;\Varid{as})\;\Varid{cs'})\mskip1.5mu\}{}\<[E]%
\\
\>[B]{}\mathrel{=}{}\<[BE]%
\>[5]{}\mbox{\commentbegin  \ensuremath{bx} is transparent, as above  \commentend}{}\<[E]%
\\
\>[B]{}\hsindent{3}{}\<[3]%
\>[3]{}\mathbf{do}\;\{\mskip1.5mu {}\<[9]%
\>[9]{}\Varid{cs}\leftarrow \Varid{get};{}\<[E]%
\\
\>[9]{}\Varid{cs'}\leftarrow \Varid{lift}\;(\Varid{sets}\;(\Varid{exec}\hsdot{\cdot }{.}bx\mathord{.}\set{L})\;bx\mathord{.}\Varid{init}_{L}\;\Varid{as}\;\Varid{cs});{}\<[E]%
\\
\>[9]{}\Varid{set}\;(\Varid{length}\;\Varid{as},\Varid{cs'});{}\<[E]%
\\
\>[9]{}\Varid{return}\;(\Varid{map}\;(bx\mathord{.}\Varid{read}_{L}\;(\Varid{take}\;(\Varid{length}\;\Varid{as})\;\Varid{cs'})))\mskip1.5mu\}{}\<[E]%
\\
\>[B]{}\mathrel{=}{}\<[BE]%
\>[5]{}\mbox{\commentbegin  observation above  \commentend}{}\<[E]%
\\
\>[B]{}\hsindent{3}{}\<[3]%
\>[3]{}\mathbf{do}\;\{\mskip1.5mu {}\<[9]%
\>[9]{}\Varid{cs}\leftarrow \Varid{get};{}\<[E]%
\\
\>[9]{}\Varid{cs'}\leftarrow \Varid{lift}\;(\Varid{sets}\;(\Varid{exec}\hsdot{\cdot }{.}bx\mathord{.}\set{L})\;bx\mathord{.}\Varid{init}_{L}\;\Varid{as}\;\Varid{cs});{}\<[E]%
\\
\>[9]{}\Varid{set}\;(\Varid{length}\;\Varid{as},\Varid{cs'});{}\<[E]%
\\
\>[9]{}\Varid{return}\;\Varid{as}\mskip1.5mu\}{}\<[E]%
\ColumnHook
\end{hscode}\resethooks
Finally, for \ensuremath{\mathrm{(I_LG_L)}} we have:
\begin{hscode}\SaveRestoreHook
\column{B}{@{}>{\hspre}c<{\hspost}@{}}%
\column{BE}{@{}l@{}}%
\column{3}{@{}>{\hspre}l<{\hspost}@{}}%
\column{5}{@{}>{\hspre}l<{\hspost}@{}}%
\column{9}{@{}>{\hspre}l<{\hspost}@{}}%
\column{E}{@{}>{\hspre}l<{\hspost}@{}}%
\>[3]{}\mathbf{do}\;\{\mskip1.5mu {}\<[9]%
\>[9]{}\Varid{cs}\leftarrow (\Varid{listIBX}\;bx)\mathord{.}\Varid{init}_{L}\;\Varid{as};{}\<[E]%
\\
\>[9]{}(\Varid{listIBX}\;bx)\mathord{.}\get{L}\;(\Varid{length}\;\Varid{as},\Varid{cs})\mskip1.5mu\}{}\<[E]%
\\
\>[B]{}\mathrel{=}{}\<[BE]%
\>[5]{}\mbox{\commentbegin  definition of \ensuremath{\Varid{listIBX}}; \ensuremath{bx} is transparent  \commentend}{}\<[E]%
\\
\>[B]{}\hsindent{3}{}\<[3]%
\>[3]{}\mathbf{do}\;\{\mskip1.5mu {}\<[9]%
\>[9]{}\Varid{cs}\leftarrow \Varid{mapM}\;(bx\mathord{.}\Varid{init}_{L})\;\Varid{as};{}\<[E]%
\\
\>[9]{}\lambda \hslambda (\Varid{n},\Varid{cs})\hsarrow{\rightarrow }{\mathpunct{.}}\Varid{gets}\;(\Varid{map}\;(bx\mathord{.}\Varid{read}_{L}))\;(\Varid{length}\;\Varid{as},\Varid{cs})\mskip1.5mu\}{}\<[E]%
\\
\>[B]{}\mathrel{=}{}\<[BE]%
\>[5]{}\mbox{\commentbegin  definition of \ensuremath{\Varid{gets}}  \commentend}{}\<[E]%
\\
\>[B]{}\hsindent{3}{}\<[3]%
\>[3]{}\mathbf{do}\;\{\mskip1.5mu {}\<[9]%
\>[9]{}\Varid{cs}\leftarrow \Varid{mapM}\;(bx\mathord{.}\Varid{init}_{L})\;\Varid{as};{}\<[E]%
\\
\>[9]{}\Varid{return}\;(\Varid{map}\;(bx\mathord{.}\Varid{read}_{L})\;(\Varid{take}\;(\Varid{length}\;\Varid{as})\;\Varid{cs}),\Varid{cs})\mskip1.5mu\}{}\<[E]%
\\
\>[B]{}\mathrel{=}{}\<[BE]%
\>[5]{}\mbox{\commentbegin  \ensuremath{\Varid{length}\;\Varid{cs}\mathrel{=}\Varid{length}\;\Varid{as}} by definition of \ensuremath{\Varid{mapM}}  \commentend}{}\<[E]%
\\
\>[B]{}\hsindent{3}{}\<[3]%
\>[3]{}\mathbf{do}\;\{\mskip1.5mu {}\<[9]%
\>[9]{}\Varid{cs}\leftarrow \Varid{mapM}\;(bx\mathord{.}\Varid{init}_{L})\;\Varid{as};{}\<[E]%
\\
\>[9]{}\Varid{return}\;(\Varid{map}\;(bx\mathord{.}\Varid{read}_{L})\;(\Varid{take}\;(\Varid{length}\;\Varid{cs})\;\Varid{cs}),\Varid{cs})\mskip1.5mu\}{}\<[E]%
\\
\>[B]{}\mathrel{=}{}\<[BE]%
\>[5]{}\mbox{\commentbegin  \ensuremath{\Varid{take}\;(\Varid{length}\;\Varid{cs})\mathrel{=}\Varid{cs}}  \commentend}{}\<[E]%
\\
\>[B]{}\hsindent{3}{}\<[3]%
\>[3]{}\mathbf{do}\;\{\mskip1.5mu {}\<[9]%
\>[9]{}\Varid{cs}\leftarrow \Varid{mapM}\;(bx\mathord{.}\Varid{init}_{L})\;\Varid{as};{}\<[E]%
\\
\>[9]{}\Varid{return}\;(\Varid{map}\;(bx\mathord{.}\Varid{read}_{L})\;\Varid{cs},\Varid{cs})\mskip1.5mu\}{}\<[E]%
\\
\>[B]{}\mathrel{=}{}\<[BE]%
\>[5]{}\mbox{\commentbegin  \ensuremath{\mathrm{(I_LG_L)}} for \ensuremath{bx}  \commentend}{}\<[E]%
\\
\>[B]{}\hsindent{3}{}\<[3]%
\>[3]{}\mathbf{do}\;\{\mskip1.5mu {}\<[9]%
\>[9]{}\Varid{cs}\leftarrow \Varid{mapM}\;(bx\mathord{.}\Varid{init}_{L})\;\Varid{as};\Varid{return}\;(\Varid{as},\Varid{cs})\mskip1.5mu\}{}\<[E]%
\\
\>[B]{}\mathrel{=}{}\<[BE]%
\>[5]{}\mbox{\commentbegin  definition of \ensuremath{\Varid{listIBX}} again  \commentend}{}\<[E]%
\\
\>[B]{}\hsindent{3}{}\<[3]%
\>[3]{}\mathbf{do}\;\{\mskip1.5mu \Varid{cs}\leftarrow ((\Varid{listIBX}\;bx)\mathord{.}\Varid{init}_{L}\;\Varid{as});\Varid{return}\;(\Varid{as},\Varid{cs})\mskip1.5mu\}{}\<[E]%
\ColumnHook
\end{hscode}\resethooks
The proofs on the right are of course symmetric, so omitted.
\end{proof}

\restatableProposition{prop:reader-wb}
\begin{prop:reader-wb}
\psnote{Added two arguments S for sense, as in paper, OK?!}
If \ensuremath{\Varid{f}\;\Varid{c}\mathbin{::}\Conid{StateTBX}\;(\Conid{Reader}\;\Conid{C})\;\Conid{S}\;\Conid{A}\;\Conid{B}} is transparent for any \ensuremath{\Varid{c}\mathbin{::}\Conid{C}}, then \ensuremath{\Varid{switch}\;\Varid{f}\mathbin{::}\Conid{StateTBX}\;(\Conid{Reader}\;\Conid{C})\;\Conid{S}\;\Conid{A}\;\Conid{B}} is well-behaved, but not necessarily transparent.
\end{prop:reader-wb}
\begin{proof}
The failure of transparency is illustrated by taking \ensuremath{\Varid{f}} to be any
non-constant function.  For example, take 
\begin{hscode}\SaveRestoreHook
\column{B}{@{}>{\hspre}l<{\hspost}@{}}%
\column{6}{@{}>{\hspre}l<{\hspost}@{}}%
\column{18}{@{}>{\hspre}l<{\hspost}@{}}%
\column{E}{@{}>{\hspre}l<{\hspost}@{}}%
\>[B]{}\tau{}\<[6]%
\>[6]{}\mathrel{=}\Conid{StateTBX}\;{}\<[18]%
\>[18]{}\Conid{Id}\;(\Conid{A},\Conid{A}){}\<[E]%
\\
\>[B]{}\alpha{}\<[6]%
\>[6]{}\mathrel{=}(\Conid{A},\Conid{A}){}\<[E]%
\\
\>[B]{}\beta{}\<[6]%
\>[6]{}\mathrel{=}\Conid{A}{}\<[E]%
\\
\>[B]{}\gamma{}\<[6]%
\>[6]{}\mathrel{=}\Conid{Bool}{}\<[E]%
\\
\>[B]{}\Varid{f}\;\Varid{b}{}\<[6]%
\>[6]{}\mathrel{=}\mathbf{if}\;\Varid{b}\;\mathbf{then}\;\Varid{fstBX}\;\mathbf{else}\;\Varid{sndBX}{}\<[E]%
\ColumnHook
\end{hscode}\resethooks
Then the \ensuremath{\get{L}} operation of \ensuremath{\Varid{switch}\;\Varid{f}} is of the form 
\begin{hscode}\SaveRestoreHook
\column{B}{@{}>{\hspre}l<{\hspost}@{}}%
\column{E}{@{}>{\hspre}l<{\hspost}@{}}%
\>[B]{}\mathbf{do}\;\{\mskip1.5mu \Varid{b}\leftarrow \Varid{lift}\;\Varid{ask};(\Varid{f}\;\Varid{b})\mathord{.}\get{L}\mskip1.5mu\}{}\<[E]%
\ColumnHook
\end{hscode}\resethooks
which is equivalent to 
\begin{hscode}\SaveRestoreHook
\column{B}{@{}>{\hspre}l<{\hspost}@{}}%
\column{E}{@{}>{\hspre}l<{\hspost}@{}}%
\>[B]{}\mathbf{do}\;\{\mskip1.5mu \Varid{b}\leftarrow \Varid{lift}\;\Varid{ask};\mathbf{if}\;\Varid{b}\;\mathbf{then}\;\Varid{fstBX}\mathord{.}\get{L}\;\mathbf{else}\;\Varid{sndBX}\mathord{.}\get{L}\mskip1.5mu\}{}\<[E]%
\ColumnHook
\end{hscode}\resethooks
which is clearly not \ensuremath{(\Conid{Reader}\;\Conid{Bool})}-pure.

We now consider the preservation of well-behavedness.  Clearly, the
\ensuremath{\Varid{get}} operations commute so \ensuremath{\mathrm{(G_LG_L)}}, \ensuremath{\mathrm{(G_RG_R)}} and \ensuremath{\mathrm{(G_LG_R)}} hold.  As usual, we prove the laws for the left side only; the rest are symmetric.

To show \ensuremath{\mathrm{(G_LS_L)}}:
\begin{hscode}\SaveRestoreHook
\column{B}{@{}>{\hspre}c<{\hspost}@{}}%
\column{BE}{@{}l@{}}%
\column{4}{@{}>{\hspre}l<{\hspost}@{}}%
\column{6}{@{}>{\hspre}l<{\hspost}@{}}%
\column{10}{@{}>{\hspre}l<{\hspost}@{}}%
\column{26}{@{}>{\hspre}l<{\hspost}@{}}%
\column{46}{@{}>{\hspre}l<{\hspost}@{}}%
\column{E}{@{}>{\hspre}l<{\hspost}@{}}%
\>[4]{}\mathbf{do}\;\{\mskip1.5mu {}\<[10]%
\>[10]{}\Varid{a}\leftarrow (\Varid{switch}\;\Varid{f})\mathord{.}\get{L};({}\<[46]%
\>[46]{}\Varid{switch}\;\Varid{f})\mathord{.}\set{L}\;\Varid{a}\mskip1.5mu\}{}\<[E]%
\\
\>[B]{}\mathrel{=}{}\<[BE]%
\>[6]{}\mbox{\commentbegin  Definition  \commentend}{}\<[E]%
\\
\>[B]{}\hsindent{4}{}\<[4]%
\>[4]{}\mathbf{do}\;\{\mskip1.5mu {}\<[10]%
\>[10]{}\Varid{c}\leftarrow \Varid{lift}\;\Varid{ask};\Varid{a}\leftarrow (\Varid{f}\;\Varid{c})\mathord{.}\get{L};\Varid{c'}\leftarrow \Varid{lift}\;\Varid{ask};(\Varid{f}\;\Varid{c'})\mathord{.}\set{L}\;\Varid{a}\mskip1.5mu\}{}\<[E]%
\\
\>[B]{}\mathrel{=}{}\<[BE]%
\>[6]{}\mbox{\commentbegin  \ensuremath{\Varid{lift}\;\Varid{ask}} commutes with any \ensuremath{(\Conid{Reader}\;\gamma)}-pure operation  \commentend}{}\<[E]%
\\
\>[B]{}\hsindent{4}{}\<[4]%
\>[4]{}\mathbf{do}\;\{\mskip1.5mu {}\<[10]%
\>[10]{}\Varid{c}\leftarrow \Varid{lift}\;\Varid{ask};\Varid{c'}\leftarrow \Varid{lift}\;\Varid{ask};\Varid{a}\leftarrow (\Varid{f}\;\Varid{c})\mathord{.}\get{L};(\Varid{f}\;\Varid{c'})\mathord{.}\set{L}\;\Varid{a}\mskip1.5mu\}{}\<[E]%
\\
\>[B]{}\mathrel{=}{}\<[BE]%
\>[6]{}\mbox{\commentbegin  \ensuremath{\Varid{lift}\;\Varid{ask}} is copyable  \commentend}{}\<[E]%
\\
\>[B]{}\hsindent{4}{}\<[4]%
\>[4]{}\mathbf{do}\;\{\mskip1.5mu {}\<[10]%
\>[10]{}\Varid{c}\leftarrow \Varid{lift}\;\Varid{ask};{}\<[26]%
\>[26]{}\Varid{a}\leftarrow (\Varid{f}\;\Varid{c})\mathord{.}\get{L};(\Varid{f}\;\Varid{c'})\mathord{.}\set{L}\;\Varid{a}\mskip1.5mu\}{}\<[E]%
\\
\>[B]{}\mathrel{=}{}\<[BE]%
\>[6]{}\mbox{\commentbegin  \ensuremath{\mathrm{(G_LS_L)}}  \commentend}{}\<[E]%
\\
\>[B]{}\hsindent{4}{}\<[4]%
\>[4]{}\mathbf{do}\;\{\mskip1.5mu {}\<[10]%
\>[10]{}\Varid{c}\leftarrow \Varid{lift}\;\Varid{ask};\Varid{return}\;()\mskip1.5mu\}{}\<[E]%
\\
\>[B]{}\mathrel{=}{}\<[BE]%
\>[6]{}\mbox{\commentbegin  \ensuremath{\Varid{lift}\;\Varid{ask}} is discardable  \commentend}{}\<[E]%
\\
\>[B]{}\hsindent{4}{}\<[4]%
\>[4]{}\Varid{return}\;(){}\<[E]%
\ColumnHook
\end{hscode}\resethooks

To show \ensuremath{\mathrm{(S_LG_L)}}:
\begin{hscode}\SaveRestoreHook
\column{B}{@{}>{\hspre}c<{\hspost}@{}}%
\column{BE}{@{}l@{}}%
\column{4}{@{}>{\hspre}l<{\hspost}@{}}%
\column{6}{@{}>{\hspre}l<{\hspost}@{}}%
\column{10}{@{}>{\hspre}l<{\hspost}@{}}%
\column{28}{@{}>{\hspre}l<{\hspost}@{}}%
\column{E}{@{}>{\hspre}l<{\hspost}@{}}%
\>[4]{}\mathbf{do}\;\{\mskip1.5mu {}\<[10]%
\>[10]{}(\Varid{switch}\;\Varid{f})\mathord{.}{}\<[28]%
\>[28]{}\set{L}\;\Varid{a};(\Varid{switch}\;\Varid{f})\mathord{.}\get{L}\mskip1.5mu\}{}\<[E]%
\\
\>[B]{}\mathrel{=}{}\<[BE]%
\>[6]{}\mbox{\commentbegin  Definition  \commentend}{}\<[E]%
\\
\>[B]{}\hsindent{4}{}\<[4]%
\>[4]{}\mathbf{do}\;\{\mskip1.5mu {}\<[10]%
\>[10]{}\Varid{c}\leftarrow \Varid{lift}\;\Varid{ask};(\Varid{f}\;\Varid{c})\mathord{.}\set{L}\;\Varid{a};\Varid{c'}\leftarrow \Varid{lift}\;\Varid{ask};(\Varid{f}\;\Varid{c'})\mathord{.}\get{L}\mskip1.5mu\}{}\<[E]%
\\
\>[B]{}\mathrel{=}{}\<[BE]%
\>[6]{}\mbox{\commentbegin  \ensuremath{\Varid{lift}\;\Varid{ask}} commutes with any \ensuremath{(\Conid{Reader}\;\gamma)}-pure operation  \commentend}{}\<[E]%
\\
\>[B]{}\hsindent{4}{}\<[4]%
\>[4]{}\mathbf{do}\;\{\mskip1.5mu {}\<[10]%
\>[10]{}\Varid{c}\leftarrow \Varid{lift}\;\Varid{ask};\Varid{c'}\leftarrow \Varid{lift}\;\Varid{ask};(\Varid{f}\;\Varid{c})\mathord{.}\set{L}\;\Varid{a};(\Varid{f}\;\Varid{c'})\mathord{.}\get{L}\mskip1.5mu\}{}\<[E]%
\\
\>[B]{}\mathrel{=}{}\<[BE]%
\>[6]{}\mbox{\commentbegin  \ensuremath{\Varid{lift}\;\Varid{ask}} is copyable  \commentend}{}\<[E]%
\\
\>[B]{}\hsindent{4}{}\<[4]%
\>[4]{}\mathbf{do}\;\{\mskip1.5mu {}\<[10]%
\>[10]{}\Varid{c}\leftarrow \Varid{lift}\;\Varid{ask};(\Varid{f}\;\Varid{c})\mathord{.}\set{L}\;\Varid{a};(\Varid{f}\;\Varid{c})\mathord{.}\get{L}\mskip1.5mu\}{}\<[E]%
\\
\>[B]{}\mathrel{=}{}\<[BE]%
\>[6]{}\mbox{\commentbegin  \ensuremath{\mathrm{(S_LG_L)}}  \commentend}{}\<[E]%
\\
\>[B]{}\hsindent{4}{}\<[4]%
\>[4]{}\mathbf{do}\;\{\mskip1.5mu {}\<[10]%
\>[10]{}\Varid{c}\leftarrow \Varid{lift}\;\Varid{ask};(\Varid{f}\;\Varid{c})\mathord{.}\set{L}\;\Varid{a};\Varid{return}\;\Varid{a}\mskip1.5mu\}{}\<[E]%
\\
\>[B]{}\mathrel{=}{}\<[BE]%
\>[6]{}\mbox{\commentbegin  Definition  \commentend}{}\<[E]%
\\
\>[B]{}\hsindent{4}{}\<[4]%
\>[4]{}\mathbf{do}\;\{\mskip1.5mu {}\<[10]%
\>[10]{}(\Varid{switch}\;\Varid{f})\mathord{.}\set{L}\;\Varid{a};\Varid{return}\;\Varid{a}\mskip1.5mu\}{}\<[E]%
\ColumnHook
\end{hscode}\resethooks
\endswithdisplay
\end{proof}

\restatableProposition{prop:partial-wb}
\begin{prop:partial-wb}
Suppose \ensuremath{\Varid{f}\mathbin{::}\Conid{A}\hsarrow{\rightarrow }{\mathpunct{.}}\Conid{Maybe}\;\Conid{B}} and \ensuremath{\Varid{g}\mathbin{::}\Conid{B}\hsarrow{\rightarrow }{\mathpunct{.}}\Conid{Maybe}\;\Conid{A}} are partial inverses; that is, for any \ensuremath{\Varid{a},\Varid{b}} we have \ensuremath{\Varid{f}\;\Varid{a}\mathrel{=}\Conid{Just}\;\Varid{b}} if and only if \ensuremath{\Varid{g}\;\Varid{b}\mathrel{=}\Conid{Just}\;\Varid{a}}, and that \ensuremath{\Varid{err}} is a zero element for monad \ensuremath{\Conid{T}}.   Then \ensuremath{\Varid{partialBX}\;\Varid{err}\;\Varid{f}\;\Varid{g}\mathbin{::}\Conid{StateTBX}\;\Conid{T}\;\Conid{S}\;\Conid{A}\;\Conid{B}} is well-behaved, where \ensuremath{\Conid{S}\mathrel{=}\{\mskip1.5mu (\Varid{a},\Varid{b})\mid \Varid{f}\;\Varid{a}\mathrel{=}\Conid{Just}\;\Varid{b}\mathrel{\wedge}\Varid{g}\;\Varid{b}\mathrel{=}\Conid{Just}\;\Varid{a}\mskip1.5mu\}}.
\end{prop:partial-wb}
\begin{proof}
  Suppose \ensuremath{\Varid{f},\Varid{g}} are partial inverses and \ensuremath{\Varid{err}} a zero element of \ensuremath{\Conid{T}},
  and let 
\begin{hscode}\SaveRestoreHook
\column{B}{@{}>{\hspre}l<{\hspost}@{}}%
\column{E}{@{}>{\hspre}l<{\hspost}@{}}%
\>[B]{}bx\mathrel{=}\Varid{partialBX}\;\Varid{err}\;\Varid{f}\;\Varid{g}\mathbin{::}\Conid{StateTBX}\;\Conid{T}\;\Conid{S}\;\Conid{A}\;\Conid{B}{}\<[E]%
\ColumnHook
\end{hscode}\resethooks
The laws
  \ensuremath{\mathrm{(G_LG_L)}}, \ensuremath{\mathrm{(G_RG_R)}} and \ensuremath{\mathrm{(G_LG_R)}} are immediate because the
  \ensuremath{\get{L}} and \ensuremath{\get{R}} operations are clearly \ensuremath{\Conid{T}}-pure.  It is also straightforward to verify that the operations maintain the invariant that the states \ensuremath{(\Varid{a},\Varid{b})} satisfy \ensuremath{\Varid{f}\;\Varid{a}\mathrel{=}\Conid{Just}\;\Varid{b}\mathrel{\wedge}\Varid{g}\;\Varid{b}\mathrel{=}\Conid{Just}\;\Varid{a}}, because the get operations do not change the state and the set operations either yield an error, or set the state to \ensuremath{(\Varid{a},\Varid{b})} where \ensuremath{\Varid{f}\;\Varid{a}\mathrel{=}\Conid{Just}\;\Varid{b}} (and therefore \ensuremath{\Varid{g}\;\Varid{b}\mathrel{=}\Conid{Just}\;\Varid{a}}, since \ensuremath{\Varid{f},\Varid{g}} are partial inverses).

For \ensuremath{\mathrm{(G_LS_L)}}, we proceed as follows:
\begin{hscode}\SaveRestoreHook
\column{B}{@{}>{\hspre}c<{\hspost}@{}}%
\column{BE}{@{}l@{}}%
\column{4}{@{}>{\hspre}l<{\hspost}@{}}%
\column{6}{@{}>{\hspre}l<{\hspost}@{}}%
\column{10}{@{}>{\hspre}l<{\hspost}@{}}%
\column{12}{@{}>{\hspre}l<{\hspost}@{}}%
\column{E}{@{}>{\hspre}l<{\hspost}@{}}%
\>[4]{}\mathbf{do}\;\{\mskip1.5mu {}\<[10]%
\>[10]{}\Varid{a}\leftarrow bx\mathord{.}\get{L};bx\mathord{.}\set{L}\;\Varid{a}\mskip1.5mu\}{}\<[E]%
\\
\>[B]{}\mathrel{=}{}\<[BE]%
\>[6]{}\mbox{\commentbegin  definition of \ensuremath{bx}, \ensuremath{\Varid{gets}}, \ensuremath{\Varid{fst}}, monad unit  \commentend}{}\<[E]%
\\
\>[B]{}\hsindent{4}{}\<[4]%
\>[4]{}\mathbf{do}\;\{\mskip1.5mu {}\<[10]%
\>[10]{}(\Varid{a},\Varid{b})\leftarrow \Varid{get};{}\<[E]%
\\
\>[10]{}\mathbf{case}\;\Varid{f}\;\Varid{a}\;\mathbf{of}{}\<[E]%
\\
\>[10]{}\hsindent{2}{}\<[12]%
\>[12]{}\Conid{Just}\;\Varid{b'}\hsarrow{\rightarrow }{\mathpunct{.}}\Varid{set}\;(\Varid{a},\Varid{b'}){}\<[E]%
\\
\>[10]{}\hsindent{2}{}\<[12]%
\>[12]{}\Conid{Nothing}\hsarrow{\rightarrow }{\mathpunct{.}}\Varid{lift}\;\Varid{err}\mskip1.5mu\}{}\<[E]%
\\
\>[B]{}\mathrel{=}{}\<[BE]%
\>[6]{}\mbox{\commentbegin  The state \ensuremath{(\Varid{a},\Varid{b})} is consistent, so \ensuremath{\Varid{f}\;\Varid{a}\mathrel{=}\Conid{Just}\;\Varid{b}}  \commentend}{}\<[E]%
\\
\>[B]{}\hsindent{4}{}\<[4]%
\>[4]{}\mathbf{do}\;\{\mskip1.5mu {}\<[10]%
\>[10]{}\Varid{a}\leftarrow \Varid{gets}\;\Varid{fst};{}\<[E]%
\\
\>[10]{}\mathbf{case}\;\Conid{Just}\;\Varid{b}\;\mathbf{of}{}\<[E]%
\\
\>[10]{}\hsindent{2}{}\<[12]%
\>[12]{}\Conid{Just}\;\Varid{b'}\hsarrow{\rightarrow }{\mathpunct{.}}\Varid{set}\;(\Varid{a},\Varid{b'}){}\<[E]%
\\
\>[10]{}\hsindent{2}{}\<[12]%
\>[12]{}\Conid{Nothing}\hsarrow{\rightarrow }{\mathpunct{.}}\Varid{lift}\;\Varid{err}\mskip1.5mu\}{}\<[E]%
\\
\>[B]{}\mathrel{=}{}\<[BE]%
\>[6]{}\mbox{\commentbegin  \ensuremath{\mathbf{case}} simplification  \commentend}{}\<[E]%
\\
\>[B]{}\hsindent{4}{}\<[4]%
\>[4]{}\mathbf{do}\;\{\mskip1.5mu {}\<[10]%
\>[10]{}(\Varid{a},\Varid{b})\leftarrow \Varid{get};{}\<[E]%
\\
\>[10]{}\Varid{set}\;(\Varid{a},\Varid{b})\mskip1.5mu\}{}\<[E]%
\\
\>[B]{}\mathrel{=}{}\<[BE]%
\>[6]{}\mbox{\commentbegin  \ensuremath{\mathrm{(GS)}}  \commentend}{}\<[E]%
\\
\>[B]{}\hsindent{4}{}\<[4]%
\>[4]{}\Varid{return}\;(){}\<[E]%
\ColumnHook
\end{hscode}\resethooks

For \ensuremath{\mathrm{(S_LG_L)}}, there are two cases.  If \ensuremath{\Varid{f}\;\Varid{a}\mathrel{=}\Conid{Nothing}}, we reason as follows:
\begin{hscode}\SaveRestoreHook
\column{B}{@{}>{\hspre}c<{\hspost}@{}}%
\column{BE}{@{}l@{}}%
\column{4}{@{}>{\hspre}l<{\hspost}@{}}%
\column{6}{@{}>{\hspre}l<{\hspost}@{}}%
\column{10}{@{}>{\hspre}l<{\hspost}@{}}%
\column{12}{@{}>{\hspre}l<{\hspost}@{}}%
\column{E}{@{}>{\hspre}l<{\hspost}@{}}%
\>[4]{}\mathbf{do}\;\{\mskip1.5mu bx\mathord{.}\set{L}\;\Varid{a};bx\mathord{.}\get{L}\mskip1.5mu\}{}\<[E]%
\\
\>[B]{}\mathrel{=}{}\<[BE]%
\>[6]{}\mbox{\commentbegin  Definition of \ensuremath{\get{L}}, \ensuremath{\set{L}}, \ensuremath{\Varid{gets}\;\Varid{fst}}  \commentend}{}\<[E]%
\\
\>[B]{}\hsindent{4}{}\<[4]%
\>[4]{}\mathbf{do}\;\{\mskip1.5mu {}\<[10]%
\>[10]{}\mathbf{case}\;\Varid{f}\;\Varid{a}\;\mathbf{of}{}\<[E]%
\\
\>[10]{}\hsindent{2}{}\<[12]%
\>[12]{}\Conid{Just}\;\Varid{b'}\hsarrow{\rightarrow }{\mathpunct{.}}\Varid{set}\;(\Varid{a},\Varid{b'}){}\<[E]%
\\
\>[10]{}\hsindent{2}{}\<[12]%
\>[12]{}\Conid{Nothing}\hsarrow{\rightarrow }{\mathpunct{.}}\Varid{lift}\;\Varid{err};{}\<[E]%
\\
\>[10]{}(\Varid{a'},\Varid{b'})\leftarrow \Varid{get};\Varid{return}\;\Varid{a'}\mskip1.5mu\}{}\<[E]%
\\
\>[B]{}\mathrel{=}{}\<[BE]%
\>[6]{}\mbox{\commentbegin  \ensuremath{\Varid{f}\;\Varid{a}\mathrel{=}\Conid{Nothing}}, simplify \ensuremath{\mathbf{case}}  \commentend}{}\<[E]%
\\
\>[B]{}\hsindent{4}{}\<[4]%
\>[4]{}\mathbf{do}\;\{\mskip1.5mu {}\<[10]%
\>[10]{}\Varid{lift}\;\Varid{err};{}\<[E]%
\\
\>[10]{}(\Varid{a'},\Varid{b'})\leftarrow \Varid{get};\Varid{return}\;\Varid{a'}\mskip1.5mu\}{}\<[E]%
\\
\>[B]{}\mathrel{=}{}\<[BE]%
\>[6]{}\mbox{\commentbegin  \ensuremath{\Varid{err}} is a zero element  \commentend}{}\<[E]%
\\
\>[B]{}\hsindent{4}{}\<[4]%
\>[4]{}\mathbf{do}\;\{\mskip1.5mu {}\<[10]%
\>[10]{}\Varid{lift}\;\Varid{err}\mskip1.5mu\}{}\<[E]%
\\
\>[B]{}\mathrel{=}{}\<[BE]%
\>[6]{}\mbox{\commentbegin  \ensuremath{\Varid{err}} is a zero element  \commentend}{}\<[E]%
\\
\>[B]{}\hsindent{4}{}\<[4]%
\>[4]{}\mathbf{do}\;\{\mskip1.5mu {}\<[10]%
\>[10]{}\Varid{lift}\;\Varid{err};\Varid{return}\;\Varid{a}\mskip1.5mu\}{}\<[E]%
\\
\>[B]{}\mathrel{=}{}\<[BE]%
\>[6]{}\mbox{\commentbegin  reverse previous steps  \commentend}{}\<[E]%
\\
\>[B]{}\hsindent{4}{}\<[4]%
\>[4]{}\mathbf{do}\;\{\mskip1.5mu {}\<[10]%
\>[10]{}\mathbf{case}\;\Varid{f}\;\Varid{a}\;\mathbf{of}{}\<[E]%
\\
\>[10]{}\hsindent{2}{}\<[12]%
\>[12]{}\Conid{Just}\;\Varid{b'}\hsarrow{\rightarrow }{\mathpunct{.}}\Varid{set}\;(\Varid{a},\Varid{b'}){}\<[E]%
\\
\>[10]{}\hsindent{2}{}\<[12]%
\>[12]{}\Conid{Nothing}\hsarrow{\rightarrow }{\mathpunct{.}}\Varid{lift}\;\Varid{err};{}\<[E]%
\\
\>[10]{}\Varid{return}\;\Varid{a'}\mskip1.5mu\}{}\<[E]%
\\
\>[B]{}\mathrel{=}{}\<[BE]%
\>[6]{}\mbox{\commentbegin  definition  \commentend}{}\<[E]%
\\
\>[B]{}\hsindent{4}{}\<[4]%
\>[4]{}\mathbf{do}\;\{\mskip1.5mu bx\mathord{.}\set{L}\;\Varid{a};\Varid{return}\;\Varid{a}\mskip1.5mu\}{}\<[E]%
\ColumnHook
\end{hscode}\resethooks

On the other hand, if \ensuremath{\Varid{f}\;\Varid{a}\mathrel{=}\Conid{Just}\;\Varid{b}} for some \ensuremath{\Varid{b}} then we reason as follows:
\begin{hscode}\SaveRestoreHook
\column{B}{@{}>{\hspre}c<{\hspost}@{}}%
\column{BE}{@{}l@{}}%
\column{4}{@{}>{\hspre}l<{\hspost}@{}}%
\column{6}{@{}>{\hspre}l<{\hspost}@{}}%
\column{10}{@{}>{\hspre}l<{\hspost}@{}}%
\column{12}{@{}>{\hspre}l<{\hspost}@{}}%
\column{E}{@{}>{\hspre}l<{\hspost}@{}}%
\>[4]{}\mathbf{do}\;\{\mskip1.5mu bx\mathord{.}\set{L}\;\Varid{a};bx\mathord{.}\get{L}\mskip1.5mu\}{}\<[E]%
\\
\>[B]{}\mathrel{=}{}\<[BE]%
\>[6]{}\mbox{\commentbegin  Definition of \ensuremath{\get{L}}, \ensuremath{\set{L}}, \ensuremath{\Varid{gets}\;\Varid{fst}}  \commentend}{}\<[E]%
\\
\>[B]{}\hsindent{4}{}\<[4]%
\>[4]{}\mathbf{do}\;\{\mskip1.5mu {}\<[10]%
\>[10]{}\mathbf{case}\;\Varid{f}\;\Varid{a}\;\mathbf{of}{}\<[E]%
\\
\>[10]{}\hsindent{2}{}\<[12]%
\>[12]{}\Conid{Just}\;\Varid{b'}\hsarrow{\rightarrow }{\mathpunct{.}}\Varid{set}\;(\Varid{a},\Varid{b'}){}\<[E]%
\\
\>[10]{}\hsindent{2}{}\<[12]%
\>[12]{}\Conid{Nothing}\hsarrow{\rightarrow }{\mathpunct{.}}\Varid{lift}\;\Varid{err};{}\<[E]%
\\
\>[10]{}(\Varid{a'},\Varid{b'})\leftarrow \Varid{get};\Varid{return}\;\Varid{a'}\mskip1.5mu\}{}\<[E]%
\\
\>[B]{}\mathrel{=}{}\<[BE]%
\>[6]{}\mbox{\commentbegin  \ensuremath{\Varid{f}\;\Varid{a}\mathrel{=}\Conid{Just}\;\Varid{b}}, simplify \ensuremath{\mathbf{case}}  \commentend}{}\<[E]%
\\
\>[B]{}\hsindent{4}{}\<[4]%
\>[4]{}\mathbf{do}\;\{\mskip1.5mu {}\<[10]%
\>[10]{}\Varid{set}\;(\Varid{a},\Varid{b});{}\<[E]%
\\
\>[10]{}(\Varid{a'},\Varid{b'})\leftarrow \Varid{get};\Varid{return}\;\Varid{a'}\mskip1.5mu\}{}\<[E]%
\\
\>[B]{}\mathrel{=}{}\<[BE]%
\>[6]{}\mbox{\commentbegin  \ensuremath{\mathrm{(SG)}}  \commentend}{}\<[E]%
\\
\>[B]{}\hsindent{4}{}\<[4]%
\>[4]{}\mathbf{do}\;\{\mskip1.5mu {}\<[10]%
\>[10]{}\Varid{set}\;(\Varid{a},\Varid{b});{}\<[E]%
\\
\>[10]{}\Varid{return}\;\Varid{a'}\mskip1.5mu\}{}\<[E]%
\\
\>[B]{}\mathrel{=}{}\<[BE]%
\>[6]{}\mbox{\commentbegin  reverse previous steps  \commentend}{}\<[E]%
\\
\>[B]{}\hsindent{4}{}\<[4]%
\>[4]{}\mathbf{do}\;\{\mskip1.5mu {}\<[10]%
\>[10]{}\mathbf{case}\;\Varid{f}\;\Varid{a}\;\mathbf{of}{}\<[E]%
\\
\>[10]{}\hsindent{2}{}\<[12]%
\>[12]{}\Conid{Just}\;\Varid{b'}\hsarrow{\rightarrow }{\mathpunct{.}}\Varid{set}\;(\Varid{a},\Varid{b'}){}\<[E]%
\\
\>[10]{}\hsindent{2}{}\<[12]%
\>[12]{}\Conid{Nothing}\hsarrow{\rightarrow }{\mathpunct{.}}\Varid{lift}\;\Varid{err};{}\<[E]%
\\
\>[10]{}\Varid{return}\;\Varid{a'}\mskip1.5mu\}{}\<[E]%
\\
\>[B]{}\mathrel{=}{}\<[BE]%
\>[6]{}\mbox{\commentbegin  definition  \commentend}{}\<[E]%
\\
\>[B]{}\hsindent{4}{}\<[4]%
\>[4]{}\mathbf{do}\;\{\mskip1.5mu bx\mathord{.}\set{L}\;\Varid{a};\Varid{return}\;\Varid{a}\mskip1.5mu\}{}\<[E]%
\ColumnHook
\end{hscode}\resethooks
\endswithdisplay
\end{proof}

\restatableProposition{prop:nondet-wb}
\begin{prop:nondet-wb}
Assume that \ensuremath{\Varid{ok}}, \ensuremath{\Varid{as}} and \ensuremath{\Varid{bs}} satisfy the following equations:
\begin{hscode}\SaveRestoreHook
\column{B}{@{}>{\hspre}l<{\hspost}@{}}%
\column{E}{@{}>{\hspre}l<{\hspost}@{}}%
\>[B]{}\Varid{a}\in\Varid{as}\;\Varid{b}\Rightarrow \Varid{ok}\;\Varid{a}\;\Varid{b}{}\<[E]%
\\
\>[B]{}\Varid{b}\in\Varid{bs}\;\Varid{a}\Rightarrow \Varid{ok}\;\Varid{a}\;\Varid{b}{}\<[E]%
\ColumnHook
\end{hscode}\resethooks
Then \ensuremath{\Varid{nondetBX}\;\Varid{ok}\;\Varid{bs}\;\Varid{as}} is well-behaved.
\end{prop:nondet-wb}
\begin{proof}
  For well-definedness on the state space \ensuremath{\{\mskip1.5mu (\Varid{a},\Varid{b})\mid \Varid{ok}\;\Varid{a}\;\Varid{b}\mskip1.5mu\}}, we reason as follows.  Suppose \ensuremath{\Varid{ok}\;\Varid{a}\;\Varid{b}} holds.  Then clearly, after doing a \ensuremath{\Varid{get}} the state is unchanged and this continues to hold.  After a \ensuremath{\Varid{set}}, if the new value of \ensuremath{\Varid{a'}} satisfies \ensuremath{\Varid{ok}\;\Varid{a'}\;\Varid{b}} then the updated state will be \ensuremath{(\Varid{a'},\Varid{b})}, so the invariant is maintained.  Otherwise, the updated state will be \ensuremath{(\Varid{a'},\Varid{b'})} where \ensuremath{\Varid{b'}\in\Varid{bs}\;\Varid{a'}}, so by assumption \ensuremath{\Varid{ok}\;\Varid{a'}\;\Varid{b'}}  holds.

For \ensuremath{\mathrm{(G_LS_L)}} we reason as follows:
\begin{hscode}\SaveRestoreHook
\column{B}{@{}>{\hspre}c<{\hspost}@{}}%
\column{BE}{@{}l@{}}%
\column{5}{@{}>{\hspre}l<{\hspost}@{}}%
\column{7}{@{}>{\hspre}l<{\hspost}@{}}%
\column{11}{@{}>{\hspre}l<{\hspost}@{}}%
\column{E}{@{}>{\hspre}l<{\hspost}@{}}%
\>[5]{}\mathbf{do}\;\{\mskip1.5mu {}\<[11]%
\>[11]{}\Varid{a}\leftarrow (\Varid{nondetBX}\;\Varid{as}\;\Varid{bs})\mathord{.}\get{L};(\Varid{nondetBX}\;\Varid{as}\;\Varid{bs})\mathord{.}\set{L}\;\Varid{a}\mskip1.5mu\}{}\<[E]%
\\
\>[B]{}\mathrel{=}{}\<[BE]%
\>[7]{}\mbox{\commentbegin  definition  \commentend}{}\<[E]%
\\
\>[B]{}\hsindent{5}{}\<[5]%
\>[5]{}\mathbf{do}\;\{\mskip1.5mu {}\<[11]%
\>[11]{}(\Varid{a},\Varid{b})\leftarrow \Varid{get};\mathbf{let}\;\Varid{a'}\mathrel{=}\Varid{a};(\Varid{a''},\Varid{b''})\leftarrow \Varid{get};{}\<[E]%
\\
\>[11]{}\mathbf{if}\;\Varid{ok}\;\Varid{a'}\;\Varid{b''}\;\mathbf{then}\;\Varid{set}\;(\Varid{a'},\Varid{b''}){}\<[E]%
\\
\>[11]{}\mathbf{else}\;\mathbf{do}\;\{\mskip1.5mu \Varid{b'}\leftarrow \Varid{lift}\;(\Varid{bs}\;\Varid{a'});\Varid{set}\;(\Varid{a'},\Varid{b'})\mskip1.5mu\}\mskip1.5mu\}{}\<[E]%
\\
\>[B]{}\mathrel{=}{}\<[BE]%
\>[7]{}\mbox{\commentbegin  inline \ensuremath{\mathbf{let}}  \commentend}{}\<[E]%
\\
\>[B]{}\hsindent{5}{}\<[5]%
\>[5]{}\mathbf{do}\;\{\mskip1.5mu {}\<[11]%
\>[11]{}(\Varid{a},\Varid{b})\leftarrow \Varid{get};(\Varid{a''},\Varid{b''})\leftarrow \Varid{get};{}\<[E]%
\\
\>[11]{}\mathbf{if}\;\Varid{ok}\;\Varid{a}\;\Varid{b''}\;\mathbf{then}\;\Varid{set}\;(\Varid{a'},\Varid{b''}){}\<[E]%
\\
\>[11]{}\mathbf{else}\;\mathbf{do}\;\{\mskip1.5mu \Varid{b'}\leftarrow \Varid{lift}\;(\Varid{bs}\;\Varid{a});\Varid{set}\;(\Varid{a},\Varid{b'})\mskip1.5mu\}\mskip1.5mu\}{}\<[E]%
\\
\>[B]{}\mathrel{=}{}\<[BE]%
\>[7]{}\mbox{\commentbegin  \ensuremath{\mathrm{(GG)}}  \commentend}{}\<[E]%
\\
\>[B]{}\hsindent{5}{}\<[5]%
\>[5]{}\mathbf{do}\;\{\mskip1.5mu {}\<[11]%
\>[11]{}(\Varid{a},\Varid{b})\leftarrow \Varid{get};{}\<[E]%
\\
\>[11]{}\mathbf{if}\;\Varid{ok}\;\Varid{a}\;\Varid{b}\;\mathbf{then}\;\Varid{set}\;(\Varid{a},\Varid{b}){}\<[E]%
\\
\>[11]{}\mathbf{else}\;\mathbf{do}\;\{\mskip1.5mu \Varid{b'}\leftarrow \Varid{lift}\;(\Varid{bs}\;\Varid{a});\Varid{set}\;(\Varid{a},\Varid{b'})\mskip1.5mu\}\mskip1.5mu\}{}\<[E]%
\\
\>[B]{}\mathrel{=}{}\<[BE]%
\>[7]{}\mbox{\commentbegin  \ensuremath{\Varid{ok}\;\Varid{a}\;\Varid{b}\equals\Conid{True}}  \commentend}{}\<[E]%
\\
\>[B]{}\hsindent{5}{}\<[5]%
\>[5]{}\mathbf{do}\;\{\mskip1.5mu {}\<[11]%
\>[11]{}(\Varid{a},\Varid{b})\leftarrow \Varid{get};\Varid{set}\;(\Varid{a},\Varid{b})\mskip1.5mu\}{}\<[E]%
\\
\>[B]{}\mathrel{=}{}\<[BE]%
\>[7]{}\mbox{\commentbegin  \ensuremath{\mathrm{(GS)}}  \commentend}{}\<[E]%
\\
\>[B]{}\hsindent{5}{}\<[5]%
\>[5]{}\Varid{return}\;(){}\<[E]%
\ColumnHook
\end{hscode}\resethooks
Note that we rely on the invariant (not explicit in the type) that the state values \ensuremath{(\Varid{a},\Varid{b})} satisfy \ensuremath{\Varid{ok}\;\Varid{a}\;\Varid{b}\mathrel{=}\Conid{True}}.

For \ensuremath{\mathrm{(S_LG_L)}} the reasoning is as follows:
\begin{hscode}\SaveRestoreHook
\column{B}{@{}>{\hspre}c<{\hspost}@{}}%
\column{BE}{@{}l@{}}%
\column{4}{@{}>{\hspre}l<{\hspost}@{}}%
\column{6}{@{}>{\hspre}l<{\hspost}@{}}%
\column{9}{@{}>{\hspre}l<{\hspost}@{}}%
\column{E}{@{}>{\hspre}l<{\hspost}@{}}%
\>[4]{}\mathbf{do}\;\{\mskip1.5mu {}\<[9]%
\>[9]{}bx\mathord{.}\set{L}\;\Varid{a};bx\mathord{.}\get{L}\mskip1.5mu\}{}\<[E]%
\\
\>[B]{}\mathrel{=}{}\<[BE]%
\>[6]{}\mbox{\commentbegin  Definition  \commentend}{}\<[E]%
\\
\>[B]{}\hsindent{4}{}\<[4]%
\>[4]{}\mathbf{do}\;\{\mskip1.5mu {}\<[9]%
\>[9]{}(\Varid{a},\Varid{b})\leftarrow \Varid{get}{}\<[E]%
\\
\>[9]{}\mathbf{if}\;\Varid{ok}\;\Varid{a'}\;\Varid{b}\;\mathbf{then}\;\Varid{set}\;(\Varid{a'},\Varid{b}){}\<[E]%
\\
\>[9]{}\mathbf{else}\;\mathbf{do}\;\{\mskip1.5mu \Varid{b'}\leftarrow \Varid{lift}\;(\Varid{bs}\;\Varid{a'});\Varid{set}\;(\Varid{a'},\Varid{b'})\mskip1.5mu\};{}\<[E]%
\\
\>[9]{}(\Varid{a''},\Varid{b''})\leftarrow \Varid{get};\Varid{return}\;\Varid{a''}\mskip1.5mu\}{}\<[E]%
\\
\>[B]{}\mathrel{=}{}\<[BE]%
\>[6]{}\mbox{\commentbegin  \ensuremath{\mathrm{(SG)}}  \commentend}{}\<[E]%
\\
\>[B]{}\hsindent{4}{}\<[4]%
\>[4]{}\mathbf{do}\;\{\mskip1.5mu {}\<[9]%
\>[9]{}(\Varid{a},\Varid{b})\leftarrow \Varid{get}{}\<[E]%
\\
\>[9]{}\mathbf{if}\;\Varid{ok}\;\Varid{a'}\;\Varid{b}\;\mathbf{then}\;\Varid{set}\;(\Varid{a'},\Varid{b}){}\<[E]%
\\
\>[9]{}\mathbf{else}\;\mathbf{do}\;\{\mskip1.5mu \Varid{b'}\leftarrow \Varid{lift}\;(\Varid{bs}\;\Varid{a'});\Varid{set}\;(\Varid{a'},\Varid{b'})\mskip1.5mu\};{}\<[E]%
\\
\>[9]{}\Varid{return}\;\Varid{a'}\mskip1.5mu\}{}\<[E]%
\ColumnHook
\end{hscode}\resethooks
 \endswithdisplay
\end{proof}

\restatableProposition{prop:signal-wb}
\begin{prop:signal-wb}
   If \ensuremath{\Conid{A}} and \ensuremath{\Conid{B}} are types equipped with a correct notion of equality
  (so \ensuremath{\Varid{a}\mathrel{=}\Varid{b}} if and only if \ensuremath{(\Varid{a}\equals\Varid{b})\mathrel{=}\Conid{True}}), and \ensuremath{bx\mathbin{::}\Conid{StateTBX}\;\Conid{T}\;\Conid{S}\;\Conid{A}\;\Conid{B}} then \ensuremath{\Varid{signalBX}\;\Varid{sigA}\;\Varid{sigB}\;bx\mathbin{::}\Conid{StateTBX}\;\Conid{T}\;\Conid{S}\;\Conid{A}\;\Conid{B}} is well-behaved.
\end{prop:signal-wb}
\begin{proof}
First, observe that the \ensuremath{\Varid{get}} operations are defined so as to
obviously be \ensuremath{\Conid{T}}-pure, and therefore \ensuremath{\mathrm{(G_LG_R)}} holds.  Let \ensuremath{\Varid{bx'}\mathrel{=}\Varid{signalBX}\;\Varid{sigA}\;\Varid{sigB}\;bx}.

For \ensuremath{\mathrm{(G_LS_L)}}, we proceed as follows::
\begin{hscode}\SaveRestoreHook
\column{B}{@{}>{\hspre}c<{\hspost}@{}}%
\column{BE}{@{}l@{}}%
\column{4}{@{}>{\hspre}l<{\hspost}@{}}%
\column{6}{@{}>{\hspre}l<{\hspost}@{}}%
\column{8}{@{}>{\hspre}l<{\hspost}@{}}%
\column{10}{@{}>{\hspre}l<{\hspost}@{}}%
\column{E}{@{}>{\hspre}l<{\hspost}@{}}%
\>[4]{}\mathbf{do}\;\{\mskip1.5mu {}\<[10]%
\>[10]{}\Varid{a}\leftarrow \Varid{bx'}\mathord{.}\get{L};\Varid{bx'}\mathord{.}\set{L}\;\Varid{a}\mskip1.5mu\}{}\<[E]%
\\
\>[B]{}\mathrel{=}{}\<[BE]%
\>[6]{}\mbox{\commentbegin  Definitions  \commentend}{}\<[E]%
\\
\>[B]{}\hsindent{4}{}\<[4]%
\>[4]{}\mathbf{do}\;\{\mskip1.5mu {}\<[10]%
\>[10]{}\Varid{a}\leftarrow bx\mathord{.}\get{L};\Varid{a'}\leftarrow bx\mathord{.}\get{L};bx\mathord{.}\set{L}\;\Varid{a};{}\<[E]%
\\
\>[4]{}\hsindent{4}{}\<[8]%
\>[8]{}\Varid{lift}\;(\mathbf{if}\;\Varid{a}\notequals\Varid{a'}\;\mathbf{then}\;\Varid{sigA}\;\Varid{a'}\;\mathbf{else}\;\Varid{return}\;())\mskip1.5mu\}{}\<[E]%
\\
\>[B]{}\mathrel{=}{}\<[BE]%
\>[6]{}\mbox{\commentbegin  \ensuremath{bx\mathord{.}\get{L}} copyable  \commentend}{}\<[E]%
\\
\>[B]{}\hsindent{4}{}\<[4]%
\>[4]{}\mathbf{do}\;\{\mskip1.5mu {}\<[10]%
\>[10]{}\Varid{a}\leftarrow bx\mathord{.}\get{L};bx\mathord{.}\set{L}\;\Varid{a};{}\<[E]%
\\
\>[4]{}\hsindent{4}{}\<[8]%
\>[8]{}\Varid{lift}\;(\mathbf{if}\;\Varid{a}\notequals\Varid{a'}\;\mathbf{then}\;\Varid{sigA}\;\Varid{a'}\;\mathbf{else}\;\Varid{return}\;())\mskip1.5mu\}{}\<[E]%
\\
\>[B]{}\mathrel{=}{}\<[BE]%
\>[6]{}\mbox{\commentbegin  \ensuremath{\Varid{a}\notequals\Varid{a}\equals\Conid{False}}  \commentend}{}\<[E]%
\\
\>[B]{}\hsindent{4}{}\<[4]%
\>[4]{}\mathbf{do}\;\{\mskip1.5mu {}\<[10]%
\>[10]{}\Varid{a}\leftarrow bx\mathord{.}\get{L};bx\mathord{.}\set{L}\;\Varid{a};\Varid{lift}\;(\Varid{return}\;())\mskip1.5mu\}{}\<[E]%
\\
\>[B]{}\mathrel{=}{}\<[BE]%
\>[6]{}\mbox{\commentbegin  \ensuremath{\mathrm{(G_LS_L)}}  \commentend}{}\<[E]%
\\
\>[B]{}\hsindent{4}{}\<[4]%
\>[4]{}\mathbf{do}\;\{\mskip1.5mu {}\<[10]%
\>[10]{}\Varid{return}\;();\Varid{lift}\;(\Varid{return}\;())\mskip1.5mu\}{}\<[E]%
\\
\>[B]{}\mathrel{=}{}\<[BE]%
\>[6]{}\mbox{\commentbegin  monad unit, \ensuremath{\Varid{lift}} monad morphism  \commentend}{}\<[E]%
\\
\>[B]{}\hsindent{4}{}\<[4]%
\>[4]{}\Varid{return}\;(){}\<[E]%
\ColumnHook
\end{hscode}\resethooks

\ensuremath{\mathrm{(S_LG_L)}}:
\begin{hscode}\SaveRestoreHook
\column{B}{@{}>{\hspre}c<{\hspost}@{}}%
\column{BE}{@{}l@{}}%
\column{4}{@{}>{\hspre}l<{\hspost}@{}}%
\column{6}{@{}>{\hspre}l<{\hspost}@{}}%
\column{10}{@{}>{\hspre}l<{\hspost}@{}}%
\column{35}{@{}>{\hspre}l<{\hspost}@{}}%
\column{53}{@{}>{\hspre}l<{\hspost}@{}}%
\column{E}{@{}>{\hspre}l<{\hspost}@{}}%
\>[4]{}\mathbf{do}\;\{\mskip1.5mu {}\<[10]%
\>[10]{}\Varid{bx'}\mathord{.}\set{L}\;\Varid{a};\Varid{bx'}\mathord{.}\get{L}\mskip1.5mu\}{}\<[E]%
\\
\>[B]{}\mathrel{=}{}\<[BE]%
\>[6]{}\mbox{\commentbegin  Definition  \commentend}{}\<[E]%
\\
\>[B]{}\hsindent{4}{}\<[4]%
\>[4]{}\mathbf{do}\;\{\mskip1.5mu {}\<[10]%
\>[10]{}\Varid{a'}\leftarrow bx\mathord{.}\get{L};bx\mathord{.}\set{L}\;\Varid{a};{}\<[E]%
\\
\>[10]{}\Varid{lift}\;(\mathbf{if}\;\Varid{a}\notequals\Varid{a'}\;\mathbf{then}\;\Varid{sigA}\;\Varid{a'}\;\mathbf{else}\;\Varid{return}\;());{}\<[E]%
\\
\>[10]{}bx\mathord{.}\get{L}\mskip1.5mu\}{}\<[E]%
\\
\>[B]{}\mathrel{=}{}\<[BE]%
\>[6]{}\mbox{\commentbegin  Monad unit  \commentend}{}\<[E]%
\\
\>[B]{}\hsindent{4}{}\<[4]%
\>[4]{}\mathbf{do}\;\{\mskip1.5mu {}\<[10]%
\>[10]{}\Varid{a'}\leftarrow bx\mathord{.}\get{L};bx\mathord{.}\set{L}\;\Varid{a};{}\<[E]%
\\
\>[10]{}\Varid{lift}\;(\mathbf{if}\;\Varid{a}\notequals\Varid{a'}\;\mathbf{then}\;\Varid{sigA}\;\Varid{a'}\;\mathbf{else}\;\Varid{return}\;());{}\<[E]%
\\
\>[10]{}\Varid{a''}\leftarrow bx\mathord{.}\get{L};\Varid{return}\;\Varid{a''}\mskip1.5mu\}{}\<[E]%
\\
\>[B]{}\mathrel{=}{}\<[BE]%
\>[6]{}\mbox{\commentbegin  Lemma~\ref{lem:liftings-commute}  \commentend}{}\<[E]%
\\
\>[B]{}\hsindent{4}{}\<[4]%
\>[4]{}\mathbf{do}\;\{\mskip1.5mu {}\<[10]%
\>[10]{}\Varid{a'}\leftarrow bx\mathord{.}\get{L};bx\mathord{.}\set{L}\;\Varid{a};{}\<[53]%
\>[53]{}\Varid{a''}\leftarrow bx\mathord{.}\get{L};{}\<[E]%
\\
\>[10]{}\Varid{lift}\;(\mathbf{if}\;\Varid{a}\notequals\Varid{a'}\;\mathbf{then}\;\Varid{sigA}\;\Varid{a'}\;\mathbf{else}\;\Varid{return}\;());\Varid{return}\;\Varid{a''}\mskip1.5mu\}{}\<[E]%
\\
\>[B]{}\mathrel{=}{}\<[BE]%
\>[6]{}\mbox{\commentbegin  \ensuremath{\mathrm{(S_LG_L)}}  \commentend}{}\<[E]%
\\
\>[B]{}\hsindent{4}{}\<[4]%
\>[4]{}\mathbf{do}\;\{\mskip1.5mu {}\<[10]%
\>[10]{}\Varid{a'}\leftarrow bx\mathord{.}\get{L};bx\mathord{.}\set{L}\;\Varid{a};\Varid{a''}\leftarrow \Varid{return}\;\Varid{a};{}\<[E]%
\\
\>[10]{}\Varid{lift}\;(\mathbf{if}\;\Varid{a}\notequals\Varid{a'}\;\mathbf{then}\;\Varid{sigA}\;\Varid{a'}\;\mathbf{else}\;\Varid{return}\;());\Varid{return}\;\Varid{a''}\mskip1.5mu\}{}\<[E]%
\\
\>[B]{}\mathrel{=}{}\<[BE]%
\>[6]{}\mbox{\commentbegin  Monad unit  \commentend}{}\<[E]%
\\
\>[B]{}\hsindent{4}{}\<[4]%
\>[4]{}\mathbf{do}\;\{\mskip1.5mu {}\<[10]%
\>[10]{}\Varid{a'}\leftarrow bx\mathord{.}\get{L};bx\mathord{.}\set{L}\;\Varid{a};{}\<[E]%
\\
\>[10]{}\Varid{lift}\;(\mathbf{if}\;\Varid{a}\notequals\Varid{a'}\;\mathbf{then}\;\Varid{sigA}\;\Varid{a'}\;\mathbf{else}\;\Varid{return}\;());\Varid{return}\;\Varid{a}\mskip1.5mu\}{}\<[E]%
\\
\>[B]{}\mathrel{=}{}\<[BE]%
\>[6]{}\mbox{\commentbegin  Definition  \commentend}{}\<[E]%
\\
\>[B]{}\hsindent{4}{}\<[4]%
\>[4]{}\mathbf{do}\;\{\mskip1.5mu (\Varid{signalBX}\;\Varid{sigA}\;\Varid{sigB}\;{}\<[35]%
\>[35]{}bx)\mathord{.}\set{L}\;\Varid{a};\Varid{return}\;\Varid{a}\mskip1.5mu\}{}\<[E]%
\ColumnHook
\end{hscode}\resethooks
\endswithdisplay
\end{proof}

\restatableProposition{prop:dynamic-wb}
\begin{prop:dynamic-wb}
For any \ensuremath{\Varid{f}}, \ensuremath{\Varid{g}}, the dynamic \bx{} \ensuremath{\Varid{dynamicBX}\;\Varid{f}\;\Varid{g}} is well-behaved.  
\end{prop:dynamic-wb}
\begin{proof}
  Let \ensuremath{bx\mathrel{=}\Varid{dynamicBX}\;\Varid{f}\;\Varid{g}} for some \ensuremath{\Varid{f},\Varid{g}}.
  For \ensuremath{\mathrm{(S_LG_L)}}, by construction, an invocation of \ensuremath{bx\mathord{.}\set{L}\;\Varid{a'}} ends
by setting the state to \ensuremath{((\Varid{a'},\Varid{b'}),\Varid{fs},\Varid{bs})} for some \ensuremath{\Varid{b'},\Varid{fs},\Varid{bs}}, and a
subsequent \ensuremath{bx\mathord{.}\get{L}} will return \ensuremath{\Varid{a'}}.
More formally, we proceed as follows:
\begin{hscode}\SaveRestoreHook
\column{B}{@{}>{\hspre}c<{\hspost}@{}}%
\column{BE}{@{}l@{}}%
\column{4}{@{}>{\hspre}l<{\hspost}@{}}%
\column{6}{@{}>{\hspre}l<{\hspost}@{}}%
\column{10}{@{}>{\hspre}l<{\hspost}@{}}%
\column{12}{@{}>{\hspre}l<{\hspost}@{}}%
\column{16}{@{}>{\hspre}l<{\hspost}@{}}%
\column{24}{@{}>{\hspre}l<{\hspost}@{}}%
\column{33}{@{}>{\hspre}l<{\hspost}@{}}%
\column{E}{@{}>{\hspre}l<{\hspost}@{}}%
\>[4]{}\mathbf{do}\;\{\mskip1.5mu bx\mathord{.}\set{L}\;\Varid{a'};bx\mathord{.}\get{L}\mskip1.5mu\}{}\<[E]%
\\
\>[B]{}\mathrel{=}{}\<[BE]%
\>[6]{}\mbox{\commentbegin  definition  \commentend}{}\<[E]%
\\
\>[B]{}\hsindent{4}{}\<[4]%
\>[4]{}\mathbf{do}\;\{\mskip1.5mu {}\<[10]%
\>[10]{}((\Varid{a},\Varid{b}),\Varid{fs},\Varid{bs})\leftarrow \Varid{get};{}\<[E]%
\\
\>[10]{}\mathbf{if}\;\Varid{a}\equals\Varid{a'}\;\mathbf{then}\;\Varid{return}\;()\;\mathbf{else}{}\<[E]%
\\
\>[10]{}\hsindent{2}{}\<[12]%
\>[12]{}\mathbf{do}\;{}\<[16]%
\>[16]{}\Varid{b'}\leftarrow \mathbf{case}\;\Varid{lookup}\;(\Varid{a'},\Varid{b})\;\Varid{fs}\;\mathbf{of}{}\<[E]%
\\
\>[16]{}\hsindent{8}{}\<[24]%
\>[24]{}\Conid{Just}\;\Varid{b'}{}\<[33]%
\>[33]{}\hsarrow{\rightarrow }{\mathpunct{.}}\Varid{return}\;\Varid{b'}{}\<[E]%
\\
\>[16]{}\hsindent{8}{}\<[24]%
\>[24]{}\Conid{Nothing}{}\<[33]%
\>[33]{}\hsarrow{\rightarrow }{\mathpunct{.}}\Varid{lift}\;(\Varid{f}\;\Varid{a'}\;\Varid{b}){}\<[E]%
\\
\>[16]{}\Varid{set}\;((\Varid{a'},\Varid{b'}),((\Varid{a'},\Varid{b}),\Varid{b'})\mathbin{:}\Varid{fs},\Varid{bs});{}\<[E]%
\\
\>[10]{}((\Varid{a''},\Varid{b''}),\Varid{fs''},\Varid{gs''})\leftarrow \Varid{get};\Varid{return}\;\Varid{a''}\mskip1.5mu\}{}\<[E]%
\ColumnHook
\end{hscode}\resethooks
We now consider three sub-cases.  

First, if \ensuremath{\Varid{a}\equals\Varid{a'}} then 
\begin{hscode}\SaveRestoreHook
\column{B}{@{}>{\hspre}c<{\hspost}@{}}%
\column{BE}{@{}l@{}}%
\column{4}{@{}>{\hspre}l<{\hspost}@{}}%
\column{6}{@{}>{\hspre}l<{\hspost}@{}}%
\column{10}{@{}>{\hspre}l<{\hspost}@{}}%
\column{12}{@{}>{\hspre}l<{\hspost}@{}}%
\column{16}{@{}>{\hspre}l<{\hspost}@{}}%
\column{24}{@{}>{\hspre}l<{\hspost}@{}}%
\column{33}{@{}>{\hspre}l<{\hspost}@{}}%
\column{E}{@{}>{\hspre}l<{\hspost}@{}}%
\>[4]{}\mathbf{do}\;\{\mskip1.5mu {}\<[10]%
\>[10]{}((\Varid{a},\Varid{b}),\Varid{fs},\Varid{bs})\leftarrow \Varid{get};{}\<[E]%
\\
\>[10]{}\mathbf{if}\;\Varid{a}\equals\Varid{a'}\;\mathbf{then}\;\Varid{return}\;()\;\mathbf{else}{}\<[E]%
\\
\>[10]{}\hsindent{2}{}\<[12]%
\>[12]{}\mathbf{do}\;{}\<[16]%
\>[16]{}\Varid{b'}\leftarrow \mathbf{case}\;\Varid{lookup}\;(\Varid{a'},\Varid{b})\;\Varid{fs}\;\mathbf{of}{}\<[E]%
\\
\>[16]{}\hsindent{8}{}\<[24]%
\>[24]{}\Conid{Just}\;\Varid{b'}{}\<[33]%
\>[33]{}\hsarrow{\rightarrow }{\mathpunct{.}}\Varid{return}\;\Varid{b'}{}\<[E]%
\\
\>[16]{}\hsindent{8}{}\<[24]%
\>[24]{}\Conid{Nothing}{}\<[33]%
\>[33]{}\hsarrow{\rightarrow }{\mathpunct{.}}\Varid{lift}\;(\Varid{f}\;\Varid{a'}\;\Varid{b}){}\<[E]%
\\
\>[16]{}\Varid{set}\;((\Varid{a'},\Varid{b'}),((\Varid{a'},\Varid{b}),\Varid{b'})\mathbin{:}\Varid{fs},\Varid{bs});{}\<[E]%
\\
\>[10]{}((\Varid{a''},\Varid{b''}),\Varid{fs''},\Varid{gs''})\leftarrow \Varid{get};\Varid{return}\;\Varid{a''}\mskip1.5mu\}{}\<[E]%
\\
\>[B]{}\mathrel{=}{}\<[BE]%
\>[6]{}\mbox{\commentbegin  \ensuremath{\Varid{a}\equals\Varid{a'}}  \commentend}{}\<[E]%
\\
\>[B]{}\hsindent{4}{}\<[4]%
\>[4]{}\mathbf{do}\;\{\mskip1.5mu {}\<[10]%
\>[10]{}((\Varid{a},\Varid{b}),\Varid{fs},\Varid{bs})\leftarrow \Varid{get};{}\<[E]%
\\
\>[10]{}\Varid{return}\;();{}\<[E]%
\\
\>[10]{}((\Varid{a''},\Varid{b''}),\Varid{fs''},\Varid{gs''})\leftarrow \Varid{get};\Varid{return}\;\Varid{a''}\mskip1.5mu\}{}\<[E]%
\\
\>[B]{}\mathrel{=}{}\<[BE]%
\>[6]{}\mbox{\commentbegin  \ensuremath{\mathrm{(GG)}}  \commentend}{}\<[E]%
\\
\>[B]{}\hsindent{4}{}\<[4]%
\>[4]{}\mathbf{do}\;\{\mskip1.5mu {}\<[10]%
\>[10]{}((\Varid{a},\Varid{b}),\Varid{fs},\Varid{bs})\leftarrow \Varid{get};{}\<[E]%
\\
\>[10]{}\Varid{return}\;\Varid{a}\mskip1.5mu\}{}\<[E]%
\\
\>[B]{}\mathrel{=}{}\<[BE]%
\>[6]{}\mbox{\commentbegin  reversing previous steps  \commentend}{}\<[E]%
\\
\>[B]{}\hsindent{4}{}\<[4]%
\>[4]{}\mathbf{do}\;\{\mskip1.5mu {}\<[10]%
\>[10]{}((\Varid{a},\Varid{b}),\Varid{fs},\Varid{bs})\leftarrow \Varid{get};{}\<[E]%
\\
\>[10]{}\mathbf{if}\;\Varid{a}\equals\Varid{a'}\;\mathbf{then}\;\Varid{return}\;()\;\mathbf{else}{}\<[E]%
\\
\>[10]{}\hsindent{2}{}\<[12]%
\>[12]{}\mathbf{do}\;{}\<[16]%
\>[16]{}\Varid{b'}\leftarrow \mathbf{case}\;\Varid{lookup}\;(\Varid{a'},\Varid{b})\;\Varid{fs}\;\mathbf{of}{}\<[E]%
\\
\>[16]{}\hsindent{8}{}\<[24]%
\>[24]{}\Conid{Just}\;\Varid{b'}{}\<[33]%
\>[33]{}\hsarrow{\rightarrow }{\mathpunct{.}}\Varid{return}\;\Varid{b'}{}\<[E]%
\\
\>[16]{}\hsindent{8}{}\<[24]%
\>[24]{}\Conid{Nothing}{}\<[33]%
\>[33]{}\hsarrow{\rightarrow }{\mathpunct{.}}\Varid{lift}\;(\Varid{f}\;\Varid{a'}\;\Varid{b}){}\<[E]%
\\
\>[16]{}\Varid{set}\;((\Varid{a'},\Varid{b'}),((\Varid{a'},\Varid{b}),\Varid{b'})\mathbin{:}\Varid{fs},\Varid{bs});{}\<[E]%
\\
\>[10]{}\Varid{return}\;\Varid{a}\mskip1.5mu\}{}\<[E]%
\\
\>[B]{}\mathrel{=}{}\<[BE]%
\>[6]{}\mbox{\commentbegin  definition, \ensuremath{\Varid{a}\mathrel{=}\Varid{a'}}  \commentend}{}\<[E]%
\\
\>[B]{}\hsindent{4}{}\<[4]%
\>[4]{}\mathbf{do}\;\{\mskip1.5mu bx\mathord{.}\set{L}\;\Varid{a'};\Varid{return}\;\Varid{a'}\mskip1.5mu\}{}\<[E]%
\ColumnHook
\end{hscode}\resethooks

Second, if \ensuremath{\Varid{a}\notequals\Varid{a'}} and \ensuremath{((\Varid{a'},\Varid{b}),\Varid{b'})\in\Varid{fs}} for some \ensuremath{\Varid{b'}}, then \ensuremath{\Varid{lookup}\;(\Varid{a'},\Varid{b})\;\Varid{fs}\mathrel{=}\Conid{Just}\;\Varid{b'}} holds, so:
\begin{hscode}\SaveRestoreHook
\column{B}{@{}>{\hspre}c<{\hspost}@{}}%
\column{BE}{@{}l@{}}%
\column{4}{@{}>{\hspre}l<{\hspost}@{}}%
\column{6}{@{}>{\hspre}l<{\hspost}@{}}%
\column{10}{@{}>{\hspre}l<{\hspost}@{}}%
\column{12}{@{}>{\hspre}l<{\hspost}@{}}%
\column{16}{@{}>{\hspre}l<{\hspost}@{}}%
\column{24}{@{}>{\hspre}l<{\hspost}@{}}%
\column{33}{@{}>{\hspre}l<{\hspost}@{}}%
\column{E}{@{}>{\hspre}l<{\hspost}@{}}%
\>[4]{}\mathbf{do}\;\{\mskip1.5mu {}\<[10]%
\>[10]{}((\Varid{a},\Varid{b}),\Varid{fs},\Varid{bs})\leftarrow \Varid{get};{}\<[E]%
\\
\>[10]{}\mathbf{if}\;\Varid{a}\equals\Varid{a'}\;\mathbf{then}\;\Varid{return}\;()\;\mathbf{else}{}\<[E]%
\\
\>[10]{}\hsindent{2}{}\<[12]%
\>[12]{}\mathbf{do}\;{}\<[16]%
\>[16]{}\Varid{b'}\leftarrow \mathbf{case}\;\Varid{lookup}\;(\Varid{a'},\Varid{b})\;\Varid{fs}\;\mathbf{of}{}\<[E]%
\\
\>[16]{}\hsindent{8}{}\<[24]%
\>[24]{}\Conid{Just}\;\Varid{b'}{}\<[33]%
\>[33]{}\hsarrow{\rightarrow }{\mathpunct{.}}\Varid{return}\;\Varid{b'}{}\<[E]%
\\
\>[16]{}\hsindent{8}{}\<[24]%
\>[24]{}\Conid{Nothing}{}\<[33]%
\>[33]{}\hsarrow{\rightarrow }{\mathpunct{.}}\Varid{lift}\;(\Varid{f}\;\Varid{a'}\;\Varid{b}){}\<[E]%
\\
\>[16]{}\Varid{set}\;((\Varid{a'},\Varid{b'}),((\Varid{a'},\Varid{b}),\Varid{b'})\mathbin{:}\Varid{fs},\Varid{bs});{}\<[E]%
\\
\>[10]{}((\Varid{a''},\Varid{b''}),\Varid{fs''},\Varid{gs''})\leftarrow \Varid{get};\Varid{return}\;\Varid{a''}\mskip1.5mu\}{}\<[E]%
\\
\>[B]{}\mathrel{=}{}\<[BE]%
\>[6]{}\mbox{\commentbegin  \ensuremath{\Varid{a}\notequals\Varid{a'}}, \ensuremath{\mathbf{case}} simplification  \commentend}{}\<[E]%
\\
\>[B]{}\hsindent{4}{}\<[4]%
\>[4]{}\mathbf{do}\;\{\mskip1.5mu {}\<[10]%
\>[10]{}((\Varid{a},\Varid{b}),\Varid{fs},\Varid{bs})\leftarrow \Varid{get};{}\<[E]%
\\
\>[10]{}\Varid{b'}\leftarrow \Varid{return}\;\Varid{b'};{}\<[E]%
\\
\>[10]{}\Varid{set}\;((\Varid{a'},\Varid{b'}),((\Varid{a'},\Varid{b}),\Varid{b'})\mathbin{:}\Varid{fs},\Varid{bs});{}\<[E]%
\\
\>[10]{}((\Varid{a''},\Varid{b''}),\Varid{fs''},\Varid{gs''})\leftarrow \Varid{get};\Varid{return}\;\Varid{a''}\mskip1.5mu\}{}\<[E]%
\\
\>[B]{}\mathrel{=}{}\<[BE]%
\>[6]{}\mbox{\commentbegin  monad unit  \commentend}{}\<[E]%
\\
\>[B]{}\hsindent{4}{}\<[4]%
\>[4]{}\mathbf{do}\;\{\mskip1.5mu {}\<[10]%
\>[10]{}((\Varid{a},\Varid{b}),\Varid{fs},\Varid{bs})\leftarrow \Varid{get};{}\<[E]%
\\
\>[10]{}\Varid{set}\;((\Varid{a'},\Varid{b'}),((\Varid{a'},\Varid{b}),\Varid{b'})\mathbin{:}\Varid{fs},\Varid{bs});{}\<[E]%
\\
\>[10]{}((\Varid{a''},\Varid{b''}),\Varid{fs''},\Varid{gs''})\leftarrow \Varid{get};\Varid{return}\;\Varid{a''}\mskip1.5mu\}{}\<[E]%
\\
\>[B]{}\mathrel{=}{}\<[BE]%
\>[6]{}\mbox{\commentbegin  \ensuremath{\mathrm{(GG)}}  \commentend}{}\<[E]%
\\
\>[B]{}\hsindent{4}{}\<[4]%
\>[4]{}\mathbf{do}\;\{\mskip1.5mu {}\<[10]%
\>[10]{}((\Varid{a},\Varid{b}),\Varid{fs},\Varid{bs})\leftarrow \Varid{get};{}\<[E]%
\\
\>[10]{}\Varid{set}\;((\Varid{a'},\Varid{b'}),((\Varid{a'},\Varid{b}),\Varid{b'})\mathbin{:}\Varid{fs},\Varid{bs});{}\<[E]%
\\
\>[10]{}\Varid{return}\;\Varid{a'}\mskip1.5mu\}{}\<[E]%
\\
\>[B]{}\mathrel{=}{}\<[BE]%
\>[6]{}\mbox{\commentbegin  reversing previous steps  \commentend}{}\<[E]%
\\
\>[B]{}\hsindent{4}{}\<[4]%
\>[4]{}\mathbf{do}\;\{\mskip1.5mu {}\<[10]%
\>[10]{}((\Varid{a},\Varid{b}),\Varid{fs},\Varid{bs})\leftarrow \Varid{get}{}\<[E]%
\\
\>[10]{}\mathbf{if}\;\Varid{a}\equals\Varid{a'}\;\mathbf{then}\;\Varid{return}\;()\;\mathbf{else}{}\<[E]%
\\
\>[10]{}\hsindent{2}{}\<[12]%
\>[12]{}\mathbf{do}\;{}\<[16]%
\>[16]{}\Varid{b'}\leftarrow \mathbf{case}\;\Varid{lookup}\;(\Varid{a'},\Varid{b})\;\Varid{fs}\;\mathbf{of}{}\<[E]%
\\
\>[16]{}\hsindent{8}{}\<[24]%
\>[24]{}\Conid{Just}\;\Varid{b'}{}\<[33]%
\>[33]{}\hsarrow{\rightarrow }{\mathpunct{.}}\Varid{return}\;\Varid{b'}{}\<[E]%
\\
\>[16]{}\hsindent{8}{}\<[24]%
\>[24]{}\Conid{Nothing}{}\<[33]%
\>[33]{}\hsarrow{\rightarrow }{\mathpunct{.}}\Varid{lift}\;(\Varid{f}\;\Varid{a'}\;\Varid{b}){}\<[E]%
\\
\>[16]{}\Varid{set}\;((\Varid{a'},\Varid{b'}),((\Varid{a'},\Varid{b}),\Varid{b'})\mathbin{:}\Varid{fs},\Varid{bs});{}\<[E]%
\\
\>[10]{}\Varid{return}\;\Varid{a'}\mskip1.5mu\}{}\<[E]%
\\
\>[B]{}\mathrel{=}{}\<[BE]%
\>[6]{}\mbox{\commentbegin  definition  \commentend}{}\<[E]%
\\
\>[B]{}\hsindent{4}{}\<[4]%
\>[4]{}\mathbf{do}\;\{\mskip1.5mu bx\mathord{.}\set{L}\;\Varid{a'};\Varid{return}\;\Varid{a'}\mskip1.5mu\}{}\<[E]%
\ColumnHook
\end{hscode}\resethooks

Finally, if \ensuremath{\Varid{a}\notequals\Varid{a'}} and there is no \ensuremath{\Varid{b'}} such that \ensuremath{((\Varid{a'},\Varid{b}),\Varid{b'})\in\Varid{fs}}, then \ensuremath{\Varid{lookup}\;(\Varid{a'},\Varid{b})\;\Varid{fs}\mathrel{=}\Conid{Nothing}}, and:
 \begin{hscode}\SaveRestoreHook
\column{B}{@{}>{\hspre}c<{\hspost}@{}}%
\column{BE}{@{}l@{}}%
\column{4}{@{}>{\hspre}l<{\hspost}@{}}%
\column{6}{@{}>{\hspre}l<{\hspost}@{}}%
\column{10}{@{}>{\hspre}l<{\hspost}@{}}%
\column{12}{@{}>{\hspre}l<{\hspost}@{}}%
\column{16}{@{}>{\hspre}l<{\hspost}@{}}%
\column{24}{@{}>{\hspre}l<{\hspost}@{}}%
\column{33}{@{}>{\hspre}l<{\hspost}@{}}%
\column{E}{@{}>{\hspre}l<{\hspost}@{}}%
\>[4]{}\mathbf{do}\;\{\mskip1.5mu {}\<[10]%
\>[10]{}((\Varid{a},\Varid{b}),\Varid{fs},\Varid{bs})\leftarrow \Varid{get};{}\<[E]%
\\
\>[10]{}\mathbf{if}\;\Varid{a}\equals\Varid{a'}\;\mathbf{then}\;\Varid{return}\;()\;\mathbf{else}{}\<[E]%
\\
\>[10]{}\hsindent{2}{}\<[12]%
\>[12]{}\mathbf{do}\;{}\<[16]%
\>[16]{}\Varid{b'}\leftarrow \mathbf{case}\;\Varid{lookup}\;(\Varid{a'},\Varid{b})\;\Varid{fs}\;\mathbf{of}{}\<[E]%
\\
\>[16]{}\hsindent{8}{}\<[24]%
\>[24]{}\Conid{Just}\;\Varid{b'}{}\<[33]%
\>[33]{}\hsarrow{\rightarrow }{\mathpunct{.}}\Varid{return}\;\Varid{b'}{}\<[E]%
\\
\>[16]{}\hsindent{8}{}\<[24]%
\>[24]{}\Conid{Nothing}{}\<[33]%
\>[33]{}\hsarrow{\rightarrow }{\mathpunct{.}}\Varid{lift}\;(\Varid{f}\;\Varid{a'}\;\Varid{b}){}\<[E]%
\\
\>[16]{}\Varid{set}\;((\Varid{a'},\Varid{b'}),((\Varid{a'},\Varid{b}),\Varid{b'})\mathbin{:}\Varid{fs},\Varid{bs});{}\<[E]%
\\
\>[10]{}((\Varid{a''},\Varid{b''}),\Varid{fs''},\Varid{gs''})\leftarrow \Varid{get};\Varid{return}\;\Varid{a''}\mskip1.5mu\}{}\<[E]%
\\
\>[B]{}\mathrel{=}{}\<[BE]%
\>[6]{}\mbox{\commentbegin  \ensuremath{\Varid{a}\notequals\Varid{a'}}, \ensuremath{\Varid{lookup}\;(\Varid{a'},\Varid{b})\mathrel{=}\Conid{Nothing}}  \commentend}{}\<[E]%
\\
\>[B]{}\hsindent{4}{}\<[4]%
\>[4]{}\mathbf{do}\;\{\mskip1.5mu {}\<[10]%
\>[10]{}((\Varid{a},\Varid{b}),\Varid{fs},\Varid{bs})\leftarrow \Varid{get};{}\<[E]%
\\
\>[10]{}\Varid{b'}\leftarrow \Varid{lift}\;(\Varid{f}\;\Varid{a'}\;\Varid{b});{}\<[E]%
\\
\>[10]{}\Varid{set}\;((\Varid{a'},\Varid{b'}),((\Varid{a'},\Varid{b}),\Varid{b'})\mathbin{:}\Varid{fs},\Varid{bs});{}\<[E]%
\\
\>[10]{}((\Varid{a''},\Varid{b''}),\Varid{fs''},\Varid{gs''})\leftarrow \Varid{get};\Varid{return}\;\Varid{a''}\mskip1.5mu\}{}\<[E]%
\\
\>[B]{}\mathrel{=}{}\<[BE]%
\>[6]{}\mbox{\commentbegin  \ensuremath{\mathrm{(SG)}}  \commentend}{}\<[E]%
\\
\>[B]{}\hsindent{4}{}\<[4]%
\>[4]{}\mathbf{do}\;\{\mskip1.5mu {}\<[10]%
\>[10]{}((\Varid{a},\Varid{b}),\Varid{fs},\Varid{bs})\leftarrow \Varid{get};{}\<[E]%
\\
\>[10]{}\Varid{b'}\leftarrow \Varid{lift}\;(\Varid{f}\;\Varid{a'}\;\Varid{b});{}\<[E]%
\\
\>[10]{}\Varid{set}\;((\Varid{a'},\Varid{b'}),((\Varid{a'},\Varid{b}),\Varid{b'})\mathbin{:}\Varid{fs},\Varid{bs});{}\<[E]%
\\
\>[10]{}\Varid{return}\;\Varid{a'}\mskip1.5mu\}{}\<[E]%
\\
\>[B]{}\mathrel{=}{}\<[BE]%
\>[6]{}\mbox{\commentbegin  reversing previous steps  \commentend}{}\<[E]%
\\
\>[B]{}\hsindent{4}{}\<[4]%
\>[4]{}\mathbf{do}\;\{\mskip1.5mu {}\<[10]%
\>[10]{}((\Varid{a},\Varid{b}),\Varid{fs},\Varid{bs})\leftarrow \Varid{get}{}\<[E]%
\\
\>[10]{}\mathbf{if}\;\Varid{a}\equals\Varid{a'}\;\mathbf{then}\;\Varid{return}\;()\;\mathbf{else}{}\<[E]%
\\
\>[10]{}\hsindent{2}{}\<[12]%
\>[12]{}\mathbf{do}\;{}\<[16]%
\>[16]{}\Varid{b'}\leftarrow \mathbf{case}\;\Varid{lookup}\;(\Varid{a'},\Varid{b})\;\Varid{fs}\;\mathbf{of}{}\<[E]%
\\
\>[16]{}\hsindent{8}{}\<[24]%
\>[24]{}\Conid{Just}\;\Varid{b'}{}\<[33]%
\>[33]{}\hsarrow{\rightarrow }{\mathpunct{.}}\Varid{return}\;\Varid{b'}{}\<[E]%
\\
\>[16]{}\hsindent{8}{}\<[24]%
\>[24]{}\Conid{Nothing}{}\<[33]%
\>[33]{}\hsarrow{\rightarrow }{\mathpunct{.}}\Varid{lift}\;(\Varid{f}\;\Varid{a'}\;\Varid{b}){}\<[E]%
\\
\>[16]{}\Varid{set}\;((\Varid{a'},\Varid{b'}),((\Varid{a'},\Varid{b}),\Varid{b'})\mathbin{:}\Varid{fs},\Varid{bs});{}\<[E]%
\\
\>[10]{}\Varid{return}\;\Varid{a'}\mskip1.5mu\}{}\<[E]%
\\
\>[B]{}\mathrel{=}{}\<[BE]%
\>[6]{}\mbox{\commentbegin  definition  \commentend}{}\<[E]%
\\
\>[B]{}\hsindent{4}{}\<[4]%
\>[4]{}\mathbf{do}\;\{\mskip1.5mu bx\mathord{.}\set{L}\;\Varid{a'};\Varid{return}\;\Varid{a'}\mskip1.5mu\}{}\<[E]%
\ColumnHook
\end{hscode}\resethooks
Therefore, \ensuremath{\mathrm{(S_LG_L)}} holds in all three cases.

For \ensuremath{\mathrm{(G_LS_L)}}, an invocation of \ensuremath{bx\mathord{.}\get{L}} in a state
\ensuremath{((\Varid{a},\Varid{b}),\Varid{fs},\Varid{bs})} returns \ensuremath{\Varid{a}}, and by construction a subsequent
\ensuremath{bx\mathord{.}\set{L}\;\Varid{a}} has no effect.  

More formally, we proceed as follows.  
  \begin{hscode}\SaveRestoreHook
\column{B}{@{}>{\hspre}c<{\hspost}@{}}%
\column{BE}{@{}l@{}}%
\column{4}{@{}>{\hspre}l<{\hspost}@{}}%
\column{6}{@{}>{\hspre}l<{\hspost}@{}}%
\column{10}{@{}>{\hspre}l<{\hspost}@{}}%
\column{13}{@{}>{\hspre}l<{\hspost}@{}}%
\column{17}{@{}>{\hspre}l<{\hspost}@{}}%
\column{25}{@{}>{\hspre}l<{\hspost}@{}}%
\column{34}{@{}>{\hspre}l<{\hspost}@{}}%
\column{E}{@{}>{\hspre}l<{\hspost}@{}}%
\>[4]{}\mathbf{do}\;\{\mskip1.5mu \Varid{a}\leftarrow bx\mathord{.}\get{L};bx\mathord{.}\set{L}\;\Varid{a}\mskip1.5mu\}{}\<[E]%
\\
\>[B]{}\mathrel{=}{}\<[BE]%
\>[6]{}\mbox{\commentbegin  Definition  \commentend}{}\<[E]%
\\
\>[B]{}\hsindent{4}{}\<[4]%
\>[4]{}\mathbf{do}\;\{\mskip1.5mu {}\<[10]%
\>[10]{}((\Varid{a},\anonymous ),\anonymous ,\anonymous )\leftarrow \Varid{get};{}\<[E]%
\\
\>[10]{}((\Varid{a}_{0},\Varid{b}_{0}),\Varid{fs},\Varid{bs})\leftarrow \Varid{get};{}\<[E]%
\\
\>[10]{}\mathbf{if}\;\Varid{a}_{0}\equals\Varid{a}\;\mathbf{then}\;\Varid{return}\;()\;\mathbf{else}{}\<[E]%
\\
\>[10]{}\hsindent{3}{}\<[13]%
\>[13]{}\mathbf{do}\;{}\<[17]%
\>[17]{}\Varid{b'}\leftarrow \mathbf{case}\;\Varid{lookup}\;(\Varid{a},\Varid{b}_{0})\;\Varid{fs}\;\mathbf{of}{}\<[E]%
\\
\>[17]{}\hsindent{8}{}\<[25]%
\>[25]{}\Conid{Just}\;\Varid{b'}{}\<[34]%
\>[34]{}\hsarrow{\rightarrow }{\mathpunct{.}}\Varid{return}\;\Varid{b'}{}\<[E]%
\\
\>[17]{}\hsindent{8}{}\<[25]%
\>[25]{}\Conid{Nothing}{}\<[34]%
\>[34]{}\hsarrow{\rightarrow }{\mathpunct{.}}\Varid{lift}\;(\Varid{f}\;\Varid{a}\;\Varid{b}){}\<[E]%
\\
\>[17]{}\Varid{set}\;((\Varid{a},\Varid{b'}),((\Varid{a},\Varid{b}),\Varid{b'})\mathbin{:}\Varid{fs},\Varid{bs})\mskip1.5mu\}{}\<[E]%
\\
\>[B]{}\mathrel{=}{}\<[BE]%
\>[6]{}\mbox{\commentbegin  \ensuremath{\mathrm{(GG)}}  \commentend}{}\<[E]%
\\
\>[B]{}\hsindent{4}{}\<[4]%
\>[4]{}\mathbf{do}\;\{\mskip1.5mu {}\<[10]%
\>[10]{}((\Varid{a}_{0},\Varid{b}_{0}),\Varid{fs},\Varid{bs})\leftarrow \Varid{get};{}\<[E]%
\\
\>[10]{}\mathbf{if}\;\Varid{a}_{0}\equals\Varid{a}_{0}\;\mathbf{then}\;\Varid{return}\;()\;\mathbf{else}{}\<[E]%
\\
\>[10]{}\hsindent{3}{}\<[13]%
\>[13]{}\mathbf{do}\;{}\<[17]%
\>[17]{}\Varid{b'}\leftarrow \mathbf{case}\;\Varid{lookup}\;(\Varid{a},\Varid{b}_{0})\;\Varid{fs}\;\mathbf{of}{}\<[E]%
\\
\>[17]{}\hsindent{8}{}\<[25]%
\>[25]{}\Conid{Just}\;\Varid{b'}{}\<[34]%
\>[34]{}\hsarrow{\rightarrow }{\mathpunct{.}}\Varid{return}\;\Varid{b'}{}\<[E]%
\\
\>[17]{}\hsindent{8}{}\<[25]%
\>[25]{}\Conid{Nothing}{}\<[34]%
\>[34]{}\hsarrow{\rightarrow }{\mathpunct{.}}\Varid{lift}\;(\Varid{f}\;\Varid{a}_{0}\;\Varid{b}){}\<[E]%
\\
\>[17]{}\Varid{set}\;((\Varid{a}_{0},\Varid{b'}),((\Varid{a}_{0},\Varid{b}),\Varid{b'})\mathbin{:}\Varid{fs},\Varid{bs})\mskip1.5mu\}{}\<[E]%
\\
\>[B]{}\mathrel{=}{}\<[BE]%
\>[6]{}\mbox{\commentbegin  \ensuremath{\Varid{a}_{0}\mathrel{=}\Varid{a}_{0}}  \commentend}{}\<[E]%
\\
\>[B]{}\hsindent{4}{}\<[4]%
\>[4]{}\mathbf{do}\;\{\mskip1.5mu {}\<[10]%
\>[10]{}((\Varid{a}_{0},\Varid{b}_{0}),\Varid{fs},\Varid{bs})\leftarrow \Varid{get};{}\<[E]%
\\
\>[10]{}\Varid{return}\;()\mskip1.5mu\}{}\<[E]%
\\
\>[B]{}\mathrel{=}{}\<[BE]%
\>[6]{}\mbox{\commentbegin  \ensuremath{\mathrm{(GG)}}  \commentend}{}\<[E]%
\\
\>[B]{}\hsindent{4}{}\<[4]%
\>[4]{}\Varid{return}\;(){}\<[E]%
\ColumnHook
\end{hscode}\resethooks
\endswithdisplay
\end{proof}

\section{Code}\label{code}

This appendix
includes all the code discussed in the paper, along with other
convenience definitions that were not discussed in the paper and
alternative definitions of e.g.\ composition that we explored while
writing the paper.
In this appendix, we have reinstated the standard Haskell use of \ensuremath{\mathbf{newtype}}s etc.

\subsection{SetBX}
\begin{hscode}\SaveRestoreHook
\column{B}{@{}>{\hspre}l<{\hspost}@{}}%
\column{3}{@{}>{\hspre}l<{\hspost}@{}}%
\column{8}{@{}>{\hspre}l<{\hspost}@{}}%
\column{E}{@{}>{\hspre}l<{\hspost}@{}}%
\>[3]{}\mbox{\enskip\{-\# LANGUAGE RankNTypes, ImpredicativeTypes  \#-\}\enskip}{}\<[E]%
\\
\>[3]{}\mathbf{module}\;\Conid{BX}\;\mathbf{where}{}\<[E]%
\\
\>[3]{}\mathbf{import}\;\Conid{\Conid{Control}.\Conid{Monad}.State}\;\Varid{as}\;\Conid{State}{}\<[E]%
\\
\>[3]{}\mathbf{import}\;\Conid{\Conid{Control}.\Conid{Monad}.Reader}\;\Varid{as}\;\Conid{Reader}{}\<[E]%
\\[\blanklineskip]%
\>[3]{}\mathbf{data}\;\Conid{BX}\;\Varid{m}\;\Varid{a}\;\Varid{b}\mathrel{=}\Conid{BX}\;\{\mskip1.5mu {}\<[E]%
\\
\>[3]{}\hsindent{5}{}\<[8]%
\>[8]{}\Varid{mgetl}\mathbin{::}\Varid{m}\;\Varid{a},{}\<[E]%
\\
\>[3]{}\hsindent{5}{}\<[8]%
\>[8]{}\Varid{msetl}\mathbin{::}\Varid{a}\hsarrow{\rightarrow }{\mathpunct{.}}\Varid{m}\;(),{}\<[E]%
\\
\>[3]{}\hsindent{5}{}\<[8]%
\>[8]{}\Varid{mgetr}\mathbin{::}\Varid{m}\;\Varid{b},{}\<[E]%
\\
\>[3]{}\hsindent{5}{}\<[8]%
\>[8]{}\Varid{msetr}\mathbin{::}\Varid{b}\hsarrow{\rightarrow }{\mathpunct{.}}\Varid{m}\;(){}\<[E]%
\\
\>[3]{}\mskip1.5mu\}{}\<[E]%
\\[\blanklineskip]%
\>[3]{}\Varid{mputlr}\mathbin{::}\Conid{Monad}\;\Varid{m}\Rightarrow \Conid{BX}\;\Varid{m}\;\Varid{a}\;\Varid{b}\hsarrow{\rightarrow }{\mathpunct{.}}\Varid{a}\hsarrow{\rightarrow }{\mathpunct{.}}\Varid{m}\;\Varid{b}{}\<[E]%
\\
\>[3]{}\Varid{mputlr}\;bx\;\Varid{a}\mathrel{=}\Varid{msetl}\;bx\;\Varid{a}\sequ \Varid{mgetr}\;bx{}\<[E]%
\\[\blanklineskip]%
\>[3]{}\Varid{mputrl}\mathbin{::}\Conid{Monad}\;\Varid{m}\Rightarrow \Conid{BX}\;\Varid{m}\;\Varid{a}\;\Varid{b}\hsarrow{\rightarrow }{\mathpunct{.}}\Varid{b}\hsarrow{\rightarrow }{\mathpunct{.}}\Varid{m}\;\Varid{a}{}\<[E]%
\\
\>[3]{}\Varid{mputrl}\;bx\;\Varid{b}\mathrel{=}\Varid{msetr}\;bx\;\Varid{b}\sequ \Varid{mgetl}\;bx{}\<[E]%
\ColumnHook
\end{hscode}\resethooks
Identity.
 \begin{hscode}\SaveRestoreHook
\column{B}{@{}>{\hspre}l<{\hspost}@{}}%
\column{3}{@{}>{\hspre}l<{\hspost}@{}}%
\column{E}{@{}>{\hspre}l<{\hspost}@{}}%
\>[3]{}\Varid{idMBX}\mathbin{::}\Conid{MonadState}\;\Varid{a}\;\Varid{m}\Rightarrow \Conid{BX}\;\Varid{m}\;\Varid{a}\;\Varid{a}{}\<[E]%
\\
\>[3]{}\Varid{idMBX}\mathrel{=}\Conid{BX}\;\Varid{get}\;\Varid{put}\;\Varid{get}\;\Varid{put}{}\<[E]%
\ColumnHook
\end{hscode}\resethooks
Duality.
\begin{hscode}\SaveRestoreHook
\column{B}{@{}>{\hspre}l<{\hspost}@{}}%
\column{3}{@{}>{\hspre}l<{\hspost}@{}}%
\column{E}{@{}>{\hspre}l<{\hspost}@{}}%
\>[3]{}\Varid{coMBX}\mathbin{::}\Conid{BX}\;\Varid{m}\;\Varid{a}\;\Varid{b}\hsarrow{\rightarrow }{\mathpunct{.}}\Conid{BX}\;\Varid{m}\;\Varid{b}\;\Varid{a}{}\<[E]%
\\
\>[3]{}\Varid{coMBX}\;bx\mathrel{=}\Conid{BX}\;(\Varid{mgetr}\;bx)\;(\Varid{msetr}\;bx)\;(\Varid{mgetl}\;bx)\;(\Varid{msetl}\;bx){}\<[E]%
\ColumnHook
\end{hscode}\resethooks
Natural transformations
 \begin{hscode}\SaveRestoreHook
\column{B}{@{}>{\hspre}l<{\hspost}@{}}%
\column{3}{@{}>{\hspre}l<{\hspost}@{}}%
\column{E}{@{}>{\hspre}l<{\hspost}@{}}%
\>[3]{}\mathbf{type}\;\Varid{f}\ntto\Varid{g}\mathrel{=}\forall \Varid{a}\hsforall \hsdot{\cdot }{.}\Varid{f}\;\Varid{a}\hsarrow{\rightarrow }{\mathpunct{.}}\Varid{g}\;\Varid{a}{}\<[E]%
\\[\blanklineskip]%
\>[3]{}\mathbf{type}\;\Varid{g}_{1} \ntgets \Varid{f} \ntto \Varid{g}_{2}\mathrel{=}(\Varid{f}\ntto\Varid{g}_{1},\Varid{f}\ntto\Varid{g}_{2}){}\<[E]%
\\[\blanklineskip]%
\>[3]{}\mathbf{type}\;\Varid{f}_{1} \ntto \Varid{g} \ntgets \Varid{f}_{2}\mathrel{=}(\Varid{f}_{1}\ntto\Varid{g},\Varid{f}_{2}\ntto\Varid{g}){}\<[E]%
\\[\blanklineskip]%
\>[3]{}\mathbf{type}\;\Conid{NTSquare}\;\Varid{f}\;\Varid{g}_{1}\;\Varid{g}_{2}\;\Varid{h}\mathrel{=}(\Varid{g}_{1} \ntgets \Varid{f} \ntto \Varid{g}_{2},\Varid{g}_{1} \ntto \Varid{h} \ntgets \Varid{g}_{2}){}\<[E]%
\ColumnHook
\end{hscode}\resethooks
Composition
\begin{hscode}\SaveRestoreHook
\column{B}{@{}>{\hspre}l<{\hspost}@{}}%
\column{3}{@{}>{\hspre}l<{\hspost}@{}}%
\column{15}{@{}>{\hspre}l<{\hspost}@{}}%
\column{16}{@{}>{\hspre}l<{\hspost}@{}}%
\column{38}{@{}>{\hspre}l<{\hspost}@{}}%
\column{E}{@{}>{\hspre}l<{\hspost}@{}}%
\>[3]{}\Varid{compMBX}\mathbin{::}{}\<[15]%
\>[15]{}(\Varid{m}_{1} \ntto \Varid{n} \ntgets \Varid{m}_{2})\hsarrow{\rightarrow }{\mathpunct{.}}\Conid{BX}\;\Varid{m}_{1}\;\Varid{a}\;\Varid{b}\hsarrow{\rightarrow }{\mathpunct{.}}\Conid{BX}\;\Varid{m}_{2}\;\Varid{b}\;\Varid{c}\hsarrow{\rightarrow }{\mathpunct{.}}{}\<[E]%
\\
\>[15]{}\hsindent{1}{}\<[16]%
\>[16]{}\Conid{BX}\;\Varid{n}\;\Varid{a}\;\Varid{c}{}\<[E]%
\\
\>[3]{}\Varid{compMBX}\;(\varphi ,\psi )\;bx_{1}\;bx_{2}\mathrel{=}\Conid{BX}\;{}\<[38]%
\>[38]{}(\varphi \;(\Varid{mgetl}\;bx_{1}))\;{}\<[E]%
\\
\>[38]{}(\lambda \hslambda \Varid{a}\hsarrow{\rightarrow }{\mathpunct{.}}\varphi \;(\Varid{msetl}\;bx_{1}\;\Varid{a}))\;{}\<[E]%
\\
\>[38]{}(\psi \;(\Varid{mgetr}\;bx_{2}))\;{}\<[E]%
\\
\>[38]{}(\lambda \hslambda \Varid{a}\hsarrow{\rightarrow }{\mathpunct{.}}\psi \;(\Varid{msetr}\;bx_{2}\;\Varid{a})){}\<[E]%
\ColumnHook
\end{hscode}\resethooks
Variant, assuming monad morphisms \ensuremath{\Varid{l}} and \ensuremath{\Varid{r}} 
 \begin{hscode}\SaveRestoreHook
\column{B}{@{}>{\hspre}l<{\hspost}@{}}%
\column{3}{@{}>{\hspre}l<{\hspost}@{}}%
\column{10}{@{}>{\hspre}l<{\hspost}@{}}%
\column{20}{@{}>{\hspre}l<{\hspost}@{}}%
\column{26}{@{}>{\hspre}l<{\hspost}@{}}%
\column{33}{@{}>{\hspre}l<{\hspost}@{}}%
\column{E}{@{}>{\hspre}l<{\hspost}@{}}%
\>[3]{}\Varid{compMBX'}\mathbin{::}(\Conid{Monad}\;\Varid{m}_{1},\Conid{Monad}\;\Varid{m}_{2},\Conid{Monad}\;\Varid{n})\Rightarrow {}\<[E]%
\\
\>[3]{}\hsindent{23}{}\<[26]%
\>[26]{}(\Varid{m}_{1} \ntto \Varid{n} \ntgets \Varid{m}_{2})\hsarrow{\rightarrow }{\mathpunct{.}}{}\<[E]%
\\
\>[3]{}\hsindent{23}{}\<[26]%
\>[26]{}\Conid{BX}\;\Varid{m}_{1}\;\Varid{a}\;\Varid{b}\hsarrow{\rightarrow }{\mathpunct{.}}\Conid{BX}\;\Varid{m}_{2}\;\Varid{b}\;\Varid{c}\hsarrow{\rightarrow }{\mathpunct{.}}\Conid{BX}\;\Varid{n}\;\Varid{a}\;\Varid{c}{}\<[E]%
\\
\>[3]{}\Varid{compMBX'}\;(\Varid{l},\Varid{r})\;bx_{1}\;bx_{2}\mathrel{=}{}\<[E]%
\\
\>[3]{}\hsindent{7}{}\<[10]%
\>[10]{}\Conid{BX}\;{}\<[20]%
\>[20]{}(\Varid{l}\;(\Varid{mgetl}\;bx_{1}))\;{}\<[E]%
\\
\>[20]{}(\lambda \hslambda \Varid{a}\hsarrow{\rightarrow }{\mathpunct{.}}\mathbf{do}\;\{\mskip1.5mu {}\<[33]%
\>[33]{}\Varid{b}\leftarrow \Varid{l}\;(\mathbf{do}\;\{\mskip1.5mu \Varid{msetl}\;bx_{1}\;\Varid{a};\Varid{mgetr}\;bx_{1}\mskip1.5mu\});{}\<[E]%
\\
\>[33]{}\Varid{r}\;(\Varid{msetl}\;bx_{2}\;\Varid{b})\mskip1.5mu\})\;{}\<[E]%
\\
\>[20]{}(\Varid{r}\;(\Varid{mgetr}\;bx_{2}))\;{}\<[E]%
\\
\>[20]{}(\lambda \hslambda \Varid{c}\hsarrow{\rightarrow }{\mathpunct{.}}\mathbf{do}\;\{\mskip1.5mu {}\<[33]%
\>[33]{}\Varid{b}\leftarrow \Varid{r}\;(\mathbf{do}\;\{\mskip1.5mu \Varid{msetr}\;bx_{2}\;\Varid{c};\Varid{mgetl}\;bx_{2}\mskip1.5mu\});{}\<[E]%
\\
\>[33]{}\Varid{l}\;(\Varid{msetr}\;bx_{1}\;\Varid{b})\mskip1.5mu\}){}\<[E]%
\ColumnHook
\end{hscode}\resethooks

\subsection{Isomorphisms}
\begin{hscode}\SaveRestoreHook
\column{B}{@{}>{\hspre}l<{\hspost}@{}}%
\column{3}{@{}>{\hspre}l<{\hspost}@{}}%
\column{E}{@{}>{\hspre}l<{\hspost}@{}}%
\>[3]{}\mathbf{module}\;\Conid{Iso}\;\mathbf{where}{}\<[E]%
\ColumnHook
\end{hscode}\resethooks
Some isomorphisms.
\begin{hscode}\SaveRestoreHook
\column{B}{@{}>{\hspre}l<{\hspost}@{}}%
\column{3}{@{}>{\hspre}l<{\hspost}@{}}%
\column{E}{@{}>{\hspre}l<{\hspost}@{}}%
\>[3]{}\mathbf{data}\;\Conid{Iso}\;\Varid{a}\;\Varid{b}\mathrel{=}\Conid{Iso}\;\{\mskip1.5mu \Varid{to}\mathbin{::}\Varid{a}\hsarrow{\rightarrow }{\mathpunct{.}}\Varid{b},\Varid{from}\mathbin{::}\Varid{b}\hsarrow{\rightarrow }{\mathpunct{.}}\Varid{a}\mskip1.5mu\}{}\<[E]%
\\[\blanklineskip]%
\>[3]{}\Varid{assocIso}\mathbin{::}\Conid{Iso}\;((\Varid{a},\Varid{b}),\Varid{c})\;(\Varid{a},(\Varid{b},\Varid{c})){}\<[E]%
\\
\>[3]{}\Varid{assocIso}\mathrel{=}\Conid{Iso}\;(\lambda \hslambda ((\Varid{a},\Varid{b}),\Varid{c})\hsarrow{\rightarrow }{\mathpunct{.}}(\Varid{a},(\Varid{b},\Varid{c})))\;(\lambda \hslambda (\Varid{a},(\Varid{b},\Varid{c}))\hsarrow{\rightarrow }{\mathpunct{.}}((\Varid{a},\Varid{b}),\Varid{c})){}\<[E]%
\\[\blanklineskip]%
\>[3]{}\Varid{swapIso}\mathbin{::}\Conid{Iso}\;(\Varid{a},\Varid{b})\;(\Varid{b},\Varid{a}){}\<[E]%
\\
\>[3]{}\Varid{swapIso}\mathrel{=}\Conid{Iso}\;(\lambda \hslambda (\Varid{a},\Varid{b})\hsarrow{\rightarrow }{\mathpunct{.}}(\Varid{b},\Varid{a}))\;(\lambda \hslambda (\Varid{a},\Varid{b})\hsarrow{\rightarrow }{\mathpunct{.}}(\Varid{b},\Varid{a})){}\<[E]%
\\[\blanklineskip]%
\>[3]{}\Varid{unitlIso}\mathbin{::}\Conid{Iso}\;\Varid{a}\;((),\Varid{a}){}\<[E]%
\\
\>[3]{}\Varid{unitlIso}\mathrel{=}\Conid{Iso}\;(\lambda \hslambda \Varid{a}\hsarrow{\rightarrow }{\mathpunct{.}}((),\Varid{a}))\;(\lambda \hslambda ((),\Varid{a})\hsarrow{\rightarrow }{\mathpunct{.}}\Varid{a}){}\<[E]%
\\[\blanklineskip]%
\>[3]{}\Varid{unitrIso}\mathbin{::}\Conid{Iso}\;\Varid{a}\;(\Varid{a},()){}\<[E]%
\\
\>[3]{}\Varid{unitrIso}\mathrel{=}\Conid{Iso}\;(\lambda \hslambda \Varid{a}\hsarrow{\rightarrow }{\mathpunct{.}}(\Varid{a},()))\;(\lambda \hslambda (\Varid{a},())\hsarrow{\rightarrow }{\mathpunct{.}}\Varid{a}){}\<[E]%
\ColumnHook
\end{hscode}\resethooks

\subsection{Lenses}
\begin{hscode}\SaveRestoreHook
\column{B}{@{}>{\hspre}l<{\hspost}@{}}%
\column{3}{@{}>{\hspre}l<{\hspost}@{}}%
\column{26}{@{}>{\hspre}l<{\hspost}@{}}%
\column{27}{@{}>{\hspre}l<{\hspost}@{}}%
\column{E}{@{}>{\hspre}l<{\hspost}@{}}%
\>[3]{}\mathbf{module}\;\Conid{Lens}\;\mathbf{where}{}\<[E]%
\\[\blanklineskip]%
\>[3]{}\mathbf{data}\;\Conid{Lens}\;\Varid{a}\;\Varid{b}\mathrel{=}\Conid{Lens}\;\{\mskip1.5mu {}\<[27]%
\>[27]{}\Varid{view}\mathbin{::}\Varid{a}\hsarrow{\rightarrow }{\mathpunct{.}}\Varid{b},{}\<[E]%
\\
\>[27]{}\Varid{update}\mathbin{::}\Varid{a}\hsarrow{\rightarrow }{\mathpunct{.}}\Varid{b}\hsarrow{\rightarrow }{\mathpunct{.}}\Varid{a},{}\<[E]%
\\
\>[27]{}\Varid{create}\mathbin{::}\Varid{b}\hsarrow{\rightarrow }{\mathpunct{.}}\Varid{a}\mskip1.5mu\}{}\<[E]%
\\[\blanklineskip]%
\>[3]{}\Varid{idLens}\mathbin{::}\Conid{Lens}\;\Varid{a}\;\Varid{a}{}\<[E]%
\\
\>[3]{}\Varid{idLens}\mathrel{=}\Conid{Lens}\;(\lambda \hslambda \Varid{a}\hsarrow{\rightarrow }{\mathpunct{.}}\Varid{a})\;(\mathbin{\char92 \char95 }\Varid{a}\hsarrow{\rightarrow }{\mathpunct{.}}\Varid{a})\;(\lambda \hslambda \Varid{a}\hsarrow{\rightarrow }{\mathpunct{.}}\Varid{a}){}\<[E]%
\\[\blanklineskip]%
\>[3]{}\Varid{compLens}\mathbin{::}\Conid{Lens}\;\Varid{b}\;\Varid{c}\hsarrow{\rightarrow }{\mathpunct{.}}\Conid{Lens}\;\Varid{a}\;\Varid{b}\hsarrow{\rightarrow }{\mathpunct{.}}\Conid{Lens}\;\Varid{a}\;\Varid{c}{}\<[E]%
\\
\>[3]{}\Varid{compLens}\;\Varid{l}_{2}\;\Varid{l}_{1}\mathrel{=}\Conid{Lens}\;{}\<[26]%
\>[26]{}(\Varid{view}\;\Varid{l}_{2}\hsdot{\cdot }{.}\Varid{view}\;\Varid{l}_{1})\;{}\<[E]%
\\
\>[26]{}(\lambda \hslambda \Varid{a}\;\Varid{c}\hsarrow{\rightarrow }{\mathpunct{.}}\Varid{update}\;\Varid{l}_{1}\;\Varid{a}\;(\Varid{update}\;\Varid{l}_{2}\;(\Varid{view}\;\Varid{l}_{1}\;\Varid{a})\;\Varid{c}))\;{}\<[E]%
\\
\>[26]{}(\Varid{create}\;\Varid{l}_{1}\hsdot{\cdot }{.}\Varid{create}\;\Varid{l}_{2}){}\<[E]%
\\[\blanklineskip]%
\>[3]{}\Varid{fstLens}\mathbin{::}\Varid{b}\hsarrow{\rightarrow }{\mathpunct{.}}\Conid{Lens}\;(\Varid{a},\Varid{b})\;\Varid{a}{}\<[E]%
\\
\>[3]{}\Varid{fstLens}\;\Varid{b}\mathrel{=}\Conid{Lens}\;\Varid{fst}\;(\lambda \hslambda (\Varid{a},\Varid{b})\;\Varid{a'}\hsarrow{\rightarrow }{\mathpunct{.}}(\Varid{a'},\Varid{b}))\;(\lambda \hslambda \Varid{a}\hsarrow{\rightarrow }{\mathpunct{.}}(\Varid{a},\Varid{b})){}\<[E]%
\\[\blanklineskip]%
\>[3]{}\Varid{sndLens}\mathbin{::}\Varid{a}\hsarrow{\rightarrow }{\mathpunct{.}}\Conid{Lens}\;(\Varid{a},\Varid{b})\;\Varid{b}{}\<[E]%
\\
\>[3]{}\Varid{sndLens}\;\Varid{a}\mathrel{=}\Conid{Lens}\;\Varid{snd}\;(\lambda \hslambda (\Varid{a},\Varid{b})\;\Varid{b'}\hsarrow{\rightarrow }{\mathpunct{.}}(\Varid{a},\Varid{b'}))\;(\lambda \hslambda \Varid{b}\hsarrow{\rightarrow }{\mathpunct{.}}(\Varid{a},\Varid{b})){}\<[E]%
\ColumnHook
\end{hscode}\resethooks

\subsection{Monadic Lenses}
\begin{hscode}\SaveRestoreHook
\column{B}{@{}>{\hspre}l<{\hspost}@{}}%
\column{3}{@{}>{\hspre}l<{\hspost}@{}}%
\column{25}{@{}>{\hspre}l<{\hspost}@{}}%
\column{28}{@{}>{\hspre}l<{\hspost}@{}}%
\column{31}{@{}>{\hspre}l<{\hspost}@{}}%
\column{43}{@{}>{\hspre}l<{\hspost}@{}}%
\column{E}{@{}>{\hspre}l<{\hspost}@{}}%
\>[3]{}\mathbf{module}\;\Conid{MLens}\;\mathbf{where}{}\<[E]%
\\
\>[3]{}\mathbf{import}\;\Conid{Lens}{}\<[E]%
\\[\blanklineskip]%
\>[3]{}\mathbf{data}\;\Conid{MLens}\;\Varid{m}\;\Varid{a}\;\Varid{b}\mathrel{=}\Conid{MLens}\;\{\mskip1.5mu {}\<[31]%
\>[31]{}\Varid{mview}\mathbin{::}\Varid{a}\hsarrow{\rightarrow }{\mathpunct{.}}\Varid{b},{}\<[E]%
\\
\>[31]{}\Varid{mupdate}\mathbin{::}\Varid{a}\hsarrow{\rightarrow }{\mathpunct{.}}\Varid{b}\hsarrow{\rightarrow }{\mathpunct{.}}\Varid{m}\;\Varid{a},{}\<[E]%
\\
\>[31]{}\Varid{mcreate}\mathbin{::}\Varid{b}\hsarrow{\rightarrow }{\mathpunct{.}}\Varid{m}\;\Varid{a}\mskip1.5mu\}{}\<[E]%
\\[\blanklineskip]%
\>[3]{}\Varid{idMLens}\mathbin{::}\Conid{Monad}\;\Varid{m}\Rightarrow \Conid{MLens}\;\Varid{m}\;\Varid{a}\;\Varid{a}{}\<[E]%
\\
\>[3]{}\Varid{idMLens}\mathrel{=}\Conid{MLens}\;(\lambda \hslambda \Varid{a}\hsarrow{\rightarrow }{\mathpunct{.}}\Varid{a})\;(\mathbin{\char92 \char95 }\Varid{a}\hsarrow{\rightarrow }{\mathpunct{.}}\Varid{return}\;\Varid{a})\;(\lambda \hslambda \Varid{a}\hsarrow{\rightarrow }{\mathpunct{.}}\Varid{return}\;\Varid{a}){}\<[E]%
\\[\blanklineskip]%
\>[3]{}(\mathbin{;})\mathbin{::}\Conid{Monad}\;\Varid{m}\Rightarrow \Conid{MLens}\;\Varid{m}\;\Varid{b}\;\Varid{c}\hsarrow{\rightarrow }{\mathpunct{.}}\Conid{MLens}\;\Varid{m}\;\Varid{a}\;\Varid{b}\hsarrow{\rightarrow }{\mathpunct{.}}\Conid{MLens}\;\Varid{m}\;\Varid{a}\;\Varid{c}{}\<[E]%
\\
\>[3]{}(\mathbin{;})\;\Varid{l}_{2}\;\Varid{l}_{1}\mathrel{=}\Conid{MLens}\;{}\<[28]%
\>[28]{}(\Varid{mview}\;\Varid{l}_{2}\hsdot{\cdot }{.}\Varid{mview}\;\Varid{l}_{1})\;{}\<[E]%
\\
\>[28]{}(\lambda \hslambda \Varid{a}\;\Varid{c}\hsarrow{\rightarrow }{\mathpunct{.}}\mathbf{do}\;\{\mskip1.5mu {}\<[43]%
\>[43]{}\Varid{b}\leftarrow \Varid{mupdate}\;\Varid{l}_{2}\;(\Varid{mview}\;\Varid{l}_{1}\;\Varid{a})\;\Varid{c};\Varid{mupdate}\;\Varid{l}_{1}\;\Varid{a}\;\Varid{b}\mskip1.5mu\})\;{}\<[E]%
\\
\>[28]{}(\lambda \hslambda \Varid{c}\hsarrow{\rightarrow }{\mathpunct{.}}\mathbf{do}\;\{\mskip1.5mu \Varid{b}\leftarrow \Varid{mcreate}\;\Varid{l}_{2}\;\Varid{c};\Varid{mcreate}\;\Varid{l}_{1}\;\Varid{b}\mskip1.5mu\}){}\<[E]%
\\[\blanklineskip]%
\>[3]{}\Varid{lens2MLens}\mathbin{::}\Conid{Monad}\;\Varid{m}\Rightarrow \Conid{Lens}\;\Varid{a}\;\Varid{b}\hsarrow{\rightarrow }{\mathpunct{.}}\Conid{MLens}\;\Varid{m}\;\Varid{a}\;\Varid{b}{}\<[E]%
\\
\>[3]{}\Varid{lens2MLens}\;\Varid{l}\mathrel{=}\Conid{MLens}\;{}\<[25]%
\>[25]{}(\Varid{view}\;\Varid{l})\;{}\<[E]%
\\
\>[25]{}(\lambda \hslambda \Varid{a}\;\Varid{b}\hsarrow{\rightarrow }{\mathpunct{.}}\Varid{return}\;(\Varid{update}\;\Varid{l}\;\Varid{a}\;\Varid{b}))\;{}\<[E]%
\\
\>[25]{}(\Varid{return}\hsdot{\cdot }{.}\Varid{create}\;\Varid{l}){}\<[E]%
\ColumnHook
\end{hscode}\resethooks

\subsection{Relational BX}
\begin{hscode}\SaveRestoreHook
\column{B}{@{}>{\hspre}l<{\hspost}@{}}%
\column{3}{@{}>{\hspre}l<{\hspost}@{}}%
\column{E}{@{}>{\hspre}l<{\hspost}@{}}%
\>[3]{}\mathbf{module}\;\Conid{RelBX}\;\mathbf{where}{}\<[E]%
\\
\>[3]{}\mathbf{import}\;\Conid{Lens}{}\<[E]%
\ColumnHook
\end{hscode}\resethooks
pointed types that have a distinguished element
\begin{hscode}\SaveRestoreHook
\column{B}{@{}>{\hspre}l<{\hspost}@{}}%
\column{3}{@{}>{\hspre}l<{\hspost}@{}}%
\column{9}{@{}>{\hspre}l<{\hspost}@{}}%
\column{E}{@{}>{\hspre}l<{\hspost}@{}}%
\>[3]{}\mathbf{class}\;\Conid{Pointed}\;\Varid{a}\;\mathbf{where}{}\<[E]%
\\
\>[3]{}\hsindent{6}{}\<[9]%
\>[9]{}\Varid{point}\mathbin{::}\Varid{a}{}\<[E]%
\ColumnHook
\end{hscode}\resethooks
Relational bx
\begin{hscode}\SaveRestoreHook
\column{B}{@{}>{\hspre}l<{\hspost}@{}}%
\column{3}{@{}>{\hspre}l<{\hspost}@{}}%
\column{29}{@{}>{\hspre}l<{\hspost}@{}}%
\column{E}{@{}>{\hspre}l<{\hspost}@{}}%
\>[3]{}\mathbf{data}\;\Conid{RelBX}\;\Varid{a}\;\Varid{b}\mathrel{=}\Conid{RelBX}\;\{\mskip1.5mu {}\<[29]%
\>[29]{}\Varid{consistent}\mathbin{::}\Varid{a}\hsarrow{\rightarrow }{\mathpunct{.}}\Varid{b}\hsarrow{\rightarrow }{\mathpunct{.}}\Conid{Bool},{}\<[E]%
\\
\>[29]{}\Varid{fwd}\mathbin{::}\Varid{a}\hsarrow{\rightarrow }{\mathpunct{.}}\Varid{b}\hsarrow{\rightarrow }{\mathpunct{.}}\Varid{b},{}\<[E]%
\\
\>[29]{}\Varid{bwd}\mathbin{::}\Varid{a}\hsarrow{\rightarrow }{\mathpunct{.}}\Varid{b}\hsarrow{\rightarrow }{\mathpunct{.}}\Varid{a}\mskip1.5mu\}{}\<[E]%
\ColumnHook
\end{hscode}\resethooks
Lenses from relational bx 
\begin{hscode}\SaveRestoreHook
\column{B}{@{}>{\hspre}l<{\hspost}@{}}%
\column{3}{@{}>{\hspre}l<{\hspost}@{}}%
\column{23}{@{}>{\hspre}l<{\hspost}@{}}%
\column{E}{@{}>{\hspre}l<{\hspost}@{}}%
\>[3]{}\Varid{lens2rel}\mathbin{::}\Conid{Eq}\;\Varid{b}\Rightarrow \Conid{Lens}\;\Varid{a}\;\Varid{b}\hsarrow{\rightarrow }{\mathpunct{.}}\Conid{RelBX}\;\Varid{a}\;\Varid{b}{}\<[E]%
\\
\>[3]{}\Varid{lens2rel}\;\Varid{l}\mathrel{=}\Conid{RelBX}\;{}\<[23]%
\>[23]{}(\lambda \hslambda \Varid{a}\;\Varid{b}\hsarrow{\rightarrow }{\mathpunct{.}}\Varid{view}\;\Varid{l}\;\Varid{a}\equals\Varid{b})\;{}\<[E]%
\\
\>[23]{}(\lambda \hslambda \Varid{a}\;\Varid{b}\hsarrow{\rightarrow }{\mathpunct{.}}\Varid{view}\;\Varid{l}\;\Varid{a})\;{}\<[E]%
\\
\>[23]{}(\lambda \hslambda \Varid{a}\;\Varid{b}\hsarrow{\rightarrow }{\mathpunct{.}}\Varid{update}\;\Varid{l}\;\Varid{a}\;\Varid{b}){}\<[E]%
\ColumnHook
\end{hscode}\resethooks
Relational BX form spans of lenses provided types pointed
\begin{hscode}\SaveRestoreHook
\column{B}{@{}>{\hspre}l<{\hspost}@{}}%
\column{3}{@{}>{\hspre}l<{\hspost}@{}}%
\column{24}{@{}>{\hspre}l<{\hspost}@{}}%
\column{30}{@{}>{\hspre}l<{\hspost}@{}}%
\column{E}{@{}>{\hspre}l<{\hspost}@{}}%
\>[3]{}\Varid{rel2lensSpan}\mathbin{::}(\Conid{Pointed}\;\Varid{a},\Conid{Pointed}\;\Varid{b})\Rightarrow \Conid{RelBX}\;\Varid{a}\;\Varid{b}\hsarrow{\rightarrow }{\mathpunct{.}}(\Conid{Lens}\;(\Varid{a},\Varid{b})\;\Varid{a},\Conid{Lens}\;(\Varid{a},\Varid{b})\;\Varid{b}){}\<[E]%
\\
\>[3]{}\Varid{rel2lensSpan}\;bx\mathrel{=}({}\<[24]%
\>[24]{}\Conid{Lens}\;{}\<[30]%
\>[30]{}\Varid{fst}\;{}\<[E]%
\\
\>[30]{}(\lambda \hslambda (\anonymous ,\Varid{b})\;\Varid{a}\hsarrow{\rightarrow }{\mathpunct{.}}(\Varid{a},\Varid{fwd}\;bx\;\Varid{a}\;\Varid{b}))\;{}\<[E]%
\\
\>[30]{}(\lambda \hslambda \Varid{a}\hsarrow{\rightarrow }{\mathpunct{.}}(\Varid{point},\Varid{point})),{}\<[E]%
\\
\>[24]{}\Conid{Lens}\;{}\<[30]%
\>[30]{}\Varid{snd}\;{}\<[E]%
\\
\>[30]{}(\lambda \hslambda (\Varid{a},\anonymous )\;\Varid{b}\hsarrow{\rightarrow }{\mathpunct{.}}(\Varid{bwd}\;bx\;\Varid{a}\;\Varid{b},\Varid{b}))\;{}\<[E]%
\\
\>[30]{}(\lambda \hslambda \Varid{b}\hsarrow{\rightarrow }{\mathpunct{.}}(\Varid{point},\Varid{point}))){}\<[E]%
\ColumnHook
\end{hscode}\resethooks

\subsection{Symmetric Lenses}
\begin{hscode}\SaveRestoreHook
\column{B}{@{}>{\hspre}l<{\hspost}@{}}%
\column{3}{@{}>{\hspre}l<{\hspost}@{}}%
\column{E}{@{}>{\hspre}l<{\hspost}@{}}%
\>[3]{}\mathbf{module}\;\Conid{SLens}\;\mathbf{where}{}\<[E]%
\\
\>[3]{}\mathbf{import}\;\Conid{Lens}{}\<[E]%
\\
\>[3]{}\mathbf{import}\;\Conid{RelBX}{}\<[E]%
\ColumnHook
\end{hscode}\resethooks
Symmetric lenses (with explicit points) 
 \begin{hscode}\SaveRestoreHook
\column{B}{@{}>{\hspre}l<{\hspost}@{}}%
\column{3}{@{}>{\hspre}l<{\hspost}@{}}%
\column{31}{@{}>{\hspre}l<{\hspost}@{}}%
\column{E}{@{}>{\hspre}l<{\hspost}@{}}%
\>[3]{}\mathbf{data}\;\Conid{SLens}\;\Varid{c}\;\Varid{a}\;\Varid{b}\mathrel{=}\Conid{SLens}\;\{\mskip1.5mu {}\<[31]%
\>[31]{}\Varid{putr}\mathbin{::}(\Varid{a},\Varid{c})\hsarrow{\rightarrow }{\mathpunct{.}}(\Varid{b},\Varid{c}),{}\<[E]%
\\
\>[31]{}\Varid{putl}\mathbin{::}(\Varid{b},\Varid{c})\hsarrow{\rightarrow }{\mathpunct{.}}(\Varid{a},\Varid{c}),{}\<[E]%
\\
\>[31]{}\Varid{missing}\mathbin{::}\Varid{c}\mskip1.5mu\}{}\<[E]%
\ColumnHook
\end{hscode}\resethooks
Dual
\begin{hscode}\SaveRestoreHook
\column{B}{@{}>{\hspre}l<{\hspost}@{}}%
\column{3}{@{}>{\hspre}l<{\hspost}@{}}%
\column{E}{@{}>{\hspre}l<{\hspost}@{}}%
\>[3]{}\Varid{dualSL}\;\Varid{sl}\mathrel{=}\Conid{SLens}\;(\Varid{putl}\;\Varid{sl})\;(\Varid{putr}\;\Varid{sl})\;\Varid{missing}{}\<[E]%
\ColumnHook
\end{hscode}\resethooks
From asymmetric lenses
\begin{hscode}\SaveRestoreHook
\column{B}{@{}>{\hspre}l<{\hspost}@{}}%
\column{3}{@{}>{\hspre}l<{\hspost}@{}}%
\column{5}{@{}>{\hspre}l<{\hspost}@{}}%
\column{12}{@{}>{\hspre}l<{\hspost}@{}}%
\column{26}{@{}>{\hspre}c<{\hspost}@{}}%
\column{26E}{@{}l@{}}%
\column{29}{@{}>{\hspre}l<{\hspost}@{}}%
\column{39}{@{}>{\hspre}l<{\hspost}@{}}%
\column{42}{@{}>{\hspre}l<{\hspost}@{}}%
\column{E}{@{}>{\hspre}l<{\hspost}@{}}%
\>[3]{}\Varid{lens2symlens}\mathbin{::}\Conid{Lens}\;\Varid{a}\;\Varid{b}\hsarrow{\rightarrow }{\mathpunct{.}}\Conid{SLens}\;(\Conid{Maybe}\;\Varid{a})\;\Varid{a}\;\Varid{b}{}\<[E]%
\\
\>[3]{}\Varid{lens2symlens}\;\Varid{l}\mathrel{=}\Conid{SLens}\;\Varid{putr}\;\Varid{putl}\;\Conid{Nothing}{}\<[E]%
\\
\>[3]{}\hsindent{2}{}\<[5]%
\>[5]{}\mathbf{where}\;{}\<[12]%
\>[12]{}\Varid{putr}\;(\Varid{a},\anonymous ){}\<[26]%
\>[26]{}\mathrel{=}{}\<[26E]%
\>[29]{}(\Varid{view}\;\Varid{l}\;\Varid{a},\Conid{Just}\;\Varid{a}){}\<[E]%
\\
\>[12]{}\Varid{putl}\;(\Varid{b'},\Varid{ma}){}\<[26]%
\>[26]{}\mathrel{=}{}\<[26E]%
\>[29]{}\mathbf{let}\;\Varid{a'}\mathrel{=}{}\<[39]%
\>[39]{}\mathbf{case}\;\Varid{ma}\;\mathbf{of}{}\<[E]%
\\
\>[39]{}\hsindent{3}{}\<[42]%
\>[42]{}\Conid{Nothing}\hsarrow{\rightarrow }{\mathpunct{.}}\Varid{create}\;\Varid{l}\;\Varid{b'}{}\<[E]%
\\
\>[39]{}\hsindent{3}{}\<[42]%
\>[42]{}\Varid{a}\hsarrow{\rightarrow }{\mathpunct{.}}\Varid{update}\;\Varid{l}\;\Varid{a'}\;\Varid{b'}{}\<[E]%
\\
\>[29]{}\mathbf{in}\;(\Varid{create}\;\Varid{l}\;\Varid{b'},\Conid{Just}\;\Varid{a'}){}\<[E]%
\ColumnHook
\end{hscode}\resethooks
From relational bx
\begin{hscode}\SaveRestoreHook
\column{B}{@{}>{\hspre}l<{\hspost}@{}}%
\column{3}{@{}>{\hspre}l<{\hspost}@{}}%
\column{19}{@{}>{\hspre}l<{\hspost}@{}}%
\column{27}{@{}>{\hspre}l<{\hspost}@{}}%
\column{44}{@{}>{\hspre}l<{\hspost}@{}}%
\column{E}{@{}>{\hspre}l<{\hspost}@{}}%
\>[3]{}\Varid{rel2symlens}\mathbin{::}{}\<[19]%
\>[19]{}(\Conid{Pointed}\;\Varid{a},\Conid{Pointed}\;\Varid{b})\Rightarrow \Conid{RelBX}\;\Varid{a}\;\Varid{b}\hsarrow{\rightarrow }{\mathpunct{.}}\Conid{SLens}\;(\Varid{a},\Varid{b})\;\Varid{a}\;\Varid{b}{}\<[E]%
\\
\>[3]{}\Varid{rel2symlens}\;bx\mathrel{=}\Conid{SLens}\;{}\<[27]%
\>[27]{}(\lambda \hslambda (\Varid{a'},(\Varid{a},\Varid{b}))\hsarrow{\rightarrow }{\mathpunct{.}}{}\<[44]%
\>[44]{}\mathbf{let}\;\Varid{b'}\mathrel{=}\Varid{fwd}\;bx\;\Varid{a'}\;\Varid{b}{}\<[E]%
\\
\>[44]{}\mathbf{in}\;(\Varid{b'},(\Varid{a'},\Varid{b'})))\;{}\<[E]%
\\
\>[27]{}(\lambda \hslambda (\Varid{b'},(\Varid{a},\Varid{b}))\hsarrow{\rightarrow }{\mathpunct{.}}{}\<[44]%
\>[44]{}\mathbf{let}\;\Varid{a'}\mathrel{=}\Varid{bwd}\;bx\;\Varid{a}\;\Varid{b'}{}\<[E]%
\\
\>[44]{}\mathbf{in}\;(\Varid{a'},(\Varid{a'},\Varid{b'})))\;{}\<[E]%
\\
\>[27]{}(\Varid{point},\Varid{point}){}\<[E]%
\ColumnHook
\end{hscode}\resethooks
To asymmetric lens, on the left...
 \begin{hscode}\SaveRestoreHook
\column{B}{@{}>{\hspre}l<{\hspost}@{}}%
\column{3}{@{}>{\hspre}l<{\hspost}@{}}%
\column{5}{@{}>{\hspre}l<{\hspost}@{}}%
\column{28}{@{}>{\hspre}l<{\hspost}@{}}%
\column{E}{@{}>{\hspre}l<{\hspost}@{}}%
\>[3]{}\Varid{symlens2lensL}\mathbin{::}\Conid{SLens}\;\Varid{c}\;\Varid{a}\;\Varid{b}\hsarrow{\rightarrow }{\mathpunct{.}}\Conid{Lens}\;(\Varid{a},\Varid{b},\Varid{c})\;\Varid{a}{}\<[E]%
\\
\>[3]{}\Varid{symlens2lensL}\;\Varid{sl}\mathrel{=}\Conid{Lens}\;{}\<[28]%
\>[28]{}(\lambda \hslambda (\Varid{a},\Varid{b},\Varid{c})\hsarrow{\rightarrow }{\mathpunct{.}}\Varid{a})\;{}\<[E]%
\\
\>[28]{}(\lambda \hslambda (\anonymous ,\anonymous ,\Varid{c})\hsarrow{\rightarrow }{\mathpunct{.}}\Varid{fixup}\;\Varid{c})\;{}\<[E]%
\\
\>[28]{}(\Varid{fixup}\;(\Varid{missing}\;\Varid{sl})){}\<[E]%
\\
\>[3]{}\hsindent{2}{}\<[5]%
\>[5]{}\mathbf{where}\;\Varid{fixup}\;\Varid{c}\;\Varid{a'}\mathrel{=}\mathbf{let}\;(\Varid{b'},\Varid{c'})\mathrel{=}\Varid{putr}\;\Varid{sl}\;(\Varid{a'},\Varid{c})\;\mathbf{in}\;(\Varid{a'},\Varid{b'},\Varid{c'}){}\<[E]%
\ColumnHook
\end{hscode}\resethooks
...and on the right
 \begin{hscode}\SaveRestoreHook
\column{B}{@{}>{\hspre}l<{\hspost}@{}}%
\column{3}{@{}>{\hspre}l<{\hspost}@{}}%
\column{5}{@{}>{\hspre}l<{\hspost}@{}}%
\column{28}{@{}>{\hspre}l<{\hspost}@{}}%
\column{E}{@{}>{\hspre}l<{\hspost}@{}}%
\>[3]{}\Varid{symlens2lensR}\mathbin{::}\Conid{SLens}\;\Varid{c}\;\Varid{a}\;\Varid{b}\hsarrow{\rightarrow }{\mathpunct{.}}\Conid{Lens}\;(\Varid{a},\Varid{b},\Varid{c})\;\Varid{b}{}\<[E]%
\\
\>[3]{}\Varid{symlens2lensR}\;\Varid{sl}\mathrel{=}\Conid{Lens}\;{}\<[28]%
\>[28]{}(\lambda \hslambda (\Varid{a},\Varid{b},\Varid{c})\hsarrow{\rightarrow }{\mathpunct{.}}\Varid{b})\;{}\<[E]%
\\
\>[28]{}(\lambda \hslambda (\anonymous ,\anonymous ,\Varid{c})\hsarrow{\rightarrow }{\mathpunct{.}}\Varid{fixup}\;\Varid{c})\;{}\<[E]%
\\
\>[28]{}(\Varid{fixup}\;(\Varid{missing}\;\Varid{sl})){}\<[E]%
\\
\>[3]{}\hsindent{2}{}\<[5]%
\>[5]{}\mathbf{where}\;\Varid{fixup}\;\Varid{c}\;\Varid{b'}\mathrel{=}\mathbf{let}\;(\Varid{a'},\Varid{c'})\mathrel{=}\Varid{putl}\;\Varid{sl}\;(\Varid{b'},\Varid{c})\;\mathbf{in}\;(\Varid{a'},\Varid{b'},\Varid{c'}){}\<[E]%
\ColumnHook
\end{hscode}\resethooks
Spans and cospans: used to simplify some definitions.
\begin{hscode}\SaveRestoreHook
\column{B}{@{}>{\hspre}l<{\hspost}@{}}%
\column{3}{@{}>{\hspre}l<{\hspost}@{}}%
\column{E}{@{}>{\hspre}l<{\hspost}@{}}%
\>[3]{}\mathbf{type}\;\Conid{Span}\;\Varid{c}\;\Varid{y1}\;\Varid{x}\;\Varid{y2}\mathrel{=}(\Varid{c}\;\Varid{x}\;\Varid{y1},\Varid{c}\;\Varid{x}\;\Varid{y2}){}\<[E]%
\\
\>[3]{}\mathbf{type}\;\Conid{Cospan}\;\Varid{c}\;\Varid{y1}\;\Varid{z}\;\Varid{y2}\mathrel{=}(\Varid{c}\;\Varid{y1}\;\Varid{z},\Varid{c}\;\Varid{y2}\;\Varid{z}){}\<[E]%
\ColumnHook
\end{hscode}\resethooks
To a span
\begin{hscode}\SaveRestoreHook
\column{B}{@{}>{\hspre}l<{\hspost}@{}}%
\column{3}{@{}>{\hspre}l<{\hspost}@{}}%
\column{E}{@{}>{\hspre}l<{\hspost}@{}}%
\>[3]{}\Varid{symlens2lensSpan}\mathbin{::}\Conid{SLens}\;\Varid{c}\;\Varid{a}\;\Varid{b}\hsarrow{\rightarrow }{\mathpunct{.}}\Conid{Span}\;\Conid{Lens}\;\Varid{a}\;(\Varid{a},\Varid{b},\Varid{c})\;\Varid{b}{}\<[E]%
\\
\>[3]{}\Varid{symlens2lensSpan}\;\Varid{sl}\mathrel{=}(\Varid{symlens2lensL}\;\Varid{sl},\Varid{symlens2lensR}\;\Varid{sl}){}\<[E]%
\ColumnHook
\end{hscode}\resethooks
From a span
\begin{hscode}\SaveRestoreHook
\column{B}{@{}>{\hspre}l<{\hspost}@{}}%
\column{3}{@{}>{\hspre}l<{\hspost}@{}}%
\column{37}{@{}>{\hspre}l<{\hspost}@{}}%
\column{39}{@{}>{\hspre}l<{\hspost}@{}}%
\column{60}{@{}>{\hspre}l<{\hspost}@{}}%
\column{E}{@{}>{\hspre}l<{\hspost}@{}}%
\>[3]{}\Varid{lensSpan2symlens}\mathbin{::}\Conid{Span}\;\Conid{Lens}\;\Varid{a}\;\Varid{c}\;\Varid{b}\hsarrow{\rightarrow }{\mathpunct{.}}\Conid{SLens}\;(\Conid{Maybe}\;\Varid{c})\;\Varid{a}\;\Varid{b}{}\<[E]%
\\
\>[3]{}\Varid{lensSpan2symlens}\;(\Varid{l}_{1},\Varid{l}_{2})\mathrel{=}\Conid{SLens}\;{}\<[37]%
\>[37]{}(\lambda \hslambda (\Varid{a},\Varid{mc})\hsarrow{\rightarrow }{\mathpunct{.}}{}\<[E]%
\\
\>[37]{}\hsindent{2}{}\<[39]%
\>[39]{}\mathbf{let}\;\Varid{c'}\mathrel{=}\mathbf{case}\;\Varid{mc}\;\mathbf{of}\;{}\<[60]%
\>[60]{}\Conid{Nothing}\hsarrow{\rightarrow }{\mathpunct{.}}\Varid{create}\;\Varid{l}_{1}\;\Varid{a}{}\<[E]%
\\
\>[60]{}\Conid{Just}\;\Varid{c}\hsarrow{\rightarrow }{\mathpunct{.}}\Varid{update}\;\Varid{l}_{1}\;\Varid{c}\;\Varid{a}{}\<[E]%
\\
\>[37]{}\hsindent{2}{}\<[39]%
\>[39]{}\mathbf{in}\;(\Varid{view}\;\Varid{l}_{2}\;\Varid{c'},\Conid{Just}\;\Varid{c'}))\;{}\<[E]%
\\
\>[37]{}(\lambda \hslambda (\Varid{b},\Varid{mc})\hsarrow{\rightarrow }{\mathpunct{.}}{}\<[E]%
\\
\>[37]{}\hsindent{2}{}\<[39]%
\>[39]{}\mathbf{let}\;\Varid{c'}\mathrel{=}\mathbf{case}\;\Varid{mc}\;\mathbf{of}\;{}\<[60]%
\>[60]{}\Conid{Nothing}\hsarrow{\rightarrow }{\mathpunct{.}}\Varid{create}\;\Varid{l}_{2}\;\Varid{b}{}\<[E]%
\\
\>[60]{}\Conid{Just}\;\Varid{c}\hsarrow{\rightarrow }{\mathpunct{.}}\Varid{update}\;\Varid{l}_{2}\;\Varid{c}\;\Varid{b}{}\<[E]%
\\
\>[37]{}\hsindent{2}{}\<[39]%
\>[39]{}\mathbf{in}\;(\Varid{view}\;\Varid{l}_{1}\;\Varid{c'},\Conid{Just}\;\Varid{c'}))\;{}\<[E]%
\\
\>[37]{}\Conid{Nothing}{}\<[E]%
\ColumnHook
\end{hscode}\resethooks

\subsection{Monadic Symmetric Lenses}
\begin{hscode}\SaveRestoreHook
\column{B}{@{}>{\hspre}l<{\hspost}@{}}%
\column{3}{@{}>{\hspre}l<{\hspost}@{}}%
\column{E}{@{}>{\hspre}l<{\hspost}@{}}%
\>[3]{}\mathbf{module}\;\Conid{SMLens}\;\mathbf{where}{}\<[E]%
\\
\>[3]{}\mathbf{import}\;\Conid{MLens}{}\<[E]%
\ColumnHook
\end{hscode}\resethooks
Symmetric lenses (with explicit points) 
 \begin{hscode}\SaveRestoreHook
\column{B}{@{}>{\hspre}l<{\hspost}@{}}%
\column{3}{@{}>{\hspre}l<{\hspost}@{}}%
\column{35}{@{}>{\hspre}l<{\hspost}@{}}%
\column{E}{@{}>{\hspre}l<{\hspost}@{}}%
\>[3]{}\mathbf{data}\;\Conid{SMLens}\;\Varid{m}\;\Varid{c}\;\Varid{a}\;\Varid{b}\mathrel{=}\Conid{SMLens}\;\{\mskip1.5mu {}\<[35]%
\>[35]{}\Varid{mputr}\mathbin{::}(\Varid{a},\Varid{c})\hsarrow{\rightarrow }{\mathpunct{.}}\Varid{m}\;(\Varid{b},\Varid{c}),{}\<[E]%
\\
\>[35]{}\Varid{mputl}\mathbin{::}(\Varid{b},\Varid{c})\hsarrow{\rightarrow }{\mathpunct{.}}\Varid{m}\;(\Varid{a},\Varid{c}),{}\<[E]%
\\
\>[35]{}\Varid{mmissing}\mathbin{::}\Varid{c}\mskip1.5mu\}{}\<[E]%
\ColumnHook
\end{hscode}\resethooks
Dual
\begin{hscode}\SaveRestoreHook
\column{B}{@{}>{\hspre}l<{\hspost}@{}}%
\column{3}{@{}>{\hspre}l<{\hspost}@{}}%
\column{E}{@{}>{\hspre}l<{\hspost}@{}}%
\>[3]{}\Varid{dualSL}\;\Varid{sl}\mathrel{=}\Conid{SMLens}\;(\Varid{mputl}\;\Varid{sl})\;(\Varid{mputr}\;\Varid{sl})\;\Varid{mmissing}{}\<[E]%
\ColumnHook
\end{hscode}\resethooks
To asymmetric MLens, on the left...
 \begin{hscode}\SaveRestoreHook
\column{B}{@{}>{\hspre}l<{\hspost}@{}}%
\column{3}{@{}>{\hspre}l<{\hspost}@{}}%
\column{5}{@{}>{\hspre}l<{\hspost}@{}}%
\column{23}{@{}>{\hspre}l<{\hspost}@{}}%
\column{30}{@{}>{\hspre}l<{\hspost}@{}}%
\column{31}{@{}>{\hspre}l<{\hspost}@{}}%
\column{E}{@{}>{\hspre}l<{\hspost}@{}}%
\>[3]{}\Varid{symMLens2MLensL}\mathbin{::}{}\<[23]%
\>[23]{}\Conid{Monad}\;\Varid{m}\Rightarrow \Conid{SMLens}\;\Varid{m}\;\Varid{c}\;\Varid{a}\;\Varid{b}\hsarrow{\rightarrow }{\mathpunct{.}}\Conid{MLens}\;\Varid{m}\;(\Varid{a},\Varid{b},\Varid{c})\;\Varid{a}{}\<[E]%
\\
\>[3]{}\Varid{symMLens2MLensL}\;\Varid{sl}\mathrel{=}\Conid{MLens}\;{}\<[31]%
\>[31]{}(\lambda \hslambda (\Varid{a},\Varid{b},\Varid{c})\hsarrow{\rightarrow }{\mathpunct{.}}\Varid{a})\;{}\<[E]%
\\
\>[31]{}(\lambda \hslambda (\anonymous ,\anonymous ,\Varid{c})\hsarrow{\rightarrow }{\mathpunct{.}}\Varid{fixup}\;\Varid{c})\;{}\<[E]%
\\
\>[31]{}(\Varid{fixup}\;(\Varid{mmissing}\;\Varid{sl})){}\<[E]%
\\
\>[3]{}\hsindent{2}{}\<[5]%
\>[5]{}\mathbf{where}\;\Varid{fixup}\;\Varid{c}\;\Varid{a'}\mathrel{=}\mathbf{do}\;\{\mskip1.5mu {}\<[30]%
\>[30]{}(\Varid{b'},\Varid{c'})\leftarrow \Varid{mputr}\;\Varid{sl}\;(\Varid{a'},\Varid{c});\Varid{return}\;(\Varid{a'},\Varid{b'},\Varid{c'})\mskip1.5mu\}{}\<[E]%
\ColumnHook
\end{hscode}\resethooks
...and on the right
 \begin{hscode}\SaveRestoreHook
\column{B}{@{}>{\hspre}l<{\hspost}@{}}%
\column{3}{@{}>{\hspre}l<{\hspost}@{}}%
\column{5}{@{}>{\hspre}l<{\hspost}@{}}%
\column{23}{@{}>{\hspre}l<{\hspost}@{}}%
\column{30}{@{}>{\hspre}l<{\hspost}@{}}%
\column{31}{@{}>{\hspre}l<{\hspost}@{}}%
\column{E}{@{}>{\hspre}l<{\hspost}@{}}%
\>[3]{}\Varid{symMLens2MLensR}\mathbin{::}{}\<[23]%
\>[23]{}\Conid{Monad}\;\Varid{m}\Rightarrow \Conid{SMLens}\;\Varid{m}\;\Varid{c}\;\Varid{a}\;\Varid{b}\hsarrow{\rightarrow }{\mathpunct{.}}\Conid{MLens}\;\Varid{m}\;(\Varid{a},\Varid{b},\Varid{c})\;\Varid{b}{}\<[E]%
\\
\>[3]{}\Varid{symMLens2MLensR}\;\Varid{sl}\mathrel{=}\Conid{MLens}\;{}\<[31]%
\>[31]{}(\lambda \hslambda (\Varid{a},\Varid{b},\Varid{c})\hsarrow{\rightarrow }{\mathpunct{.}}\Varid{b})\;{}\<[E]%
\\
\>[31]{}(\lambda \hslambda (\anonymous ,\anonymous ,\Varid{c})\hsarrow{\rightarrow }{\mathpunct{.}}\Varid{fixup}\;\Varid{c})\;{}\<[E]%
\\
\>[31]{}(\Varid{fixup}\;(\Varid{mmissing}\;\Varid{sl})){}\<[E]%
\\
\>[3]{}\hsindent{2}{}\<[5]%
\>[5]{}\mathbf{where}\;\Varid{fixup}\;\Varid{c}\;\Varid{b'}\mathrel{=}\mathbf{do}\;\{\mskip1.5mu {}\<[30]%
\>[30]{}(\Varid{a'},\Varid{c'})\leftarrow \Varid{mputl}\;\Varid{sl}\;(\Varid{b'},\Varid{c});\Varid{return}\;(\Varid{a'},\Varid{b'},\Varid{c'})\mskip1.5mu\}{}\<[E]%
\ColumnHook
\end{hscode}\resethooks
Spans and cospans: used to simplify some definitions.
\begin{hscode}\SaveRestoreHook
\column{B}{@{}>{\hspre}l<{\hspost}@{}}%
\column{3}{@{}>{\hspre}l<{\hspost}@{}}%
\column{E}{@{}>{\hspre}l<{\hspost}@{}}%
\>[3]{}\mathbf{type}\;\Conid{Span}\;\Varid{c}\;\Varid{y1}\;\Varid{x}\;\Varid{y2}\mathrel{=}(\Varid{c}\;\Varid{x}\;\Varid{y1},\Varid{c}\;\Varid{x}\;\Varid{y2}){}\<[E]%
\\
\>[3]{}\mathbf{type}\;\Conid{Cospan}\;\Varid{c}\;\Varid{y1}\;\Varid{z}\;\Varid{y2}\mathrel{=}(\Varid{c}\;\Varid{y1}\;\Varid{z},\Varid{c}\;\Varid{y2}\;\Varid{z}){}\<[E]%
\ColumnHook
\end{hscode}\resethooks
To a span
\begin{hscode}\SaveRestoreHook
\column{B}{@{}>{\hspre}l<{\hspost}@{}}%
\column{3}{@{}>{\hspre}l<{\hspost}@{}}%
\column{24}{@{}>{\hspre}l<{\hspost}@{}}%
\column{28}{@{}>{\hspre}l<{\hspost}@{}}%
\column{E}{@{}>{\hspre}l<{\hspost}@{}}%
\>[3]{}\Varid{symlens2lensSpan}\mathbin{::}{}\<[24]%
\>[24]{}\Conid{Monad}\;\Varid{m}\Rightarrow \Conid{SMLens}\;\Varid{m}\;\Varid{c}\;\Varid{a}\;\Varid{b}\hsarrow{\rightarrow }{\mathpunct{.}}\Conid{Span}\;(\Conid{MLens}\;\Varid{m})\;\Varid{a}\;(\Varid{a},\Varid{b},\Varid{c})\;\Varid{b}{}\<[E]%
\\
\>[3]{}\Varid{symlens2lensSpan}\;\Varid{sl}\mathrel{=}({}\<[28]%
\>[28]{}\Varid{symMLens2MLensL}\;\Varid{sl},\Varid{symMLens2MLensR}\;\Varid{sl}){}\<[E]%
\ColumnHook
\end{hscode}\resethooks
and from a span
\begin{hscode}\SaveRestoreHook
\column{B}{@{}>{\hspre}l<{\hspost}@{}}%
\column{3}{@{}>{\hspre}l<{\hspost}@{}}%
\column{5}{@{}>{\hspre}l<{\hspost}@{}}%
\column{15}{@{}>{\hspre}l<{\hspost}@{}}%
\column{24}{@{}>{\hspre}l<{\hspost}@{}}%
\column{28}{@{}>{\hspre}l<{\hspost}@{}}%
\column{32}{@{}>{\hspre}l<{\hspost}@{}}%
\column{50}{@{}>{\hspre}l<{\hspost}@{}}%
\column{E}{@{}>{\hspre}l<{\hspost}@{}}%
\>[3]{}\Varid{lensSpan2symlens}\mathbin{::}{}\<[24]%
\>[24]{}\Conid{Monad}\;\Varid{m}\Rightarrow \Conid{Span}\;(\Conid{MLens}\;\Varid{m})\;\Varid{a}\;\Varid{c}\;\Varid{b}\hsarrow{\rightarrow }{\mathpunct{.}}\Conid{SMLens}\;\Varid{m}\;(\Conid{Maybe}\;\Varid{c})\;\Varid{a}\;\Varid{b}{}\<[E]%
\\
\>[3]{}\Varid{lensSpan2symlens}\;(\Varid{l}_{1},\Varid{l}_{2}){}\<[E]%
\\
\>[3]{}\hsindent{2}{}\<[5]%
\>[5]{}\mathrel{=}\Conid{SMLens}\;{}\<[15]%
\>[15]{}(\lambda \hslambda (\Varid{a},\Varid{mc})\hsarrow{\rightarrow }{\mathpunct{.}}{}\<[28]%
\>[28]{}\mathbf{do}\;{}\<[32]%
\>[32]{}\Varid{c'}\leftarrow \mathbf{case}\;\Varid{mc}\;\mathbf{of}\;{}\<[50]%
\>[50]{}\Conid{Nothing}\hsarrow{\rightarrow }{\mathpunct{.}}\Varid{mcreate}\;\Varid{l}_{1}\;\Varid{a}{}\<[E]%
\\
\>[50]{}\Conid{Just}\;\Varid{c}\hsarrow{\rightarrow }{\mathpunct{.}}\Varid{mupdate}\;\Varid{l}_{1}\;\Varid{c}\;\Varid{a}{}\<[E]%
\\
\>[32]{}\Varid{return}\;(\Varid{mview}\;\Varid{l}_{2}\;\Varid{c'},\Conid{Just}\;\Varid{c'}))\;{}\<[E]%
\\
\>[15]{}(\lambda \hslambda (\Varid{b},\Varid{mc})\hsarrow{\rightarrow }{\mathpunct{.}}{}\<[28]%
\>[28]{}\mathbf{do}\;{}\<[32]%
\>[32]{}\Varid{c'}\leftarrow \mathbf{case}\;\Varid{mc}\;\mathbf{of}\;{}\<[50]%
\>[50]{}\Conid{Nothing}\hsarrow{\rightarrow }{\mathpunct{.}}\Varid{mcreate}\;\Varid{l}_{2}\;\Varid{b}{}\<[E]%
\\
\>[50]{}\Conid{Just}\;\Varid{c}\hsarrow{\rightarrow }{\mathpunct{.}}\Varid{mupdate}\;\Varid{l}_{2}\;\Varid{c}\;\Varid{b}{}\<[E]%
\\
\>[32]{}\Varid{return}\;(\Varid{mview}\;\Varid{l}_{1}\;\Varid{c'},\Conid{Just}\;\Varid{c'}))\;{}\<[E]%
\\
\>[15]{}\Conid{Nothing}{}\<[E]%
\ColumnHook
\end{hscode}\resethooks
Composition (naive)
\begin{hscode}\SaveRestoreHook
\column{B}{@{}>{\hspre}l<{\hspost}@{}}%
\column{3}{@{}>{\hspre}l<{\hspost}@{}}%
\column{5}{@{}>{\hspre}l<{\hspost}@{}}%
\column{24}{@{}>{\hspre}l<{\hspost}@{}}%
\column{29}{@{}>{\hspre}l<{\hspost}@{}}%
\column{30}{@{}>{\hspre}l<{\hspost}@{}}%
\column{E}{@{}>{\hspre}l<{\hspost}@{}}%
\>[3]{}(\mathbin{;})\mathbin{::}\Conid{Monad}\;\Varid{m}\Rightarrow {}\<[29]%
\>[29]{}\Conid{SMLens}\;\Varid{m}\;\Varid{c}_{1}\;\Varid{a}\;\Varid{b}\hsarrow{\rightarrow }{\mathpunct{.}}\Conid{SMLens}\;\Varid{m}\;\Varid{c}_{2}\;\Varid{b}\;\Varid{c}\hsarrow{\rightarrow }{\mathpunct{.}}\Conid{SMLens}\;\Varid{m}\;(\Varid{c}_{1},\Varid{c}_{2})\;\Varid{a}\;\Varid{c}{}\<[E]%
\\
\>[3]{}(\mathbin{;})\;\Varid{sl}_{1}\;\Varid{sl}_{2}\mathrel{=}\Conid{SMLens}\;\Varid{mput}_{R}\;\Varid{mput}_{L}\;\Varid{mMissing}\;\mathbf{where}{}\<[E]%
\\
\>[3]{}\hsindent{2}{}\<[5]%
\>[5]{}\Varid{mput}_{R}\;(\Varid{a},(\Varid{c}_{1},\Varid{c}_{2})){}\<[24]%
\>[24]{}\mathrel{=}\mathbf{do}\;{}\<[30]%
\>[30]{}(\Varid{b},\Varid{c}_{1}')\leftarrow \Varid{mputr}\;\Varid{sl}_{1}\;(\Varid{a},\Varid{c}_{1}){}\<[E]%
\\
\>[30]{}(\Varid{c},\Varid{c}_{2}')\leftarrow \Varid{mputr}\;\Varid{sl}_{2}\;(\Varid{b},\Varid{c}_{2}){}\<[E]%
\\
\>[30]{}\Varid{return}\;(\Varid{c},(\Varid{c}_{1}',\Varid{c}_{2}')){}\<[E]%
\\
\>[3]{}\hsindent{2}{}\<[5]%
\>[5]{}\Varid{mput}_{L}\;(\Varid{c},(\Varid{c}_{1},\Varid{c}_{2})){}\<[24]%
\>[24]{}\mathrel{=}\mathbf{do}\;{}\<[30]%
\>[30]{}(\Varid{b},\Varid{c}_{2}')\leftarrow \Varid{mputl}\;\Varid{sl}_{2}\;(\Varid{c},\Varid{c}_{2}){}\<[E]%
\\
\>[30]{}(\Varid{a},\Varid{c}_{1}')\leftarrow \Varid{mputl}\;\Varid{sl}_{1}\;(\Varid{b},\Varid{c}_{1}){}\<[E]%
\\
\>[30]{}\Varid{return}\;(\Varid{a},(\Varid{c}_{1}',\Varid{c}_{2}')){}\<[E]%
\\
\>[3]{}\hsindent{2}{}\<[5]%
\>[5]{}\Varid{mMissing}{}\<[24]%
\>[24]{}\mathrel{=}(\Varid{mmissing}\;\Varid{sl}_{1},\Varid{mmissing}\;\Varid{sl}_{2}){}\<[E]%
\ColumnHook
\end{hscode}\resethooks

\subsection{StateTBX}
\begin{hscode}\SaveRestoreHook
\column{B}{@{}>{\hspre}l<{\hspost}@{}}%
\column{3}{@{}>{\hspre}l<{\hspost}@{}}%
\column{E}{@{}>{\hspre}l<{\hspost}@{}}%
\>[3]{}\mbox{\enskip\{-\# LANGUAGE RankNTypes, FlexibleContexts  \#-\}\enskip}{}\<[E]%
\\[\blanklineskip]%
\>[3]{}\mathbf{module}\;\Conid{StateTBX}\;\mathbf{where}{}\<[E]%
\\
\>[3]{}\mathbf{import}\;\Conid{\Conid{Control}.\Conid{Monad}.State}\;\Varid{as}\;\Conid{State}{}\<[E]%
\\
\>[3]{}\mathbf{import}\;\Conid{Control}.\Conid{Monad}.\Conid{Id}\;\Varid{as}\;\Conid{Id}{}\<[E]%
\\
\>[3]{}\mathbf{import}\;\Conid{BX}{}\<[E]%
\\
\>[3]{}\mathbf{import}\;\Conid{Iso}{}\<[E]%
\\
\>[3]{}\mathbf{import}\;\Conid{Lens}{}\<[E]%
\\
\>[3]{}\mathbf{import}\;\Conid{RelBX}{}\<[E]%
\\
\>[3]{}\mathbf{import}\;\Conid{SLens}\;\Varid{as}\;\Conid{SLens}{}\<[E]%
\\
\>[3]{}\mathbf{import}\;\Conid{MLens}\;\Varid{as}\;\Conid{MLens}{}\<[E]%
\\
\>[3]{}\mathbf{import}\;\Conid{SMLens}\;\Varid{as}\;\Conid{SMLens}{}\<[E]%
\ColumnHook
\end{hscode}\resethooks
The interface
\begin{hscode}\SaveRestoreHook
\column{B}{@{}>{\hspre}l<{\hspost}@{}}%
\column{3}{@{}>{\hspre}l<{\hspost}@{}}%
\column{8}{@{}>{\hspre}l<{\hspost}@{}}%
\column{E}{@{}>{\hspre}l<{\hspost}@{}}%
\>[3]{}\mathbf{data}\;\Conid{StateTBX}\;\Varid{m}\;\Varid{s}\;\Varid{a}\;\Varid{b}\mathrel{=}\Conid{StateTBX}\;\{\mskip1.5mu {}\<[E]%
\\
\>[3]{}\hsindent{5}{}\<[8]%
\>[8]{}\Varid{getl}\mathbin{::}\Conid{StateT}\;\Varid{s}\;\Varid{m}\;\Varid{a},{}\<[E]%
\\
\>[3]{}\hsindent{5}{}\<[8]%
\>[8]{}\Varid{setl}\mathbin{::}\Varid{a}\hsarrow{\rightarrow }{\mathpunct{.}}\Conid{StateT}\;\Varid{s}\;\Varid{m}\;(),{}\<[E]%
\\
\>[3]{}\hsindent{5}{}\<[8]%
\>[8]{}\Varid{initl}\mathbin{::}\Varid{a}\hsarrow{\rightarrow }{\mathpunct{.}}\Varid{m}\;\Varid{s},{}\<[E]%
\\
\>[3]{}\hsindent{5}{}\<[8]%
\>[8]{}\Varid{getr}\mathbin{::}\Conid{StateT}\;\Varid{s}\;\Varid{m}\;\Varid{b},{}\<[E]%
\\
\>[3]{}\hsindent{5}{}\<[8]%
\>[8]{}\Varid{setr}\mathbin{::}\Varid{b}\hsarrow{\rightarrow }{\mathpunct{.}}\Conid{StateT}\;\Varid{s}\;\Varid{m}\;(),{}\<[E]%
\\
\>[3]{}\hsindent{5}{}\<[8]%
\>[8]{}\Varid{initr}\mathbin{::}\Varid{b}\hsarrow{\rightarrow }{\mathpunct{.}}\Varid{m}\;\Varid{s}{}\<[E]%
\\
\>[3]{}\mskip1.5mu\}{}\<[E]%
\ColumnHook
\end{hscode}\resethooks
Variations on initialisation
\begin{hscode}\SaveRestoreHook
\column{B}{@{}>{\hspre}l<{\hspost}@{}}%
\column{3}{@{}>{\hspre}l<{\hspost}@{}}%
\column{12}{@{}>{\hspre}l<{\hspost}@{}}%
\column{16}{@{}>{\hspre}l<{\hspost}@{}}%
\column{E}{@{}>{\hspre}l<{\hspost}@{}}%
\>[3]{}\Varid{init2run}\mathbin{::}{}\<[16]%
\>[16]{}\Conid{Monad}\;\Varid{m}\Rightarrow (\Varid{a}\hsarrow{\rightarrow }{\mathpunct{.}}\Varid{m}\;\Varid{s})\hsarrow{\rightarrow }{\mathpunct{.}}\Conid{StateT}\;\Varid{s}\;\Varid{m}\;\Varid{x}\hsarrow{\rightarrow }{\mathpunct{.}}\Varid{a}\hsarrow{\rightarrow }{\mathpunct{.}}\Varid{m}\;(\Varid{x},\Varid{s}){}\<[E]%
\\
\>[3]{}\Varid{init2run}\;\Varid{init}\;\Varid{m}\;\Varid{a}\mathrel{=}\mathbf{do}\;\{\mskip1.5mu \Varid{s}\leftarrow \Varid{init}\;\Varid{a};\Varid{runStateT}\;\Varid{m}\;\Varid{s}\mskip1.5mu\}{}\<[E]%
\\[\blanklineskip]%
\>[3]{}\Varid{runl}\mathbin{::}{}\<[12]%
\>[12]{}\Conid{Monad}\;\Varid{m}\Rightarrow \Conid{StateTBX}\;\Varid{m}\;\Varid{s}\;\Varid{a}\;\Varid{b}\hsarrow{\rightarrow }{\mathpunct{.}}\Conid{StateT}\;\Varid{s}\;\Varid{m}\;\Varid{x}\hsarrow{\rightarrow }{\mathpunct{.}}\Varid{a}\hsarrow{\rightarrow }{\mathpunct{.}}\Varid{m}\;(\Varid{x},\Varid{s}){}\<[E]%
\\
\>[3]{}\Varid{runl}\;bx\mathrel{=}\Varid{init2run}\;(\Varid{initl}\;bx){}\<[E]%
\\[\blanklineskip]%
\>[3]{}\Varid{runr}\mathbin{::}{}\<[12]%
\>[12]{}\Conid{Monad}\;\Varid{m}\Rightarrow \Conid{StateTBX}\;\Varid{m}\;\Varid{s}\;\Varid{a}\;\Varid{b}\hsarrow{\rightarrow }{\mathpunct{.}}\Conid{StateT}\;\Varid{s}\;\Varid{m}\;\Varid{x}\hsarrow{\rightarrow }{\mathpunct{.}}\Varid{b}\hsarrow{\rightarrow }{\mathpunct{.}}\Varid{m}\;(\Varid{x},\Varid{s}){}\<[E]%
\\
\>[3]{}\Varid{runr}\;bx\mathrel{=}\Varid{init2run}\;(\Varid{initr}\;bx){}\<[E]%
\\[\blanklineskip]%
\>[3]{}\Varid{run2init}\mathbin{::}{}\<[16]%
\>[16]{}\Conid{Monad}\;\Varid{m}\Rightarrow (\forall \Varid{x}\hsforall \hsdot{\cdot }{.}\Conid{StateT}\;\Varid{s}\;\Varid{m}\;\Varid{x}\hsarrow{\rightarrow }{\mathpunct{.}}\Varid{a}\hsarrow{\rightarrow }{\mathpunct{.}}\Varid{m}\;(\Varid{x},\Varid{s}))\hsarrow{\rightarrow }{\mathpunct{.}}\Varid{a}\hsarrow{\rightarrow }{\mathpunct{.}}\Varid{m}\;\Varid{s}{}\<[E]%
\\
\>[3]{}\Varid{run2init}\;\Varid{run}\;\Varid{a}\mathrel{=}\mathbf{do}\;\{\mskip1.5mu ((),\Varid{s})\leftarrow \Varid{run}\;(\Varid{return}\;())\;\Varid{a};\Varid{return}\;\Varid{s}\mskip1.5mu\}{}\<[E]%
\ColumnHook
\end{hscode}\resethooks
An alternative `PutBX' or push--pull style
\begin{hscode}\SaveRestoreHook
\column{B}{@{}>{\hspre}l<{\hspost}@{}}%
\column{3}{@{}>{\hspre}l<{\hspost}@{}}%
\column{E}{@{}>{\hspre}l<{\hspost}@{}}%
\>[3]{}\Varid{put}_{L}^{R}\mathbin{::}\Conid{Monad}\;\Varid{m}\Rightarrow \Conid{StateTBX}\;\Varid{m}\;\Varid{s}\;\Varid{a}\;\Varid{b}\hsarrow{\rightarrow }{\mathpunct{.}}\Varid{a}\hsarrow{\rightarrow }{\mathpunct{.}}\Conid{StateT}\;\Varid{s}\;\Varid{m}\;\Varid{b}{}\<[E]%
\\
\>[3]{}\Varid{put}_{L}^{R}\;bx\;\Varid{a}\mathrel{=}\mathbf{do}\;\{\mskip1.5mu \Varid{setl}\;bx\;\Varid{a};\Varid{getr}\;bx\mskip1.5mu\}{}\<[E]%
\\[\blanklineskip]%
\>[3]{}\Varid{put}_{R}^{L}\mathbin{::}\Conid{Monad}\;\Varid{m}\Rightarrow \Conid{StateTBX}\;\Varid{m}\;\Varid{s}\;\Varid{a}\;\Varid{b}\hsarrow{\rightarrow }{\mathpunct{.}}\Varid{b}\hsarrow{\rightarrow }{\mathpunct{.}}\Conid{StateT}\;\Varid{s}\;\Varid{m}\;\Varid{a}{}\<[E]%
\\
\>[3]{}\Varid{put}_{R}^{L}\;bx\;\Varid{b}\mathrel{=}\mathbf{do}\;\{\mskip1.5mu \Varid{setr}\;bx\;\Varid{b};\Varid{getl}\;bx\mskip1.5mu\}{}\<[E]%
\ColumnHook
\end{hscode}\resethooks
Identity is easy
\begin{hscode}\SaveRestoreHook
\column{B}{@{}>{\hspre}l<{\hspost}@{}}%
\column{3}{@{}>{\hspre}l<{\hspost}@{}}%
\column{43}{@{}>{\hspre}l<{\hspost}@{}}%
\column{E}{@{}>{\hspre}l<{\hspost}@{}}%
\>[3]{}\Varid{idBX}\mathbin{::}\Conid{Monad}\;\Varid{m}\Rightarrow \Conid{StateTBX}\;\Varid{m}\;\Varid{a}\;\Varid{a}\;\Varid{a}{}\<[E]%
\\
\>[3]{}\Varid{idBX}\mathrel{=}\Conid{StateTBX}\;\Varid{get}\;\Varid{put}\;\Varid{return}\;\Varid{get}\;\Varid{put}\;{}\<[43]%
\>[43]{}\Varid{return}{}\<[E]%
\ColumnHook
\end{hscode}\resethooks
Duality
\begin{hscode}\SaveRestoreHook
\column{B}{@{}>{\hspre}l<{\hspost}@{}}%
\column{3}{@{}>{\hspre}l<{\hspost}@{}}%
\column{23}{@{}>{\hspre}l<{\hspost}@{}}%
\column{E}{@{}>{\hspre}l<{\hspost}@{}}%
\>[3]{}\Varid{coBX}\mathbin{::}\Conid{StateTBX}\;\Varid{m}\;\Varid{s}\;\Varid{a}\;\Varid{b}\hsarrow{\rightarrow }{\mathpunct{.}}\Conid{StateTBX}\;\Varid{m}\;\Varid{s}\;\Varid{b}\;\Varid{a}{}\<[E]%
\\
\>[3]{}\Varid{coBX}\;bx\mathrel{=}\Conid{StateTBX}\;{}\<[23]%
\>[23]{}(\Varid{getr}\;bx)\;(\Varid{setr}\;bx)\;(\Varid{initr}\;bx)\;{}\<[E]%
\\
\>[23]{}(\Varid{getl}\;bx)\;(\Varid{setl}\;bx)\;(\Varid{initl}\;bx){}\<[E]%
\ColumnHook
\end{hscode}\resethooks
Monad morphisms injecting \ensuremath{\Conid{StateT}\;\Varid{s}\;\Varid{m}} (respectively \ensuremath{\Conid{StateT}\;\Varid{t}\;\Varid{m}}) into
\ensuremath{\Conid{StateT}\;(\Varid{s},\Varid{t})\;\Varid{m}}. 
\begin{hscode}\SaveRestoreHook
\column{B}{@{}>{\hspre}l<{\hspost}@{}}%
\column{3}{@{}>{\hspre}l<{\hspost}@{}}%
\column{19}{@{}>{\hspre}l<{\hspost}@{}}%
\column{20}{@{}>{\hspre}l<{\hspost}@{}}%
\column{E}{@{}>{\hspre}l<{\hspost}@{}}%
\>[3]{}\Varid{left}\mathbin{::}\Conid{Monad}\;\Varid{m}\Rightarrow \Conid{StateT}\;\Varid{s}\;\Varid{m}\;\Varid{a}\hsarrow{\rightarrow }{\mathpunct{.}}\Conid{StateT}\;(\Varid{s},\Varid{t})\;\Varid{m}\;\Varid{a}{}\<[E]%
\\
\>[3]{}\Varid{left}\;\Varid{ma}\mathrel{=}\mathbf{do}\;\{\mskip1.5mu {}\<[19]%
\>[19]{}(\Varid{s},\Varid{t})\leftarrow \Varid{get};{}\<[E]%
\\
\>[19]{}(\Varid{a},\Varid{s'})\leftarrow \Varid{lift}\;(\Varid{runStateT}\;\Varid{ma}\;\Varid{s});{}\<[E]%
\\
\>[19]{}\Varid{put}\;(\Varid{s'},\Varid{t});{}\<[E]%
\\
\>[19]{}\Varid{return}\;\Varid{a}\mskip1.5mu\}{}\<[E]%
\\[\blanklineskip]%
\>[3]{}\Varid{right}\mathbin{::}\Conid{Monad}\;\Varid{m}\Rightarrow \Conid{StateT}\;\Varid{t}\;\Varid{m}\;\Varid{a}\hsarrow{\rightarrow }{\mathpunct{.}}\Conid{StateT}\;(\Varid{s},\Varid{t})\;\Varid{m}\;\Varid{a}{}\<[E]%
\\
\>[3]{}\Varid{right}\;\Varid{ma}\mathrel{=}\mathbf{do}\;\{\mskip1.5mu {}\<[20]%
\>[20]{}(\Varid{s},\Varid{t})\leftarrow \Varid{get};{}\<[E]%
\\
\>[20]{}(\Varid{a},\Varid{t'})\leftarrow \Varid{lift}\;(\Varid{runStateT}\;\Varid{ma}\;\Varid{t});{}\<[E]%
\\
\>[20]{}\Varid{put}\;(\Varid{s},\Varid{t'});{}\<[E]%
\\
\>[20]{}\Varid{return}\;\Varid{a}\mskip1.5mu\}{}\<[E]%
\ColumnHook
\end{hscode}\resethooks
Composition:
given \ensuremath{\Varid{l}\mathbin{::}\Conid{StateTBX}\;\Varid{m}\;\Varid{s}\;\Varid{a}\;\Varid{b}} and \ensuremath{\Varid{l'}\mathbin{::}\Conid{StateTBX}\;\Varid{m}\;\Varid{s}\;\Varid{b}\;\Varid{c}},
we want
\begin{hscode}\SaveRestoreHook
\column{B}{@{}>{\hspre}l<{\hspost}@{}}%
\column{E}{@{}>{\hspre}l<{\hspost}@{}}%
\>[B]{}\Varid{compBX}\;\Varid{l}\;\Varid{l'}\mathbin{::}\Conid{StateTBX}\;\Varid{m}\;(\Varid{s},\Varid{t})\;\Varid{a}\;\Varid{c}{}\<[E]%
\ColumnHook
\end{hscode}\resethooks
satisfying the monad and bx laws
\begin{hscode}\SaveRestoreHook
\column{B}{@{}>{\hspre}l<{\hspost}@{}}%
\column{3}{@{}>{\hspre}l<{\hspost}@{}}%
\column{22}{@{}>{\hspre}l<{\hspost}@{}}%
\column{E}{@{}>{\hspre}l<{\hspost}@{}}%
\>[3]{}\vartheta\mathbin{::}\Conid{Monad}\;\Varid{m}\Rightarrow \Conid{MLens}\;\Varid{m}\;\Varid{s}\;\Varid{v}\hsarrow{\rightarrow }{\mathpunct{.}}\Conid{StateT}\;\Varid{v}\;\Varid{m}\ntto\Conid{StateT}\;\Varid{s}\;\Varid{m}{}\<[E]%
\\
\>[3]{}\vartheta\;\Varid{l}\;\Varid{m}\mathrel{=}\mathbf{do}\;{}\<[22]%
\>[22]{}\Varid{s}\leftarrow \Varid{get}{}\<[E]%
\\
\>[22]{}\mathbf{let}\;\Varid{v}\mathrel{=}\Varid{mview}\;\Varid{l}\;\Varid{s}{}\<[E]%
\\
\>[22]{}(\Varid{a},\Varid{v'})\leftarrow \Varid{lift}\;(\Varid{runStateT}\;\Varid{m}\;\Varid{v}){}\<[E]%
\\
\>[22]{}\Varid{s'}\leftarrow \Varid{lift}\;(\Varid{mupdate}\;\Varid{l}\;\Varid{s}\;\Varid{v'}){}\<[E]%
\\
\>[22]{}\Varid{put}\;\Varid{s'}{}\<[E]%
\\
\>[22]{}\Varid{return}\;\Varid{a}{}\<[E]%
\ColumnHook
\end{hscode}\resethooks
The m-lenses induced by two composable bxs.
\begin{hscode}\SaveRestoreHook
\column{B}{@{}>{\hspre}l<{\hspost}@{}}%
\column{3}{@{}>{\hspre}l<{\hspost}@{}}%
\column{11}{@{}>{\hspre}l<{\hspost}@{}}%
\column{25}{@{}>{\hspre}l<{\hspost}@{}}%
\column{31}{@{}>{\hspre}l<{\hspost}@{}}%
\column{32}{@{}>{\hspre}l<{\hspost}@{}}%
\column{37}{@{}>{\hspre}l<{\hspost}@{}}%
\column{38}{@{}>{\hspre}l<{\hspost}@{}}%
\column{E}{@{}>{\hspre}l<{\hspost}@{}}%
\>[3]{}\Varid{mlensL}\mathbin{::}\Conid{Monad}\;\Varid{m}\Rightarrow {}\<[25]%
\>[25]{}\Conid{StateTBX}\;\Varid{m}\;\Varid{s}_{1}\;\Varid{a}\;\Varid{b}\hsarrow{\rightarrow }{\mathpunct{.}}{}\<[E]%
\\
\>[25]{}\Conid{StateTBX}\;\Varid{m}\;\Varid{s}_{2}\;\Varid{b}\;\Varid{c}\hsarrow{\rightarrow }{\mathpunct{.}}{}\<[E]%
\\
\>[25]{}\Conid{MLens}\;\Varid{m}\;(\Varid{s}_{1},\Varid{s}_{2})\;\Varid{s}_{1}{}\<[E]%
\\
\>[3]{}\Varid{mlensL}\;bx_{1}\;bx_{2}\mathrel{=}\Conid{MLens}\;\Varid{view}\;\Varid{update}\;\Varid{create}\;\mathbf{where}{}\<[E]%
\\
\>[3]{}\hsindent{8}{}\<[11]%
\>[11]{}\Varid{view}\;(\Varid{s}_{1},\Varid{s}_{2}){}\<[31]%
\>[31]{}\mathrel{=}\Varid{s}_{1}{}\<[E]%
\\
\>[3]{}\hsindent{8}{}\<[11]%
\>[11]{}\Varid{update}\;(\Varid{s}_{1},\Varid{s}_{2})\;\Varid{s}_{1}'{}\<[31]%
\>[31]{}\mathrel{=}\mathbf{do}\;{}\<[37]%
\>[37]{}\Varid{b}\leftarrow \Varid{evalStateT}\;(\Varid{getr}\;bx_{1})\;\Varid{s}_{1}'{}\<[E]%
\\
\>[37]{}\Varid{s}_{2}'\leftarrow \Varid{execStateT}\;(\Varid{setl}\;bx_{2}\;\Varid{b})\;\Varid{s}_{2}{}\<[E]%
\\
\>[37]{}\Varid{return}\;(\Varid{s}_{1}',\Varid{s}_{2}'){}\<[E]%
\\
\>[3]{}\hsindent{8}{}\<[11]%
\>[11]{}\Varid{create}\;\Varid{s}_{1}{}\<[31]%
\>[31]{}\mathrel{=}\mathbf{do}\;{}\<[37]%
\>[37]{}\Varid{b}\leftarrow \Varid{evalStateT}\;(\Varid{getr}\;bx_{1})\;\Varid{s}_{1}{}\<[E]%
\\
\>[37]{}\Varid{s}_{2}\leftarrow \Varid{initl}\;bx_{2}\;\Varid{b}{}\<[E]%
\\
\>[37]{}\Varid{return}\;(\Varid{s}_{1},\Varid{s}_{2}){}\<[E]%
\\[\blanklineskip]%
\>[3]{}\Varid{mlensR}\mathbin{::}\Conid{Monad}\;\Varid{m}\Rightarrow {}\<[25]%
\>[25]{}\Conid{StateTBX}\;\Varid{m}\;\Varid{s}_{1}\;\Varid{a}\;\Varid{b}\hsarrow{\rightarrow }{\mathpunct{.}}{}\<[E]%
\\
\>[25]{}\Conid{StateTBX}\;\Varid{m}\;\Varid{s}_{2}\;\Varid{b}\;\Varid{c}\hsarrow{\rightarrow }{\mathpunct{.}}{}\<[E]%
\\
\>[25]{}\Conid{MLens}\;\Varid{m}\;(\Varid{s}_{1},\Varid{s}_{2})\;\Varid{s}_{2}{}\<[E]%
\\
\>[3]{}\Varid{mlensR}\;bx_{1}\;bx_{2}\mathrel{=}\Conid{MLens}\;\Varid{view}\;\Varid{update}\;\Varid{create}\;\mathbf{where}{}\<[E]%
\\
\>[3]{}\hsindent{8}{}\<[11]%
\>[11]{}\Varid{view}\;(\Varid{s}_{1},\Varid{s}_{2}){}\<[32]%
\>[32]{}\mathrel{=}\Varid{s}_{2}{}\<[E]%
\\
\>[3]{}\hsindent{8}{}\<[11]%
\>[11]{}\Varid{update}\;(\Varid{s}_{1},\Varid{s}_{2})\;\Varid{s}_{2}'{}\<[32]%
\>[32]{}\mathrel{=}\mathbf{do}\;{}\<[38]%
\>[38]{}(\Varid{b},\anonymous )\leftarrow \Varid{runStateT}\;(\Varid{getl}\;bx_{2})\;\Varid{s}_{2}'{}\<[E]%
\\
\>[38]{}((),\Varid{s}_{1}')\leftarrow \Varid{runStateT}\;(\Varid{setr}\;bx_{1}\;\Varid{b})\;\Varid{s}_{1}{}\<[E]%
\\
\>[38]{}\Varid{return}\;(\Varid{s}_{1}',\Varid{s}_{2}'){}\<[E]%
\\
\>[3]{}\hsindent{8}{}\<[11]%
\>[11]{}\Varid{create}\;\Varid{s}_{2}{}\<[32]%
\>[32]{}\mathrel{=}\mathbf{do}\;{}\<[38]%
\>[38]{}\Varid{b}\leftarrow \Varid{evalStateT}\;(\Varid{getl}\;bx_{2})\;\Varid{s}_{2}{}\<[E]%
\\
\>[38]{}\Varid{s}_{1}\leftarrow \Varid{initr}\;bx_{1}\;\Varid{b}{}\<[E]%
\\
\>[38]{}\Varid{return}\;(\Varid{s}_{1},\Varid{s}_{2}){}\<[E]%
\ColumnHook
\end{hscode}\resethooks
Composition in terms of m-lenses
\begin{hscode}\SaveRestoreHook
\column{B}{@{}>{\hspre}l<{\hspost}@{}}%
\column{3}{@{}>{\hspre}l<{\hspost}@{}}%
\column{10}{@{}>{\hspre}l<{\hspost}@{}}%
\column{17}{@{}>{\hspre}l<{\hspost}@{}}%
\column{20}{@{}>{\hspre}l<{\hspost}@{}}%
\column{27}{@{}>{\hspre}l<{\hspost}@{}}%
\column{30}{@{}>{\hspre}l<{\hspost}@{}}%
\column{41}{@{}>{\hspre}l<{\hspost}@{}}%
\column{E}{@{}>{\hspre}l<{\hspost}@{}}%
\>[3]{}\Varid{compBX}\mathbin{::}(\Conid{Monad}\;\Varid{m})\Rightarrow {}\<[27]%
\>[27]{}\Conid{StateTBX}\;\Varid{m}\;\Varid{s}_{1}\;\Varid{a}\;\Varid{b}\hsarrow{\rightarrow }{\mathpunct{.}}{}\<[E]%
\\
\>[27]{}\Conid{StateTBX}\;\Varid{m}\;\Varid{s}_{2}\;\Varid{b}\;\Varid{c}\hsarrow{\rightarrow }{\mathpunct{.}}{}\<[E]%
\\
\>[27]{}\Conid{StateTBX}\;\Varid{m}\;(\Varid{s}_{1},\Varid{s}_{2})\;\Varid{a}\;\Varid{c}{}\<[E]%
\\
\>[3]{}\Varid{compBX}\;bx_{1}\;bx_{2}\mathrel{=}{}\<[E]%
\\
\>[3]{}\hsindent{17}{}\<[20]%
\>[20]{}\Conid{StateTBX}\;{}\<[30]%
\>[30]{}(\varphi \;(\Varid{getl}\;bx_{1}))\;(\varphi \hsdot{\cdot }{.}(\Varid{setl}\;bx_{1}))\;{}\<[E]%
\\
\>[30]{}(\lambda \hslambda \Varid{a}\hsarrow{\rightarrow }{\mathpunct{.}}\mathbf{do}\;{}\<[41]%
\>[41]{}(\Varid{b},\Varid{s})\leftarrow \Varid{runl}\;bx_{1}\;(\Varid{getr}\;bx_{1})\;\Varid{a}{}\<[E]%
\\
\>[41]{}\Varid{t}\leftarrow \Varid{initl}\;bx_{2}\;\Varid{b}{}\<[E]%
\\
\>[41]{}\Varid{return}\;(\Varid{s},\Varid{t}))\;{}\<[E]%
\\
\>[30]{}(\psi \;(\Varid{getr}\;bx_{2}))\;(\psi \hsdot{\cdot }{.}(\Varid{setr}\;bx_{2}))\;{}\<[E]%
\\
\>[30]{}(\lambda \hslambda \Varid{c}\hsarrow{\rightarrow }{\mathpunct{.}}\mathbf{do}\;{}\<[41]%
\>[41]{}(\Varid{b},\Varid{t})\leftarrow \Varid{runr}\;bx_{2}\;(\Varid{getl}\;bx_{2})\;\Varid{c}{}\<[E]%
\\
\>[41]{}\Varid{s}\leftarrow \Varid{initr}\;bx_{1}\;\Varid{b}{}\<[E]%
\\
\>[41]{}\Varid{return}\;(\Varid{s},\Varid{t})){}\<[E]%
\\
\>[3]{}\hsindent{7}{}\<[10]%
\>[10]{}\mathbf{where}\;{}\<[17]%
\>[17]{}\varphi \mathrel{=}\vartheta\;(\Varid{mlensL}\;bx_{1}\;bx_{2}){}\<[E]%
\\
\>[17]{}\psi \mathrel{=}\vartheta\;(\Varid{mlensR}\;bx_{1}\;bx_{2}){}\<[E]%
\ColumnHook
\end{hscode}\resethooks
Alternative definition using \ensuremath{\Varid{left}} and \ensuremath{\Varid{right}}
\begin{hscode}\SaveRestoreHook
\column{B}{@{}>{\hspre}l<{\hspost}@{}}%
\column{3}{@{}>{\hspre}l<{\hspost}@{}}%
\column{10}{@{}>{\hspre}l<{\hspost}@{}}%
\column{20}{@{}>{\hspre}l<{\hspost}@{}}%
\column{28}{@{}>{\hspre}l<{\hspost}@{}}%
\column{31}{@{}>{\hspre}l<{\hspost}@{}}%
\column{E}{@{}>{\hspre}l<{\hspost}@{}}%
\>[3]{}\Varid{compBX'}\mathbin{::}(\Conid{Monad}\;\Varid{m})\Rightarrow {}\<[28]%
\>[28]{}\Conid{StateTBX}\;\Varid{m}\;\Varid{s}\;\Varid{a}\;\Varid{b}\hsarrow{\rightarrow }{\mathpunct{.}}{}\<[E]%
\\
\>[28]{}\Conid{StateTBX}\;\Varid{m}\;\Varid{t}\;\Varid{b}\;\Varid{c}\hsarrow{\rightarrow }{\mathpunct{.}}{}\<[E]%
\\
\>[28]{}\Conid{StateTBX}\;\Varid{m}\;(\Varid{s},\Varid{t})\;\Varid{a}\;\Varid{c}{}\<[E]%
\\
\>[3]{}\Varid{compBX'}\;bx_{1}\;bx_{2}\mathrel{=}{}\<[E]%
\\
\>[3]{}\hsindent{7}{}\<[10]%
\>[10]{}\Conid{StateTBX}\;{}\<[20]%
\>[20]{}(\Varid{left}\;(\Varid{getl}\;bx_{1}))\;{}\<[E]%
\\
\>[20]{}(\lambda \hslambda \Varid{a}\hsarrow{\rightarrow }{\mathpunct{.}}\mathbf{do}\;{}\<[31]%
\>[31]{}\Varid{b}\leftarrow \Varid{left}\;(\Varid{setl}\;bx_{1}\;\Varid{a}\sequ \Varid{getr}\;bx_{1}){}\<[E]%
\\
\>[31]{}\Varid{right}\;(\Varid{setl}\;bx_{2}\;\Varid{b}))\;{}\<[E]%
\\
\>[20]{}(\lambda \hslambda \Varid{a}\hsarrow{\rightarrow }{\mathpunct{.}}\mathbf{do}\;{}\<[31]%
\>[31]{}(\Varid{b},\Varid{s})\leftarrow \Varid{runl}\;bx_{1}\;(\Varid{getr}\;bx_{1})\;\Varid{a}{}\<[E]%
\\
\>[31]{}\Varid{t}\leftarrow \Varid{initl}\;bx_{2}\;\Varid{b}{}\<[E]%
\\
\>[31]{}\Varid{return}\;(\Varid{s},\Varid{t}))\;{}\<[E]%
\\
\>[20]{}(\Varid{right}\;(\Varid{getr}\;bx_{2}))\;{}\<[E]%
\\
\>[20]{}(\lambda \hslambda \Varid{c}\hsarrow{\rightarrow }{\mathpunct{.}}\mathbf{do}\;{}\<[31]%
\>[31]{}\Varid{b}\leftarrow \Varid{right}\;(\Varid{setr}\;bx_{2}\;\Varid{c}\sequ \Varid{getl}\;bx_{2}){}\<[E]%
\\
\>[31]{}\Varid{left}\;(\Varid{setr}\;bx_{1}\;\Varid{b}))\;{}\<[E]%
\\
\>[20]{}(\lambda \hslambda \Varid{c}\hsarrow{\rightarrow }{\mathpunct{.}}\mathbf{do}\;{}\<[31]%
\>[31]{}(\Varid{b},\Varid{t})\leftarrow \Varid{runr}\;bx_{2}\;(\Varid{getl}\;bx_{2})\;\Varid{c}{}\<[E]%
\\
\>[31]{}\Varid{s}\leftarrow \Varid{initr}\;bx_{1}\;\Varid{b}{}\<[E]%
\\
\>[31]{}\Varid{return}\;(\Varid{s},\Varid{t})){}\<[E]%
\ColumnHook
\end{hscode}\resethooks
Direct definition
\begin{hscode}\SaveRestoreHook
\column{B}{@{}>{\hspre}l<{\hspost}@{}}%
\column{3}{@{}>{\hspre}l<{\hspost}@{}}%
\column{11}{@{}>{\hspre}l<{\hspost}@{}}%
\column{20}{@{}>{\hspre}l<{\hspost}@{}}%
\column{28}{@{}>{\hspre}l<{\hspost}@{}}%
\column{31}{@{}>{\hspre}l<{\hspost}@{}}%
\column{E}{@{}>{\hspre}l<{\hspost}@{}}%
\>[3]{}\Varid{compBX0}\mathbin{::}(\Conid{Monad}\;\Varid{m})\Rightarrow {}\<[28]%
\>[28]{}\Conid{StateTBX}\;\Varid{m}\;\Varid{s}\;\Varid{a}\;\Varid{b}\hsarrow{\rightarrow }{\mathpunct{.}}{}\<[E]%
\\
\>[28]{}\Conid{StateTBX}\;\Varid{m}\;\Varid{t}\;\Varid{b}\;\Varid{c}\hsarrow{\rightarrow }{\mathpunct{.}}{}\<[E]%
\\
\>[28]{}\Conid{StateTBX}\;\Varid{m}\;(\Varid{s},\Varid{t})\;\Varid{a}\;\Varid{c}{}\<[E]%
\\
\>[3]{}\Varid{compBX0}\;bx_{1}\;bx_{2}\mathrel{=}{}\<[E]%
\\
\>[3]{}\hsindent{8}{}\<[11]%
\>[11]{}\Conid{StateTBX}\;(\mathbf{do}\;\{\mskip1.5mu (\Varid{s},\Varid{t})\leftarrow \Varid{get};\Varid{lift}\;(\Varid{evalStateT}\;(\Varid{getl}\;bx_{1})\;\Varid{s})\mskip1.5mu\})\;{}\<[E]%
\\
\>[11]{}\hsindent{9}{}\<[20]%
\>[20]{}(\lambda \hslambda \Varid{a}\hsarrow{\rightarrow }{\mathpunct{.}}\mathbf{do}\;{}\<[31]%
\>[31]{}(\Varid{s},\Varid{t})\leftarrow \Varid{get}{}\<[E]%
\\
\>[31]{}\Varid{s'}\leftarrow \Varid{lift}\;(\Varid{execStateT}\;(\Varid{setl}\;bx_{1}\;\Varid{a})\;\Varid{s}){}\<[E]%
\\
\>[31]{}\Varid{b'}\leftarrow \Varid{lift}\;(\Varid{evalStateT}\;(\Varid{getr}\;bx_{1})\;\Varid{s'}){}\<[E]%
\\
\>[31]{}\Varid{t'}\leftarrow \Varid{lift}\;(\Varid{execStateT}\;(\Varid{setl}\;bx_{2}\;\Varid{b'})\;\Varid{t}){}\<[E]%
\\
\>[31]{}\Varid{put}\;(\Varid{s'},\Varid{t'}))\;{}\<[E]%
\\
\>[11]{}\hsindent{9}{}\<[20]%
\>[20]{}(\lambda \hslambda \Varid{a}\hsarrow{\rightarrow }{\mathpunct{.}}\mathbf{do}\;{}\<[31]%
\>[31]{}(\Varid{b},\Varid{s})\leftarrow \Varid{runl}\;bx_{1}\;(\Varid{getr}\;bx_{1})\;\Varid{a}{}\<[E]%
\\
\>[31]{}\Varid{t}\leftarrow \Varid{initl}\;bx_{2}\;\Varid{b}{}\<[E]%
\\
\>[31]{}\Varid{return}\;(\Varid{s},\Varid{t}))\;{}\<[E]%
\\
\>[11]{}\hsindent{9}{}\<[20]%
\>[20]{}(\mathbf{do}\;\{\mskip1.5mu (\Varid{s},\Varid{t})\leftarrow \Varid{get};\Varid{lift}\;(\Varid{evalStateT}\;(\Varid{getr}\;bx_{2})\;\Varid{t})\mskip1.5mu\})\;{}\<[E]%
\\
\>[11]{}\hsindent{9}{}\<[20]%
\>[20]{}(\lambda \hslambda \Varid{c}\hsarrow{\rightarrow }{\mathpunct{.}}\mathbf{do}\;{}\<[31]%
\>[31]{}(\Varid{s},\Varid{t})\leftarrow \Varid{get}{}\<[E]%
\\
\>[31]{}\Varid{t'}\leftarrow \Varid{lift}\;(\Varid{execStateT}\;(\Varid{setr}\;bx_{2}\;\Varid{c})\;\Varid{t}){}\<[E]%
\\
\>[31]{}\Varid{b'}\leftarrow \Varid{lift}\;(\Varid{evalStateT}\;(\Varid{getl}\;bx_{2})\;\Varid{t'}){}\<[E]%
\\
\>[31]{}\Varid{s'}\leftarrow \Varid{lift}\;(\Varid{execStateT}\;(\Varid{setr}\;bx_{1}\;\Varid{b'})\;\Varid{s}){}\<[E]%
\\
\>[31]{}\Varid{put}\;(\Varid{s'},\Varid{t'}))\;{}\<[E]%
\\
\>[11]{}\hsindent{9}{}\<[20]%
\>[20]{}(\lambda \hslambda \Varid{c}\hsarrow{\rightarrow }{\mathpunct{.}}\mathbf{do}\;{}\<[31]%
\>[31]{}(\Varid{b},\Varid{t})\leftarrow \Varid{runr}\;bx_{2}\;(\Varid{getl}\;bx_{2})\;\Varid{c}{}\<[E]%
\\
\>[31]{}\Varid{s}\leftarrow \Varid{initr}\;bx_{1}\;\Varid{b}{}\<[E]%
\\
\>[31]{}\Varid{return}\;(\Varid{s},\Varid{t})){}\<[E]%
\ColumnHook
\end{hscode}\resethooks
Isomorphisms
 \begin{hscode}\SaveRestoreHook
\column{B}{@{}>{\hspre}l<{\hspost}@{}}%
\column{3}{@{}>{\hspre}l<{\hspost}@{}}%
\column{26}{@{}>{\hspre}l<{\hspost}@{}}%
\column{E}{@{}>{\hspre}l<{\hspost}@{}}%
\>[3]{}\Varid{iso2BX}\mathbin{::}\Conid{Monad}\;\Varid{m}\Rightarrow \Conid{Iso}\;\Varid{a}\;\Varid{b}\hsarrow{\rightarrow }{\mathpunct{.}}\Conid{StateTBX}\;\Varid{m}\;\Varid{a}\;\Varid{a}\;\Varid{b}{}\<[E]%
\\
\>[3]{}\Varid{iso2BX}\;\Varid{iso}\mathrel{=}\Conid{StateTBX}\;{}\<[26]%
\>[26]{}\Varid{get}\;\Varid{put}\;\Varid{return}\;{}\<[E]%
\\
\>[26]{}(\mathbf{do}\;\{\mskip1.5mu \Varid{a}\leftarrow \Varid{get};\Varid{return}\;(\Varid{to}\;\Varid{iso}\;\Varid{a})\mskip1.5mu\})\;{}\<[E]%
\\
\>[26]{}(\lambda \hslambda \Varid{b}\hsarrow{\rightarrow }{\mathpunct{.}}\mathbf{do}\;\{\mskip1.5mu \Varid{put}\;(\Varid{from}\;\Varid{iso}\;\Varid{b})\mskip1.5mu\})\;{}\<[E]%
\\
\>[26]{}(\lambda \hslambda \Varid{b}\hsarrow{\rightarrow }{\mathpunct{.}}\Varid{return}\;(\Varid{from}\;\Varid{iso}\;\Varid{b})){}\<[E]%
\\[\blanklineskip]%
\>[3]{}\Varid{assocBX}\mathbin{::}\Conid{Monad}\;\Varid{m}\Rightarrow {}\<[26]%
\>[26]{}\Conid{StateTBX}\;\Varid{m}\;((\Varid{a},\Varid{b}),\Varid{c})\;((\Varid{a},\Varid{b}),\Varid{c})\;(\Varid{a},(\Varid{b},\Varid{c})){}\<[E]%
\\
\>[3]{}\Varid{assocBX}\mathrel{=}\Varid{iso2BX}\;\Varid{assocIso}{}\<[E]%
\\[\blanklineskip]%
\>[3]{}\Varid{swapBX}\mathbin{::}\Conid{Monad}\;\Varid{m}\Rightarrow \Conid{StateTBX}\;\Varid{m}\;(\Varid{a},\Varid{b})\;(\Varid{a},\Varid{b})\;(\Varid{b},\Varid{a}){}\<[E]%
\\
\>[3]{}\Varid{swapBX}\mathrel{=}\Varid{iso2BX}\;\Varid{swapIso}{}\<[E]%
\\[\blanklineskip]%
\>[3]{}\Varid{unitlBX}\mathbin{::}\Conid{Monad}\;\Varid{m}\Rightarrow \Conid{StateTBX}\;\Varid{m}\;\Varid{a}\;\Varid{a}\;((),\Varid{a}){}\<[E]%
\\
\>[3]{}\Varid{unitlBX}\mathrel{=}\Varid{iso2BX}\;\Varid{unitlIso}{}\<[E]%
\\[\blanklineskip]%
\>[3]{}\Varid{unitrBX}\mathbin{::}\Conid{Monad}\;\Varid{m}\Rightarrow \Conid{StateTBX}\;\Varid{m}\;\Varid{a}\;\Varid{a}\;(\Varid{a},()){}\<[E]%
\\
\>[3]{}\Varid{unitrBX}\mathrel{=}\Varid{iso2BX}\;\Varid{unitrIso}{}\<[E]%
\ColumnHook
\end{hscode}\resethooks
Lenses
\begin{hscode}\SaveRestoreHook
\column{B}{@{}>{\hspre}l<{\hspost}@{}}%
\column{3}{@{}>{\hspre}l<{\hspost}@{}}%
\column{25}{@{}>{\hspre}l<{\hspost}@{}}%
\column{30}{@{}>{\hspre}l<{\hspost}@{}}%
\column{33}{@{}>{\hspre}l<{\hspost}@{}}%
\column{38}{@{}>{\hspre}l<{\hspost}@{}}%
\column{44}{@{}>{\hspre}l<{\hspost}@{}}%
\column{E}{@{}>{\hspre}l<{\hspost}@{}}%
\>[3]{}\Varid{lens2BX}\mathbin{::}\Conid{Monad}\;\Varid{m}\Rightarrow \Conid{Lens}\;\Varid{a}\;\Varid{b}\hsarrow{\rightarrow }{\mathpunct{.}}\Conid{StateTBX}\;\Varid{m}\;\Varid{a}\;\Varid{a}\;\Varid{b}{}\<[E]%
\\
\>[3]{}\Varid{lens2BX}\;\Varid{l}\mathrel{=}\Conid{StateTBX}\;{}\<[25]%
\>[25]{}\Varid{get}\;\Varid{put}\;\Varid{return}\;{}\<[E]%
\\
\>[25]{}(\mathbf{do}\;\{\mskip1.5mu \Varid{a}\leftarrow \Varid{get};\Varid{return}\;(\Varid{view}\;\Varid{l}\;\Varid{a})\mskip1.5mu\})\;{}\<[E]%
\\
\>[25]{}(\lambda \hslambda \Varid{b}\hsarrow{\rightarrow }{\mathpunct{.}}\mathbf{do}\;\{\mskip1.5mu \Varid{a}\leftarrow \Varid{get};\Varid{put}\;(\Varid{update}\;\Varid{l}\;\Varid{a}\;\Varid{b})\mskip1.5mu\})\;{}\<[E]%
\\
\>[25]{}(\lambda \hslambda \Varid{b}\hsarrow{\rightarrow }{\mathpunct{.}}\Varid{return}\;(\Varid{create}\;\Varid{l}\;\Varid{b})){}\<[E]%
\\[\blanklineskip]%
\>[3]{}\Varid{lensSpan2BX}\mathbin{::}\Conid{Monad}\;\Varid{m}\Rightarrow {}\<[30]%
\>[30]{}\Conid{Lens}\;\Varid{c}\;\Varid{a}\hsarrow{\rightarrow }{\mathpunct{.}}\Conid{Lens}\;\Varid{c}\;\Varid{b}\hsarrow{\rightarrow }{\mathpunct{.}}\Conid{StateTBX}\;\Varid{m}\;\Varid{c}\;\Varid{a}\;\Varid{b}{}\<[E]%
\\
\>[3]{}\Varid{lensSpan2BX}\;\Varid{l}_{1}\;\Varid{l}_{2}\mathrel{=}\Conid{StateTBX}\;{}\<[33]%
\>[33]{}(\mathbf{do}\;{}\<[38]%
\>[38]{}\Varid{c}\leftarrow \Varid{get}{}\<[E]%
\\
\>[38]{}\Varid{return}\;(\Varid{view}\;\Varid{l}_{1}\;\Varid{c}))\;{}\<[E]%
\\
\>[33]{}(\lambda \hslambda \Varid{a}\hsarrow{\rightarrow }{\mathpunct{.}}\mathbf{do}\;{}\<[44]%
\>[44]{}\Varid{c}\leftarrow \Varid{get}{}\<[E]%
\\
\>[44]{}\Varid{put}\;(\Varid{update}\;\Varid{l}_{1}\;\Varid{c}\;\Varid{a}))\;{}\<[E]%
\\
\>[33]{}(\lambda \hslambda \Varid{a}\hsarrow{\rightarrow }{\mathpunct{.}}\Varid{return}\;(\Varid{create}\;\Varid{l}_{1}\;\Varid{a}))\;{}\<[E]%
\\
\>[33]{}(\mathbf{do}\;{}\<[38]%
\>[38]{}\Varid{c}\leftarrow \Varid{get}{}\<[E]%
\\
\>[38]{}\Varid{return}\;(\Varid{view}\;\Varid{l}_{2}\;\Varid{c}))\;{}\<[E]%
\\
\>[33]{}(\lambda \hslambda \Varid{b}\hsarrow{\rightarrow }{\mathpunct{.}}\mathbf{do}\;{}\<[44]%
\>[44]{}\Varid{c}\leftarrow \Varid{get}{}\<[E]%
\\
\>[44]{}\Varid{put}\;(\Varid{update}\;\Varid{l}_{2}\;\Varid{c}\;\Varid{b}))\;{}\<[E]%
\\
\>[33]{}(\lambda \hslambda \Varid{b}\hsarrow{\rightarrow }{\mathpunct{.}}\Varid{return}\;(\Varid{create}\;\Varid{l}_{2}\;\Varid{b})){}\<[E]%
\ColumnHook
\end{hscode}\resethooks
Monadic lenses
\begin{hscode}\SaveRestoreHook
\column{B}{@{}>{\hspre}l<{\hspost}@{}}%
\column{3}{@{}>{\hspre}l<{\hspost}@{}}%
\column{5}{@{}>{\hspre}l<{\hspost}@{}}%
\column{19}{@{}>{\hspre}l<{\hspost}@{}}%
\column{23}{@{}>{\hspre}l<{\hspost}@{}}%
\column{25}{@{}>{\hspre}l<{\hspost}@{}}%
\column{26}{@{}>{\hspre}l<{\hspost}@{}}%
\column{31}{@{}>{\hspre}l<{\hspost}@{}}%
\column{34}{@{}>{\hspre}l<{\hspost}@{}}%
\column{E}{@{}>{\hspre}l<{\hspost}@{}}%
\>[3]{}\Varid{mlens2BX}\mathbin{::}\Conid{Monad}\;\Varid{m}\Rightarrow \Conid{MLens}\;\Varid{m}\;\Varid{a}\;\Varid{b}\hsarrow{\rightarrow }{\mathpunct{.}}\Conid{StateTBX}\;\Varid{m}\;\Varid{a}\;\Varid{a}\;\Varid{b}{}\<[E]%
\\
\>[3]{}\Varid{mlens2BX}\;\Varid{l}\mathrel{=}\Conid{StateTBX}\;{}\<[26]%
\>[26]{}\Varid{get}\;\Varid{put}\;\Varid{return}\;\Varid{view}\;\Varid{upd}\;\Varid{create}\;\mathbf{where}{}\<[E]%
\\
\>[3]{}\hsindent{2}{}\<[5]%
\>[5]{}\Varid{view}{}\<[23]%
\>[23]{}\mathrel{=}\Varid{gets}\;(\Varid{mview}\;\Varid{l}){}\<[E]%
\\
\>[3]{}\hsindent{2}{}\<[5]%
\>[5]{}\Varid{upd}\;\Varid{b}{}\<[23]%
\>[23]{}\mathrel{=}\mathbf{do}\;\{\mskip1.5mu {}\<[31]%
\>[31]{}\Varid{a}\leftarrow \Varid{get};{}\<[E]%
\\
\>[31]{}\Varid{a'}\leftarrow \Varid{lift}\;(\Varid{mupdate}\;\Varid{l}\;\Varid{a}\;\Varid{b});{}\<[E]%
\\
\>[31]{}\Varid{put}\;\Varid{a'}\mskip1.5mu\}{}\<[E]%
\\
\>[3]{}\hsindent{2}{}\<[5]%
\>[5]{}\Varid{create}\;\Varid{b}{}\<[23]%
\>[23]{}\mathrel{=}\Varid{mcreate}\;\Varid{l}\;\Varid{b}{}\<[E]%
\\[\blanklineskip]%
\>[3]{}\Varid{mlensSpan2BX}\mathbin{::}\Conid{Monad}\;\Varid{m}\Rightarrow {}\<[31]%
\>[31]{}\Conid{MLens}\;\Varid{m}\;\Varid{c}\;\Varid{a}\hsarrow{\rightarrow }{\mathpunct{.}}\Conid{MLens}\;\Varid{m}\;\Varid{c}\;\Varid{b}\hsarrow{\rightarrow }{\mathpunct{.}}{}\<[E]%
\\
\>[31]{}\Conid{StateTBX}\;\Varid{m}\;\Varid{c}\;\Varid{a}\;\Varid{b}{}\<[E]%
\\
\>[3]{}\Varid{mlensSpan2BX}\;\Varid{l}_{1}\;\Varid{l}_{2}\mathrel{=}\Conid{StateTBX}\;{}\<[34]%
\>[34]{}\Varid{viewL}\;\Varid{updL}\;\Varid{createL}\;{}\<[E]%
\\
\>[34]{}\Varid{viewR}\;\Varid{updR}\;\Varid{createR}\;\mathbf{where}{}\<[E]%
\\
\>[3]{}\hsindent{2}{}\<[5]%
\>[5]{}\Varid{viewL}{}\<[19]%
\>[19]{}\mathrel{=}\Varid{gets}\;(\Varid{mview}\;\Varid{l}_{1}){}\<[E]%
\\
\>[3]{}\hsindent{2}{}\<[5]%
\>[5]{}\Varid{updL}\;\Varid{a}{}\<[19]%
\>[19]{}\mathrel{=}\mathbf{do}\;{}\<[25]%
\>[25]{}\Varid{c}\leftarrow \Varid{get}{}\<[E]%
\\
\>[25]{}\Varid{c'}\leftarrow \Varid{lift}\;(\Varid{mupdate}\;\Varid{l}_{1}\;\Varid{c}\;\Varid{a}){}\<[E]%
\\
\>[25]{}\Varid{put}\;\Varid{c'}{}\<[E]%
\\
\>[3]{}\hsindent{2}{}\<[5]%
\>[5]{}\Varid{createL}\;\Varid{a}{}\<[19]%
\>[19]{}\mathrel{=}\Varid{mcreate}\;\Varid{l}_{1}\;\Varid{a}{}\<[E]%
\\
\>[3]{}\hsindent{2}{}\<[5]%
\>[5]{}\Varid{viewR}{}\<[19]%
\>[19]{}\mathrel{=}\Varid{gets}\;(\Varid{mview}\;\Varid{l}_{2}){}\<[E]%
\\
\>[3]{}\hsindent{2}{}\<[5]%
\>[5]{}\Varid{updR}\;\Varid{b}{}\<[19]%
\>[19]{}\mathrel{=}\mathbf{do}\;{}\<[25]%
\>[25]{}\Varid{c}\leftarrow \Varid{get}{}\<[E]%
\\
\>[25]{}\Varid{c'}\leftarrow \Varid{lift}\;(\Varid{mupdate}\;\Varid{l}_{2}\;\Varid{c}\;\Varid{b}){}\<[E]%
\\
\>[25]{}\Varid{put}\;\Varid{c'}{}\<[E]%
\\
\>[3]{}\hsindent{2}{}\<[5]%
\>[5]{}\Varid{createR}\;\Varid{b}{}\<[19]%
\>[19]{}\mathrel{=}\Varid{mcreate}\;\Varid{l}_{2}\;\Varid{b}{}\<[E]%
\ColumnHook
\end{hscode}\resethooks
Relational bxs. 
\begin{hscode}\SaveRestoreHook
\column{B}{@{}>{\hspre}l<{\hspost}@{}}%
\column{3}{@{}>{\hspre}l<{\hspost}@{}}%
\column{14}{@{}>{\hspre}l<{\hspost}@{}}%
\column{24}{@{}>{\hspre}l<{\hspost}@{}}%
\column{E}{@{}>{\hspre}l<{\hspost}@{}}%
\>[3]{}\Varid{rel2BX}\mathbin{::}{}\<[14]%
\>[14]{}(\Conid{Monad}\;\Varid{m},\Conid{Pointed}\;\Varid{a},\Conid{Pointed}\;\Varid{b})\Rightarrow {}\<[E]%
\\
\>[14]{}\Conid{RelBX}\;\Varid{a}\;\Varid{b}\hsarrow{\rightarrow }{\mathpunct{.}}\Conid{StateTBX}\;\Varid{m}\;(\Varid{a},\Varid{b})\;\Varid{a}\;\Varid{b}{}\<[E]%
\\
\>[3]{}\Varid{rel2BX}\;bx\mathrel{=}\Conid{StateTBX}\;(\mathbf{do}\;\{\mskip1.5mu (\Varid{a},\Varid{b})\leftarrow \Varid{get};\Varid{return}\;\Varid{a}\mskip1.5mu\})\;{}\<[E]%
\\
\>[3]{}\hsindent{21}{}\<[24]%
\>[24]{}(\lambda \hslambda \Varid{a'}\hsarrow{\rightarrow }{\mathpunct{.}}\mathbf{do}\;\{\mskip1.5mu (\Varid{a},\Varid{b})\leftarrow \Varid{get};\Varid{put}\;(\Varid{a'},\Varid{fwd}\;bx\;\Varid{a'}\;\Varid{b})\mskip1.5mu\})\;{}\<[E]%
\\
\>[3]{}\hsindent{21}{}\<[24]%
\>[24]{}(\lambda \hslambda \Varid{a}\hsarrow{\rightarrow }{\mathpunct{.}}\Varid{return}\;(\Varid{a},\Varid{point}))\;{}\<[E]%
\\
\>[3]{}\hsindent{21}{}\<[24]%
\>[24]{}(\mathbf{do}\;\{\mskip1.5mu (\Varid{a},\Varid{b})\leftarrow \Varid{get};\Varid{return}\;\Varid{b}\mskip1.5mu\})\;{}\<[E]%
\\
\>[3]{}\hsindent{21}{}\<[24]%
\>[24]{}(\lambda \hslambda \Varid{b'}\hsarrow{\rightarrow }{\mathpunct{.}}\mathbf{do}\;\{\mskip1.5mu (\Varid{a},\Varid{b})\leftarrow \Varid{get};\Varid{put}\;(\Varid{bwd}\;bx\;\Varid{a}\;\Varid{b'},\Varid{b'})\mskip1.5mu\})\;{}\<[E]%
\\
\>[3]{}\hsindent{21}{}\<[24]%
\>[24]{}(\lambda \hslambda \Varid{b}\hsarrow{\rightarrow }{\mathpunct{.}}\Varid{return}\;(\Varid{point},\Varid{b})){}\<[E]%
\ColumnHook
\end{hscode}\resethooks
Symmetric lenses 
\begin{hscode}\SaveRestoreHook
\column{B}{@{}>{\hspre}l<{\hspost}@{}}%
\column{3}{@{}>{\hspre}l<{\hspost}@{}}%
\column{26}{@{}>{\hspre}l<{\hspost}@{}}%
\column{28}{@{}>{\hspre}l<{\hspost}@{}}%
\column{29}{@{}>{\hspre}l<{\hspost}@{}}%
\column{33}{@{}>{\hspre}l<{\hspost}@{}}%
\column{35}{@{}>{\hspre}l<{\hspost}@{}}%
\column{37}{@{}>{\hspre}l<{\hspost}@{}}%
\column{40}{@{}>{\hspre}l<{\hspost}@{}}%
\column{E}{@{}>{\hspre}l<{\hspost}@{}}%
\>[3]{}\Varid{symlens2bx}\mathbin{::}\Conid{Monad}\;\Varid{m}\Rightarrow {}\<[29]%
\>[29]{}\Conid{SLens}\;\Varid{c}\;\Varid{a}\;\Varid{b}\hsarrow{\rightarrow }{\mathpunct{.}}\Conid{StateTBX}\;\Varid{m}\;(\Varid{a},\Varid{b},\Varid{c})\;\Varid{a}\;\Varid{b}{}\<[E]%
\\
\>[3]{}\Varid{symlens2bx}\;\Varid{l}\mathrel{=}\Conid{StateTBX}\;{}\<[28]%
\>[28]{}(\mathbf{do}\;{}\<[33]%
\>[33]{}(\Varid{a},\Varid{b},\Varid{c})\leftarrow \Varid{get}{}\<[E]%
\\
\>[33]{}\Varid{return}\;\Varid{a})\;{}\<[E]%
\\
\>[28]{}(\lambda \hslambda \Varid{a'}\hsarrow{\rightarrow }{\mathpunct{.}}\mathbf{do}\;{}\<[40]%
\>[40]{}(\Varid{a},\Varid{b},\Varid{c})\leftarrow \Varid{get}{}\<[E]%
\\
\>[40]{}\mathbf{let}\;(\Varid{b'},\Varid{c'})\mathrel{=}\Varid{putr}\;\Varid{l}\;(\Varid{a'},\Varid{c}){}\<[E]%
\\
\>[40]{}\Varid{put}\;(\Varid{a'},\Varid{b'},\Varid{c'}))\;{}\<[E]%
\\
\>[28]{}(\lambda \hslambda \Varid{a}\hsarrow{\rightarrow }{\mathpunct{.}}\mathbf{do}\;{}\<[40]%
\>[40]{}\mathbf{let}\;(\Varid{b},\Varid{c})\mathrel{=}\Varid{putr}\;\Varid{l}\;(\Varid{a},\Varid{missing}\;\Varid{l}){}\<[E]%
\\
\>[40]{}\Varid{return}\;(\Varid{a},\Varid{b},\Varid{c}))\;{}\<[E]%
\\
\>[28]{}(\mathbf{do}\;{}\<[33]%
\>[33]{}(\Varid{a},\Varid{b},\Varid{c})\leftarrow \Varid{get}{}\<[E]%
\\
\>[33]{}\Varid{return}\;(\Varid{b}))\;{}\<[E]%
\\
\>[28]{}(\lambda \hslambda \Varid{b'}\hsarrow{\rightarrow }{\mathpunct{.}}\mathbf{do}\;{}\<[40]%
\>[40]{}(\Varid{a},\Varid{b},\Varid{c})\leftarrow \Varid{get}{}\<[E]%
\\
\>[40]{}\mathbf{let}\;(\Varid{a'},\Varid{c'})\mathrel{=}\Varid{putl}\;\Varid{l}\;(\Varid{b'},\Varid{c}){}\<[E]%
\\
\>[40]{}\Varid{put}\;(\Varid{a'},\Varid{b'},\Varid{c'}))\;{}\<[E]%
\\
\>[28]{}(\lambda \hslambda \Varid{b}\hsarrow{\rightarrow }{\mathpunct{.}}\mathbf{do}\;{}\<[40]%
\>[40]{}\mathbf{let}\;(\Varid{a},\Varid{c})\mathrel{=}\Varid{putl}\;\Varid{l}\;(\Varid{b},\Varid{missing}\;\Varid{l}){}\<[E]%
\\
\>[40]{}\Varid{return}\;(\Varid{a},\Varid{b},\Varid{c})){}\<[E]%
\\[\blanklineskip]%
\>[3]{}\Varid{bx2symlens}\mathbin{::}\Conid{StateTBX}\;\Conid{Id}\;\Varid{c}\;\Varid{a}\;\Varid{b}\hsarrow{\rightarrow }{\mathpunct{.}}\Conid{SLens}\;(\Conid{Maybe}\;\Varid{c})\;\Varid{a}\;\Varid{b}{}\<[E]%
\\
\>[3]{}\Varid{bx2symlens}\;bx\mathrel{=}\Conid{SLens}\;{}\<[26]%
\>[26]{}(\lambda \hslambda (\Varid{a},\Varid{mc})\hsarrow{\rightarrow }{\mathpunct{.}}{}\<[E]%
\\
\>[26]{}\hsindent{3}{}\<[29]%
\>[29]{}\mathbf{let}\;\Varid{m}\mathrel{=}(\Varid{setl}\;bx\;\Varid{a}\sequ \Varid{getr}\;bx)\;\mathbf{in}{}\<[E]%
\\
\>[26]{}\hsindent{3}{}\<[29]%
\>[29]{}\mathbf{let}\;(\Varid{b'},\Varid{c'})\mathrel{=}{}\<[E]%
\\
\>[29]{}\hsindent{6}{}\<[35]%
\>[35]{}\mathbf{case}\;\Varid{mc}\;\mathbf{of}{}\<[E]%
\\
\>[35]{}\hsindent{2}{}\<[37]%
\>[37]{}\Conid{Nothing}\hsarrow{\rightarrow }{\mathpunct{.}}\Varid{runIdentity}\;(\Varid{runl}\;bx\;\Varid{m}\;\Varid{a});{}\<[E]%
\\
\>[35]{}\hsindent{2}{}\<[37]%
\>[37]{}\Conid{Just}\;\Varid{c}\hsarrow{\rightarrow }{\mathpunct{.}}\Varid{runIdentity}\;(\Varid{runStateT}\;\Varid{m}\;\Varid{c}){}\<[E]%
\\
\>[26]{}\hsindent{3}{}\<[29]%
\>[29]{}\mathbf{in}\;(\Varid{b'},\Conid{Just}\;\Varid{c'}))\;{}\<[E]%
\\
\>[26]{}(\lambda \hslambda (\Varid{b},\Varid{mc})\hsarrow{\rightarrow }{\mathpunct{.}}{}\<[E]%
\\
\>[26]{}\hsindent{3}{}\<[29]%
\>[29]{}\mathbf{let}\;\Varid{m}\mathrel{=}(\Varid{setr}\;bx\;\Varid{b}\sequ \Varid{getl}\;bx)\;\mathbf{in}{}\<[E]%
\\
\>[26]{}\hsindent{3}{}\<[29]%
\>[29]{}\mathbf{let}\;(\Varid{a'},\Varid{c'})\mathrel{=}{}\<[E]%
\\
\>[29]{}\hsindent{6}{}\<[35]%
\>[35]{}\mathbf{case}\;\Varid{mc}\;\mathbf{of}{}\<[E]%
\\
\>[35]{}\hsindent{2}{}\<[37]%
\>[37]{}\Conid{Nothing}\hsarrow{\rightarrow }{\mathpunct{.}}\Varid{runIdentity}\;(\Varid{runr}\;bx\;\Varid{m}\;\Varid{b});{}\<[E]%
\\
\>[35]{}\hsindent{2}{}\<[37]%
\>[37]{}\Conid{Just}\;\Varid{c}\hsarrow{\rightarrow }{\mathpunct{.}}\Varid{runIdentity}\;(\Varid{runStateT}\;\Varid{m}\;\Varid{c}){}\<[E]%
\\
\>[26]{}\hsindent{3}{}\<[29]%
\>[29]{}\mathbf{in}\;(\Varid{a'},\Conid{Just}\;\Varid{c'}))\;{}\<[E]%
\\
\>[26]{}\Conid{Nothing}{}\<[E]%
\ColumnHook
\end{hscode}\resethooks
Monadic symmetric lenses
\begin{hscode}\SaveRestoreHook
\column{B}{@{}>{\hspre}l<{\hspost}@{}}%
\column{3}{@{}>{\hspre}l<{\hspost}@{}}%
\column{4}{@{}>{\hspre}l<{\hspost}@{}}%
\column{13}{@{}>{\hspre}l<{\hspost}@{}}%
\column{19}{@{}>{\hspre}l<{\hspost}@{}}%
\column{21}{@{}>{\hspre}l<{\hspost}@{}}%
\column{25}{@{}>{\hspre}l<{\hspost}@{}}%
\column{29}{@{}>{\hspre}l<{\hspost}@{}}%
\column{30}{@{}>{\hspre}l<{\hspost}@{}}%
\column{32}{@{}>{\hspre}l<{\hspost}@{}}%
\column{34}{@{}>{\hspre}l<{\hspost}@{}}%
\column{36}{@{}>{\hspre}l<{\hspost}@{}}%
\column{41}{@{}>{\hspre}l<{\hspost}@{}}%
\column{52}{@{}>{\hspre}l<{\hspost}@{}}%
\column{54}{@{}>{\hspre}l<{\hspost}@{}}%
\column{61}{@{}>{\hspre}l<{\hspost}@{}}%
\column{63}{@{}>{\hspre}l<{\hspost}@{}}%
\column{E}{@{}>{\hspre}l<{\hspost}@{}}%
\>[3]{}\Varid{symMLens2bx}\mathbin{::}\Conid{Monad}\;\Varid{m}\Rightarrow {}\<[30]%
\>[30]{}\Conid{SMLens}\;\Varid{m}\;\Varid{c}\;\Varid{a}\;\Varid{b}\hsarrow{\rightarrow }{\mathpunct{.}}\Conid{StateTBX}\;\Varid{m}\;(\Varid{a},\Varid{b},\Varid{c})\;\Varid{a}\;\Varid{b}{}\<[E]%
\\
\>[3]{}\Varid{symMLens2bx}\;\Varid{l}\mathrel{=}\Conid{StateTBX}\;{}\<[29]%
\>[29]{}(\mathbf{do}\;{}\<[34]%
\>[34]{}(\Varid{a},\Varid{b},\Varid{c})\leftarrow \Varid{get}{}\<[E]%
\\
\>[34]{}\Varid{return}\;\Varid{a})\;{}\<[E]%
\\
\>[29]{}(\lambda \hslambda \Varid{a'}\hsarrow{\rightarrow }{\mathpunct{.}}\mathbf{do}\;{}\<[41]%
\>[41]{}(\Varid{a},\Varid{b},\Varid{c})\leftarrow \Varid{get}{}\<[E]%
\\
\>[41]{}(\Varid{b'},\Varid{c'})\leftarrow \Varid{lift}\;(\Varid{mputr}\;\Varid{l}\;(\Varid{a'},\Varid{c})){}\<[E]%
\\
\>[41]{}\Varid{put}\;(\Varid{a'},\Varid{b'},\Varid{c'}))\;{}\<[E]%
\\
\>[29]{}(\lambda \hslambda \Varid{a}\hsarrow{\rightarrow }{\mathpunct{.}}\mathbf{do}\;{}\<[41]%
\>[41]{}(\Varid{b},\Varid{c})\leftarrow \Varid{mputr}\;\Varid{l}\;(\Varid{a},\Varid{mmissing}\;\Varid{l}){}\<[E]%
\\
\>[41]{}\Varid{return}\;(\Varid{a},\Varid{b},\Varid{c}))\;{}\<[E]%
\\
\>[29]{}(\mathbf{do}\;{}\<[34]%
\>[34]{}(\Varid{a},\Varid{b},\Varid{c})\leftarrow \Varid{get}{}\<[E]%
\\
\>[34]{}\Varid{return}\;\Varid{b})\;{}\<[E]%
\\
\>[29]{}(\lambda \hslambda \Varid{b'}\hsarrow{\rightarrow }{\mathpunct{.}}\mathbf{do}\;{}\<[41]%
\>[41]{}(\Varid{a},\Varid{b},\Varid{c})\leftarrow \Varid{get}{}\<[E]%
\\
\>[41]{}(\Varid{a'},\Varid{c'})\leftarrow \Varid{lift}\;(\Varid{mputl}\;\Varid{l}\;(\Varid{b'},\Varid{c})){}\<[E]%
\\
\>[41]{}\Varid{put}\;(\Varid{a'},\Varid{b'},\Varid{c'}))\;{}\<[E]%
\\
\>[29]{}(\lambda \hslambda \Varid{b}\hsarrow{\rightarrow }{\mathpunct{.}}\mathbf{do}\;{}\<[41]%
\>[41]{}(\Varid{a},\Varid{c})\leftarrow \Varid{mputl}\;\Varid{l}\;(\Varid{b},\Varid{mmissing}\;\Varid{l}){}\<[E]%
\\
\>[41]{}\Varid{return}\;(\Varid{a},\Varid{b},\Varid{c})){}\<[E]%
\\[\blanklineskip]%
\>[3]{}\Varid{bx2symMLens}\mathbin{::}{}\<[19]%
\>[19]{}\Conid{Monad}\;\Varid{m}\Rightarrow \Conid{StateTBX}\;\Varid{m}\;\Varid{c}\;\Varid{a}\;\Varid{b}\hsarrow{\rightarrow }{\mathpunct{.}}{}\<[E]%
\\
\>[19]{}\Conid{SMLens}\;\Varid{m}\;(\Conid{Maybe}\;\Varid{c})\;\Varid{a}\;\Varid{b}{}\<[E]%
\\
\>[3]{}\Varid{bx2symMLens}\;bx\mathrel{=}\Conid{SMLens}\;\Varid{mputlr}\;\Varid{mputrl}\;\Varid{missing}\;\mathbf{where}{}\<[E]%
\\
\>[3]{}\hsindent{1}{}\<[4]%
\>[4]{}\Varid{mputlr}\;(\Varid{a},\Varid{ms})\mathrel{=}{}\<[21]%
\>[21]{}\mathbf{do}\;{}\<[25]%
\>[25]{}\Varid{s}\leftarrow \mathbf{case}\;\Varid{ms}\;\mathbf{of}{}\<[E]%
\\
\>[25]{}\hsindent{7}{}\<[32]%
\>[32]{}\Conid{Just}\;\Varid{s'}{}\<[41]%
\>[41]{}\hsarrow{\rightarrow }{\mathpunct{.}}\Varid{return}\;\Varid{s'}{}\<[E]%
\\
\>[25]{}\hsindent{7}{}\<[32]%
\>[32]{}\Conid{Nothing}{}\<[41]%
\>[41]{}\hsarrow{\rightarrow }{\mathpunct{.}}\Varid{initl}\;bx\;\Varid{a}{}\<[E]%
\\
\>[25]{}(\Varid{b},\Varid{s'})\leftarrow {}\<[36]%
\>[36]{}(\Varid{runStateT}\;(\mathbf{do}\;\{\mskip1.5mu {}\<[54]%
\>[54]{}\Varid{setl}\;bx\;{}\<[63]%
\>[63]{}\Varid{a};\Varid{getr}\;bx\mskip1.5mu\})\;\Varid{s}){}\<[E]%
\\
\>[25]{}\Varid{return}\;(\Varid{b},\Conid{Just}\;\Varid{s'}){}\<[E]%
\\
\>[3]{}\hsindent{1}{}\<[4]%
\>[4]{}\Varid{mputrl}\;(\Varid{b},\Varid{ms})\mathrel{=}{}\<[21]%
\>[21]{}\mathbf{do}\;{}\<[25]%
\>[25]{}\Varid{s}\leftarrow \mathbf{case}\;\Varid{ms}\;\mathbf{of}{}\<[E]%
\\
\>[25]{}\hsindent{7}{}\<[32]%
\>[32]{}\Conid{Just}\;\Varid{s'}{}\<[41]%
\>[41]{}\hsarrow{\rightarrow }{\mathpunct{.}}\Varid{return}\;\Varid{s'}{}\<[E]%
\\
\>[25]{}\hsindent{7}{}\<[32]%
\>[32]{}\Conid{Nothing}{}\<[41]%
\>[41]{}\hsarrow{\rightarrow }{\mathpunct{.}}\Varid{initr}\;bx\;\Varid{b}{}\<[E]%
\\
\>[25]{}(\Varid{a},\Varid{s'})\leftarrow \Varid{runStateT}\;(\mathbf{do}\;\{\mskip1.5mu {}\<[52]%
\>[52]{}\Varid{setr}\;bx\;{}\<[61]%
\>[61]{}\Varid{b};\Varid{getl}\;bx\mskip1.5mu\})\;\Varid{s}{}\<[E]%
\\
\>[25]{}\Varid{return}\;(\Varid{a},\Conid{Just}\;\Varid{s'}){}\<[E]%
\\
\>[3]{}\hsindent{1}{}\<[4]%
\>[4]{}\Varid{missing}{}\<[13]%
\>[13]{}\mathrel{=}\Conid{Nothing}{}\<[E]%
\ColumnHook
\end{hscode}\resethooks
Constants
\begin{hscode}\SaveRestoreHook
\column{B}{@{}>{\hspre}l<{\hspost}@{}}%
\column{3}{@{}>{\hspre}l<{\hspost}@{}}%
\column{12}{@{}>{\hspre}c<{\hspost}@{}}%
\column{12E}{@{}l@{}}%
\column{18}{@{}>{\hspre}l<{\hspost}@{}}%
\column{25}{@{}>{\hspre}l<{\hspost}@{}}%
\column{E}{@{}>{\hspre}l<{\hspost}@{}}%
\>[3]{}\Varid{constBX}{}\<[12]%
\>[12]{}\mathbin{::}{}\<[12E]%
\>[18]{}\Conid{Monad}\;\tau\Rightarrow \alpha\hsarrow{\rightarrow }{\mathpunct{.}}\Conid{StateTBX}\;\tau\;\alpha\;()\;\alpha{}\<[E]%
\\
\>[3]{}\Varid{constBX}\;\Varid{a}\mathrel{=}\Conid{StateTBX}\;{}\<[25]%
\>[25]{}(\Varid{return}\;())\;{}\<[E]%
\\
\>[25]{}(\Varid{const}\;(\Varid{return}\;()))\;{}\<[E]%
\\
\>[25]{}(\Varid{const}\;(\Varid{return}\;\Varid{a}))\;{}\<[E]%
\\
\>[25]{}\Varid{get}\;\Varid{put}\;\Varid{return}{}\<[E]%
\ColumnHook
\end{hscode}\resethooks
Pairs
\begin{hscode}\SaveRestoreHook
\column{B}{@{}>{\hspre}l<{\hspost}@{}}%
\column{3}{@{}>{\hspre}l<{\hspost}@{}}%
\column{24}{@{}>{\hspre}l<{\hspost}@{}}%
\column{35}{@{}>{\hspre}l<{\hspost}@{}}%
\column{E}{@{}>{\hspre}l<{\hspost}@{}}%
\>[3]{}\Varid{fstBX}\mathbin{::}(\Conid{Monad}\;\Varid{m})\Rightarrow \Varid{b}\hsarrow{\rightarrow }{\mathpunct{.}}\Conid{StateTBX}\;\Varid{m}\;(\Varid{a},\Varid{b})\;(\Varid{a},\Varid{b})\;\Varid{a}{}\<[E]%
\\
\>[3]{}\Varid{fstBX}\;\Varid{b}_{0}\mathrel{=}\Conid{StateTBX}\;{}\<[24]%
\>[24]{}(\Varid{get})\;{}\<[E]%
\\
\>[24]{}(\Varid{put})\;{}\<[E]%
\\
\>[24]{}\Varid{return}\;{}\<[E]%
\\
\>[24]{}(\Varid{gets}\;\Varid{fst})\;{}\<[E]%
\\
\>[24]{}(\lambda \hslambda \Varid{a}\hsarrow{\rightarrow }{\mathpunct{.}}\mathbf{do}\;{}\<[35]%
\>[35]{}(\anonymous ,\Varid{b})\leftarrow \Varid{get}{}\<[E]%
\\
\>[35]{}\Varid{put}\;(\Varid{a},\Varid{b}))\;{}\<[E]%
\\
\>[24]{}(\lambda \hslambda \Varid{a}\hsarrow{\rightarrow }{\mathpunct{.}}\Varid{return}\;(\Varid{a},\Varid{b}_{0})){}\<[E]%
\\[\blanklineskip]%
\>[3]{}\Varid{sndBX}\mathbin{::}\Conid{Monad}\;\Varid{m}\Rightarrow \Varid{a}\hsarrow{\rightarrow }{\mathpunct{.}}\Conid{StateTBX}\;\Varid{m}\;(\Varid{a},\Varid{b})\;(\Varid{a},\Varid{b})\;\Varid{b}{}\<[E]%
\\
\>[3]{}\Varid{sndBX}\;\Varid{a}_{0}\mathrel{=}\Conid{StateTBX}\;{}\<[24]%
\>[24]{}(\Varid{get})\;{}\<[E]%
\\
\>[24]{}(\Varid{put})\;{}\<[E]%
\\
\>[24]{}\Varid{return}\;{}\<[E]%
\\
\>[24]{}(\Varid{gets}\;\Varid{snd})\;{}\<[E]%
\\
\>[24]{}(\lambda \hslambda \Varid{b}\hsarrow{\rightarrow }{\mathpunct{.}}\mathbf{do}\;{}\<[35]%
\>[35]{}(\Varid{a},\anonymous )\leftarrow \Varid{get}{}\<[E]%
\\
\>[35]{}\Varid{put}\;(\Varid{a},\Varid{b}))\;{}\<[E]%
\\
\>[24]{}(\lambda \hslambda \Varid{b}\hsarrow{\rightarrow }{\mathpunct{.}}\Varid{return}\;(\Varid{a}_{0},\Varid{b})){}\<[E]%
\ColumnHook
\end{hscode}\resethooks
Products
 \begin{hscode}\SaveRestoreHook
\column{B}{@{}>{\hspre}l<{\hspost}@{}}%
\column{3}{@{}>{\hspre}l<{\hspost}@{}}%
\column{25}{@{}>{\hspre}l<{\hspost}@{}}%
\column{30}{@{}>{\hspre}l<{\hspost}@{}}%
\column{35}{@{}>{\hspre}l<{\hspost}@{}}%
\column{47}{@{}>{\hspre}l<{\hspost}@{}}%
\column{E}{@{}>{\hspre}l<{\hspost}@{}}%
\>[3]{}\Varid{pairBX}\mathbin{::}\Conid{Monad}\;\Varid{m}\Rightarrow {}\<[25]%
\>[25]{}\Conid{StateTBX}\;\Varid{m}\;\Varid{s}_{1}\;\Varid{a}_{1}\;\Varid{b}_{1}\hsarrow{\rightarrow }{\mathpunct{.}}{}\<[E]%
\\
\>[25]{}\Conid{StateTBX}\;\Varid{m}\;\Varid{s}_{2}\;\Varid{a}_{2}\;\Varid{b}_{2}\hsarrow{\rightarrow }{\mathpunct{.}}{}\<[E]%
\\
\>[25]{}\Conid{StateTBX}\;\Varid{m}\;(\Varid{s}_{1},\Varid{s}_{2})\;(\Varid{a}_{1},\Varid{a}_{2})\;(\Varid{b}_{1},\Varid{b}_{2}){}\<[E]%
\\
\>[3]{}\Varid{pairBX}\;bx_{1}\;bx_{2}\mathrel{=}\Conid{StateTBX}\;{}\<[30]%
\>[30]{}(\mathbf{do}\;{}\<[35]%
\>[35]{}\Varid{a}_{1}\leftarrow \Varid{left}\;(\Varid{getl}\;bx_{1}){}\<[E]%
\\
\>[35]{}\Varid{a}_{2}\leftarrow \Varid{right}\;(\Varid{getl}\;bx_{2}){}\<[E]%
\\
\>[35]{}\Varid{return}\;(\Varid{a}_{1},\Varid{a}_{2}))\;{}\<[E]%
\\
\>[30]{}(\lambda \hslambda (\Varid{a}_{1},\Varid{a}_{2})\hsarrow{\rightarrow }{\mathpunct{.}}\mathbf{do}\;{}\<[47]%
\>[47]{}\Varid{left}\;(\Varid{setl}\;bx_{1}\;\Varid{a}_{1}){}\<[E]%
\\
\>[47]{}\Varid{right}\;(\Varid{setl}\;bx_{2}\;\Varid{a}_{2}))\;{}\<[E]%
\\
\>[30]{}(\lambda \hslambda (\Varid{a}_{1},\Varid{a}_{2})\hsarrow{\rightarrow }{\mathpunct{.}}\mathbf{do}\;{}\<[47]%
\>[47]{}\Varid{s}_{1}\leftarrow \Varid{initl}\;bx_{1}\;\Varid{a}_{1}{}\<[E]%
\\
\>[47]{}\Varid{s}_{2}\leftarrow \Varid{initl}\;bx_{2}\;\Varid{a}_{2}{}\<[E]%
\\
\>[47]{}\Varid{return}\;(\Varid{s}_{1},\Varid{s}_{2}))\;{}\<[E]%
\\
\>[30]{}(\mathbf{do}\;{}\<[35]%
\>[35]{}\Varid{b}_{1}\leftarrow \Varid{left}\;(\Varid{getr}\;bx_{1}){}\<[E]%
\\
\>[35]{}\Varid{b}_{2}\leftarrow \Varid{right}\;(\Varid{getr}\;bx_{2}){}\<[E]%
\\
\>[35]{}\Varid{return}\;(\Varid{b}_{1},\Varid{b}_{2}))\;{}\<[E]%
\\
\>[30]{}(\lambda \hslambda (\Varid{b}_{1},\Varid{b}_{2})\hsarrow{\rightarrow }{\mathpunct{.}}\mathbf{do}\;{}\<[47]%
\>[47]{}\Varid{left}\;(\Varid{setr}\;bx_{1}\;\Varid{b}_{1}){}\<[E]%
\\
\>[47]{}\Varid{right}\;(\Varid{setr}\;bx_{2}\;\Varid{b}_{2}))\;{}\<[E]%
\\
\>[30]{}(\lambda \hslambda (\Varid{b}_{1},\Varid{b}_{2})\hsarrow{\rightarrow }{\mathpunct{.}}\mathbf{do}\;{}\<[47]%
\>[47]{}\Varid{s}_{1}\leftarrow \Varid{initr}\;bx_{1}\;\Varid{b}_{1}{}\<[E]%
\\
\>[47]{}\Varid{s}_{2}\leftarrow \Varid{initr}\;bx_{2}\;\Varid{b}_{2}{}\<[E]%
\\
\>[47]{}\Varid{return}\;(\Varid{s}_{1},\Varid{s}_{2})){}\<[E]%
\ColumnHook
\end{hscode}\resethooks
Sums
\begin{hscode}\SaveRestoreHook
\column{B}{@{}>{\hspre}l<{\hspost}@{}}%
\column{3}{@{}>{\hspre}l<{\hspost}@{}}%
\column{5}{@{}>{\hspre}l<{\hspost}@{}}%
\column{11}{@{}>{\hspre}l<{\hspost}@{}}%
\column{12}{@{}>{\hspre}l<{\hspost}@{}}%
\column{13}{@{}>{\hspre}l<{\hspost}@{}}%
\column{25}{@{}>{\hspre}l<{\hspost}@{}}%
\column{29}{@{}>{\hspre}c<{\hspost}@{}}%
\column{29E}{@{}l@{}}%
\column{30}{@{}>{\hspre}l<{\hspost}@{}}%
\column{32}{@{}>{\hspre}l<{\hspost}@{}}%
\column{36}{@{}>{\hspre}l<{\hspost}@{}}%
\column{38}{@{}>{\hspre}l<{\hspost}@{}}%
\column{42}{@{}>{\hspre}l<{\hspost}@{}}%
\column{45}{@{}>{\hspre}l<{\hspost}@{}}%
\column{E}{@{}>{\hspre}l<{\hspost}@{}}%
\>[3]{}\Varid{inlBX}\mathbin{::}\Conid{Monad}\;\Varid{m}\Rightarrow \Varid{x}\hsarrow{\rightarrow }{\mathpunct{.}}\Conid{StateTBX}\;\Varid{m}\;(\Varid{x},\Conid{Maybe}\;\Varid{y})\;\Varid{x}\;(\Conid{Either}\;\Varid{x}\;\Varid{y}){}\<[E]%
\\
\>[3]{}\Varid{inlBX}\;\Varid{initX}\mathrel{=}\Conid{StateTBX}\;\get{A}\;\set{A}\;\Varid{initA}\;\get{B}\;\set{B}\;\Varid{initB}{}\<[E]%
\\
\>[3]{}\hsindent{2}{}\<[5]%
\>[5]{}\mathbf{where}\;{}\<[12]%
\>[12]{}\get{A}{}\<[29]%
\>[29]{}\mathrel{=}{}\<[29E]%
\>[32]{}\mathbf{do}\;{}\<[36]%
\>[36]{}\{\mskip1.5mu (\Varid{x},\anonymous )\leftarrow \Varid{get};\Varid{return}\;\Varid{x}\mskip1.5mu\}{}\<[E]%
\\
\>[12]{}\get{B}{}\<[29]%
\>[29]{}\mathrel{=}{}\<[29E]%
\>[32]{}\mathbf{do}\;{}\<[36]%
\>[36]{}(\Varid{x},\Varid{my})\leftarrow \Varid{get}{}\<[E]%
\\
\>[36]{}\mathbf{case}\;\Varid{my}\;\mathbf{of}{}\<[E]%
\\
\>[36]{}\hsindent{2}{}\<[38]%
\>[38]{}\Conid{Just}\;\Varid{y}\hsarrow{\rightarrow }{\mathpunct{.}}\Varid{return}\;(\Conid{Right}\;\Varid{y}){}\<[E]%
\\
\>[36]{}\hsindent{2}{}\<[38]%
\>[38]{}\Conid{Nothing}\hsarrow{\rightarrow }{\mathpunct{.}}\Varid{return}\;(\Conid{Left}\;\Varid{x}){}\<[E]%
\\
\>[12]{}\set{A}\;\Varid{x'}{}\<[29]%
\>[29]{}\mathrel{=}{}\<[29E]%
\>[32]{}\mathbf{do}\;\{\mskip1.5mu (\Varid{x},\Varid{my})\leftarrow \Varid{get};\Varid{put}\;(\Varid{x'},\Varid{my})\mskip1.5mu\}{}\<[E]%
\\
\>[12]{}\set{B}\;(\Conid{Left}\;\Varid{x}){}\<[29]%
\>[29]{}\mathrel{=}{}\<[29E]%
\>[32]{}\mathbf{do}\;\{\mskip1.5mu \Varid{put}\;(\Varid{x},\Conid{Nothing})\mskip1.5mu\}{}\<[E]%
\\
\>[12]{}\set{B}\;(\Conid{Right}\;\Varid{y}){}\<[29]%
\>[29]{}\mathrel{=}{}\<[29E]%
\>[32]{}\mathbf{do}\;\{\mskip1.5mu (\Varid{x},\anonymous )\leftarrow \Varid{get};\Varid{put}\;(\Varid{x},\Conid{Just}\;\Varid{y})\mskip1.5mu\}{}\<[E]%
\\
\>[12]{}\Varid{initA}\;\Varid{x}{}\<[29]%
\>[29]{}\mathrel{=}{}\<[29E]%
\>[32]{}\Varid{return}\;(\Varid{x},\Conid{Nothing}){}\<[E]%
\\
\>[12]{}\Varid{initB}\;(\Conid{Left}\;\Varid{x}){}\<[29]%
\>[29]{}\mathrel{=}{}\<[29E]%
\>[32]{}\Varid{return}\;(\Varid{x},\Conid{Nothing}){}\<[E]%
\\
\>[12]{}\Varid{initB}\;(\Conid{Right}\;\Varid{y}){}\<[29]%
\>[29]{}\mathrel{=}{}\<[29E]%
\>[32]{}\Varid{return}\;(\Varid{initX},\Conid{Just}\;\Varid{y}){}\<[E]%
\\[\blanklineskip]%
\>[3]{}\Varid{inrBX}\mathbin{::}\Conid{Monad}\;\Varid{m}\Rightarrow \Varid{y}\hsarrow{\rightarrow }{\mathpunct{.}}\Conid{StateTBX}\;\Varid{m}\;(\Varid{y},\Conid{Maybe}\;\Varid{x})\;\Varid{y}\;(\Conid{Either}\;\Varid{x}\;\Varid{y}){}\<[E]%
\\
\>[3]{}\Varid{inrBX}\;\Varid{initY}\mathrel{=}\Conid{StateTBX}\;\get{A}\;\set{A}\;\Varid{initA}\;\get{B}\;\set{B}\;\Varid{initB}{}\<[E]%
\\
\>[3]{}\hsindent{2}{}\<[5]%
\>[5]{}\mathbf{where}\;{}\<[12]%
\>[12]{}\get{A}{}\<[29]%
\>[29]{}\mathrel{=}{}\<[29E]%
\>[32]{}\mathbf{do}\;\{\mskip1.5mu (\Varid{y},\anonymous )\leftarrow \Varid{get};\Varid{return}\;\Varid{y}\mskip1.5mu\}{}\<[E]%
\\
\>[12]{}\get{B}{}\<[29]%
\>[29]{}\mathrel{=}{}\<[29E]%
\>[32]{}\mathbf{do}\;{}\<[36]%
\>[36]{}(\Varid{y},\Varid{mx})\leftarrow \Varid{get}{}\<[E]%
\\
\>[36]{}\mathbf{case}\;\Varid{mx}\;\mathbf{of}{}\<[E]%
\\
\>[36]{}\hsindent{2}{}\<[38]%
\>[38]{}\Conid{Just}\;\Varid{x}\hsarrow{\rightarrow }{\mathpunct{.}}\Varid{return}\;(\Conid{Left}\;\Varid{x}){}\<[E]%
\\
\>[36]{}\hsindent{2}{}\<[38]%
\>[38]{}\Conid{Nothing}\hsarrow{\rightarrow }{\mathpunct{.}}\Varid{return}\;(\Conid{Right}\;\Varid{y}){}\<[E]%
\\
\>[12]{}\set{A}\;\Varid{y'}{}\<[29]%
\>[29]{}\mathrel{=}{}\<[29E]%
\>[32]{}\mathbf{do}\;\{\mskip1.5mu (\Varid{y},\Varid{mx})\leftarrow \Varid{get};\Varid{put}\;(\Varid{y'},\Varid{mx})\mskip1.5mu\}{}\<[E]%
\\
\>[12]{}\set{B}\;(\Conid{Left}\;\Varid{x}){}\<[29]%
\>[29]{}\mathrel{=}{}\<[29E]%
\>[32]{}\mathbf{do}\;\{\mskip1.5mu (\Varid{y},\anonymous )\leftarrow \Varid{get};\Varid{put}\;(\Varid{y},\Conid{Just}\;\Varid{x})\mskip1.5mu\}{}\<[E]%
\\
\>[12]{}\set{B}\;(\Conid{Right}\;\Varid{y}){}\<[29]%
\>[29]{}\mathrel{=}{}\<[29E]%
\>[32]{}\mathbf{do}\;\{\mskip1.5mu \Varid{put}\;(\Varid{y},\Conid{Nothing})\mskip1.5mu\}{}\<[E]%
\\
\>[12]{}\Varid{initA}\;\Varid{y}{}\<[29]%
\>[29]{}\mathrel{=}{}\<[29E]%
\>[32]{}\Varid{return}\;(\Varid{y},\Conid{Nothing}){}\<[E]%
\\
\>[12]{}\Varid{initB}\;(\Conid{Left}\;\Varid{x}){}\<[29]%
\>[29]{}\mathrel{=}{}\<[29E]%
\>[32]{}\Varid{return}\;(\Varid{initY},\Conid{Just}\;\Varid{x}){}\<[E]%
\\
\>[12]{}\Varid{initB}\;(\Conid{Right}\;\Varid{y}){}\<[29]%
\>[29]{}\mathrel{=}{}\<[29E]%
\>[32]{}\Varid{return}\;(\Varid{y},\Conid{Nothing}){}\<[E]%
\\[\blanklineskip]%
\>[3]{}\Varid{sumBX}\mathbin{::}{}\<[13]%
\>[13]{}\Conid{Monad}\;\Varid{m}\Rightarrow {}\<[25]%
\>[25]{}\Conid{StateTBX}\;\Varid{m}\;\Varid{s}_{1}\;\Varid{a}_{1}\;\Varid{b}_{1}\hsarrow{\rightarrow }{\mathpunct{.}}{}\<[E]%
\\
\>[13]{}\Conid{StateTBX}\;\Varid{m}\;\Varid{s}_{2}\;\Varid{a}_{2}\;\Varid{b}_{2}\hsarrow{\rightarrow }{\mathpunct{.}}{}\<[E]%
\\
\>[13]{}\Conid{StateTBX}\;\Varid{m}\;(\Conid{Bool},\Varid{s}_{1},\Varid{s}_{2})\;(\Conid{Either}\;\Varid{a}_{1}\;\Varid{a}_{2})\;(\Conid{Either}\;\Varid{b}_{1}\;\Varid{b}_{2}){}\<[E]%
\\
\>[3]{}\Varid{sumBX}\;bx_{1}\;bx_{2}\mathrel{=}\Conid{StateTBX}\;\get{A}\;\set{A}\;\Varid{initA}\;\get{B}\;\set{B}\;\Varid{initB}{}\<[E]%
\\
\>[3]{}\hsindent{2}{}\<[5]%
\>[5]{}\mathbf{where}\;\get{A}{}\<[30]%
\>[30]{}\mathrel{=}\mathbf{do}\;{}\<[36]%
\>[36]{}(\Varid{b},\Varid{s}_{1},\Varid{s}_{2})\leftarrow \Varid{get};{}\<[E]%
\\
\>[36]{}\mathbf{if}\;\Varid{b}\;\mathbf{then}{}\<[E]%
\\
\>[36]{}\hsindent{2}{}\<[38]%
\>[38]{}\mathbf{do}\;{}\<[42]%
\>[42]{}(\Varid{a}_{1},\anonymous )\leftarrow \Varid{lift}\;(\Varid{runStateT}\;(\Varid{getl}\;bx_{1})\;\Varid{s}_{1}){}\<[E]%
\\
\>[42]{}\Varid{return}\;(\Conid{Left}\;\Varid{a}_{1}){}\<[E]%
\\
\>[36]{}\mathbf{else}\;\mathbf{do}\;{}\<[45]%
\>[45]{}(\Varid{a}_{2},\anonymous )\leftarrow \Varid{lift}\;(\Varid{runStateT}\;(\Varid{getl}\;bx_{2})\;\Varid{s}_{2}){}\<[E]%
\\
\>[45]{}\Varid{return}\;(\Conid{Right}\;\Varid{a}_{2}){}\<[E]%
\\
\>[5]{}\hsindent{6}{}\<[11]%
\>[11]{}\get{B}{}\<[30]%
\>[30]{}\mathrel{=}\mathbf{do}\;{}\<[36]%
\>[36]{}(\Varid{b},\Varid{s}_{1},\Varid{s}_{2})\leftarrow \Varid{get};{}\<[E]%
\\
\>[36]{}\mathbf{if}\;\Varid{b}\;\mathbf{then}{}\<[E]%
\\
\>[36]{}\hsindent{2}{}\<[38]%
\>[38]{}\mathbf{do}\;{}\<[42]%
\>[42]{}(\Varid{b}_{1},\anonymous )\leftarrow \Varid{lift}\;(\Varid{runStateT}\;(\Varid{getr}\;bx_{1})\;\Varid{s}_{1}){}\<[E]%
\\
\>[42]{}\Varid{return}\;(\Conid{Left}\;\Varid{b}_{1}){}\<[E]%
\\
\>[36]{}\mathbf{else}\;\mathbf{do}\;{}\<[45]%
\>[45]{}(\Varid{b}_{2},\anonymous )\leftarrow \Varid{lift}\;(\Varid{runStateT}\;(\Varid{getr}\;bx_{2})\;\Varid{s}_{2}){}\<[E]%
\\
\>[45]{}\Varid{return}\;(\Conid{Right}\;\Varid{b}_{2}){}\<[E]%
\\
\>[5]{}\hsindent{6}{}\<[11]%
\>[11]{}\set{A}\;(\Conid{Left}\;\Varid{a}_{1}){}\<[30]%
\>[30]{}\mathrel{=}\mathbf{do}\;{}\<[36]%
\>[36]{}(\Varid{b},\Varid{s}_{1},\Varid{s}_{2})\leftarrow \Varid{get}{}\<[E]%
\\
\>[36]{}((),\Varid{s}_{1}')\leftarrow \Varid{lift}\;(\Varid{runStateT}\;(\Varid{setl}\;bx_{1}\;\Varid{a}_{1})\;\Varid{s}_{1}){}\<[E]%
\\
\>[36]{}\Varid{put}\;(\Conid{True},\Varid{s}_{1}',\Varid{s}_{2}){}\<[E]%
\\
\>[5]{}\hsindent{6}{}\<[11]%
\>[11]{}\set{A}\;(\Conid{Right}\;\Varid{a}_{2}){}\<[30]%
\>[30]{}\mathrel{=}\mathbf{do}\;{}\<[36]%
\>[36]{}(\Varid{b},\Varid{s}_{1},\Varid{s}_{2})\leftarrow \Varid{get}{}\<[E]%
\\
\>[36]{}((),\Varid{s}_{2}')\leftarrow \Varid{lift}\;(\Varid{runStateT}\;(\Varid{setl}\;bx_{2}\;\Varid{a}_{2})\;\Varid{s}_{2}){}\<[E]%
\\
\>[36]{}\Varid{put}\;(\Conid{False},\Varid{s}_{1},\Varid{s}_{2}'){}\<[E]%
\\
\>[5]{}\hsindent{6}{}\<[11]%
\>[11]{}\set{B}\;(\Conid{Left}\;\Varid{b}_{1}){}\<[30]%
\>[30]{}\mathrel{=}\mathbf{do}\;{}\<[36]%
\>[36]{}(\Varid{b},\Varid{s}_{1},\Varid{s}_{2})\leftarrow \Varid{get}{}\<[E]%
\\
\>[36]{}((),\Varid{s}_{1}')\leftarrow \Varid{lift}\;(\Varid{runStateT}\;(\Varid{setr}\;bx_{1}\;\Varid{b}_{1})\;\Varid{s}_{1}){}\<[E]%
\\
\>[36]{}\Varid{put}\;(\Conid{True},\Varid{s}_{1}',\Varid{s}_{2}){}\<[E]%
\\
\>[5]{}\hsindent{6}{}\<[11]%
\>[11]{}\set{B}\;(\Conid{Right}\;\Varid{b}_{2}){}\<[30]%
\>[30]{}\mathrel{=}\mathbf{do}\;{}\<[36]%
\>[36]{}(\Varid{b},\Varid{s}_{1},\Varid{s}_{2})\leftarrow \Varid{get}{}\<[E]%
\\
\>[36]{}((),\Varid{s}_{2}')\leftarrow \Varid{lift}\;(\Varid{runStateT}\;(\Varid{setr}\;bx_{2}\;\Varid{b}_{2})\;\Varid{s}_{2}){}\<[E]%
\\
\>[36]{}\Varid{put}\;(\Conid{False},\Varid{s}_{1},\Varid{s}_{2}'){}\<[E]%
\\
\>[5]{}\hsindent{6}{}\<[11]%
\>[11]{}\Varid{initA}\;(\Conid{Left}\;\Varid{a}_{1}){}\<[30]%
\>[30]{}\mathrel{=}\mathbf{do}\;{}\<[36]%
\>[36]{}\Varid{s}_{1}\leftarrow \Varid{initl}\;bx_{1}\;\Varid{a}_{1}{}\<[E]%
\\
\>[36]{}\Varid{return}\;(\Conid{True},\Varid{s}_{1},\bot ){}\<[E]%
\\
\>[5]{}\hsindent{6}{}\<[11]%
\>[11]{}\Varid{initA}\;(\Conid{Right}\;\Varid{a}_{2}){}\<[30]%
\>[30]{}\mathrel{=}\mathbf{do}\;{}\<[36]%
\>[36]{}\Varid{s}_{2}\leftarrow \Varid{initl}\;bx_{2}\;\Varid{a}_{2}{}\<[E]%
\\
\>[36]{}\Varid{return}\;(\Conid{False},\bot ,\Varid{s}_{2}){}\<[E]%
\\
\>[5]{}\hsindent{6}{}\<[11]%
\>[11]{}\Varid{initB}\;(\Conid{Left}\;\Varid{b}_{1}){}\<[30]%
\>[30]{}\mathrel{=}\mathbf{do}\;{}\<[36]%
\>[36]{}\Varid{s}_{1}\leftarrow \Varid{initr}\;bx_{1}\;\Varid{b}_{1}{}\<[E]%
\\
\>[36]{}\Varid{return}\;(\Conid{True},\Varid{s}_{1},\bot ){}\<[E]%
\\
\>[5]{}\hsindent{6}{}\<[11]%
\>[11]{}\Varid{initB}\;(\Conid{Right}\;\Varid{b}_{2}){}\<[30]%
\>[30]{}\mathrel{=}\mathbf{do}\;{}\<[36]%
\>[36]{}\Varid{s}_{2}\leftarrow \Varid{initr}\;bx_{2}\;\Varid{b}_{2}{}\<[E]%
\\
\>[36]{}\Varid{return}\;(\Conid{False},\bot ,\Varid{s}_{2}){}\<[E]%
\ColumnHook
\end{hscode}\resethooks
List 
\begin{hscode}\SaveRestoreHook
\column{B}{@{}>{\hspre}l<{\hspost}@{}}%
\column{3}{@{}>{\hspre}l<{\hspost}@{}}%
\column{5}{@{}>{\hspre}l<{\hspost}@{}}%
\column{11}{@{}>{\hspre}l<{\hspost}@{}}%
\column{14}{@{}>{\hspre}l<{\hspost}@{}}%
\column{23}{@{}>{\hspre}c<{\hspost}@{}}%
\column{23E}{@{}l@{}}%
\column{25}{@{}>{\hspre}l<{\hspost}@{}}%
\column{26}{@{}>{\hspre}l<{\hspost}@{}}%
\column{32}{@{}>{\hspre}l<{\hspost}@{}}%
\column{40}{@{}>{\hspre}c<{\hspost}@{}}%
\column{40E}{@{}l@{}}%
\column{43}{@{}>{\hspre}l<{\hspost}@{}}%
\column{47}{@{}>{\hspre}l<{\hspost}@{}}%
\column{E}{@{}>{\hspre}l<{\hspost}@{}}%
\>[3]{}\Varid{listBX}\mathbin{::}{}\<[14]%
\>[14]{}\Conid{Monad}\;\Varid{m}\Rightarrow {}\<[26]%
\>[26]{}\Conid{StateTBX}\;\Varid{m}\;\Varid{s}\;\Varid{a}\;\Varid{b}\hsarrow{\rightarrow }{\mathpunct{.}}{}\<[E]%
\\
\>[26]{}\Conid{StateTBX}\;\Varid{m}\;(\Conid{Int},[\mskip1.5mu \Varid{s}\mskip1.5mu])\;[\mskip1.5mu \Varid{a}\mskip1.5mu]\;[\mskip1.5mu \Varid{b}\mskip1.5mu]{}\<[E]%
\\
\>[3]{}\Varid{listBX}\;bx\mathrel{=}\Conid{StateTBX}\;{}\<[25]%
\>[25]{}\get{L}\;\set{L}\;\Varid{init}_{L}\;\get{R}\;\set{R}\;\Varid{init}_{R}{}\<[E]%
\\
\>[3]{}\hsindent{2}{}\<[5]%
\>[5]{}\mathbf{where}\;\get{L}{}\<[23]%
\>[23]{}\mathrel{=}{}\<[23E]%
\>[26]{}\mathbf{do}\;\{\mskip1.5mu {}\<[32]%
\>[32]{}(\Varid{n},\Varid{cs})\leftarrow \Varid{get};{}\<[E]%
\\
\>[32]{}\Varid{mapM}\;(\Varid{lift}\hsdot{\cdot }{.}\Varid{evalStateT}\;(\Varid{getl}\;bx))\;(\Varid{take}\;\Varid{n}\;\Varid{cs})\mskip1.5mu\}{}\<[E]%
\\
\>[5]{}\hsindent{6}{}\<[11]%
\>[11]{}\get{R}{}\<[23]%
\>[23]{}\mathrel{=}{}\<[23E]%
\>[26]{}\mathbf{do}\;\{\mskip1.5mu {}\<[32]%
\>[32]{}(\Varid{n},\Varid{cs})\leftarrow \Varid{get};{}\<[E]%
\\
\>[32]{}\Varid{mapM}\;(\Varid{lift}\hsdot{\cdot }{.}\Varid{evalStateT}\;(\Varid{getr}\;bx))\;(\Varid{take}\;\Varid{n}\;\Varid{cs})\mskip1.5mu\}{}\<[E]%
\\
\>[5]{}\hsindent{6}{}\<[11]%
\>[11]{}\set{L}\;\Varid{as}{}\<[23]%
\>[23]{}\mathrel{=}{}\<[23E]%
\>[26]{}\mathbf{do}\;\{\mskip1.5mu {}\<[32]%
\>[32]{}(\anonymous ,\Varid{cs})\leftarrow \Varid{get};{}\<[E]%
\\
\>[32]{}\Varid{cs'}\leftarrow \Varid{sets}\;(\Varid{setl}\;bx)\;(\Varid{initl}\;bx)\;\Varid{cs}\;\Varid{as};{}\<[E]%
\\
\>[32]{}\Varid{put}\;(\Varid{length}\;\Varid{as},\Varid{cs'})\mskip1.5mu\}{}\<[E]%
\\
\>[5]{}\hsindent{6}{}\<[11]%
\>[11]{}\set{R}\;\Varid{bs}{}\<[23]%
\>[23]{}\mathrel{=}{}\<[23E]%
\>[26]{}\mathbf{do}\;\{\mskip1.5mu {}\<[32]%
\>[32]{}(\anonymous ,\Varid{cs})\leftarrow \Varid{get};{}\<[E]%
\\
\>[32]{}\Varid{cs'}\leftarrow \Varid{sets}\;(\Varid{setr}\;bx)\;(\Varid{initr}\;bx)\;\Varid{cs}\;\Varid{bs};{}\<[E]%
\\
\>[32]{}\Varid{put}\;(\Varid{length}\;\Varid{bs},\Varid{cs'})\mskip1.5mu\}{}\<[E]%
\\
\>[5]{}\hsindent{6}{}\<[11]%
\>[11]{}\Varid{init}_{L}\;\Varid{as}{}\<[23]%
\>[23]{}\mathrel{=}{}\<[23E]%
\>[26]{}\mathbf{do}\;\{\mskip1.5mu {}\<[32]%
\>[32]{}\Varid{cs}\leftarrow \Varid{mapM}\;(\Varid{initl}\;bx)\;\Varid{as};{}\<[E]%
\\
\>[32]{}\Varid{return}\;(\Varid{length}\;\Varid{as},\Varid{cs})\mskip1.5mu\}{}\<[E]%
\\
\>[5]{}\hsindent{6}{}\<[11]%
\>[11]{}\Varid{init}_{R}\;\Varid{bs}{}\<[23]%
\>[23]{}\mathrel{=}{}\<[23E]%
\>[26]{}\mathbf{do}\;\{\mskip1.5mu {}\<[32]%
\>[32]{}\Varid{cs}\leftarrow \Varid{mapM}\;(\Varid{initr}\;bx)\;\Varid{bs};{}\<[E]%
\\
\>[32]{}\Varid{return}\;(\Varid{length}\;\Varid{bs},\Varid{cs})\mskip1.5mu\}{}\<[E]%
\\[\blanklineskip]%
\>[5]{}\hsindent{6}{}\<[11]%
\>[11]{}\Varid{sets}\;\Varid{set}\;\Varid{init}\;[\mskip1.5mu \mskip1.5mu]\;[\mskip1.5mu \mskip1.5mu]{}\<[40]%
\>[40]{}\mathrel{=}{}\<[40E]%
\>[43]{}\Varid{return}\;[\mskip1.5mu \mskip1.5mu]{}\<[E]%
\\
\>[5]{}\hsindent{6}{}\<[11]%
\>[11]{}\Varid{sets}\;\Varid{set}\;\Varid{init}\;(\Varid{c}\mathbin{:}\Varid{cs})\;(\Varid{x}\mathbin{:}\Varid{xs}){}\<[40]%
\>[40]{}\mathrel{=}{}\<[40E]%
\>[43]{}\mathbf{do}\;{}\<[47]%
\>[47]{}\Varid{c'}\leftarrow \Varid{lift}\;(\Varid{execStateT}\;(\Varid{set}\;\Varid{x})\;\Varid{c}){}\<[E]%
\\
\>[47]{}\Varid{cs'}\leftarrow \Varid{sets}\;\Varid{set}\;\Varid{init}\;\Varid{cs}\;\Varid{xs}{}\<[E]%
\\
\>[47]{}\Varid{return}\;(\Varid{c'}\mathbin{:}\Varid{cs'}){}\<[E]%
\\
\>[5]{}\hsindent{6}{}\<[11]%
\>[11]{}\Varid{sets}\;\Varid{set}\;\Varid{init}\;\Varid{cs}\;[\mskip1.5mu \mskip1.5mu]{}\<[40]%
\>[40]{}\mathrel{=}{}\<[40E]%
\>[43]{}\Varid{return}\;\Varid{cs}{}\<[E]%
\\
\>[5]{}\hsindent{6}{}\<[11]%
\>[11]{}\Varid{sets}\;\Varid{set}\;\Varid{init}\;[\mskip1.5mu \mskip1.5mu]\;\Varid{xs}{}\<[40]%
\>[40]{}\mathrel{=}{}\<[40E]%
\>[43]{}\Varid{lift}\;(\Varid{mapM}\;\Varid{init}\;\Varid{xs}){}\<[E]%
\ColumnHook
\end{hscode}\resethooks

\subsection{Composers example}
\begin{hscode}\SaveRestoreHook
\column{B}{@{}>{\hspre}l<{\hspost}@{}}%
\column{3}{@{}>{\hspre}l<{\hspost}@{}}%
\column{E}{@{}>{\hspre}l<{\hspost}@{}}%
\>[3]{}\mbox{\enskip\{-\# LANGUAGE MultiParamTypeClasses, ScopedTypeVariables  \#-\}\enskip}{}\<[E]%
\ColumnHook
\end{hscode}\resethooks
\begin{hscode}\SaveRestoreHook
\column{B}{@{}>{\hspre}l<{\hspost}@{}}%
\column{3}{@{}>{\hspre}l<{\hspost}@{}}%
\column{E}{@{}>{\hspre}l<{\hspost}@{}}%
\>[3]{}\mathbf{module}\;\Conid{Composers}\;\mathbf{where}{}\<[E]%
\\
\>[3]{}\mathbf{import}\;\Conid{\Conid{Data}.List}\;\Varid{as}\;\Conid{List}{}\<[E]%
\\
\>[3]{}\mathbf{import}\;\Conid{Data}.\ensuremath{\mathsf{Set}}\;\Varid{as}\;\ensuremath{\mathsf{Set}}{}\<[E]%
\\
\>[3]{}\mathbf{import}\;\Conid{\Conid{Control}.\Conid{Monad}.State}\;\Varid{as}\;\Conid{State}{}\<[E]%
\\
\>[3]{}\mathbf{import}\;\Conid{Control}.\Conid{Monad}.\Conid{Id}{}\<[E]%
\\
\>[3]{}\mathbf{import}\;\Conid{SLens}{}\<[E]%
\\
\>[3]{}\mathbf{import}\;\Conid{StateTBX}{}\<[E]%
\ColumnHook
\end{hscode}\resethooks
Here is a version of the familiar Composers example
\cite{composers},
see the Bx wiki;
versions of this have been used in many papers including e.g. the
Symmetric Lens paper \cite{symlens}.

Assumption: Name is a key in all our datastructures: the user is
required not to give as argument any view that contains more than one
element for a given name.

NB We are not saying this version is better than any other version:
it's just an illustration.
\begin{hscode}\SaveRestoreHook
\column{B}{@{}>{\hspre}l<{\hspost}@{}}%
\column{3}{@{}>{\hspre}l<{\hspost}@{}}%
\column{5}{@{}>{\hspre}l<{\hspost}@{}}%
\column{7}{@{}>{\hspre}l<{\hspost}@{}}%
\column{9}{@{}>{\hspre}l<{\hspost}@{}}%
\column{11}{@{}>{\hspre}l<{\hspost}@{}}%
\column{14}{@{}>{\hspre}l<{\hspost}@{}}%
\column{17}{@{}>{\hspre}l<{\hspost}@{}}%
\column{23}{@{}>{\hspre}l<{\hspost}@{}}%
\column{24}{@{}>{\hspre}l<{\hspost}@{}}%
\column{E}{@{}>{\hspre}l<{\hspost}@{}}%
\>[3]{}\Varid{composers}\mathbin{::}\Conid{SLens}\;{}\<[23]%
\>[23]{}[\mskip1.5mu (\Conid{Name},\Conid{Dates})\mskip1.5mu]\;{}\<[E]%
\\
\>[23]{}(\ensuremath{\mathsf{Set}}\;(\Conid{Name},\Conid{Nation},\Conid{Dates}))\;{}\<[E]%
\\
\>[23]{}[\mskip1.5mu (\Conid{Name},\Conid{Nation})\mskip1.5mu]{}\<[E]%
\\
\>[3]{}\Varid{composers}\mathrel{=}\Conid{SLens}\;\{\mskip1.5mu {}\<[24]%
\>[24]{}\Varid{putr}\mathrel{=}\Varid{putMN},\Varid{putl}\mathrel{=}\Varid{putNM},{}\<[E]%
\\
\>[24]{}\Varid{missing}\mathrel{=}[\mskip1.5mu \mskip1.5mu]\mskip1.5mu\}{}\<[E]%
\\
\>[3]{}\hsindent{2}{}\<[5]%
\>[5]{}\mathbf{where}{}\<[E]%
\\
\>[5]{}\hsindent{2}{}\<[7]%
\>[7]{}\Varid{putMN}\;(\Varid{m},\Varid{c})\mathrel{=}(\Varid{n},\Varid{c'}){}\<[E]%
\\
\>[7]{}\hsindent{2}{}\<[9]%
\>[9]{}\mathbf{where}{}\<[E]%
\\
\>[9]{}\hsindent{2}{}\<[11]%
\>[11]{}\Varid{n}\mathrel{=}\Varid{selectNNfromNND}\;\Varid{tripleList}{}\<[E]%
\\
\>[9]{}\hsindent{2}{}\<[11]%
\>[11]{}\Varid{c'}\mathrel{=}\Varid{selectNDfromNND}\;\Varid{tripleList}{}\<[E]%
\\
\>[9]{}\hsindent{2}{}\<[11]%
\>[11]{}\Varid{tripleList}\mathrel{=}\Varid{h}\;\Varid{c}\;[\mskip1.5mu \mskip1.5mu]\;(\Varid{\ensuremath{\mathsf{Set}}.toList}\;\Varid{m}){}\<[E]%
\\
\>[9]{}\hsindent{2}{}\<[11]%
\>[11]{}\Varid{h}\;[\mskip1.5mu \mskip1.5mu]\;\Varid{sel}\;\Varid{leftover}\mathrel{=}\Varid{reverse}\;\Varid{sel}\plus \Varid{sort}\;\Varid{leftover}{}\<[E]%
\\
\>[9]{}\hsindent{2}{}\<[11]%
\>[11]{}\Varid{h}\;((\Varid{nn},\anonymous )\mathbin{:}\Varid{cs})\;\Varid{ss}\;\Varid{ls}\mathrel{=}\Varid{h}\;\Varid{cs}\;(\Varid{ps}\plus \Varid{ss})\;\Varid{ns}{}\<[E]%
\\
\>[11]{}\hsindent{3}{}\<[14]%
\>[14]{}\mathbf{where}\;(\Varid{ps},\Varid{ns})\mathrel{=}\Varid{selectNNDonKey}\;\Varid{nn}\;\Varid{ls}{}\<[E]%
\\
\>[5]{}\hsindent{2}{}\<[7]%
\>[7]{}\Varid{putNM}\;(\Varid{n},\Varid{c})\mathrel{=}(\Varid{m},\Varid{c'}){}\<[E]%
\\
\>[7]{}\hsindent{2}{}\<[9]%
\>[9]{}\mathbf{where}{}\<[E]%
\\
\>[7]{}\hsindent{2}{}\<[9]%
\>[9]{}\Varid{m}\mathrel{=}\Varid{\ensuremath{\mathsf{Set}}.fromList}\;\Varid{tripleList}{}\<[E]%
\\
\>[7]{}\hsindent{2}{}\<[9]%
\>[9]{}\Varid{c'}\mathrel{=}\Varid{selectNDfromNND}\;\Varid{tripleList}{}\<[E]%
\\
\>[7]{}\hsindent{2}{}\<[9]%
\>[9]{}\Varid{tripleList}\mathrel{=}\Varid{k}\;\Varid{n}\;[\mskip1.5mu \mskip1.5mu]\;\Varid{c}{}\<[E]%
\\
\>[7]{}\hsindent{2}{}\<[9]%
\>[9]{}\Varid{k}\;[\mskip1.5mu \mskip1.5mu]\;\Varid{selected}\;\anonymous \mathrel{=}\Varid{\Conid{List}.reverse}\;\Varid{selected}{}\<[E]%
\\
\>[7]{}\hsindent{2}{}\<[9]%
\>[9]{}\Varid{k}\;(\Varid{h}\mathord{@}(\Varid{nn},\anonymous )\mathbin{:}\Varid{nts})\;\Varid{ss}\;\Varid{ls}\mathrel{=}\Varid{k}\;\Varid{nts}\;(\Varid{newTriple}\mathbin{:}\Varid{ss})\;\Varid{ns}{}\<[E]%
\\
\>[9]{}\hsindent{2}{}\<[11]%
\>[11]{}\mathbf{where}\;(\Varid{ps},\Varid{ns})\mathrel{=}\Varid{selectNDonKey}\;\Varid{nn}\;\Varid{ls}{}\<[E]%
\\
\>[11]{}\hsindent{6}{}\<[17]%
\>[17]{}\Varid{newTriple}\mathrel{=}\Varid{newTripleFromList}\;\Varid{h}\;(\lambda \hslambda (\anonymous ,\Varid{d})\hsarrow{\rightarrow }{\mathpunct{.}}\Varid{d})\;\Varid{ps}{}\<[E]%
\ColumnHook
\end{hscode}\resethooks
where the useful `select' statements are packaged as follows
\begin{hscode}\SaveRestoreHook
\column{B}{@{}>{\hspre}l<{\hspost}@{}}%
\column{3}{@{}>{\hspre}l<{\hspost}@{}}%
\column{20}{@{}>{\hspre}l<{\hspost}@{}}%
\column{E}{@{}>{\hspre}l<{\hspost}@{}}%
\>[3]{}\mathbf{type}\;\Conid{NND}\mathrel{=}(\Conid{Name},\Conid{Nation},\Conid{Dates}){}\<[E]%
\\
\>[3]{}\Varid{selectNNDonKey}\mathbin{::}\Conid{Name}\hsarrow{\rightarrow }{\mathpunct{.}}[\mskip1.5mu \Conid{NND}\mskip1.5mu]\hsarrow{\rightarrow }{\mathpunct{.}}([\mskip1.5mu \Conid{NND}\mskip1.5mu],[\mskip1.5mu \Conid{NND}\mskip1.5mu]){}\<[E]%
\\
\>[3]{}\Varid{selectNNDonKey}\;\Varid{n}\mathrel{=}\Varid{\Conid{List}.partition}\;(\lambda \hslambda (\Varid{nn},\anonymous ,\anonymous )\hsarrow{\rightarrow }{\mathpunct{.}}\Varid{nn}\equals\Varid{n}){}\<[E]%
\\[\blanklineskip]%
\>[3]{}\mathbf{type}\;\Conid{ND}\mathrel{=}(\Conid{Name},\Conid{Dates}){}\<[E]%
\\
\>[3]{}\Varid{selectNDonKey}\mathbin{::}\Conid{Name}\hsarrow{\rightarrow }{\mathpunct{.}}[\mskip1.5mu \Conid{ND}\mskip1.5mu]\hsarrow{\rightarrow }{\mathpunct{.}}([\mskip1.5mu \Conid{ND}\mskip1.5mu],[\mskip1.5mu \Conid{ND}\mskip1.5mu]){}\<[E]%
\\
\>[3]{}\Varid{selectNDonKey}\;\Varid{n}\mathrel{=}\Varid{\Conid{List}.partition}\;(\lambda \hslambda (\Varid{nn},\anonymous )\hsarrow{\rightarrow }{\mathpunct{.}}\Varid{nn}\equals\Varid{n}){}\<[E]%
\\[\blanklineskip]%
\>[3]{}\Varid{selectNDfromNND}\mathbin{::}[\mskip1.5mu \Conid{NND}\mskip1.5mu]\hsarrow{\rightarrow }{\mathpunct{.}}[\mskip1.5mu \Conid{ND}\mskip1.5mu]{}\<[E]%
\\
\>[3]{}\Varid{selectNDfromNND}{}\<[20]%
\>[20]{}\mathrel{=}\Varid{\Conid{List}.map}\;(\lambda \hslambda (\Varid{nn},\Varid{nt},\Varid{dd})\hsarrow{\rightarrow }{\mathpunct{.}}(\Varid{nn},\Varid{dd})){}\<[E]%
\\[\blanklineskip]%
\>[3]{}\mathbf{type}\;\Conid{NN}\mathrel{=}(\Conid{Name},\Conid{Nation}){}\<[E]%
\\
\>[3]{}\Varid{selectNNfromNND}\mathbin{::}[\mskip1.5mu \Conid{NND}\mskip1.5mu]\hsarrow{\rightarrow }{\mathpunct{.}}[\mskip1.5mu \Conid{NN}\mskip1.5mu]{}\<[E]%
\\
\>[3]{}\Varid{selectNNfromNND}{}\<[20]%
\>[20]{}\mathrel{=}\Varid{\Conid{List}.map}\;(\lambda \hslambda (\Varid{nn},\Varid{nt},\Varid{dd})\hsarrow{\rightarrow }{\mathpunct{.}}(\Varid{nn},\Varid{nt})){}\<[E]%
\\
\>[3]{}\Varid{mkNNDfromNN}\mathbin{::}[\mskip1.5mu \Conid{NN}\mskip1.5mu]\hsarrow{\rightarrow }{\mathpunct{.}}[\mskip1.5mu \Conid{NND}\mskip1.5mu]{}\<[E]%
\\
\>[3]{}\Varid{mkNNDfromNN}\mathrel{=}\Varid{\Conid{List}.map}\;(\lambda \hslambda (\Varid{nn},\Varid{nt})\hsarrow{\rightarrow }{\mathpunct{.}}(\Varid{nn},\Varid{nt},\Varid{dates0})){}\<[E]%
\ColumnHook
\end{hscode}\resethooks
This last helper function abstracts how to make a new triple
\begin{hscode}\SaveRestoreHook
\column{B}{@{}>{\hspre}l<{\hspost}@{}}%
\column{3}{@{}>{\hspre}l<{\hspost}@{}}%
\column{E}{@{}>{\hspre}l<{\hspost}@{}}%
\>[3]{}\Varid{newTripleFromList}\mathbin{::}\Conid{NN}\hsarrow{\rightarrow }{\mathpunct{.}}(\Varid{a}\hsarrow{\rightarrow }{\mathpunct{.}}\Conid{Dates})\hsarrow{\rightarrow }{\mathpunct{.}}[\mskip1.5mu \Varid{a}\mskip1.5mu]\hsarrow{\rightarrow }{\mathpunct{.}}\Conid{NND}{}\<[E]%
\\
\>[3]{}\Varid{newTripleFromList}\;(\Varid{nn},\Varid{nt})\;\anonymous \;[\mskip1.5mu \mskip1.5mu]\mathrel{=}(\Varid{nn},\Varid{nt},\Varid{dates0}){}\<[E]%
\\
\>[3]{}\Varid{newTripleFromList}\;(\Varid{nn},\Varid{nt})\;\Varid{f}\;(\Varid{a}\mathbin{:\char95 })\mathrel{=}(\Varid{nn},\Varid{nt},\Varid{f}\;\Varid{a}){}\<[E]%
\ColumnHook
\end{hscode}\resethooks
Now here is the same functionality as a bx. There are no effects other
than the state ones induced directly by the BX, so the underlying
monad is the Identity monad.
\begin{hscode}\SaveRestoreHook
\column{B}{@{}>{\hspre}l<{\hspost}@{}}%
\column{3}{@{}>{\hspre}l<{\hspost}@{}}%
\column{5}{@{}>{\hspre}l<{\hspost}@{}}%
\column{7}{@{}>{\hspre}l<{\hspost}@{}}%
\column{9}{@{}>{\hspre}l<{\hspost}@{}}%
\column{11}{@{}>{\hspre}l<{\hspost}@{}}%
\column{15}{@{}>{\hspre}l<{\hspost}@{}}%
\column{28}{@{}>{\hspre}l<{\hspost}@{}}%
\column{42}{@{}>{\hspre}l<{\hspost}@{}}%
\column{E}{@{}>{\hspre}l<{\hspost}@{}}%
\>[3]{}\Varid{composersBx}\mathbin{::}\Conid{StateTBX}\;{}\<[28]%
\>[28]{}\Conid{Id}\;{}\<[E]%
\\
\>[28]{}[\mskip1.5mu (\Conid{Name},\Conid{Nation},\Conid{Dates})\mskip1.5mu]\;{}\<[E]%
\\
\>[28]{}(\ensuremath{\mathsf{Set}}\;(\Conid{Name},\Conid{Nation},\Conid{Dates}))\;{}\<[E]%
\\
\>[28]{}[\mskip1.5mu (\Conid{Name},\Conid{Nation})\mskip1.5mu]{}\<[E]%
\\
\>[3]{}\Varid{composersBx}\mathrel{=}\Conid{StateTBX}\;\Varid{getl}\;\Varid{setl}\;\Varid{initl}\;\Varid{getr}\;\Varid{setr}\;\Varid{initr}{}\<[E]%
\\
\>[3]{}\hsindent{2}{}\<[5]%
\>[5]{}\mathbf{where}{}\<[E]%
\\
\>[5]{}\hsindent{2}{}\<[7]%
\>[7]{}\Varid{getl}\mathrel{=}\Varid{state}\;(\lambda \hslambda \Varid{l}\hsarrow{\rightarrow }{\mathpunct{.}}(\Varid{\ensuremath{\mathsf{Set}}.fromList}\;\Varid{l},\Varid{l})){}\<[E]%
\\
\>[5]{}\hsindent{2}{}\<[7]%
\>[7]{}\Varid{setl}\mathrel{=}(\lambda \hslambda \Varid{m}\hsarrow{\rightarrow }{\mathpunct{.}}\Varid{state}\;(\lambda \hslambda \Varid{l}\hsarrow{\rightarrow }{\mathpunct{.}}((),\Varid{f}\;(\Varid{\ensuremath{\mathsf{Set}}.toList}\;\Varid{m})\;[\mskip1.5mu \mskip1.5mu]\;\Varid{l}))){}\<[E]%
\\
\>[5]{}\hsindent{2}{}\<[7]%
\>[7]{}\Varid{initl}\mathrel{=}(\lambda \hslambda \Varid{m}\hsarrow{\rightarrow }{\mathpunct{.}}\Varid{return}\;(\Varid{\ensuremath{\mathsf{Set}}.toList}\;\Varid{m})){}\<[E]%
\\[\blanklineskip]%
\>[5]{}\hsindent{2}{}\<[7]%
\>[7]{}\Varid{f}\;\Varid{leftovers}\;\Varid{upd}\;[\mskip1.5mu \mskip1.5mu]\mathrel{=}(\Varid{reverse}\;\Varid{upd})\plus \Varid{sort}\;\Varid{leftovers}{}\<[E]%
\\
\>[5]{}\hsindent{2}{}\<[7]%
\>[7]{}\Varid{f}\;\Varid{leftovers}\;\Varid{upd}\;((\Varid{nn},\Varid{na},\Varid{nd})\mathbin{:}\Varid{rs}){}\<[42]%
\>[42]{}\mathrel{=}\Varid{f}\;\Varid{ns}\;(\Varid{ps}\plus \Varid{upd})\;\Varid{rs}{}\<[E]%
\\
\>[7]{}\hsindent{2}{}\<[9]%
\>[9]{}\mathbf{where}\;(\Varid{ps},\Varid{ns})\mathrel{=}\Varid{selectNNDonKey}\;\Varid{nn}\;\Varid{leftovers}{}\<[E]%
\\[\blanklineskip]%
\>[5]{}\hsindent{2}{}\<[7]%
\>[7]{}\Varid{getr}\mathrel{=}\Varid{state}\;(\lambda \hslambda \Varid{l}\hsarrow{\rightarrow }{\mathpunct{.}}(\Varid{selectNNfromNND}\;\Varid{l},\Varid{l})){}\<[E]%
\\
\>[5]{}\hsindent{2}{}\<[7]%
\>[7]{}\Varid{setr}\mathrel{=}(\lambda \hslambda \Varid{n}\hsarrow{\rightarrow }{\mathpunct{.}}\Varid{state}\;(\lambda \hslambda \Varid{l}\hsarrow{\rightarrow }{\mathpunct{.}}((),\Varid{g}\;\Varid{n}\;[\mskip1.5mu \mskip1.5mu]\;\Varid{l}))){}\<[E]%
\\
\>[5]{}\hsindent{2}{}\<[7]%
\>[7]{}\Varid{initr}\mathrel{=}(\lambda \hslambda \Varid{n}\hsarrow{\rightarrow }{\mathpunct{.}}\Varid{return}\;(\Varid{mkNNDfromNN}\;\Varid{n})){}\<[E]%
\\[\blanklineskip]%
\>[5]{}\hsindent{2}{}\<[7]%
\>[7]{}\Varid{g}\;[\mskip1.5mu \mskip1.5mu]\;\Varid{updated}\;\Varid{stateNotSeenYet}\mathrel{=}\Varid{reverse}\;\Varid{updated}{}\<[E]%
\\
\>[5]{}\hsindent{2}{}\<[7]%
\>[7]{}\Varid{g}\;(\Varid{h}\mathord{@}(\Varid{nn},\Varid{na})\mathbin{:}\Varid{todo})\;\Varid{updated}\;\Varid{stateNotSeenYet}\mathrel{=}{}\<[E]%
\\
\>[7]{}\hsindent{4}{}\<[11]%
\>[11]{}\Varid{g}\;\Varid{todo}\;(\Varid{newTriple}\mathbin{:}\Varid{updated})\;\Varid{ns}{}\<[E]%
\\
\>[7]{}\hsindent{2}{}\<[9]%
\>[9]{}\mathbf{where}\;(\Varid{ps},\Varid{ns})\mathrel{=}\Varid{selectNNDonKey}\;\Varid{nn}\;\Varid{stateNotSeenYet}{}\<[E]%
\\
\>[9]{}\hsindent{6}{}\<[15]%
\>[15]{}\Varid{newTriple}\mathrel{=}\Varid{newTripleFromList}\;\Varid{h}\;(\lambda \hslambda (\anonymous ,\anonymous ,\Varid{d})\hsarrow{\rightarrow }{\mathpunct{.}}\Varid{d})\;\Varid{ps}{}\<[E]%
\ColumnHook
\end{hscode}\resethooks
Now let's see how to use both the symmetric lens and the bx versions,
and demonstrate them behaving the same.

\begin{enumerate}

\item
Initialise both with no composers at all.
\begin{hscode}\SaveRestoreHook
\column{B}{@{}>{\hspre}l<{\hspost}@{}}%
\column{3}{@{}>{\hspre}l<{\hspost}@{}}%
\column{E}{@{}>{\hspre}l<{\hspost}@{}}%
\>[3]{}(\Varid{m}_{1},\Varid{c}_{1})\mathrel{=}\Varid{putl}\;\Varid{composers}\;([\mskip1.5mu \mskip1.5mu],\Varid{missing}\;\Varid{composers}){}\<[E]%
\\
\>[3]{}\Varid{s}_{1}\mathrel{=}\Varid{runIdentity}\;(\Varid{initr}\;\Varid{composersBx}\;[\mskip1.5mu \mskip1.5mu]){}\<[E]%
\ColumnHook
\end{hscode}\resethooks
\item
Now suppose the owner of the left-hand view likes JS Bach.
\begin{hscode}\SaveRestoreHook
\column{B}{@{}>{\hspre}l<{\hspost}@{}}%
\column{3}{@{}>{\hspre}l<{\hspost}@{}}%
\column{13}{@{}>{\hspre}c<{\hspost}@{}}%
\column{13E}{@{}l@{}}%
\column{16}{@{}>{\hspre}l<{\hspost}@{}}%
\column{E}{@{}>{\hspre}l<{\hspost}@{}}%
\>[3]{}\Varid{jsBach}\mathrel{=}{}\<[13]%
\>[13]{}({}\<[13E]%
\>[16]{}\Conid{Name}\;\text{\tt \char34 J.~S.~Bach\char34},\Conid{Nation}\;\text{\tt \char34 German\char34},{}\<[E]%
\\
\>[16]{}\Conid{Dates}\;(\Conid{Just}\;(\Conid{Date}\;\text{\tt \char34 1685\char34},\Conid{Date}\;\text{\tt \char34 1750\char34}))){}\<[E]%
\\
\>[3]{}\Varid{onlyBach}\mathrel{=}\Varid{\ensuremath{\mathsf{Set}}.fromList}\;([\mskip1.5mu \Varid{jsBach}\mskip1.5mu]){}\<[E]%
\ColumnHook
\end{hscode}\resethooks
Putting this into the symmetric lens version:
\begin{hscode}\SaveRestoreHook
\column{B}{@{}>{\hspre}l<{\hspost}@{}}%
\column{3}{@{}>{\hspre}l<{\hspost}@{}}%
\column{E}{@{}>{\hspre}l<{\hspost}@{}}%
\>[3]{}(\Varid{n1sl},\Varid{c}_{2})\mathrel{=}\Varid{putr}\;\Varid{composers}\;(\Varid{onlyBach},\Varid{c}_{1}){}\<[E]%
\ColumnHook
\end{hscode}\resethooks
and into the bx version (the underscore represents the result of the monadic
computation; we could use \ensuremath{\Varid{evalState}} if we didn't like it, but this is
just standard Haskell-monad-cruft, nothing to do with our formalism
specifically:
\begin{hscode}\SaveRestoreHook
\column{B}{@{}>{\hspre}l<{\hspost}@{}}%
\column{3}{@{}>{\hspre}l<{\hspost}@{}}%
\column{E}{@{}>{\hspre}l<{\hspost}@{}}%
\>[3]{}(\anonymous ,\Varid{s}_{2})\mathrel{=}\Varid{runState}\Varid{n}\;(\mathbf{do}\;\{\mskip1.5mu \Varid{setl}\;\Varid{composersBx}\;\Varid{onlyBach}\mskip1.5mu\})\;\Varid{s}_{1}{}\<[E]%
\ColumnHook
\end{hscode}\resethooks
\item
Let's check that what the owner of the right-hand view sees is the
same in both cases. n1 is that, for the symmetric lens (we got told,
whether we liked it or not). For the bx:
\begin{hscode}\SaveRestoreHook
\column{B}{@{}>{\hspre}l<{\hspost}@{}}%
\column{3}{@{}>{\hspre}l<{\hspost}@{}}%
\column{E}{@{}>{\hspre}l<{\hspost}@{}}%
\>[3]{}(\Varid{n1bx},\Varid{s}_{3})\mathrel{=}\Varid{runState}\Varid{n}\;(\mathbf{do}\;\{\mskip1.5mu \Varid{n}\leftarrow \Varid{getr}\;\Varid{composersBx};\Varid{return}\;\Varid{n}\mskip1.5mu\})\;\Varid{s}_{2}{}\<[E]%
\\
\>[3]{}\Varid{ok1}\mathrel{=}(\Varid{n1sl}\equals\Varid{n1bx}){}\<[E]%
\ColumnHook
\end{hscode}\resethooks
\item
The RH view owner also likes John Tavener:
\begin{hscode}\SaveRestoreHook
\column{B}{@{}>{\hspre}l<{\hspost}@{}}%
\column{3}{@{}>{\hspre}l<{\hspost}@{}}%
\column{20}{@{}>{\hspre}l<{\hspost}@{}}%
\column{E}{@{}>{\hspre}l<{\hspost}@{}}%
\>[3]{}\Varid{johnTavener}\mathrel{=}({}\<[20]%
\>[20]{}\Conid{Name}\;\text{\tt \char34 John~Tavener\char34},\Conid{Nation}\;\text{\tt \char34 British\char34}){}\<[E]%
\ColumnHook
\end{hscode}\resethooks
and decides to append:
\begin{hscode}\SaveRestoreHook
\column{B}{@{}>{\hspre}l<{\hspost}@{}}%
\column{3}{@{}>{\hspre}l<{\hspost}@{}}%
\column{E}{@{}>{\hspre}l<{\hspost}@{}}%
\>[3]{}\Varid{bachTavener}\mathrel{=}\Varid{n1sl}\plus [\mskip1.5mu \Varid{johnTavener}\mskip1.5mu]{}\<[E]%
\ColumnHook
\end{hscode}\resethooks
Putting this into the symmetric lens version:
\begin{hscode}\SaveRestoreHook
\column{B}{@{}>{\hspre}l<{\hspost}@{}}%
\column{3}{@{}>{\hspre}l<{\hspost}@{}}%
\column{E}{@{}>{\hspre}l<{\hspost}@{}}%
\>[3]{}(\Varid{m1sl},\Varid{c3})\mathrel{=}\Varid{putl}\;\Varid{composers}\;(\Varid{bachTavener},\Varid{c}_{2}){}\<[E]%
\ColumnHook
\end{hscode}\resethooks
and into the bx version:
\begin{hscode}\SaveRestoreHook
\column{B}{@{}>{\hspre}l<{\hspost}@{}}%
\column{3}{@{}>{\hspre}l<{\hspost}@{}}%
\column{E}{@{}>{\hspre}l<{\hspost}@{}}%
\>[3]{}(\anonymous ,\Varid{s4})\mathrel{=}\Varid{runState}\Varid{n}\;(\mathbf{do}\;\{\mskip1.5mu \Varid{setr}\;\Varid{composersBx}\;\Varid{bachTavener}\mskip1.5mu\})\;\Varid{s}_{3}{}\<[E]%
\ColumnHook
\end{hscode}\resethooks
yields the same result for the LH view owner:
\begin{hscode}\SaveRestoreHook
\column{B}{@{}>{\hspre}l<{\hspost}@{}}%
\column{3}{@{}>{\hspre}l<{\hspost}@{}}%
\column{E}{@{}>{\hspre}l<{\hspost}@{}}%
\>[3]{}(\Varid{m1bx},\Varid{s5})\mathrel{=}\Varid{runState}\Varid{n}\;(\mathbf{do}\;\{\mskip1.5mu \Varid{m}\leftarrow \Varid{getl}\;\Varid{composersBx};\Varid{return}\;\Varid{m}\mskip1.5mu\})\;\Varid{s4}{}\<[E]%
\\
\>[3]{}\Varid{ok2}\mathrel{=}(\Varid{m1sl}\equals\Varid{m1bx}){}\<[E]%
\ColumnHook
\end{hscode}\resethooks
(Note that Haskell's \ensuremath{\ensuremath{\mathsf{Set}}} equality compares the contents of \ensuremath{\ensuremath{\mathsf{Set}}}s ignoring multiplicity and order.)

\item
The LH owner looks up Tavener's dates:
\begin{hscode}\SaveRestoreHook
\column{B}{@{}>{\hspre}l<{\hspost}@{}}%
\column{3}{@{}>{\hspre}l<{\hspost}@{}}%
\column{E}{@{}>{\hspre}l<{\hspost}@{}}%
\>[3]{}\Varid{datesJT}\mathrel{=}\Conid{Dates}\;(\Conid{Just}\;(\Conid{Date}\;\text{\tt \char34 1944\char34},\Conid{Date}\;\text{\tt \char34 2013\char34})){}\<[E]%
\ColumnHook
\end{hscode}\resethooks
and fixes their view:
\begin{hscode}\SaveRestoreHook
\column{B}{@{}>{\hspre}l<{\hspost}@{}}%
\column{3}{@{}>{\hspre}l<{\hspost}@{}}%
\column{E}{@{}>{\hspre}l<{\hspost}@{}}%
\>[3]{}(\Varid{yesJT},\Varid{noJT})\mathrel{=}\Varid{\ensuremath{\mathsf{Set}}.partition}\;(\lambda \hslambda (\Varid{nn},\Varid{na},\Varid{dd})\hsarrow{\rightarrow }{\mathpunct{.}}\Varid{nn}\equals\Conid{Name}\;\text{\tt \char34 John~Tavener\char34})\;\Varid{m1sl}{}\<[E]%
\\
\>[3]{}\Varid{fixedYesJT}\mathrel{=}\Varid{\ensuremath{\mathsf{Set}}.map}\;(\lambda \hslambda (\Varid{nn},\Varid{na},\Varid{dd})\hsarrow{\rightarrow }{\mathpunct{.}}(\Varid{nn},\Varid{na},\Varid{datesJT}))\;\Varid{yesJT}{}\<[E]%
\\
\>[3]{}\Varid{m}_{2}\mathrel{=}\Varid{\ensuremath{\mathsf{Set}}.union}\;\Varid{noJT}\;\Varid{fixedYesJT}{}\<[E]%
\ColumnHook
\end{hscode}\resethooks
and puts it back in the symmetric lens version:
\begin{hscode}\SaveRestoreHook
\column{B}{@{}>{\hspre}l<{\hspost}@{}}%
\column{3}{@{}>{\hspre}l<{\hspost}@{}}%
\column{E}{@{}>{\hspre}l<{\hspost}@{}}%
\>[3]{}(\Varid{n2sl},\Varid{c4})\mathrel{=}\Varid{putr}\;\Varid{composers}\;(\Varid{m}_{2},\Varid{c3}){}\<[E]%
\ColumnHook
\end{hscode}\resethooks
and in the bx version:
\begin{hscode}\SaveRestoreHook
\column{B}{@{}>{\hspre}l<{\hspost}@{}}%
\column{3}{@{}>{\hspre}l<{\hspost}@{}}%
\column{E}{@{}>{\hspre}l<{\hspost}@{}}%
\>[3]{}(\anonymous ,\Varid{s6})\mathrel{=}\Varid{runState}\Varid{n}\;(\mathbf{do}\;\{\mskip1.5mu \Varid{setl}\;\Varid{composersBx}\;\Varid{m}_{2}\mskip1.5mu\})\;\Varid{s5}{}\<[E]%
\ColumnHook
\end{hscode}\resethooks
Checking result from the other side:
\begin{hscode}\SaveRestoreHook
\column{B}{@{}>{\hspre}l<{\hspost}@{}}%
\column{3}{@{}>{\hspre}l<{\hspost}@{}}%
\column{E}{@{}>{\hspre}l<{\hspost}@{}}%
\>[3]{}(\Varid{n2bx},\Varid{s7})\mathrel{=}\Varid{runState}\Varid{n}\;(\mathbf{do}\;\{\mskip1.5mu \Varid{n}\leftarrow \Varid{getr}\;\Varid{composersBx};\Varid{return}\;\Varid{n}\mskip1.5mu\})\;\Varid{s6}{}\<[E]%
\\
\>[3]{}\Varid{ok3}\mathrel{=}(\Varid{n2sl}\equals\Varid{n2bx}){}\<[E]%
\ColumnHook
\end{hscode}\resethooks
\item
Back to the RH view owner
\begin{hscode}\SaveRestoreHook
\column{B}{@{}>{\hspre}l<{\hspost}@{}}%
\column{3}{@{}>{\hspre}l<{\hspost}@{}}%
\column{9}{@{}>{\hspre}l<{\hspost}@{}}%
\column{13}{@{}>{\hspre}l<{\hspost}@{}}%
\column{E}{@{}>{\hspre}l<{\hspost}@{}}%
\>[3]{}\Varid{n3}\mathrel{=}{}\<[9]%
\>[9]{}(({}\<[13]%
\>[13]{}\Conid{Name}\;\text{\tt \char34 Hendrik~Andriessen\char34},\Conid{Nation}\;\text{\tt \char34 Dutch\char34})\mathbin{:}\Varid{n2sl}){}\<[E]%
\\
\>[9]{}\plus [\mskip1.5mu (\Conid{Name}\;\text{\tt \char34 J-B~Lully\char34},\Conid{Nation}\;\text{\tt \char34 French\char34})\mskip1.5mu]{}\<[E]%
\\
\>[3]{}(\Varid{m3sl},\Varid{c5})\mathrel{=}\Varid{putl}\;\Varid{composers}\;(\Varid{n3},\Varid{c4}){}\<[E]%
\\
\>[3]{}(\anonymous ,\Varid{s8})\mathrel{=}\Varid{runState}\Varid{n}\;(\mathbf{do}\;\{\mskip1.5mu \Varid{setr}\;\Varid{composersBx}\;\Varid{n3}\mskip1.5mu\})\;\Varid{s6}{}\<[E]%
\ColumnHook
\end{hscode}\resethooks
\end{enumerate}

To note:

\begin{itemize}

\item
we have shown alternating sets on the two sides, as that is natural
for symmetric lenses; for bx, any order of sets works equally well
(and there is no need to wonder about what complement to use).

\item
We've shown the fine-grained version to facilitate comparison, but
we can also combine steps:
\begin{hscode}\SaveRestoreHook
\column{B}{@{}>{\hspre}l<{\hspost}@{}}%
\column{3}{@{}>{\hspre}l<{\hspost}@{}}%
\column{16}{@{}>{\hspre}l<{\hspost}@{}}%
\column{E}{@{}>{\hspre}l<{\hspost}@{}}%
\>[3]{}(\Varid{n'},\Varid{s'})\mathrel{=}\Varid{runIdentity}\;(\Varid{runr}\;\Varid{composersBx}\;(\mathbf{do}\;\{\mskip1.5mu {}\<[E]%
\\
\>[3]{}\hsindent{13}{}\<[16]%
\>[16]{}\Varid{setl}\;\Varid{composersBx}\;\Varid{onlyBach};{}\<[E]%
\\
\>[3]{}\hsindent{13}{}\<[16]%
\>[16]{}\Varid{n}\leftarrow \Varid{getr}\;\Varid{composersBx};{}\<[E]%
\\
\>[3]{}\hsindent{13}{}\<[16]%
\>[16]{}\Varid{return}\;\Varid{n}\mskip1.5mu\})[\mskip1.5mu \mskip1.5mu]){}\<[E]%
\ColumnHook
\end{hscode}\resethooks
etc. 

\end{itemize}

Auxiliary definitions; only the \ensuremath{\Conid{Show}} instance for Dates is noteworthy
\begin{hscode}\SaveRestoreHook
\column{B}{@{}>{\hspre}l<{\hspost}@{}}%
\column{3}{@{}>{\hspre}l<{\hspost}@{}}%
\column{5}{@{}>{\hspre}l<{\hspost}@{}}%
\column{7}{@{}>{\hspre}l<{\hspost}@{}}%
\column{11}{@{}>{\hspre}l<{\hspost}@{}}%
\column{13}{@{}>{\hspre}l<{\hspost}@{}}%
\column{E}{@{}>{\hspre}l<{\hspost}@{}}%
\>[3]{}\mathbf{newtype}\;\Conid{Name}\mathrel{=}\Conid{Name}\;\{\mskip1.5mu \Varid{unName}\mathbin{::}\Conid{String}\mskip1.5mu\}{}\<[E]%
\\
\>[3]{}\hsindent{2}{}\<[5]%
\>[5]{}\mathbf{deriving}\;(\Conid{Eq},\Conid{Ord}){}\<[E]%
\\
\>[3]{}\mathbf{instance}\;\Conid{Show}\;\Conid{Name}\;\mathbf{where}{}\<[E]%
\\
\>[3]{}\hsindent{2}{}\<[5]%
\>[5]{}\Varid{show}\;(\Conid{Name}\;\Varid{n})\mathrel{=}\Varid{n}{}\<[E]%
\\[\blanklineskip]%
\>[3]{}\mathbf{newtype}\;\Conid{Nation}\mathrel{=}\Conid{Nation}\;\{\mskip1.5mu \Varid{unNation}\mathbin{::}\Conid{String}\mskip1.5mu\}{}\<[E]%
\\
\>[3]{}\hsindent{2}{}\<[5]%
\>[5]{}\mathbf{deriving}\;(\Conid{Eq},\Conid{Ord}){}\<[E]%
\\
\>[3]{}\mathbf{instance}\;\Conid{Show}\;\Conid{Nation}\;\mathbf{where}{}\<[E]%
\\
\>[3]{}\hsindent{2}{}\<[5]%
\>[5]{}\Varid{show}\;(\Conid{Nation}\;\Varid{n})\mathrel{=}\Varid{n}{}\<[E]%
\\[\blanklineskip]%
\>[3]{}\mathbf{newtype}\;\Conid{Date}\mathrel{=}\Conid{Date}\;\{\mskip1.5mu \Varid{unDate}\mathbin{::}\Conid{String}\mskip1.5mu\}{}\<[E]%
\\
\>[3]{}\hsindent{2}{}\<[5]%
\>[5]{}\mathbf{deriving}\;(\Conid{Eq},\Conid{Ord}){}\<[E]%
\\
\>[3]{}\mathbf{instance}\;\Conid{Show}\;\Conid{Date}\;\mathbf{where}{}\<[E]%
\\
\>[3]{}\hsindent{2}{}\<[5]%
\>[5]{}\Varid{show}\;(\Conid{Date}\;\Varid{d})\mathrel{=}\Varid{d}{}\<[E]%
\\[\blanklineskip]%
\>[3]{}\mathbf{newtype}\;\Conid{Dates}\mathrel{=}\Conid{Dates}\;\{\mskip1.5mu \Varid{unDates}\mathbin{::}\Conid{Maybe}\;(\Conid{Date},\Conid{Date})\mskip1.5mu\}{}\<[E]%
\\
\>[3]{}\hsindent{2}{}\<[5]%
\>[5]{}\mathbf{deriving}\;(\Conid{Eq},\Conid{Ord}){}\<[E]%
\\
\>[3]{}\mathbf{instance}\;\Conid{Show}\;\Conid{Dates}\;\mathbf{where}{}\<[E]%
\\
\>[3]{}\hsindent{2}{}\<[5]%
\>[5]{}\Varid{show}\;(\Conid{Dates}\;\Varid{d})\mathrel{=}\Varid{h}\;\Varid{d}{}\<[E]%
\\
\>[5]{}\hsindent{2}{}\<[7]%
\>[7]{}\mathbf{where}\;\Varid{h}\;\Conid{Nothing}\mathrel{=}\text{\tt \char34 ????\char34}{}\<[E]%
\\
\>[7]{}\hsindent{6}{}\<[13]%
\>[13]{}\Varid{h}\;(\Conid{Just}\;(\Varid{dob},\Varid{dod}))\mathrel{=}\Varid{show}\;\Varid{dob}\plus \text{\tt \char34 --\char34}\plus \Varid{show}\;\Varid{dod}{}\<[E]%
\\[\blanklineskip]%
\>[3]{}\Varid{dates0}\mathbin{::}\Conid{Dates}{}\<[E]%
\\
\>[3]{}\Varid{dates0}{}\<[11]%
\>[11]{}\mathrel{=}\Conid{Dates}\;\Conid{Nothing}{}\<[E]%
\ColumnHook
\end{hscode}\resethooks

\subsection{Examples}
\begin{hscode}\SaveRestoreHook
\column{B}{@{}>{\hspre}l<{\hspost}@{}}%
\column{3}{@{}>{\hspre}l<{\hspost}@{}}%
\column{E}{@{}>{\hspre}l<{\hspost}@{}}%
\>[3]{}\mbox{\enskip\{-\# LANGUAGE RankNTypes, FlexibleContexts  \#-\}\enskip}{}\<[E]%
\\[\blanklineskip]%
\>[3]{}\mathbf{module}\;\Conid{Examples}\;\mathbf{where}{}\<[E]%
\\
\>[3]{}\mathbf{import}\;\Conid{\Conid{Control}.\Conid{Monad}.State}\;\Varid{as}\;\Conid{State}{}\<[E]%
\\
\>[3]{}\mathbf{import}\;\Conid{\Conid{Control}.\Conid{Monad}.Reader}\;\Varid{as}\;\Conid{Reader}{}\<[E]%
\\
\>[3]{}\mathbf{import}\;\Conid{\Conid{Control}.\Conid{Monad}.Writer}\;\Varid{as}\;\Conid{Writer}{}\<[E]%
\\
\>[3]{}\mathbf{import}\;\Conid{Control}.\Conid{Monad}.\Conid{Id}\;\Varid{as}\;\Conid{Id}{}\<[E]%
\\
\>[3]{}\mathbf{import}\;\Conid{\Conid{Data}.Map}\;\Varid{as}\;\Conid{Map}\;(\Conid{Map}){}\<[E]%
\\
\>[3]{}\mathbf{import}\;\Conid{BX}{}\<[E]%
\\
\>[3]{}\mathbf{import}\;\Conid{StateTBX}{}\<[E]%
\ColumnHook
\end{hscode}\resethooks
Failure.
\begin{hscode}\SaveRestoreHook
\column{B}{@{}>{\hspre}l<{\hspost}@{}}%
\column{3}{@{}>{\hspre}l<{\hspost}@{}}%
\column{7}{@{}>{\hspre}l<{\hspost}@{}}%
\column{14}{@{}>{\hspre}l<{\hspost}@{}}%
\column{23}{@{}>{\hspre}l<{\hspost}@{}}%
\column{E}{@{}>{\hspre}l<{\hspost}@{}}%
\>[3]{}\Varid{inv}\mathbin{::}\Conid{StateTBX}\;\Conid{Maybe}\;\Conid{Float}\;\Conid{Float}\;\Conid{Float}{}\<[E]%
\\
\>[3]{}\Varid{inv}\mathrel{=}\Conid{StateTBX}\;\Varid{get}\;\set{L}\;\Varid{init}_{L}\;(\Varid{gets}\;(\lambda \hslambda \Varid{a}\hsarrow{\rightarrow }{\mathpunct{.}}\mathrm{1}\mathbin{/}\Varid{a}))\;\set{R}\;\Varid{init}_{R}{}\<[E]%
\\
\>[3]{}\hsindent{4}{}\<[7]%
\>[7]{}\mathbf{where}\;{}\<[14]%
\>[14]{}\set{L}\;\Varid{a}{}\<[23]%
\>[23]{}\mathrel{=}\Varid{try}\;\Varid{put}\;\Varid{a}{}\<[E]%
\\
\>[14]{}\set{R}\;\Varid{b}{}\<[23]%
\>[23]{}\mathrel{=}\Varid{try}\;\Varid{put}\;(\mathrm{1}\mathbin{/}\Varid{b}){}\<[E]%
\\
\>[14]{}\Varid{try}\;\Varid{m}\;\Varid{a}{}\<[23]%
\>[23]{}\mathrel{=}\mathbf{if}\;\Varid{a}\notequals\mathrm{0.0}\;\mathbf{then}\;\Varid{m}\;\Varid{a}\;\mathbf{else}\;\Varid{lift}\;\Conid{Nothing}{}\<[E]%
\\
\>[14]{}\Varid{init}_{L}\;\Varid{a}{}\<[23]%
\>[23]{}\mathrel{=}\mathbf{if}\;\Varid{a}\notequals\mathrm{0.0}\;\mathbf{then}\;\Conid{Just}\;\Varid{a}\;\mathbf{else}\;\Conid{Nothing}{}\<[E]%
\\
\>[14]{}\Varid{init}_{R}\;\Varid{a}{}\<[23]%
\>[23]{}\mathrel{=}\mathbf{if}\;\Varid{a}\notequals\mathrm{0.0}\;\mathbf{then}\;\Conid{Just}\;(\mathrm{1}\mathbin{/}\Varid{a})\;\mathbf{else}\;\Conid{Nothing}{}\<[E]%
\ColumnHook
\end{hscode}\resethooks
A generalization
\begin{hscode}\SaveRestoreHook
\column{B}{@{}>{\hspre}l<{\hspost}@{}}%
\column{3}{@{}>{\hspre}l<{\hspost}@{}}%
\column{5}{@{}>{\hspre}l<{\hspost}@{}}%
\column{7}{@{}>{\hspre}l<{\hspost}@{}}%
\column{17}{@{}>{\hspre}l<{\hspost}@{}}%
\column{21}{@{}>{\hspre}c<{\hspost}@{}}%
\column{21E}{@{}l@{}}%
\column{25}{@{}>{\hspre}l<{\hspost}@{}}%
\column{33}{@{}>{\hspre}l<{\hspost}@{}}%
\column{E}{@{}>{\hspre}l<{\hspost}@{}}%
\>[3]{}\Varid{divZeroBX}\mathbin{::}{}\<[17]%
\>[17]{}(\Conid{Fractional}\;\Varid{a},\Conid{Eq}\;\Varid{a},\Conid{Monad}\;\Varid{m})\Rightarrow (\forall \Varid{x}\hsforall \hsdot{\cdot }{.}\Varid{m}\;\Varid{x})\hsarrow{\rightarrow }{\mathpunct{.}}\Conid{StateTBX}\;\Varid{m}\;\Varid{a}\;\Varid{a}\;\Varid{a}{}\<[E]%
\\
\>[3]{}\Varid{divZeroBX}\;\Varid{divZero}\mathrel{=}\Conid{StateTBX}\;{}\<[33]%
\>[33]{}\Varid{get}\;\set{L}\;\Varid{init}_{L}\;(\Varid{gets}\;(\lambda \hslambda \Varid{a}\hsarrow{\rightarrow }{\mathpunct{.}}(\mathrm{1}\mathbin{/}\Varid{a})))\;\set{R}\;\Varid{init}_{R}{}\<[E]%
\\
\>[3]{}\hsindent{2}{}\<[5]%
\>[5]{}\mathbf{where}{}\<[E]%
\\
\>[5]{}\hsindent{2}{}\<[7]%
\>[7]{}\set{L}\;\Varid{a}{}\<[21]%
\>[21]{}\mathrel{=}{}\<[21E]%
\>[25]{}\mathbf{do}\;\{\mskip1.5mu \Varid{lift}\;(\Varid{guard}\;(\Varid{a}\notequals\mathrm{0}));\Varid{put}\;\Varid{a}\mskip1.5mu\}{}\<[E]%
\\
\>[5]{}\hsindent{2}{}\<[7]%
\>[7]{}\set{R}\;\Varid{b}{}\<[21]%
\>[21]{}\mathrel{=}{}\<[21E]%
\>[25]{}\mathbf{do}\;\{\mskip1.5mu \Varid{lift}\;(\Varid{guard}\;(\Varid{b}\notequals\mathrm{0}));\Varid{put}\;(\mathrm{1}\mathbin{/}\Varid{b})\mskip1.5mu\}{}\<[E]%
\\
\>[5]{}\hsindent{2}{}\<[7]%
\>[7]{}\Varid{init}_{L}\;\Varid{a}{}\<[21]%
\>[21]{}\mathrel{=}{}\<[21E]%
\>[25]{}\mathbf{do}\;\{\mskip1.5mu \Varid{guard}\;(\Varid{a}\notequals\mathrm{0});\Varid{return}\;\Varid{a}\mskip1.5mu\}{}\<[E]%
\\
\>[5]{}\hsindent{2}{}\<[7]%
\>[7]{}\Varid{init}_{R}\;\Varid{b}{}\<[21]%
\>[21]{}\mathrel{=}{}\<[21E]%
\>[25]{}\mathbf{do}\;\{\mskip1.5mu \Varid{guard}\;(\Varid{b}\notequals\mathrm{0});\Varid{return}\;(\mathrm{1}\mathbin{/}\Varid{b})\mskip1.5mu\}{}\<[E]%
\\
\>[5]{}\hsindent{2}{}\<[7]%
\>[7]{}\Varid{guard}\;\Varid{b}{}\<[21]%
\>[21]{}\mathrel{=}{}\<[21E]%
\>[25]{}\mathbf{if}\;\Varid{b}\;\mathbf{then}\;\Varid{return}\;()\;\mathbf{else}\;\Varid{divZero}{}\<[E]%
\ColumnHook
\end{hscode}\resethooks
Uses \ensuremath{\Varid{readS}} to trap errors
\begin{hscode}\SaveRestoreHook
\column{B}{@{}>{\hspre}l<{\hspost}@{}}%
\column{3}{@{}>{\hspre}l<{\hspost}@{}}%
\column{5}{@{}>{\hspre}l<{\hspost}@{}}%
\column{12}{@{}>{\hspre}l<{\hspost}@{}}%
\column{16}{@{}>{\hspre}l<{\hspost}@{}}%
\column{26}{@{}>{\hspre}l<{\hspost}@{}}%
\column{33}{@{}>{\hspre}l<{\hspost}@{}}%
\column{E}{@{}>{\hspre}l<{\hspost}@{}}%
\>[3]{}\Varid{readSome}\mathbin{::}{}\<[16]%
\>[16]{}(\Conid{Read}\;\Varid{a},\Conid{Show}\;\Varid{a})\Rightarrow {}\<[E]%
\\
\>[16]{}\Conid{StateTBX}\;\Conid{Maybe}\;(\Varid{a},\Conid{String})\;\Varid{a}\;\Conid{String}{}\<[E]%
\\
\>[3]{}\Varid{readSome}\mathrel{=}\Conid{StateTBX}\;(\Varid{gets}\;\Varid{fst})\;\set{L}\;\Varid{init}_{L}\;(\Varid{gets}\;\Varid{snd})\;\set{R}\;\Varid{init}_{R}{}\<[E]%
\\
\>[3]{}\hsindent{2}{}\<[5]%
\>[5]{}\mathbf{where}\;{}\<[12]%
\>[12]{}\set{L}\;\Varid{a'}\mathrel{=}\Varid{put}\;(\Varid{a'},\Varid{show}\;\Varid{a'}){}\<[E]%
\\
\>[12]{}\set{R}\;\Varid{b'}\mathrel{=}\mathbf{do}\;{}\<[26]%
\>[26]{}(\anonymous ,\Varid{b})\leftarrow \Varid{get}{}\<[E]%
\\
\>[26]{}\mathbf{if}\;\Varid{b}\equals\Varid{b'}\;\mathbf{then}\;\Varid{return}\;(){}\<[E]%
\\
\>[26]{}\mathbf{else}\;\mathbf{case}\;\Varid{reads}\;\Varid{b'}\;\mathbf{of}{}\<[E]%
\\
\>[26]{}\hsindent{7}{}\<[33]%
\>[33]{}((\Varid{a'},\text{\tt \char34 \char34})\mathbin{:\char95 })\hsarrow{\rightarrow }{\mathpunct{.}}\Varid{put}\;(\Varid{a'},\Varid{b'}){}\<[E]%
\\
\>[26]{}\hsindent{7}{}\<[33]%
\>[33]{}\anonymous \hsarrow{\rightarrow }{\mathpunct{.}}\Varid{lift}\;\Conid{Nothing}{}\<[E]%
\\
\>[12]{}\Varid{init}_{L}\;\Varid{a}\mathrel{=}\Varid{return}\;(\Varid{a},\Varid{show}\;\Varid{a}){}\<[E]%
\\
\>[12]{}\Varid{init}_{R}\;\Varid{b}\mathrel{=}\mathbf{do}\;{}\<[26]%
\>[26]{}\Varid{a}\leftarrow \Varid{choices}\;[\mskip1.5mu \Varid{a}\mid (\Varid{a},\text{\tt \char34 \char34})\leftarrow \Varid{reads}\;\Varid{b}\mskip1.5mu]{}\<[E]%
\\
\>[26]{}\Varid{return}\;(\Varid{a},\Varid{b}){}\<[E]%
\\
\>[12]{}\Varid{choices}\;[\mskip1.5mu \mskip1.5mu]\mathrel{=}\Varid{mzero}{}\<[E]%
\\
\>[12]{}\Varid{choices}\;(\Varid{a}\mathbin{:}\Varid{as})\mathrel{=}\Varid{return}\;\Varid{a}\mathbin{`\Varid{mplus}`}\Varid{choices}\;\Varid{as}{}\<[E]%
\ColumnHook
\end{hscode}\resethooks
A generalization
\begin{hscode}\SaveRestoreHook
\column{B}{@{}>{\hspre}l<{\hspost}@{}}%
\column{3}{@{}>{\hspre}l<{\hspost}@{}}%
\column{5}{@{}>{\hspre}l<{\hspost}@{}}%
\column{12}{@{}>{\hspre}l<{\hspost}@{}}%
\column{18}{@{}>{\hspre}l<{\hspost}@{}}%
\column{21}{@{}>{\hspre}l<{\hspost}@{}}%
\column{24}{@{}>{\hspre}l<{\hspost}@{}}%
\column{27}{@{}>{\hspre}l<{\hspost}@{}}%
\column{29}{@{}>{\hspre}l<{\hspost}@{}}%
\column{E}{@{}>{\hspre}l<{\hspost}@{}}%
\>[3]{}\Varid{readableBX}\mathbin{::}{}\<[18]%
\>[18]{}(\Conid{Read}\;\Varid{a},\Conid{Show}\;\Varid{a},\Conid{MonadPlus}\;\Varid{m})\Rightarrow {}\<[E]%
\\
\>[18]{}\Conid{StateTBX}\;\Varid{m}\;(\Varid{a},\Conid{String})\;\Varid{a}\;\Conid{String}{}\<[E]%
\\
\>[3]{}\Varid{readableBX}\mathrel{=}\Conid{StateTBX}\;\get{A}\;\set{A}\;\Varid{initA}\;\get{B}\;\set{B}\;\Varid{initB}{}\<[E]%
\\
\>[3]{}\hsindent{2}{}\<[5]%
\>[5]{}\mathbf{where}\;{}\<[12]%
\>[12]{}\get{A}{}\<[21]%
\>[21]{}\mathrel{=}\mathbf{do}\;\{\mskip1.5mu (\Varid{a},\Varid{s})\leftarrow \Varid{get};\Varid{return}\;\Varid{a}\mskip1.5mu\}{}\<[E]%
\\
\>[12]{}\get{B}{}\<[21]%
\>[21]{}\mathrel{=}\mathbf{do}\;\{\mskip1.5mu (\Varid{a},\Varid{s})\leftarrow \Varid{get};\Varid{return}\;\Varid{s}\mskip1.5mu\}{}\<[E]%
\\
\>[12]{}\set{A}\;\Varid{a'}{}\<[21]%
\>[21]{}\mathrel{=}\Varid{put}\;(\Varid{a'},\Varid{show}\;\Varid{a'}){}\<[E]%
\\
\>[12]{}\set{B}\;\Varid{b'}{}\<[21]%
\>[21]{}\mathrel{=}\mathbf{do}\;{}\<[27]%
\>[27]{}(\anonymous ,\Varid{b})\leftarrow \Varid{get}{}\<[E]%
\\
\>[27]{}\mathbf{case}\;((\Varid{b}\equals\Varid{b'}),\Varid{reads}\;\Varid{b'})\;\mathbf{of}{}\<[E]%
\\
\>[27]{}\hsindent{2}{}\<[29]%
\>[29]{}(\Conid{True},\anonymous )\hsarrow{\rightarrow }{\mathpunct{.}}\Varid{return}\;(){}\<[E]%
\\
\>[27]{}\hsindent{2}{}\<[29]%
\>[29]{}(\anonymous ,(\Varid{a'},\text{\tt \char34 \char34})\mathbin{:\char95 })\hsarrow{\rightarrow }{\mathpunct{.}}\Varid{put}\;(\Varid{a'},\Varid{b'}){}\<[E]%
\\
\>[27]{}\hsindent{2}{}\<[29]%
\>[29]{}\anonymous \hsarrow{\rightarrow }{\mathpunct{.}}\Varid{lift}\;\Varid{mzero}{}\<[E]%
\\
\>[12]{}\Varid{initA}\;\Varid{a}\mathrel{=}\Varid{return}\;(\Varid{a},\Varid{show}\;\Varid{a}){}\<[E]%
\\
\>[12]{}\Varid{initB}\;\Varid{b}\mathrel{=}\mathbf{case}\;\Varid{reads}\;\Varid{b}\;\mathbf{of}{}\<[E]%
\\
\>[12]{}\hsindent{12}{}\<[24]%
\>[24]{}(\Varid{a},\text{\tt \char34 \char34})\mathbin{:\char95 }\hsarrow{\rightarrow }{\mathpunct{.}}\Varid{return}\;(\Varid{a},\Varid{b}){}\<[E]%
\\
\>[12]{}\hsindent{12}{}\<[24]%
\>[24]{}\anonymous \hsarrow{\rightarrow }{\mathpunct{.}}\Varid{mzero}{}\<[E]%
\ColumnHook
\end{hscode}\resethooks
The JTL example from the paper (\Example~\ref{ex:nondeterminism})
\begin{hscode}\SaveRestoreHook
\column{B}{@{}>{\hspre}l<{\hspost}@{}}%
\column{3}{@{}>{\hspre}l<{\hspost}@{}}%
\column{5}{@{}>{\hspre}l<{\hspost}@{}}%
\column{14}{@{}>{\hspre}l<{\hspost}@{}}%
\column{16}{@{}>{\hspre}l<{\hspost}@{}}%
\column{20}{@{}>{\hspre}l<{\hspost}@{}}%
\column{22}{@{}>{\hspre}l<{\hspost}@{}}%
\column{33}{@{}>{\hspre}l<{\hspost}@{}}%
\column{E}{@{}>{\hspre}l<{\hspost}@{}}%
\>[3]{}\Varid{nondetBX}\mathbin{::}{}\<[16]%
\>[16]{}(\Varid{a}\hsarrow{\rightarrow }{\mathpunct{.}}\Varid{b}\hsarrow{\rightarrow }{\mathpunct{.}}\Conid{Bool})\hsarrow{\rightarrow }{\mathpunct{.}}(\Varid{a}\hsarrow{\rightarrow }{\mathpunct{.}}[\mskip1.5mu \Varid{b}\mskip1.5mu])\hsarrow{\rightarrow }{\mathpunct{.}}(\Varid{b}\hsarrow{\rightarrow }{\mathpunct{.}}[\mskip1.5mu \Varid{a}\mskip1.5mu])\hsarrow{\rightarrow }{\mathpunct{.}}\Conid{StateTBX}\;[\mskip1.5mu \mskip1.5mu]\;(\Varid{a},\Varid{b})\;\Varid{a}\;\Varid{b}{}\<[E]%
\\
\>[3]{}\Varid{nondetBX}\;\Varid{ok}\;\Varid{bs}\;\Varid{as}\mathrel{=}\Conid{StateTBX}\;{}\<[33]%
\>[33]{}\get{L}\;\set{L}\;\Varid{init}_{L}\;\get{R}\;\set{R}\;\Varid{init}_{R}\;\mathbf{where}{}\<[E]%
\\
\>[3]{}\hsindent{2}{}\<[5]%
\>[5]{}\get{L}{}\<[14]%
\>[14]{}\mathrel{=}\mathbf{do}\;\{\mskip1.5mu (\Varid{a},\Varid{b})\leftarrow \Varid{get};\Varid{return}\;\Varid{a}\mskip1.5mu\}{}\<[E]%
\\
\>[3]{}\hsindent{2}{}\<[5]%
\>[5]{}\get{R}{}\<[14]%
\>[14]{}\mathrel{=}\mathbf{do}\;\{\mskip1.5mu (\Varid{a},\Varid{b})\leftarrow \Varid{get};\Varid{return}\;\Varid{b}\mskip1.5mu\}{}\<[E]%
\\
\>[3]{}\hsindent{2}{}\<[5]%
\>[5]{}\set{L}\;\Varid{a'}{}\<[14]%
\>[14]{}\mathrel{=}\mathbf{do}\;{}\<[20]%
\>[20]{}(\Varid{a},\Varid{b})\leftarrow \Varid{get}{}\<[E]%
\\
\>[20]{}\mathbf{if}\;\Varid{ok}\;\Varid{a'}\;\Varid{b}\;\mathbf{then}\;\Varid{put}\;(\Varid{a'},\Varid{b})\;\mathbf{else}{}\<[E]%
\\
\>[20]{}\hsindent{2}{}\<[22]%
\>[22]{}\mathbf{do}\;\{\mskip1.5mu \Varid{b'}\leftarrow \Varid{lift}\;(\Varid{bs}\;\Varid{a'});\Varid{put}\;(\Varid{a'},\Varid{b'})\mskip1.5mu\}{}\<[E]%
\\
\>[3]{}\hsindent{2}{}\<[5]%
\>[5]{}\set{R}\;\Varid{b'}{}\<[14]%
\>[14]{}\mathrel{=}\mathbf{do}\;{}\<[20]%
\>[20]{}(\Varid{a},\Varid{b})\leftarrow \Varid{get}{}\<[E]%
\\
\>[20]{}\mathbf{if}\;\Varid{ok}\;\Varid{a}\;\Varid{b'}\;\mathbf{then}\;\Varid{put}\;(\Varid{a},\Varid{b'})\;\mathbf{else}{}\<[E]%
\\
\>[20]{}\hsindent{2}{}\<[22]%
\>[22]{}\mathbf{do}\;\{\mskip1.5mu \Varid{a'}\leftarrow \Varid{lift}\;(\Varid{as}\;\Varid{b'});\Varid{put}\;(\Varid{a'},\Varid{b'})\mskip1.5mu\}{}\<[E]%
\\
\>[3]{}\hsindent{2}{}\<[5]%
\>[5]{}\Varid{init}_{L}\;\Varid{a}{}\<[14]%
\>[14]{}\mathrel{=}\mathbf{do}\;\{\mskip1.5mu \Varid{b}\leftarrow \Varid{bs}\;\Varid{a};\Varid{return}\;(\Varid{a},\Varid{b})\mskip1.5mu\}{}\<[E]%
\\
\>[3]{}\hsindent{2}{}\<[5]%
\>[5]{}\Varid{init}_{R}\;\Varid{b}{}\<[14]%
\>[14]{}\mathrel{=}\mathbf{do}\;\{\mskip1.5mu \Varid{a}\leftarrow \Varid{as}\;\Varid{b};\Varid{return}\;(\Varid{a},\Varid{b})\mskip1.5mu\}{}\<[E]%
\ColumnHook
\end{hscode}\resethooks
Switching between two lenses on the same state space, based on a boolean flag
\begin{hscode}\SaveRestoreHook
\column{B}{@{}>{\hspre}l<{\hspost}@{}}%
\column{3}{@{}>{\hspre}l<{\hspost}@{}}%
\column{5}{@{}>{\hspre}l<{\hspost}@{}}%
\column{11}{@{}>{\hspre}l<{\hspost}@{}}%
\column{25}{@{}>{\hspre}c<{\hspost}@{}}%
\column{25E}{@{}l@{}}%
\column{28}{@{}>{\hspre}l<{\hspost}@{}}%
\column{38}{@{}>{\hspre}l<{\hspost}@{}}%
\column{E}{@{}>{\hspre}l<{\hspost}@{}}%
\>[3]{}\Varid{switchBX}\mathbin{::}\Conid{MonadReader}\;\Conid{Bool}\;\Varid{m}\Rightarrow {}\<[38]%
\>[38]{}\Conid{StateTBX}\;\Varid{m}\;\Varid{s}\;\Varid{a}\;\Varid{b}\hsarrow{\rightarrow }{\mathpunct{.}}{}\<[E]%
\\
\>[38]{}\Conid{StateTBX}\;\Varid{m}\;\Varid{s}\;\Varid{a}\;\Varid{b}\hsarrow{\rightarrow }{\mathpunct{.}}{}\<[E]%
\\
\>[38]{}\Conid{StateTBX}\;\Varid{m}\;\Varid{s}\;\Varid{a}\;\Varid{b}{}\<[E]%
\\
\>[3]{}\Varid{switchBX}\;bx_{1}\;bx_{2}\mathrel{=}\Conid{StateTBX}\;\get{A}\;\set{A}\;\Varid{initA}\;\get{B}\;\set{B}\;\Varid{initB}{}\<[E]%
\\
\>[3]{}\hsindent{2}{}\<[5]%
\>[5]{}\mathbf{where}\;\get{A}{}\<[25]%
\>[25]{}\mathrel{=}{}\<[25E]%
\>[28]{}\Varid{switch}\;(\Varid{getl}\;bx_{1})\;(\Varid{getl}\;bx_{2}){}\<[E]%
\\
\>[5]{}\hsindent{6}{}\<[11]%
\>[11]{}\get{B}{}\<[25]%
\>[25]{}\mathrel{=}{}\<[25E]%
\>[28]{}\Varid{switch}\;(\Varid{getr}\;bx_{2})\;(\Varid{getr}\;bx_{2}){}\<[E]%
\\
\>[5]{}\hsindent{6}{}\<[11]%
\>[11]{}\set{A}\;\Varid{a}{}\<[25]%
\>[25]{}\mathrel{=}{}\<[25E]%
\>[28]{}\Varid{switch}\;(\Varid{setl}\;bx_{1}\;\Varid{a})\;(\Varid{setl}\;bx_{2}\;\Varid{a}){}\<[E]%
\\
\>[5]{}\hsindent{6}{}\<[11]%
\>[11]{}\set{B}\;\Varid{b}{}\<[25]%
\>[25]{}\mathrel{=}{}\<[25E]%
\>[28]{}\Varid{switch}\;(\Varid{setr}\;bx_{1}\;\Varid{b})\;(\Varid{setr}\;bx_{2}\;\Varid{b}){}\<[E]%
\\
\>[5]{}\hsindent{6}{}\<[11]%
\>[11]{}\Varid{initA}\;\Varid{a}{}\<[25]%
\>[25]{}\mathrel{=}{}\<[25E]%
\>[28]{}\Varid{switch}\;(\Varid{initl}\;bx_{1}\;\Varid{a})\;(\Varid{initl}\;bx_{2}\;\Varid{a}){}\<[E]%
\\
\>[5]{}\hsindent{6}{}\<[11]%
\>[11]{}\Varid{initB}\;\Varid{b}{}\<[25]%
\>[25]{}\mathrel{=}{}\<[25E]%
\>[28]{}\Varid{switch}\;(\Varid{initr}\;bx_{1}\;\Varid{b})\;(\Varid{initr}\;bx_{2}\;\Varid{b}){}\<[E]%
\\
\>[5]{}\hsindent{6}{}\<[11]%
\>[11]{}\Varid{switch}\;\Varid{m}_{1}\;\Varid{m}_{2}{}\<[25]%
\>[25]{}\mathrel{=}{}\<[25E]%
\>[28]{}\mathbf{do}\;\{\mskip1.5mu \Varid{b}\leftarrow \Varid{ask};\mathbf{if}\;\Varid{b}\;\mathbf{then}\;\Varid{m}_{1}\;\mathbf{else}\;\Varid{m}_{2}\mskip1.5mu\}{}\<[E]%
\ColumnHook
\end{hscode}\resethooks
Generalized version
\begin{hscode}\SaveRestoreHook
\column{B}{@{}>{\hspre}l<{\hspost}@{}}%
\column{3}{@{}>{\hspre}l<{\hspost}@{}}%
\column{5}{@{}>{\hspre}l<{\hspost}@{}}%
\column{11}{@{}>{\hspre}l<{\hspost}@{}}%
\column{25}{@{}>{\hspre}c<{\hspost}@{}}%
\column{25E}{@{}l@{}}%
\column{28}{@{}>{\hspre}l<{\hspost}@{}}%
\column{36}{@{}>{\hspre}l<{\hspost}@{}}%
\column{E}{@{}>{\hspre}l<{\hspost}@{}}%
\>[3]{}\Varid{switchBX'}\mathbin{::}\Conid{MonadReader}\;\Varid{c}\;\Varid{m}\Rightarrow {}\<[36]%
\>[36]{}(\Varid{c}\hsarrow{\rightarrow }{\mathpunct{.}}\Conid{StateTBX}\;\Varid{m}\;\Varid{s}\;\Varid{a}\;\Varid{b})\hsarrow{\rightarrow }{\mathpunct{.}}{}\<[E]%
\\
\>[36]{}\Conid{StateTBX}\;\Varid{m}\;\Varid{s}\;\Varid{a}\;\Varid{b}{}\<[E]%
\\
\>[3]{}\Varid{switchBX'}\;\Varid{f}\mathrel{=}\Conid{StateTBX}\;\get{A}\;\set{A}\;\Varid{initA}\;\get{B}\;\set{B}\;\Varid{initB}{}\<[E]%
\\
\>[3]{}\hsindent{2}{}\<[5]%
\>[5]{}\mathbf{where}\;\get{A}{}\<[25]%
\>[25]{}\mathrel{=}{}\<[25E]%
\>[28]{}\Varid{switch}\;\Varid{getl}{}\<[E]%
\\
\>[5]{}\hsindent{6}{}\<[11]%
\>[11]{}\get{B}{}\<[25]%
\>[25]{}\mathrel{=}{}\<[25E]%
\>[28]{}\Varid{switch}\;\Varid{getr}{}\<[E]%
\\
\>[5]{}\hsindent{6}{}\<[11]%
\>[11]{}\set{A}\;\Varid{a}{}\<[25]%
\>[25]{}\mathrel{=}{}\<[25E]%
\>[28]{}\Varid{switch}\;(\lambda \hslambda bx\hsarrow{\rightarrow }{\mathpunct{.}}\Varid{setl}\;bx\;\Varid{a}){}\<[E]%
\\
\>[5]{}\hsindent{6}{}\<[11]%
\>[11]{}\set{B}\;\Varid{b}{}\<[25]%
\>[25]{}\mathrel{=}{}\<[25E]%
\>[28]{}\Varid{switch}\;(\lambda \hslambda bx\hsarrow{\rightarrow }{\mathpunct{.}}\Varid{setr}\;bx\;\Varid{b}){}\<[E]%
\\
\>[5]{}\hsindent{6}{}\<[11]%
\>[11]{}\Varid{initA}\;\Varid{a}{}\<[25]%
\>[25]{}\mathrel{=}{}\<[25E]%
\>[28]{}\Varid{switch}\;(\lambda \hslambda bx\hsarrow{\rightarrow }{\mathpunct{.}}\Varid{initl}\;bx\;\Varid{a}){}\<[E]%
\\
\>[5]{}\hsindent{6}{}\<[11]%
\>[11]{}\Varid{initB}\;\Varid{b}{}\<[25]%
\>[25]{}\mathrel{=}{}\<[25E]%
\>[28]{}\Varid{switch}\;(\lambda \hslambda bx\hsarrow{\rightarrow }{\mathpunct{.}}\Varid{initr}\;bx\;\Varid{b}){}\<[E]%
\\
\>[5]{}\hsindent{6}{}\<[11]%
\>[11]{}\Varid{switch}\;\Varid{op}{}\<[25]%
\>[25]{}\mathrel{=}{}\<[25E]%
\>[28]{}\mathbf{do}\;\{\mskip1.5mu \Varid{c}\leftarrow \Varid{ask};\Varid{op}\;(\Varid{f}\;\Varid{c})\mskip1.5mu\}{}\<[E]%
\ColumnHook
\end{hscode}\resethooks
Logging BX: writes the sequence of sets
\begin{hscode}\SaveRestoreHook
\column{B}{@{}>{\hspre}l<{\hspost}@{}}%
\column{3}{@{}>{\hspre}l<{\hspost}@{}}%
\column{5}{@{}>{\hspre}l<{\hspost}@{}}%
\column{11}{@{}>{\hspre}l<{\hspost}@{}}%
\column{17}{@{}>{\hspre}l<{\hspost}@{}}%
\column{25}{@{}>{\hspre}l<{\hspost}@{}}%
\column{69}{@{}>{\hspre}l<{\hspost}@{}}%
\column{E}{@{}>{\hspre}l<{\hspost}@{}}%
\>[3]{}\Varid{loggingBX}\mathbin{::}{}\<[17]%
\>[17]{}(\Conid{Eq}\;\Varid{a},\Conid{Eq}\;\Varid{b},\Conid{MonadWriter}\;(\Conid{Either}\;\Varid{a}\;\Varid{b})\;\Varid{m})\Rightarrow {}\<[E]%
\\
\>[17]{}\Conid{StateTBX}\;\Varid{m}\;\Varid{s}\;\Varid{a}\;\Varid{b}\hsarrow{\rightarrow }{\mathpunct{.}}\Conid{StateTBX}\;\Varid{m}\;\Varid{s}\;\Varid{a}\;\Varid{b}{}\<[E]%
\\
\>[3]{}\Varid{loggingBX}\;bx\mathrel{=}\Conid{StateTBX}\;(\Varid{getl}\;bx)\;\set{A}\;(\Varid{initl}\;bx)\;(\Varid{getr}\;bx)\;\set{B}\;{}\<[69]%
\>[69]{}(\Varid{initr}\;bx){}\<[E]%
\\
\>[3]{}\hsindent{2}{}\<[5]%
\>[5]{}\mathbf{where}\;\set{A}\;\Varid{a'}\mathrel{=}\mathbf{do}\;{}\<[25]%
\>[25]{}\Varid{a}\leftarrow \Varid{getl}\;bx{}\<[E]%
\\
\>[25]{}\mathbf{if}\;\Varid{a}\notequals\Varid{a'}\;\mathbf{then}\;\Varid{tell}\;(\Conid{Left}\;\Varid{a'})\;\mathbf{else}\;\Varid{return}\;(){}\<[E]%
\\
\>[25]{}\Varid{setl}\;bx\;\Varid{a'}{}\<[E]%
\\
\>[5]{}\hsindent{6}{}\<[11]%
\>[11]{}\set{B}\;\Varid{b'}\mathrel{=}\mathbf{do}\;{}\<[25]%
\>[25]{}\Varid{b}\leftarrow \Varid{getr}\;bx{}\<[E]%
\\
\>[25]{}\mathbf{if}\;\Varid{b}\notequals\Varid{b'}\;\mathbf{then}\;\Varid{tell}\;(\Conid{Right}\;\Varid{b'})\;\mathbf{else}\;\Varid{return}\;(){}\<[E]%
\\
\>[25]{}\Varid{setr}\;bx\;\Varid{b'}{}\<[E]%
\ColumnHook
\end{hscode}\resethooks
I/O: user interaction
\begin{hscode}\SaveRestoreHook
\column{B}{@{}>{\hspre}l<{\hspost}@{}}%
\column{3}{@{}>{\hspre}l<{\hspost}@{}}%
\column{5}{@{}>{\hspre}l<{\hspost}@{}}%
\column{11}{@{}>{\hspre}l<{\hspost}@{}}%
\column{21}{@{}>{\hspre}l<{\hspost}@{}}%
\column{24}{@{}>{\hspre}l<{\hspost}@{}}%
\column{28}{@{}>{\hspre}l<{\hspost}@{}}%
\column{30}{@{}>{\hspre}l<{\hspost}@{}}%
\column{E}{@{}>{\hspre}l<{\hspost}@{}}%
\>[3]{}\Varid{interactiveBX}\mathbin{::}{}\<[21]%
\>[21]{}(\Conid{Read}\;\Varid{a},\Conid{Read}\;\Varid{b})\Rightarrow {}\<[E]%
\\
\>[21]{}(\Varid{a}\hsarrow{\rightarrow }{\mathpunct{.}}\Varid{b}\hsarrow{\rightarrow }{\mathpunct{.}}\Conid{Bool})\hsarrow{\rightarrow }{\mathpunct{.}}\Conid{StateTBX}\;\Conid{IO}\;(\Varid{a},\Varid{b})\;\Varid{a}\;\Varid{b}{}\<[E]%
\\
\>[3]{}\Varid{interactiveBX}\;\Varid{r}\mathrel{=}\Conid{StateTBX}\;\get{A}\;\set{A}\;\Varid{initA}\;\get{B}\;\set{B}\;\Varid{initB}{}\<[E]%
\\
\>[3]{}\hsindent{2}{}\<[5]%
\>[5]{}\mathbf{where}\;\get{A}{}\<[21]%
\>[21]{}\mathrel{=}{}\<[24]%
\>[24]{}\mathbf{do}\;\{\mskip1.5mu (\Varid{a},\Varid{b})\leftarrow \Varid{get};\Varid{return}\;\Varid{a}\mskip1.5mu\}{}\<[E]%
\\
\>[5]{}\hsindent{6}{}\<[11]%
\>[11]{}\set{A}\;\Varid{a'}{}\<[21]%
\>[21]{}\mathrel{=}{}\<[24]%
\>[24]{}\mathbf{do}\;\{\mskip1.5mu (\Varid{a},\Varid{b})\leftarrow \Varid{get};\Varid{fixB}\;\Varid{a'}\;\Varid{b}\mskip1.5mu\}{}\<[E]%
\\
\>[5]{}\hsindent{6}{}\<[11]%
\>[11]{}\Varid{fixA}\;\Varid{a}\;\Varid{b}{}\<[21]%
\>[21]{}\mathrel{=}{}\<[24]%
\>[24]{}\mathbf{if}\;{}\<[30]%
\>[30]{}\Varid{r}\;\Varid{a}\;\Varid{b}{}\<[E]%
\\
\>[24]{}\mathbf{then}\;{}\<[30]%
\>[30]{}\Varid{put}\;(\Varid{a},\Varid{b}){}\<[E]%
\\
\>[24]{}\mathbf{else}\;{}\<[30]%
\>[30]{}\mathbf{do}\;\{\mskip1.5mu \Varid{a'}\leftarrow \Varid{lift}\;(\Varid{initA}\;\Varid{b});\Varid{fixA}\;\Varid{a'}\;\Varid{b}\mskip1.5mu\}{}\<[E]%
\\
\>[5]{}\hsindent{6}{}\<[11]%
\>[11]{}\Varid{initA}\;\Varid{b}{}\<[21]%
\>[21]{}\mathrel{=}{}\<[24]%
\>[24]{}\mathbf{do}\;{}\<[28]%
\>[28]{}\Varid{print}\;\text{\tt \char34 Please~restore~consistency:\char34}{}\<[E]%
\\
\>[28]{}\Varid{str}\leftarrow \Varid{getLine}{}\<[E]%
\\
\>[28]{}\Varid{return}\;(\Varid{read}\;\Varid{str}){}\<[E]%
\\
\>[5]{}\hsindent{6}{}\<[11]%
\>[11]{}\get{B}{}\<[21]%
\>[21]{}\mathrel{=}{}\<[24]%
\>[24]{}\mathbf{do}\;\{\mskip1.5mu (\Varid{a},\Varid{b})\leftarrow \Varid{get};\Varid{return}\;\Varid{b}\mskip1.5mu\}{}\<[E]%
\\
\>[5]{}\hsindent{6}{}\<[11]%
\>[11]{}\set{B}\;\Varid{b'}{}\<[21]%
\>[21]{}\mathrel{=}{}\<[24]%
\>[24]{}\mathbf{do}\;\{\mskip1.5mu (\Varid{a},\Varid{b})\leftarrow \Varid{get};\Varid{fixA}\;\Varid{a}\;\Varid{b'}\mskip1.5mu\}{}\<[E]%
\\
\>[5]{}\hsindent{6}{}\<[11]%
\>[11]{}\Varid{fixB}\;\Varid{a}\;\Varid{b}{}\<[21]%
\>[21]{}\mathrel{=}{}\<[24]%
\>[24]{}\mathbf{if}\;{}\<[30]%
\>[30]{}\Varid{r}\;\Varid{a}\;\Varid{b}{}\<[E]%
\\
\>[24]{}\mathbf{then}\;{}\<[30]%
\>[30]{}\Varid{put}\;(\Varid{a},\Varid{b}){}\<[E]%
\\
\>[24]{}\mathbf{else}\;{}\<[30]%
\>[30]{}\mathbf{do}\;\{\mskip1.5mu \Varid{b'}\leftarrow \Varid{lift}\;(\Varid{initB}\;\Varid{b});\Varid{fixA}\;\Varid{a}\;\Varid{b'}\mskip1.5mu\}{}\<[E]%
\\
\>[5]{}\hsindent{6}{}\<[11]%
\>[11]{}\Varid{initB}\;\Varid{a}{}\<[21]%
\>[21]{}\mathrel{=}{}\<[24]%
\>[24]{}\mathbf{do}\;{}\<[28]%
\>[28]{}\Varid{print}\;\text{\tt \char34 Please~restore~consistency:\char34}{}\<[E]%
\\
\>[28]{}\Varid{str}\leftarrow \Varid{getLine}{}\<[E]%
\\
\>[28]{}\Varid{return}\;(\Varid{read}\;\Varid{str}){}\<[E]%
\ColumnHook
\end{hscode}\resethooks
and signalling changes
\begin{hscode}\SaveRestoreHook
\column{B}{@{}>{\hspre}l<{\hspost}@{}}%
\column{3}{@{}>{\hspre}l<{\hspost}@{}}%
\column{5}{@{}>{\hspre}l<{\hspost}@{}}%
\column{15}{@{}>{\hspre}l<{\hspost}@{}}%
\column{16}{@{}>{\hspre}l<{\hspost}@{}}%
\column{23}{@{}>{\hspre}l<{\hspost}@{}}%
\column{36}{@{}>{\hspre}l<{\hspost}@{}}%
\column{E}{@{}>{\hspre}l<{\hspost}@{}}%
\>[3]{}\Varid{signalBX}\mathbin{::}{}\<[16]%
\>[16]{}(\Conid{Eq}\;\Varid{a},\Conid{Eq}\;\Varid{b},\Conid{Monad}\;\Varid{m})\Rightarrow {}\<[E]%
\\
\>[16]{}(\Varid{a}\hsarrow{\rightarrow }{\mathpunct{.}}\Varid{m}\;())\hsarrow{\rightarrow }{\mathpunct{.}}(\Varid{b}\hsarrow{\rightarrow }{\mathpunct{.}}\Varid{m}\;())\hsarrow{\rightarrow }{\mathpunct{.}}{}\<[E]%
\\
\>[16]{}\Conid{StateTBX}\;\Varid{m}\;\Varid{s}\;\Varid{a}\;\Varid{b}\hsarrow{\rightarrow }{\mathpunct{.}}\Conid{StateTBX}\;\Varid{m}\;\Varid{s}\;\Varid{a}\;\Varid{b}{}\<[E]%
\\
\>[3]{}\Varid{signalBX}\;\Varid{sigA}\;\Varid{sigB}\;\Varid{t}\mathrel{=}\Conid{StateTBX}\;{}\<[36]%
\>[36]{}(\Varid{getl}\;\Varid{t})\;\set{L}\;(\Varid{initl}\;\Varid{t})\;{}\<[E]%
\\
\>[36]{}(\Varid{getr}\;\Varid{t})\;\set{R}\;(\Varid{initr}\;\Varid{t})\;\mathbf{where}{}\<[E]%
\\
\>[3]{}\hsindent{2}{}\<[5]%
\>[5]{}\set{L}\;\Varid{a'}{}\<[15]%
\>[15]{}\mathrel{=}\mathbf{do}\;\{\mskip1.5mu {}\<[23]%
\>[23]{}\Varid{a}\leftarrow \Varid{getl}\;\Varid{t};\Varid{setl}\;\Varid{t}\;\Varid{a'};{}\<[E]%
\\
\>[23]{}\Varid{lift}\;(\mathbf{if}\;\Varid{a}\notequals\Varid{a'}\;\mathbf{then}\;\Varid{sigA}\;\Varid{a'}\;\mathbf{else}\;\Varid{return}\;())\mskip1.5mu\}{}\<[E]%
\\
\>[3]{}\hsindent{2}{}\<[5]%
\>[5]{}\set{R}\;\Varid{b'}{}\<[15]%
\>[15]{}\mathrel{=}\mathbf{do}\;\{\mskip1.5mu {}\<[23]%
\>[23]{}\Varid{b}\leftarrow \Varid{getr}\;\Varid{t};\Varid{setr}\;\Varid{t}\;\Varid{b'};{}\<[E]%
\\
\>[23]{}\Varid{lift}\;(\mathbf{if}\;\Varid{b}\notequals\Varid{b'}\;\mathbf{then}\;\Varid{sigB}\;\Varid{b'}\;\mathbf{else}\;\Varid{return}\;())\mskip1.5mu\}{}\<[E]%
\\[\blanklineskip]%
\>[3]{}\Varid{alertBX}\mathbin{::}(\Conid{Eq}\;\Varid{a},\Conid{Eq}\;\Varid{b})\Rightarrow \Conid{StateTBX}\;\Conid{IO}\;\Varid{s}\;\Varid{a}\;\Varid{b}\hsarrow{\rightarrow }{\mathpunct{.}}\Conid{StateTBX}\;\Conid{IO}\;\Varid{s}\;\Varid{a}\;\Varid{b}{}\<[E]%
\\
\>[3]{}\Varid{alertBX}\mathrel{=}\Varid{signalBX}\;{}\<[23]%
\>[23]{}(\mathbin{\char92 \char95 }\hsarrow{\rightarrow }{\mathpunct{.}}\Varid{putStrLn}\;\text{\tt \char34 Left\char34})\;{}\<[E]%
\\
\>[23]{}(\mathbin{\char92 \char95 }\hsarrow{\rightarrow }{\mathpunct{.}}\Varid{putStrLn}\;\text{\tt \char34 Right\char34}){}\<[E]%
\ColumnHook
\end{hscode}\resethooks
where
\begin{hscode}\SaveRestoreHook
\column{B}{@{}>{\hspre}l<{\hspost}@{}}%
\column{3}{@{}>{\hspre}l<{\hspost}@{}}%
\column{E}{@{}>{\hspre}l<{\hspost}@{}}%
\>[3]{}\Varid{fst3}\;(\Varid{a},\anonymous ,\anonymous )\mathrel{=}\Varid{a}{}\<[E]%
\\
\>[3]{}\Varid{snd3}\;(\anonymous ,\Varid{a},\anonymous )\mathrel{=}\Varid{a}{}\<[E]%
\\
\>[3]{}\Varid{thd3}\;(\anonymous ,\anonymous ,\Varid{a})\mathrel{=}\Varid{a}{}\<[E]%
\ColumnHook
\end{hscode}\resethooks
Model-transformation-by-example (\Example~\ref{ex:dynamic} from the paper)
\begin{hscode}\SaveRestoreHook
\column{B}{@{}>{\hspre}l<{\hspost}@{}}%
\column{3}{@{}>{\hspre}l<{\hspost}@{}}%
\column{5}{@{}>{\hspre}l<{\hspost}@{}}%
\column{12}{@{}>{\hspre}l<{\hspost}@{}}%
\column{18}{@{}>{\hspre}l<{\hspost}@{}}%
\column{26}{@{}>{\hspre}l<{\hspost}@{}}%
\column{29}{@{}>{\hspre}l<{\hspost}@{}}%
\column{30}{@{}>{\hspre}l<{\hspost}@{}}%
\column{38}{@{}>{\hspre}l<{\hspost}@{}}%
\column{48}{@{}>{\hspre}l<{\hspost}@{}}%
\column{75}{@{}>{\hspre}l<{\hspost}@{}}%
\column{E}{@{}>{\hspre}l<{\hspost}@{}}%
\>[3]{}\Varid{dynamicBX'}\mathbin{::}{}\<[18]%
\>[18]{}(\Conid{Eq}\;\alpha,\Conid{Eq}\;\beta,\Conid{Monad}\;\tau)\Rightarrow {}\<[E]%
\\
\>[18]{}(\alpha\hsarrow{\rightarrow }{\mathpunct{.}}\beta\hsarrow{\rightarrow }{\mathpunct{.}}\tau\;\beta)\hsarrow{\rightarrow }{\mathpunct{.}}(\alpha\hsarrow{\rightarrow }{\mathpunct{.}}\beta\hsarrow{\rightarrow }{\mathpunct{.}}\tau\;\alpha)\hsarrow{\rightarrow }{\mathpunct{.}}{}\<[E]%
\\
\>[18]{}\Conid{StateTBX}\;\tau\;((\alpha,\beta),[\mskip1.5mu ((\alpha,\beta),\beta)\mskip1.5mu],[\mskip1.5mu ((\alpha,\beta),\alpha)\mskip1.5mu])\;\alpha\;\beta{}\<[E]%
\\
\>[3]{}\Varid{dynamicBX'}\;\Varid{f}\;\Varid{g}\mathrel{=}\Conid{StateTBX}\;{}\<[30]%
\>[30]{}(\Varid{gets}\;(\Varid{fst}\hsdot{\cdot }{.}\Varid{fst3}))\;\set{L}\;\bot \;{}\<[E]%
\\
\>[30]{}(\Varid{gets}\;(\Varid{snd}\hsdot{\cdot }{.}\Varid{fst3}))\;\set{R}\;\bot {}\<[E]%
\\
\>[3]{}\hsindent{2}{}\<[5]%
\>[5]{}\mathbf{where}\;{}\<[12]%
\>[12]{}\set{L}\;\Varid{a'}\mathrel{=}\mathbf{do}\;{}\<[26]%
\>[26]{}((\Varid{a},\Varid{b}),\Varid{fs},\Varid{bs})\leftarrow \Varid{get}{}\<[E]%
\\
\>[26]{}\mathbf{case}\;\Varid{lookup}\;(\Varid{a'},\Varid{b})\;\Varid{fs}\;\mathbf{of}{}\<[E]%
\\
\>[26]{}\hsindent{3}{}\<[29]%
\>[29]{}\Conid{Just}\;\Varid{b'}{}\<[38]%
\>[38]{}\hsarrow{\rightarrow }{\mathpunct{.}}\Varid{put}\;({}\<[48]%
\>[48]{}(\Varid{a'},\Varid{b'}),\Varid{fs},\Varid{bs}){}\<[E]%
\\
\>[26]{}\hsindent{3}{}\<[29]%
\>[29]{}\Conid{Nothing}{}\<[38]%
\>[38]{}\hsarrow{\rightarrow }{\mathpunct{.}}\mathbf{do}\;\{\mskip1.5mu \Varid{b'}\leftarrow \Varid{lift}\;(\Varid{f}\;\Varid{a'}\;\Varid{b});\Varid{put}\;({}\<[75]%
\>[75]{}(\Varid{a'},\Varid{b'}),((\Varid{a'},\Varid{b}),\Varid{b'})\mathbin{:}\Varid{fs},\Varid{bs})\mskip1.5mu\}{}\<[E]%
\\[\blanklineskip]%
\>[12]{}\set{R}\;\Varid{b'}\mathrel{=}\mathbf{do}\;{}\<[26]%
\>[26]{}((\Varid{a},\Varid{b}),\Varid{fs},\Varid{bs})\leftarrow \Varid{get}{}\<[E]%
\\
\>[26]{}\mathbf{case}\;\Varid{lookup}\;(\Varid{a},\Varid{b'})\;\Varid{bs}\;\mathbf{of}{}\<[E]%
\\
\>[26]{}\hsindent{3}{}\<[29]%
\>[29]{}\Conid{Just}\;\Varid{a'}{}\<[38]%
\>[38]{}\hsarrow{\rightarrow }{\mathpunct{.}}\Varid{put}\;({}\<[48]%
\>[48]{}(\Varid{a'},\Varid{b'}),\Varid{fs},\Varid{bs}){}\<[E]%
\\
\>[26]{}\hsindent{3}{}\<[29]%
\>[29]{}\Conid{Nothing}{}\<[38]%
\>[38]{}\hsarrow{\rightarrow }{\mathpunct{.}}\mathbf{do}\;\{\mskip1.5mu \Varid{a'}\leftarrow \Varid{lift}\;(\Varid{g}\;\Varid{a}\;\Varid{b'});\Varid{put}\;({}\<[75]%
\>[75]{}(\Varid{a'},\Varid{b'}),\Varid{fs},((\Varid{a},\Varid{b'}),\Varid{a'})\mathbin{:}\Varid{bs})\mskip1.5mu\}{}\<[E]%
\ColumnHook
\end{hscode}\resethooks

Some test cases
 \begin{hscode}\SaveRestoreHook
\column{B}{@{}>{\hspre}l<{\hspost}@{}}%
\column{3}{@{}>{\hspre}l<{\hspost}@{}}%
\column{10}{@{}>{\hspre}l<{\hspost}@{}}%
\column{18}{@{}>{\hspre}l<{\hspost}@{}}%
\column{23}{@{}>{\hspre}l<{\hspost}@{}}%
\column{E}{@{}>{\hspre}l<{\hspost}@{}}%
\>[3]{}\Varid{l0}\mathbin{::}{}\<[10]%
\>[10]{}\Conid{Monad}\;\Varid{m}\Rightarrow \Varid{b}\hsarrow{\rightarrow }{\mathpunct{.}}\Varid{c}\hsarrow{\rightarrow }{\mathpunct{.}}{}\<[E]%
\\
\>[10]{}\Conid{StateTBX}\;\Varid{m}\;((\Varid{a},\Varid{b}),(\Varid{c},\Varid{a}))\;(\Varid{a},\Varid{b})\;(\Varid{c},\Varid{a}){}\<[E]%
\\
\>[3]{}\Varid{l0}\;\Varid{b}\;\Varid{c}\mathrel{=}\Varid{fstBX}\;\Varid{b}\mathbin{`\Varid{compBX}`}\Varid{coBX}\;(\Varid{sndBX}\;\Varid{c}){}\<[E]%
\\[\blanklineskip]%
\>[3]{}\Varid{l}\mathrel{=}\Varid{l0}\;\text{\tt \char34 b\char34}\;\text{\tt \char34 c\char34}{}\<[E]%
\\[\blanklineskip]%
\>[3]{}\Varid{foo}\mathrel{=}\Varid{runl}\;\Varid{l}\;{}\<[18]%
\>[18]{}(\mathbf{do}\;{}\<[23]%
\>[23]{}\Varid{setr}\;\Varid{l}\;(\text{\tt \char34 x\char34},\text{\tt \char34 y\char34}){}\<[E]%
\\
\>[23]{}(\Varid{a},\Varid{b})\leftarrow \Varid{getl}\;\Varid{l}{}\<[E]%
\\
\>[23]{}\Varid{lift}\;(\Varid{print}\;\Varid{a}){}\<[E]%
\\
\>[23]{}\Varid{lift}\;(\Varid{print}\;\Varid{b}){}\<[E]%
\\
\>[23]{}\Varid{setl}\;\Varid{l}\;(\text{\tt \char34 z\char34},\text{\tt \char34 w\char34}){}\<[E]%
\\
\>[23]{}(\Varid{c},\Varid{d})\leftarrow \Varid{getr}\;\Varid{l}{}\<[E]%
\\
\>[23]{}\Varid{lift}\;(\Varid{print}\;\Varid{c}){}\<[E]%
\\
\>[23]{}\Varid{lift}\;(\Varid{print}\;\Varid{d}))\;(\text{\tt \char34 a\char34},\text{\tt \char34 b\char34}){}\<[E]%
\ColumnHook
\end{hscode}\resethooks
\begin{hscode}\SaveRestoreHook
\column{B}{@{}>{\hspre}l<{\hspost}@{}}%
\column{3}{@{}>{\hspre}l<{\hspost}@{}}%
\column{20}{@{}>{\hspre}l<{\hspost}@{}}%
\column{25}{@{}>{\hspre}l<{\hspost}@{}}%
\column{E}{@{}>{\hspre}l<{\hspost}@{}}%
\>[3]{}\Varid{l'}\mathbin{::}(\Conid{Read}\;\Varid{a},\Conid{Ord}\;\Varid{a})\Rightarrow \Conid{StateTBX}\;\Conid{IO}\;(\Varid{a},\Varid{a})\;\Varid{a}\;\Varid{a}{}\<[E]%
\\
\>[3]{}\Varid{l'}\mathrel{=}\Varid{interactiveBX}\;(\langle ){}\<[E]%
\\[\blanklineskip]%
\>[3]{}\Varid{bar}\mathrel{=}\Varid{runStateT}\;{}\<[20]%
\>[20]{}(\mathbf{do}\;{}\<[25]%
\>[25]{}\Varid{setr}\;\Varid{l'}\;\text{\tt \char34 abc\char34}{}\<[E]%
\\
\>[25]{}\Varid{a}\leftarrow \Varid{getl}\;\Varid{l'}{}\<[E]%
\\
\>[25]{}\Varid{lift}\;(\Varid{print}\;\Varid{a}){}\<[E]%
\\
\>[25]{}\Varid{setl}\;\Varid{l'}\;\text{\tt \char34 def\char34}{}\<[E]%
\\
\>[25]{}\Varid{b}\leftarrow \Varid{getr}\;\Varid{l'}{}\<[E]%
\\
\>[25]{}\Varid{lift}\;(\Varid{print}\;\Varid{b}))\;{}\<[E]%
\\
\>[20]{}(\text{\tt \char34 abc\char34},\text{\tt \char34 xyz\char34}){}\<[E]%
\ColumnHook
\end{hscode}\resethooks
\begin{hscode}\SaveRestoreHook
\column{B}{@{}>{\hspre}l<{\hspost}@{}}%
\column{3}{@{}>{\hspre}l<{\hspost}@{}}%
\column{9}{@{}>{\hspre}l<{\hspost}@{}}%
\column{13}{@{}>{\hspre}l<{\hspost}@{}}%
\column{23}{@{}>{\hspre}l<{\hspost}@{}}%
\column{28}{@{}>{\hspre}l<{\hspost}@{}}%
\column{E}{@{}>{\hspre}l<{\hspost}@{}}%
\>[3]{}\Varid{baz}\mathrel{=}\mathbf{let}\;bx\mathrel{=}(\Varid{divZeroBX}\;(\Varid{fail}\;\text{\tt \char34 divZero\char34})){}\<[E]%
\\
\>[3]{}\hsindent{6}{}\<[9]%
\>[9]{}\mathbf{in}\;{}\<[13]%
\>[13]{}\Varid{runl}\;bx\;{}\<[23]%
\>[23]{}(\mathbf{do}\;{}\<[28]%
\>[28]{}\Varid{setr}\;bx\;\mathrm{17.0}{}\<[E]%
\\
\>[28]{}\Varid{a}\leftarrow \Varid{getl}\;bx{}\<[E]%
\\
\>[28]{}\Varid{lift}\;(\Varid{print}\;\Varid{a}){}\<[E]%
\\
\>[28]{}\Varid{setl}\;bx\;\mathrm{42.0}{}\<[E]%
\\
\>[28]{}\Varid{a}\leftarrow \Varid{getr}\;bx{}\<[E]%
\\
\>[28]{}\Varid{lift}\;(\Varid{print}\;\Varid{a}){}\<[E]%
\\
\>[28]{}\Varid{setl}\;bx\;\mathrm{0.0}{}\<[E]%
\\
\>[28]{}\Varid{lift}\;(\Varid{print}\;\text{\tt \char34 foo\char34}))\;\mathrm{1.0}{}\<[E]%
\ColumnHook
\end{hscode}\resethooks
\begin{hscode}\SaveRestoreHook
\column{B}{@{}>{\hspre}l<{\hspost}@{}}%
\column{3}{@{}>{\hspre}l<{\hspost}@{}}%
\column{11}{@{}>{\hspre}l<{\hspost}@{}}%
\column{20}{@{}>{\hspre}l<{\hspost}@{}}%
\column{28}{@{}>{\hspre}l<{\hspost}@{}}%
\column{E}{@{}>{\hspre}l<{\hspost}@{}}%
\>[3]{}\Varid{xyzzy}\mathrel{=}\mathbf{let}\;bx\mathrel{=}\Varid{listBX}\;(\Varid{divZeroBX}\;(\Varid{fail}\;\text{\tt \char34 divZero\char34})){}\<[E]%
\\
\>[3]{}\hsindent{8}{}\<[11]%
\>[11]{}\mathbf{in}\;\Varid{runl}\;{}\<[20]%
\>[20]{}bx\;(\mathbf{do}\;{}\<[28]%
\>[28]{}\Varid{b}\leftarrow \Varid{getr}\;bx{}\<[E]%
\\
\>[28]{}\Varid{lift}\;(\Varid{print}\;\Varid{b}){}\<[E]%
\\
\>[28]{}\Varid{setr}\;bx\;[\mskip1.5mu \mathrm{5.0},\mathrm{6.0},\mathrm{7.0},\mathrm{8.0}\mskip1.5mu]{}\<[E]%
\\
\>[28]{}\Varid{a}\leftarrow \Varid{getl}\;bx{}\<[E]%
\\
\>[28]{}\Varid{lift}\;(\Varid{print}\;\Varid{a}))\;{}\<[E]%
\\
\>[20]{}[\mskip1.5mu \mathrm{1.0},\mathrm{2.0},\mathrm{3.0}\mskip1.5mu]{}\<[E]%
\ColumnHook
\end{hscode}\resethooks

\end{document}